\pdfoutput=1
\PassOptionsToPackage{names,dvipsnames}{xcolor}
\documentclass[acmsmall,screen,nonacm]{acmart}
\setcopyright{acmcopyright}
\acmPrice{15.00}
\acmDOI{10.1145/3495528}
\copyrightyear{2022}
\acmJournal{TOPLAS}
\acmVolume{1}
\acmNumber{1}
\acmMonth{1}
\acmYear{2022}

\usepackage[T1]{fontenc} %
\usepackage[utf8x]{inputenc}
\usepackage{anyfontsize} % kill warnings about font sizes 
\usepackage{amsmath}     % Mathematics FTW!
\usepackage{amssymb}     % needed for \mathbb and mathcal
\usepackage{amsbsy}      % for \boldsymbol
\usepackage{semantic}    % for PLT
\usepackage{stmaryrd}   % provides some PLy symbols, like Scott Brackets
\usepackage{mathpartir}  % for \mathpar, \inferrules
\usepackage{braket}      % for \set and \braket
\usepackage{mathtools}   % for \bigtimes
\usepackage{thmtools}    % to duplicate theorems in the appendix
\usepackage{hyperref}
\usepackage[capitalise]{cleveref}
\usepackage{xspace}
\usepackage{mathrsfs} % for mathscr font
\usepackage{proof}
\usepackage{rotating}
\usepackage{tikz-cd}
\usepackage{xcolor}
\usepackage{changepage}
\usepackage{centernot}
\usepackage{comment}
\usepackage{listings}
\usepackage{centernot}
\usepackage{multirow}
\usepackage{xparse}
\usepackage{pgothic} % for gothic punctuation (?)
\usepackage{tensor} % for left subscripts
\usepackage{bbold} % for blackboard numbers
\usepackage{subcaption}
\usepackage[all]{xy}

\colorlet{darkblue}{RoyalBlue}
\colorlet{darkgreen}{green!45!black}
\definecolor{darkred}{rgb}{0.55, 0.0, 0.0}
\definecolor{darkviolet}{rgb}{0.58, 0.0, 0.83}
\definecolor{lightblue}{rgb}{0.68, 0.85, 0.9}
\usepackage{lstcoq}
\makeatletter
\newcommand{\xMapsto}[2][]{\ext@arrow 0599{\Mapstofill@}{#1}{#2}}
\def\Mapstofill@{\arrowfill@{\Mapstochar\Relbar}\Relbar\Rightarrow}
\makeatother

\definecolor{propcolor}{HTML}{3F7D31}
\definecolor{mypurple}{HTML}{5B069D}
\definecolor{mistyrose}{rgb}{1.0, 0.89, 0.88}
\definecolor{pastelgray}{rgb}{0.81, 0.81, 0.77}

\newtheoremstyle{theoremthin}% name of the style to be used
{\smallskipamount}% measure of space to leave above the theorem. E.g.: 3pt
{\smallskipamount}% measure of space to leave below the theorem. E.g.: 3pt
{\itshape}% name of font to use in the body of the theorem
{0pt}% measure of space to indent
{\scshape}% name of head font
{.}% punctuation between head and body
{ }% space after theorem head; " " = normal interword space
{}

\theoremstyle{theoremthin}

\newtheorem{theorem}{Theorem}
\newtheorem{lemma}[theorem]{Lemma}
\newtheorem{proposition}[theorem]{Proposition}
\newtheorem{corollary}[theorem]{Corollary}
\newtheorem{definition}{Definition}
\newtheorem{property}{Property}

\theoremstyle{remark}

\makeatletter
\renewenvironment{proof}[1][\proofname]{\par
  \pushQED{\qed}%
  \normalfont
  \topsep0pt \partopsep0pt % no space before
  \trivlist
\item[\hskip\labelsep
  \itshape
  #1\@addpunct{.}]\ignorespaces
}{%
  \popQED\endtrivlist\@endpefalse
  \addvspace{2pt plus 2pt} % some space after
}
\makeatother

\newcommand{\?}{{\ensuremath{\operatorname{\boldsymbol{?}}}}} %unknown type
\newcommand{\Unk}[1]{\?_{#1}} %unknown type
\mathlig{[]}{\square}
\mathlig{.1}{._{1}}
\mathlig{.2}{._{2}}

\usepackage{scalerel}
\newcommand{\axiom}{\texttt{ax}}

\newcommand{\pcons}{\ensuremath{\mathcal{C}}\xspace}
\newcommand{\pconst}[1]{\ensuremath{\mathcal{C}_{/{#1}}}\xspace}
\newcommand{\pgrad}{\ensuremath{\mathcal{G}}\xspace}
\newcommand{\pnorm}{\ensuremath{\mathcal{N}}\xspace}

\newcommand{\psafe}{\ensuremath{\mathcal{S}}\xspace}

\newcommand{\sortOfPiName}{s_\Pi}
\newcommand{\castOfPiName}{c_\Pi}
\newcommand{\sortOfPi}[2]{\sortOfPiName(#1,#2)}
\newcommand{\castOfPi}[1]{\castOfPiName (#1)}
\newcommand{\stalk}[0]{\operatorname{germ}}
\newcommand{\stalkCIC}[2]{\stalk_{#1}\,#2}
\newcommand{\stalkCode}[2]{\code{\stalk}_{#1}\,#2}

\newcommand{\Coq}{\ensuremath{\mathrm{Coq}}\xspace}
\newcommand{\CIC}{\ensuremath{\mathsf{CIC}}\xspace}
\newcommand{\CICp}[1]{\ensuremath{\mathsf{CIC}\!+\!{#1}}\xspace}
\newcommand{\CICs}{\ensuremath{\mathsf{CIC}^{\uparrow}}\xspace}
\newcommand{\ExTT}{{\ensuremath{\mathsf{ExTT}}}\xspace}

\newcommand{\STLC}{\ensuremath{\mathsf{STLC}}\xspace}
\newcommand{\GTLC}{\ensuremath{\mathsf{GTLC}}\xspace}

\newcommand{\GCIC}{\ensuremath{\mathsf{GCIC}}\xspace}
\newcommand{\GCICP}{\ensuremath{\GCIC^{\pgrad}}\xspace}
\newcommand{\GCICT}{\ensuremath{\GCIC^{\pnorm}}\xspace}
\newcommand{\GCICs}{\ensuremath{\GCIC^{\uparrow}}\xspace}
\newcommand{\CCIC}{\ensuremath{\mathsf{CastCIC}}\xspace}
\newcommand{\CCICP}{\ensuremath{\CCIC^{\pgrad}}\xspace}
\newcommand{\CCICT}{\ensuremath{\CCIC^{\pnorm}}\xspace}
\newcommand{\CCICs}{\ensuremath{\CCIC^{\uparrow}}\xspace}
\newcommand{\CICIR}{\ensuremath{\mathsf{CIC^{IR}}}\xspace}
\newcommand{\CICIRQ}{\ensuremath{\mathsf{CIC^{IR}_{QIT}}}\xspace}
\newcommand{\bareModel}[0]{discrete model}
\newcommand{\Agda}[0]{Agda\xspace}
\newcommand{\Idris}[0]{Idris\xspace}
\newcommand{\IR}[0]{\mathrm{IR}}% target(s) of the CastCIC models

\newcommand{\UIP}[0]{UIP}

\newcommand{\ie}{\emph{i.e.,}\xspace}

\newcommand{\eg}{\emph{e.g.,}\xspace}
\newcommand{\cf}{\emph{cf.}\xspace}
\newcommand{\aka}{a.k.a.\xspace}
\newcommand{\resp}{resp.\xspace}

\renewcommand{\terms}{\operatorname{Term}}
\newcommand{\types}{\operatorname{Type}}
\renewcommand{\H}[0]{\operatorname{Head}}
\renewcommand{\C}[0]{\H}

\renewcommand{\P}{\operatorname{\Pi}}

\renewcommand{\l}{\operatorname{\lambda}}
 
\newcommand{\conv}{\equiv}
\newcommand{\cons}{\sim}

\newcommand{\acons}{\cons_{\alpha}}

\newcommand{\redCCIC}[0]{\leadsto}
\newcommand{\transRel}[1]{\operatorname{{#1}^\ast}}
\newcommand{\rtred}{\transRel{\redCCIC}}
\newcommand{\redIR}[0]{\leadsto_{\IR}}
\newcommand{\redIRn}{\leadsto^{+}_{\IR}}

\newcommand\optRedName{\rho}
\newcommand\optRed[1]{\optRedName(#1)}
\def\paraRed{\Rrightarrow}

\renewcommand{\comp}{\operatorname{\rightsquigarrow}}

\newcommand{\can}[1]{\operatorname{canonical}#1}
\newcommand{\neu}[1]{\operatorname{neutral}#1}

\newcommand{\pre}{\sqsubseteq}
\newcommand{\gpre}{\sqsupseteq}
\newcommand{\equiprecise}{\sqsupseteq\sqsubseteq}

\newcommand{\apre}{\pre^{\mathrm{G}}_{\alpha}}
\newcommand{\capre}{\pre_{\alpha}}
\newcommand{\cdpre}{\pre_{\redCCIC}}
\newcommand{\caequipre}{\equiprecise_{\alpha}}
\newcommand{\obsEquiv}{\approx^{obs}}
\newcommand{\obsApprox}{\preccurlyeq^{obs}}
\newcommand{\obsRef}{\sqsubseteq^{obs}}

\newcommand{\GG}{\mathbb{\Gamma}}
\newcommand{\DD}{\mathbb{\Delta}}
\newcommand{\fs}[1]{#1_{1}}
\newcommand{\sn}[1]{#1_{2}}

\newcommand{\letin}[3]{\text{let}\,#1=#2\,\text{in}\,#3}

\newcommand{\nat}{\mathbb{N}}
\newcommand{\bool}{\mathbb{B}}
\newcommand{\btrue}{\mathtt{true}}
\newcommand{\bfalse}{\mathtt{false}}
\newcommand{\unit}{\mathtt{unit}}
\newcommand{\ind}{\mathtt{ind}}

\newcommand{\ulev}[1]{\scalebox{0.7}{@\{#1\}}}
\newcommand{\pars}{\operatorname{\mathbf{Params}}}
\newcommand{\args}{\operatorname{\mathbf{Args}}}

\newcommand{\natzero}[0]{0}
\newcommand{\natsuc}[0]{\texttt{suc}}

\newcommand{\unitK}{\ensuremath{\mathtt{()}}}

\newcommand{\cas}{\operatorname{\mathtt{cast}}}
\newcommand{\ascdom}[3]{\langle {#3} {\ \Leftarrow \ } {#2} \rangle\,{#1}}%{\cast{#2}{#3}{#1}}
\newcommand{\ascdombg}[4]{\colorbox{#4}{$\left\langle {#3} {\ \Leftarrow \ } {#2} \right\rangle$}\,{#1}}
\newcommand{\cast}[3]{\ascdom{#3}{#1}{#2}}
\newcommand{\castbg}[4]{\ascdombg{#4}{#2}{#3}{#1}}

\newcommand{\ascName}{\mathrel{::}}
\newcommand{\asc}[2]{#1\ascName#2}
\newcommand{\ocast}[3]{#3::#1 \Rightarrow #2}

\newcommand{\err}{\operatorname{\mathtt{err}}}

\newcommand{\cicty}{\vdash_{\CIC}}

\newcommand{\caty}{\vdash_{\cas}}
\newcommand{\irty}[0]{\vdash_{\IR}}

\newcommand{\catyJ}[3]{\tcol{#1}\vdash_{\cas} \tcol{#2} : \tcol{#3}}

\newcommand{\inferty}{\operatorname{\triangleright}}
\newcommand{\pcheckty}[1]{\operatorname{\vcenter{\hbox{${\scriptscriptstyle\blacktriangleright}$}}_{#1}}}
\newcommand{\checkty}{\operatorname{\triangleleft}}

\newcommand{\sourcecolor}[1]{\mathcolor{ForestGreen}{#1}}
\newcommand{\scol}[1]{\sourcecolor{#1}}
\newcommand{\targetcolor}[1]{\mathcolor{RoyalBlue}{#1}}
\newcommand{\tcol}[1]{\targetcolor{#1}}
\newcommand{\staticcolor}[1]{\mathcolor{black}{#1}}
\newcommand{\ccol}[1]{\staticcolor{#1}}
\newcommand{\sinferelab}[6]{\targetcolor{#4} \vdash \sourcecolor{#2} \comp \targetcolor{#5} \inferty \targetcolor{#6}}
\newcommand{\spcheckelab}[7]{\targetcolor{#5} \vdash \sourcecolor{#3} \comp \targetcolor{#6} \pcheckty{#1} \targetcolor{#7}}
\newcommand{\scheckelab}[6]{\targetcolor{#4} \vdash \sourcecolor{#2} \checkty \targetcolor{#6} \comp \targetcolor{#5}}

\newcommand{\inferelab}[6]{\sinferelab{#1}{#2}{#3}{#4}{#5}{#6}}
\newcommand{\pcheckelab}[7]{\spcheckelab{#1}{#2}{#3}{#4}{#5}{#6}{#7}}
\newcommand{\checkelab}[6]{\scheckelab{#1}{#2}{#3}{#4}{#5}{#6}}

\newcommand{\id}[1]{\mathtt{id}_{#1}}

\newcommand{\orr}[1]{\mathbf{#1}}
\newcommand{\match}[4]{\operatorname{\ind}_{#1}(#2,#3,#4)}

\newcommand{\bB}[0]{\mathbb{B}}
\newcommand{\pidom}[0]{\text{d}}
\newcommand{\picod}[0]{\text{c}}
\newcommand{\eras}{\varepsilon}
\newcommand{\poset}{\ensuremath{\square^\leq}}

\renewcommand{\U}[0]{\square \!\!\!\! \square}%{\mathbb{U}}
\newcommand{\El}[0]{\mathrm{El}}

\newcommand{\ElRel}[0]{\El{}_\varepsilon}

\newcommand{\code}[1]{\widehat{#1}}
\newcommand{\natU}[0]{\code{\nat}}
\newcommand{\uU}[0]{\code{\U}}
\newcommand{\unkU}[1][]{\code{\?}_{#1}}%{?\!\! ?}
\newcommand{\errU}[1][]{\code{\err_{#1}}}%{?\!\! ?}
\newcommand{\PiU}[0]{\code{\Pi}}
\newcommand{\hd}[1]{\mathrm{head}\,#1}
\newcommand{\dash}[0]{\text{-}}

\colorlet{darkred}{red!80!gray}
\newcommand{\bareTy}[1]{\llbracket{}{\tcol{#1}}\rrbracket{}}%{\mathcolor{darkred}{\llbracket{}} #1 \mathcolor{darkred}{\rrbracket{}}}
\newcommand{\bareTm}[1]{[{\tcol{#1}}]}%{{\mathcolor{darkred}{[} #1 \mathcolor{darkred}{]}}}

\newcommand{\monTyAux}[2]{\{\!\mid\!{\tcol{#1}}\!\mid\!\}_{#2}}%{\mathcolor{blue}{\llbracket{}} #1 \mathcolor{blue}{\rrbracket{}_{#2}}}
\newcommand{\monTmAux}[2]{\{{\tcol{#1}}\}_{#2}}%{\mathcolor{blue}{[} #1 \mathcolor{blue}{]_{#2}}}

\newcommand{\monTy}[1]{\monTyAux{#1}{}}
\newcommand{\monTm}[1]{\monTmAux{#1}{}}
\newcommand{\monTyRel}[1]{\monTyAux{#1}{\varepsilon}}
\newcommand{\monTmRel}[1]{\monTmAux{#1}{\varepsilon}}

\newcommand{\binlogrelTy}[1]{\Lbag{}{#1}\Rbag{}}%{\mathcolor{darkgreen}{\llbracket{}} #1 \mathcolor{darkgreen}{\rrbracket{}}}
\newcommand{\binlogrelTm}[1]{\lbag{}{#1}\!\rbag{}}%{{\mathcolor{darkgreen}{[} #1 \mathcolor{darkgreen}{]}}}

\newcommand{\disc}[0]{\text{dis}}
\newcommand{\mon}[0]{\text{mon}}
\newcommand{\binlrel}[0]{\text{rel}}

\newcommand{\unkInj}[2]{[{#1};{#2}]}

\usepackage{float}
\floatstyle{boxed} 
\restylefloat{figure}

\makeatletter
\def\mathcolor#1#{\@mathcolor{#1}}
\def\@mathcolor#1#2#3{%
	\protect\leavevmode
	\begingroup\color#1{#2}#3%
\endgroup
}
\makeatother

\definecolor{grey}{rgb}{0.5, 0.5, 0.5}

\makeatletter
\def\slashedarrowfill@#1#2#3#4#5{%
  $\m@th\thickmuskip0mu\medmuskip\thickmuskip\thinmuskip\thickmuskip
  \relax#5#1\mkern-7mu%
  \cleaders\hbox{$#5\mkern-2mu#2\mkern-2mu$}\hfill
  \mathclap{#3}\mathclap{#2}%
  \cleaders\hbox{$#5\mkern-2mu#2\mkern-2mu$}\hfill
  \mkern-7mu#4$%
}
\def\rightslashedarrowfill@{%
  \slashedarrowfill@\relbar\relbar\mapstochar\rightarrow}
\newcommand\xslashedrightarrow[2][]{%
  \ext@arrow 0055{\rightslashedarrowfill@}{#1}{#2}}
\makeatother

\newcommand{\printf}[0]{\coqe{printf}}

\newcommand{\boxedrule}[1]{\flushleft \fbox{#1} \vspace{-1em}}

\newcommand{\eppair}[0]{ep-pair\xspace}
\newcommand{\eppairs}[0]{\eppair{}s\xspace}
\newcommand{\Pmon}[0]{\P^{\text{mon}}}
\newcommand{\pairSigma}[2]{({#1};{#2})}

\newcommand{\unkMon}[1][]{\ensuremath{\liftErr{\?}_{#1}}}

\newcommand{\liftErr}[1]{\ddot{#1}}%{\colorbox{blue!20}{$#1$}}
\newcommand{\liftNat}[0]{\liftErr{\nat}}
\newcommand{\liftBool}[0]{\liftErr{\bool}}
\newcommand{\liftSigma}[0]{{\liftErr{\Sigma}}}
\newcommand{\lifttop}[1]{\top_{#1}}
\newcommand{\liftbot}[1]{\bot_{#1}}
\newcommand{\liftZero}[0]{\liftErr{\mathbb{0}}}
\newcommand{\precision}[2]{\tensor[_{#1}]{\sqsubseteq}{_{#2}}}
\newcommand{\suc}[0]{\texttt{suc}}
\newcommand{\distr}[0]{\vartriangleleft}%{\xslashedrightarrow{}}
\newcommand{\distrcont}[0]{\vartriangleleft^\omega}%{\xslashedrightarrow{}}
\newcommand{\upcast}[0]{\uparrow}
\newcommand{\downcast}[0]{\downarrow}
\renewcommand{\id}[0]{\mathrm{id}}

\newcommand{\pred}[1]{\mathrm{pred}\,#1}
\newcommand{\Ue}[2]{#1 \sqsubseteq #2}

\newcommand{\tocont}{\to^\omega}%{\xrightarrow{\mathrm{cont}}}
\newcommand{\Picont}{\Pi^\omega}

\newif\ifappendix\appendixtrue

\newif\ifleftovers\leftoversfalse

\newcommand{\mparagraph}[1]{\emph{#1}.}

\makeatletter

\DeclareDocumentCommand \inferrule { s O{} m m o}{%
	\IfBooleanTF{#1}%
	{%
		\mpr@inferstar[#2]{#3}{#4}%
	}{%
		\mpr@inferrule[#2]{#3}{#4}%
	}%
	\IfValueT{#5}%
	{%
		\textsc{#5}%
		\my@name@inferrule{#5}%
	}%
}
\NewDocumentCommand \my@name@inferrule { m }{%
	\def\@currentlabelname{\textsc{#1}}%
}

\DeclareDocumentCommand \redrule {m m o}{%
	\IfValueTF{#3}
	{\textsc{#3}: \; #1 \redCCIC #2 \my@name@inferrule{#3}}
	{#1 \redCCIC #2}
}

\DeclareDocumentCommand \Longredrule {m m o}{%
	\IfValueTF{#3}
	{\textsc{#3 :} \hfill \vspace{-0.5em} \\ \redrule{#1}{#2}}
	{#1 \redCCIC #2}
}

\newcommand*{\ilabel}[1]{%
  \phantomsection% Correct hyper reference link
  \label{#1}% Print and store label
}

\makeatother

\newcommand{\myflushright}{\text{~} \and \hfill}

\newcommand{\Gvect}[2]{\mathtt{vec}~#1~#2}
\newcommand{\Gnil}[1]{\mathtt{nil}~#1}
\newcommand{\Gcons}[4]{\mathtt{cons}~#1~#2~#3~#4}
\newcommand{\GvectRec}[4]{\mathtt{vec\_rect}~#1~#2~#3~#4}
\newcommand{\GnilU}[1]{\mathtt{nil}_\?~#1}
\newcommand{\GnilULevel}[2]{\mathtt{nil}_\?\ulev{#2} ~#1}
\newcommand{\GconsUName}{\mathtt{cons}_\?}
\newcommand{\GconsU}[4]{\GconsUName~#1~#2~#3~#4}
\newcommand{\GconsULevel}[5]{\GconsUName \ulev{#5}~#1~#2~#3~#4}

\newcommand\divrelCtx[1]{#1 \Vdash}
\newcommand\divrelTy[3]{#1 \Vdash \tcol{#2} \sim #3}
\newcommand\divrelTm[4]{#1 \Vdash \tcol{#2} \sim #3 : #4}
\newcommand\divrelNe[4]{#1 \Vdash \tcol{#2} \sim_{\mathrm{ne}} #3 : #4}
\newcommand\divbind[3]{\tcol{#1} \sim #2 : #3}
\newcommand\whred{\twoheadrightarrow_{\mathrm{wh}}}
\newcommand\natDivrel{\nat_\varepsilon}
\newcommand\boolDivrel{\bool_\varepsilon}
\newcommand\unkDivrel{\?_\varepsilon}
\newcommand\errDivrel{\err_\varepsilon}
\newcommand\univDivrel{\square_\varepsilon}
\newcommand\piDivrel{\P_\varepsilon}
\newcommand\neDivrel{\texttt{ne}_\varepsilon}

\definecolor{shadecolor}{gray}{0.93}

\usepackage{booktabs} % For formal tables
\usepackage{enumitem}
\usepackage{pifont}% http://ctan.org/pkg/pifont

\usepackage{tcolorbox}

\newcommand{\cmark}{\ding{51}}
\newcommand{\xmark}{\ding{55}}

\newcommand{\parsub}[1]{[#1]}
\newcommand{\subs}[2]{[#1 / #2]}

\definecolor{defgreen}{rgb}{0,0.6,0}

\crefname{section}{\textsection\!}{\textsection\!}
\crefname{proposition}{Prop.}{Props.}



\begin{document}

\renewcommand{\footnotemark}{\mbox{}}

\title{Gradualizing the Calculus of Inductive Constructions
}
\author{Meven Lennon-Bertrand}
\affiliation{%
  \institution{Gallinette Project-Team, Inria}
  \city{Nantes}
  \country{France}
}
\author{Kenji Maillard}
\affiliation{%
  \institution{Gallinette Project-Team, Inria}
  \city{Nantes}
  \country{France}
}
\author{Nicolas Tabareau}
\affiliation{%
  \institution{Gallinette Project-Team, Inria}
  \city{Nantes}
  \country{France}
}
\author{\'Eric Tanter}
\affiliation{%
  \institution{PLEIAD Lab, Computer Science Department (DCC), University of Chile}
  \city{Santiago}
  \country{Chile}
}
\titlenote{This work is partially funded by ANID FONDECYT Regular
  Project 1190058, and Inria {\'E}quipe Associ{\'e}e GECO.
  \begin{tcolorbox}[width=\linewidth, sharp corners=all, colback=white!95!black]
  \textbf{To appear in ACM Transactions in Programming Languages and Systems, 2022, 
doi:10.1145/3495528}
\end{tcolorbox}}

\begin{abstract}
We investigate gradual variations on the Calculus of Inductive Construction (CIC) for swifter prototyping with imprecise types and terms. We observe, with a no-go theorem, a crucial tradeoff between graduality and the key properties of normalization and closure of universes under dependent product that CIC enjoys. Beyond this Fire Triangle of Graduality, we explore the gradualization of CIC with three different compromises, each relaxing one edge of the Fire Triangle. We develop a parametrized presentation of Gradual CIC (GCIC) that encompasses all three variations, and develop their metatheory. We first present a bidirectional elaboration of GCIC to a dependently-typed cast calculus, CastCIC, which elucidates the interrelation between typing, conversion, and the gradual guarantees. We use a syntactic model of CastCIC to inform the design of a safe, confluent reduction, and establish, when applicable, normalization. We study the static and dynamic gradual guarantees as well as the stronger notion of graduality with embedding-projection pairs formulated by New and Ahmed, using appropriate semantic model constructions. This work informs and paves the way towards the development of malleable proof assistants and dependently-typed programming languages.
\end{abstract}

\maketitle

\section{Introduction}
\label{sec:intro}

Gradual typing arose as an approach to selectively and soundly relax static type
checking by endowing programmers with imprecise
types~\cite{siekTaha:sfp2006,siekAl:snapl2015}. Optimistically well-typed
programs are safeguarded by runtime checks that detect violations of
statically-expressed assumptions. A gradual version of the simply-typed lambda
calculus (\STLC) enjoys such expressiveness that it can embed the untyped lambda
calculus. This means that gradually-typed languages tend to accommodate at least
two kinds of effects, non-termination and runtime errors.
The smoothness of the static-to-dynamic checking spectrum afforded by gradual languages
is usually captured by (static and dynamic) gradual
guarantees
which stipulate that typing and reduction are monotone with respect
to precision~\cite{siekAl:snapl2015}.

Originally formulated in terms of simple types, the extension of gradual typing
to a wide variety of typing disciplines has been an extremely active topic of
research, both in theory and in practice. As part of this quest towards more
sophisticated type disciplines, gradual typing was bound to meet with full-blown
dependent types. This encounter saw various premises in a variety of approaches
to integrate (some form of) dynamic checking with (some form of) dependent
types~\cite{dagandAl:jfp2018,knowlesFlanagan:toplas2010,lehmannTanter:popl2017,ouAl:tcs2004,tanterTabareau:dls2015,wadlerFindler:esop2009}.
Naturally, the highly-expressive setting of dependent types, in which terms and
types are not distinct and computation happens as part of typing, raises a lot
of subtle challenges for gradualization. In the most elaborate effort to date,
\citet{eremondiAl:icfp2019} present a gradual dependently-typed programming
language, GDTL, which can be seen as an effort to gradualize a two-phase
programming language such as \Idris~\cite{brady:jfp2013}. A key idea of GDTL is to adopt an
approximate form of computation at compile-time, called {\em approximate
normalization}, which ensures termination and totality of typing, while adopting
a standard gradual reduction semantics with errors and non-termination at runtime.
The metatheory of GDTL however still needs to be extended to account for inductive types.

This paper addresses the open challenge of gradualizing a full-blown dependent type theory, namely the Calculus of Inductive Constructions (hereafter, \CIC)~\cite{coquandHuet:ic1988,paulinMohring:appa2015}, identifying and addressing the corresponding metatheoretic challenges. In doing so, we build upon several threads of prior work in the type theory and gradual typing literature: syntactic models of type theories to justify extensions of \CIC~\cite{BoulierPT17}, in particular the exceptional type theory of~\citet{pedrotTabareau:esop2018}, an effective re-characterization of the dynamic gradual guarantee as {\em graduality} with embedding-projection pairs~\cite{newAhmed:icfp2018}, as well as the work on GDTL~\cite{eremondiAl:icfp2019}.

\paragraph{\textbf{Motivation.}} We believe that studying the gradualization of a full-blown dependent type theory like \CIC is in and of itself an important scientific endeavor, which is very likely to inform the gradual typing research community in its drive towards supporting ever more challenging typing disciplines. In this light, the aim of this paper is not to put forth a unique design or solution, but to explore the space of possibilities. Nor is this paper about a concrete implementation of gradual \CIC and an evaluation of its applicability; these are challenging perspectives of their own, which first require the theoretical landscape to be unveiled.

This being said, as~\citet{eremondiAl:icfp2019}, we can highlight a number of practical motivating scenarios for gradualizing \CIC, anticipating what could be achieved in a hypothetical gradual version of \Coq, for instance.

\begin{example}[Smoother development with indexed types]
\label{ex:indices}

\CIC, which underpins languages and proof assistants such as \Coq, Agda and \Idris, among others, is a very powerful system to program in, but at the same time extremely demanding. Mixing programs and their specifications is attractive but challenging.

Consider the classical example of length-indexed lists, of type \coqe{vec A n} as defined in \Coq:\footnote{We use the notation $\square_i$ for the predicative universe of types \texttt{Type$_i$}, and omit the universe level $i$ when not required.}

\begin{coq}
Inductive vec (A : Type) : nat -> Type :=
| nil  : vec A 0
| cons : A -> forall n : nat, vec A n -> vec A (S n).
\end{coq}

Indexing the inductive type by its length allows us to define a {\em total} \coqe{head} function, which can only be applied to non-empty lists:
\begin{center}
\coqe{head : forall A n, vec A (S n) -> A}
\end{center}

Developing functions over such structures can be tricky. For instance, what type should the \coqe{filter} function be given?

\begin{center}
\coqe{filter : forall A n (f : A -> bool), vec A n -> vec A \hole}%
\end{center}
The size of the resulting list depends on how many elements in the list actually match the given predicate \coqe{f}!
Dealing with this level of intricate specification can (and does) scare programmers away from mixing programs and specifications. The truth is that many libraries, such as MathComp~\cite{mahboubiTassi:mathcomp}, give up on mixing programs and specifications even for simple structures such as these, which are instead dealt with as ML-like lists with extrinsically-established properties. This
tells a lot about the current intricacies of dependently-typed programming.

Instead of avoiding the obstacle altogether, gradual dependent types provide a uniform and flexible mechanism to a tailored adoption of dependencies. For instance, one could give \coqe{filter} the following gradual type, which makes use of the {\em unknown term} $\?$ in an index position:

\begin{center}
\coqe{filter : forall A n (f : A -> bool), vec A n -> vec A ?}%
\end{center}

This imprecise type means that uses of \coqe{filter} will be optimistically accepted by the typechecker, although subject to associated checks during reduction. For instance:

\begin{center}
\coqe{head nat ? (filter nat 4 even [ 0 ; 1 ; 2 ; 3 ])}
\end{center}
typechecks, and is successfully convertible to $0$, while:

\begin{center}
\coqe{head nat ? (filter nat 2 even [ 1 ; 3 ])}
\end{center}
typechecks but fails upon reduction, when discovering that the assumption that the argument to head is non-empty is
in fact incorrect.
\end{example}

\begin{example}[Defining general recursive functions]
\label{ex:rec}

Another challenge of working in \CIC is to convince the type checker that recursive definitions are well founded. This can either require tight syntactic restrictions, or sophisticated arguments involving accessibility predicates. At any given stage of a development, one might not be in a position to follow any of these. In such cases, a workaround is to adopt the ``fuel pattern'', \ie~parametrizing a function with a clearly syntactically decreasing argument in order to please the typechecker, and to use an arbitrary initial fuel value. In practice, one sometimes requires a simpler way to unplug termination checking, and for that purpose, many proof assistants support external commands or parameters to deactivate termination checking.\footnote{such as \texttt{Unset Guard Checking} in \Coq, or \texttt{\{-\# TERMINATING \#-\}} in \Agda.}

Because the use of the unknown type allows the definition of fix-point combinators~\cite{siekTaha:sfp2006,eremondiAl:icfp2019}, one can use this added expressiveness to bypass termination checking locally.
This just means that the external facilities provided by specific proof assistant implementations now become internalized in the language.
\end{example}

\begin{example}[Large elimination, gradually]
\label{ex:elim}
One of the argued benefit of dynamically-typed languages, which is accommodated by gradual typing, is the ability to define functions that can return values of different types depending on their inputs, such as:

\begin{coq}
                  def foo(n)(m) { if (n > m) then m + 1 else m > 0 }
\end{coq}

In a gradually-typed language, one can give this function the type $\?$, or even $\nat -> \nat -> \?$ in order to enforce proper argument types, and remain flexible in the treatment of the returned value.
Of course, one knows very well that in a dependently-typed language, with large elimination, we can simply give \texttt{foo} the dependent type:

\begin{center}
\coqe{foo : forall (n m : nat), if (n > m) then nat else bool}
\end{center}

Lifting the term-level comparison \coqe{n > m} to the type level is extremely expressive, but hard to work with as well, both for the implementer of the function and its clients.

In a gradual dependently-typed setting, one can explore the whole spectrum of type-level precision for such a function, starting from the least precise to the most precise, for instance:

\begin{coq}
    foo : ?
    foo : nat -> nat -> ?
    foo : nat -> nat -> if ? then nat else ?
    foo : forall (n m : nat), if (n > m) then nat else ?
    foo : forall (n m : nat), if (n > m) then nat else bool
\end{coq}

At each stage from top to bottom, there is less flexibility (but more guarantees!) for both the implementer of \coqe{foo} and its clients. The gradual guarantee ensures that if the function is actually faithful to the most precise type
then giving it any of the less precise types above does not introduce any new failure~\cite{siekAl:snapl2015}.
\end{example}

\begin{example}[Gradually refining specifications]
\label{ex:specif}

Let us come back to the \coqe{filter} function from \cref{ex:indices}. Its fully-precise type requires appealing to a type-level function that counts the number of elements in the list that satisfy the predicate (notice the dependency to the input vector \coqe{v}):

\begin{center}
\coqe{filter : forall A n (f : A -> bool) (v : vec A n), vec A (count_if A n f v)}%
\end{center}

Anticipating the need for this function, a gradual specification could adopt the above signature for
\coqe{filter} but leave \coqe{count_if} unspecified:

\begin{coq}
Definition count_if A n (f : A -> bool) (v: vec A n) : nat := ?.
\end{coq}

This situation does not affect the behavior of the program compared to leaving the return type index unknown. More interestingly, one could immediately define the base case, which trivially specifies that there are no matching elements in an empty vector:

\begin{center}
\begin{coq}
Definition count_if A n (f : A -> bool) (v : vec A n) : nat :=
  match v with
  | nil _ _ => 0
  | cons _ _ _ => ?
  end.
\end{coq}
\end{center}

This slight increment in precision provides a little more static checking, for instance:
\begin{center}
\coqe{head nat ? (filter nat 4 even [])}
\end{center}
does not typecheck, instead of failing during reduction.

Again, the gradual guarantee ensures that such incremental refinements in precision towards the proper fully-precise version do not introduce spurious errors.
Note that this is in stark contrast with the use of axioms (which will be discussed in more depth in \cref{sec:tradeoffs}). Indeed, replacing correct code with an axiom can simply break typing! For instance, with the following definitions:

\begin{coq}
Axiom to_be_done : nat.
Definition count_if A n (f : A -> bool) (v: vec A n) : nat := to_be_done.
\end{coq}
\noindent the definition of \coqe{filter} does not typecheck anymore, as the axiom at the type-level is not convertible to any given value.
\end{example}

\paragraph{Note: Gradual programs or proofs?}

When adapting the ideas of gradual typing to a dependent type theory, one might
expect to deal with programs rather than proofs.
This observation is however misleading: from the point of view of the Curry-Howard correspondence, proofs and programs are intrinsically related, so that gradualizing the latter begs for a gradualization of the former. The examples above illustrate mixed programs and specifications, which naturally also appeal to proofs: dealing with indexed types typically requires exhibiting equality proofs to rewrite terms.
Moreover, there are settings in which one must consider computationally-relevant proofs, such as constructive algebra and analysis, homotopy type theory, etc. In such settings, using axioms to bypass unwanted proofs breaks reduction, and because typing requires reduction, the use of axioms can simply prevent typing, as illustrated in \cref{ex:specif}.

\paragraph{\textbf{Contribution.}}
This article reports on the following contributions:

\begin{itemize}[leftmargin=*]
\item We analyze, from a type theoretic point of view, the fundamental tradeoffs
  involved in gradualizing a dependent type theory such as \CIC
  (\cref{sec:tradeoffs}), and establish a no-go theorem,
  the Fire Triangle of Graduality, which does apply to \CIC. In essence, this result tells us that a gradual type theory\footnote{Note that we sometimes use ``dependent type theory'' in order to differentiate from the Gradual Type Theory of
  \citet{newAl:popl2019}, which is simply typed. But by default, in
  this article, the expression "type theory" is used to refer to a
  type theory with full dependent types, such as \CIC.}  cannot satisfy at the same time normalization, graduality, and conservativity with respect to \CIC. We explain each property and carefully analyze what it means in the type theoretic setting.

\item We present an approach to gradualizing \CIC (\cref{sec:gcic-overview}),
  parametrized by two knobs for controlling universe constraints on the
  dependent function space, resulting in
  three meaningful variants of Gradual \CIC (\GCIC), that reflect distinct resolutions of the
  Fire Triangle of Graduality. Each variant sacrifices one key property.
\item We give a bidirectional and mutually-recursive elaboration of
  \GCIC to a dependently-typed cast calculus \CCIC
  (\cref{sec:gcic-to-ccic}). This elaboration is based on a
  bidirectional presentation of \CIC, which has been recently studied
  in details by \citet{LennonBertrand2021}, and of which we give a comprehensive
  summary in \cref{sec:bidirectional-cic}.
Like \GCIC, \CCIC is parametrized, and encompasses three variants. We develop the metatheory of \GCIC, \CCIC and elaboration. In particular, we prove type safety for all variants, as well as the gradual guarantees and normalization, each for two of the three variants.
  \item To further develop the metatheory of \CCIC, we appeal to various
  models (\cref{sec:realizing-ccic}). First, to prove strong normalization of two \CCIC variants, we provide a syntactic  model of \CCIC with a translation to \CIC extended with induction-recursion~\cite{Martin-Lof1996,DybjerS03,GhaniMF15}. Second, to prove the stronger notion of graduality with embedding-projection pairs~\cite{newAhmed:icfp2018} for a normalizing variant, we provide a model of \CCIC that captures the notion of monotonicity with
    respect to precision. Finally, we present an extension of Scott's model
    based on $\omega$-complete partial orders~\cite{scott76} to prove
    graduality for the variant with divergence.
\item We describe how to handle indexed inductive types in \GCIC, either directly or via different encodings, under some constraints on indices (\cref{sec:giit}).
\end{itemize}

We then elucidate the current limitations of this work regarding three important features of \CIC---impredicativity, $\eta$-equality and propositional equality (\cref{sec:future-cic}).
We finally discuss related work (\cref{sec:related}) and conclude (\cref{sec:conclusion}). Some detailed proofs are omitted from the main text and can be found in appendix.

\section{Fundamental Tradeoffs in Gradual Dependent Type Theory}
\label{sec:tradeoffs}
Before exposing a specific approach to gradualizing \CIC, we present a general analysis of the main properties at stake and tensions that arise when gradualizing a dependent type theory.

We start by recalling two cornerstones of type theory, namely progress and normalization, and allude to the need to reconsider them carefully in a gradual setting (\cref{sec:norm-canon-endang}). We explain why the obvious approach based on axioms is unsatisfying (\cref{sec:axiom}), as well as why simply using a type theory with exceptions~\cite{pedrotTabareau:esop2018} is not enough either (\cref{sec:extt}). We then turn to the gradual approach, recalling its essential properties in the simply-typed setting (\cref{sec:grad-simple}), and revisiting them in the context of a dependent type theory (\cref{sec:graduality}). This finally leads us to establish a fundamental impossibility in the gradualization of \CIC, which means that at least one of the desired properties has to be sacrificed (\cref{sec:fire-triangle}).

\subsection{Safety and Normalization, Endangered}
\label{sec:norm-canon-endang}

As a well-behaved typed programming language, \CIC enjoys
(type) {\bf Safety} (\psafe), meaning that well-typed closed terms cannot get
stuck, \ie the normal forms of closed terms of a given type are
exactly the canonical forms of that type.
In \CIC, a closed canonical form is a term whose typing derivation ends
with an introduction rule, \ie a $\lambda$-abstraction for a function
type, and a constructor for an inductive type.
For instance, any closed term of type \coqe{bool} is convertible (and
reduces) to either \coqe{true} or \coqe{false}.
Note that an open term can reduce to an open canonical form called a {\em neutral term},
such as \coqe{not x}.

As a logically consistent type theory, \CIC enjoys (strong) {\bf Normalization}
(\pnorm), meaning that any term is convertible to its (unique) normal form.
 \pnorm together with \psafe imply {\em canonicity}: any closed term of a given
type \emph{must} reduce to a canonical form of that type.
When applied to the empty type \coqe{False}, canonicity ensures {\em logical consistency}:
because there is no canonical form for \coqe{False}, there is no
closed proof of \coqe{False}.
Note that \pnorm also has an important consequence in \CIC. Indeed, in
this system, conversion---which coarsely means syntactic equality
up-to reduction---is used in the type-checking algorithm.
\pnorm ensures that one can devise a sound and complete decision
procedure (\aka a reduction strategy) in order to decide conversion, and hence, typing.

In the gradual setting, the two cornerstones \psafe and \pnorm must be considered with care.
First, any closed term can be ascribed the unknown type \coqe{?}
first and then any other type: for instance, \coqe{0::?::bool} is a
well-typed closed term of type \coqe{bool}.%
\footnote{\label{ftn:ascription}We write $a :: A$ for a type ascription, which is syntactic sugar
  for $(\lambda x:A.x)\;a$~\cite{siekTaha:sfp2006}; in
  other systems, it can be taken as a primitive notion~\cite{garciaAl:popl2016}.}
However, such a term
cannot possibly reduce to either \coqe{true} or \coqe{false}, so some
concessions must be made with respect to safety---at least, the notion
of canonical forms must be extended.

Second, \pnorm is endangered. The quintessential example of non-termination in the untyped lambda calculus is the
term \mbox{$\Omega := \delta~\delta$} where \mbox{$\delta := (\lambda~
  x.~x~x)$}. In the simply-typed lambda calculus (hereafter \STLC), as in \CIC,
{\em self-applications} like $\delta~\delta$ and $x~x$ are ill-typed. However,
when introducing gradual types, one usually expects to accommodate such idioms,
and therefore in a standard gradually-typed calculus such as
\GTLC~\cite{siekTaha:sfp2006}, a variant of $\Omega$ that uses $(\lambda~ x :
\?.~x~x)$ for $\delta$ is well-typed and diverges, that is, admits no normal form.
The reason is that the argument type of $\delta$, the unknown type $\?$, is {\em consistent} with the type of $\delta$ itself, $\? \to \?$, and at runtime, nothing prevents reduction from going on forever.
Therefore, if one aims at ensuring \pnorm in a gradual setting, some care must be taken to restrict expressiveness.

\subsection{The Axiomatic Approach}
\label{sec:axiom}
Let us first address the elephant in the room: why would one want to gradualize \CIC instead of simply postulating
an axiom for any term (be it a program or a proof) that one does not feel like providing (yet)?

Indeed, we can augment \CIC with a general-purpose wildcard axiom $\axiom$:
\begin{center}
\coqe{Axiom _st : forall A, A.}
\end{center}

The resulting theory, called \CICp{\axiom}, has an obvious practical benefit: we can use
\coqe{(_st A)}, hereafter noted \coqe{_st_A}, as a wildcard whenever we are
asked to exhibit an inhabitant of some type \coqe{A} and we do not (yet) want to.
This is exactly what admitted definitions are in \Coq, for instance, and they do play an
important practical role at some stages of any \Coq development.

However, we cannot use the axiom \coqe{_st_A} in any meaningful way as a value {\em at the
  type level}.
For instance, going back to \cref{ex:indices}, one might be tempted to give to the \coqe{filter} function on
vectors the type
\coqe{forall A n (f : A -> bool), vec A n -> vec A _st_nat}%
, in order to avoid the complications related to specifying the
size of the vector produced by \coqe{filter}.
The problem is that the term:
\begin{center}
\coqe{head nat _st_nat (filter nat 4 even [ 0 ; 1 ; 2 ; 3 ])}
\end{center}
does not typecheck because the type of the filtering expression, \coqe{vec A _st_nat},
is not convertible to \coqe{vec A (S _st_nat)},
as required by the domain type of \coqe{head nat _st_nat}.

So the axiomatic approach is not useful for making dependently-typed programming
any more pleasing.
That is, using axioms goes in total opposition to the gradual typing
criteria~\cite{siekAl:snapl2015} when it comes to the smoothness of
the static-to-dynamic checking spectrum: given a well-typed term,
making it ``less precise'' by using axioms for some subterms actually
results in programs that do not typecheck or reduce anymore.

Because \CICp{\axiom} amounts to working in \CIC
with an initial context extended with $\axiom$, this theory
satisfies normalization (\pnorm) as much as \CIC, so conversion remains decidable. However,
\CICp{\axiom} lacks a satisfying notion of safety because 
there is an {\em infinite} number of open canonical normal forms (more
adequately called {\em stuck terms}) that inhabit any type \coqe{A}.
For instance, in \coqe{bool}, we not only have the normal forms \coqe{true},
\coqe{false}, and \coqe{_st_bool}, but an infinite number of terms stuck on
eliminations of $\axiom$, such as \coqe{match _st_A with ...}
or $\axiom_{\nat \rightarrow \bool}\;1$.

\subsection{The Exceptional Approach}
\label{sec:extt}

\citet{pedrotTabareau:esop2018} present the exceptional type theory \ExTT,
 demonstrating that it is possible to extend a
type theory with a wildcard term while enjoying a satisfying notion of safety,
which coincides with that of programming languages with exceptions.

\ExTT is essentially \CICp{\err}, that is, it
extends \CIC with an indexed error term \coqe{raise_A} that can inhabit any type \coqe{A}. But instead
of being treated as a computational black box like \coqe{_st_A}, \coqe{raise_A} is endowed with computational content
emulating exceptions in programming languages, which propagate instead of being stuck.
For instance, in \ExTT we have the following conversion:
\begin{center}
\coqe{match raise_bool return nat with | true -> O | false -> 1 end}
\quad$\equiv$\quad \coqe{raise_nat}
\end{center}

Notably, such exceptions are {\em call-by-name exceptions}, so one can only
discriminate exceptions on positive types (\ie~inductive types), not on negative
types (\ie~function types). In particular, in \ExTT, \coqe{raise_A_->_B } and
\coqe{fun _ : A => raise_B} are convertible, and the latter is considered to be
in normal form. So \coqe{raise_A} is a normal form of \coqe{A} only if \coqe{A} is a positive type.

\ExTT has a number of interesting properties: it is
normalizing (\pnorm) and safe
(\psafe), taking \coqe{raise_A} into account as usual in programming languages where exceptions are possible outcomes of computation: the normal forms of
closed terms of a positive type (\eg~\coqe{bool}) are either the
constructors of that type (\eg~\coqe{true} and \coqe{false}) or
\coqe{raise} at that type (\eg~\coqe{err_bool}).
As a consequence, \ExTT  does not satisfy full canonicity, but it
does satisfy a weaker form of it. In particular, \ExTT enjoys
(weak) logical consistency: any closed proof of \coqe{False} is convertible to
\coqe{raise_False}, which is discriminable at \coqe{False}.
It has been shown that we can still reason soundly in an
exceptional type theory, either using a parametricity
requirement~\cite{pedrotTabareau:esop2018}, or more flexibly, using
different universe hierarchies~\cite{pedrotAl:icfp2019}.

It is also important to highlight that this weak form of logical
consistency is the {\em most} one can expect in
a theory with effects. Indeed, \citet{pedrotTabareau:popl2020} have
shown that it is not possible to define a type theory with full
dependent elimination that has observable effects (from which
exceptions are a particular case) and at the same time validates
traditional canonicity.
Settling for less, as explained in \cref{sec:axiom} for the axiomatic
approach, leads to an infinite number of stuck terms, even in the
case of booleans, which is in opposition to the type safety criterion of gradual languages,
which only accounts for runtime type errors.

Unfortunately, while \ExTT solves the safety issue of the axiomatic approach, it still suffers from the same limitation as the axiomatic approach regarding type-level computation. Indeed, even though we can use \coqe{raise_A} to inhabit any type, we cannot use it in any meaningful way as a value at the type level. The term:
\begin{center}
\coqe{head nat err_nat (filter nat 4 even [ 0 ; 1 ; 2 ; 3 ])}
\end{center}
does not typecheck, because \coqe{vec A raise_nat} is still not convertible to \coqe{vec A (S raise_nat)}. The reason is that \coqe{raise_nat} behaves like an extra constructor to \coqe{nat}, so
\coqe{S raise_nat} is itself a normal form, and normal forms with different head constructors (\coqe{S} and \coqe{raise_nat}) are not convertible.

\subsection{The Gradual Approach: Simple Types}
\label{sec:grad-simple}

Before going on with our exploration of the fundamental challenges in gradual dependent type theory, we review some key concepts and expected properties in the context of simple types~\cite{siekAl:snapl2015,newAhmed:icfp2018,garciaAl:popl2016}.

\paragraph{Static semantics} Gradually-typed languages introduce
 the unknown type, written \?, which is used to indicate the lack of static typing information~\cite{siekTaha:sfp2006}.
 One can understand such an unknown type in terms of an {\em abstraction} of the
 set of possible types that it stands for~\cite{garciaAl:popl2016}. This
 interpretation provides a naive but natural understanding of the meaning of partially-specified types, for instance $\bool \to \?$ denotes the set of all function types with $\bool$ as domain. Given imprecise types, a gradual type system relaxes all type predicates and functions in order to optimistically account for occurrences of \?.  In a simple type system, the predicate on types is equality, whose relaxed counterpart is called {\em consistency}.\footnote{Not to be confused
with logical consistency!} For instance, given a function \coqe{f} of type $\bool \to \?$, the expression \coqe{(f true) + 1} is well-typed because \coqe{f} could {\em plausibly} return a number, given that its codomain is \?, which is consistent with \coqe{nat}.

Note that there are other ways to consider imprecise types, for instance by restricting the unknown type to denote base types (in which case $\?$ would not be consistent with any function type), or to only allow imprecision in certain parts of the syntax of types, such as effects~\cite{banadosAl:jfp2016}, security labels~\cite{fennellThiemann:csf2013,toroAl:toplas2018}, annotations~\cite{thiemannFennell:esop2014}, or only at the top-level~\cite{biermanAl:ecoop2010}. Here, we do not consider these specialized approaches, which have benefits and challenges of their own, and stick to the mainstream setting of gradual typing in which the unknown type is consistent with any type and can occur anywhere in the syntax of types.

\paragraph{Dynamic semantics}
Having optimistically relaxed typing based on consistency, a gradual language must detect inconsistencies at runtime if it is to satisfy safety (\psafe), which therefore has to be formulated in a way that encompasses runtime errors. For instance, if the function \coqe{f} above returns \coqe{false}, then an error must be raised to avoid reducing to \coqe{false + 1}---a closed stuck term, denoting a violation of safety. The traditional approach to do so is to avoid giving a direct reduction semantics to gradual programs, and instead, to elaborate them to an intermediate language with runtime casts, in which casts between inconsistent types raise errors~\cite{siekTaha:sfp2006}. Alternatively---and equivalently from a semantics point of view---one can define the reduction of gradual programs directly on gradual typing derivations augmented with evidence about consistency judgments, and report errors when transitivity of such judgments is unjustified~\cite{garciaAl:popl2016}. There are many ways to realize each of these approaches, which vary in terms of efficiency and eagerness of checking~\cite{hermanAl:hosc10,tobinFelleisen:popl2008,siekAl:popl10,siekAl:esop2009,toroTanter:scp2020,banados:arxiv2020}.

\paragraph{Conservativity} A first important property of a gradual language is that it is a {\em conservative extension} of a related static typing discipline: the gradual and static systems should coincide on static terms. This property is hereafter called {\bf
  Conservativity}~(\pcons), and parametrized with the considered static system. For instance, we write that \GTLC satisfies \pconst{\STLC}. Technically,
  \citet{siekTaha:sfp2006} prove that typing and reduction of \GTLC and \STLC coincide on their common set of terms (\ie~terms that are fully precise). An important aspect of \pcons is that the type formation rules and typing rules themselves are also preserved, modulo the presence of \? as a new type and the adequate lifting of predicates and functions~\cite{garciaAl:popl2016}. While this aspect is often left implicit, it ensures that the gradual type system does not behave in ad hoc ways on imprecise terms.

Note that, despite its many issues, \CICp{\axiom} (\cref{sec:axiom}) satisfies \pconst{\CIC}: all pure (\ie axiom-free) \CIC terms behave as they would in \CIC. More precisely,
two \CIC terms are convertible in \CICp{\axiom} iff they are convertible in \CIC.
Importantly, this does not mean that \CICp{\axiom} is a conservative extension of
  \CIC \emph{as a logic}---which it clearly is not!

\paragraph{Gradual guarantees}
The early accounts of gradual typing emphasized consistency as the central idea. However,
\citet{siekAl:snapl2015} observed that this characterization left too many possibilities for the impact of type information on program behavior, compared to what was originally intended~\cite{siekTaha:sfp2006}. Consequently, \citet{siekAl:snapl2015} brought forth {\em type precision} (denoted \coqe{precise}) as the key notion, from which consistency can be derived: two types \coqe{A} and \coqe{B} are consistent if and only if there exists \coqe{T} such that \coqe{T precise A} and \coqe{T precise B}. The unknown type \? is the most imprecise type of all, \ie~\coqe{T precise ?} for any \coqe{T}.
Precision is a preorder that can be used to capture the intended {\em
  monotonicity} of the static-to-dynamic spectrum afforded by gradual
typing.
The static and dynamic {\em gradual guarantees} specify that typing
and reduction should be {\em monotone with respect to precision}:
losing precision should not introduce new static or dynamic errors.
These properties require precision to be extended from types to terms. \citet{siekAl:snapl2015} present a natural extension that is purely syntactic: a term is more precise than another if they are syntactically equal except for their type annotations, which can be more precise in the former.

The {\em static gradual guarantee} (SGG) ensures that imprecision does not break typeability:
\begin{definition}[SGG]
If \coqe{t precise u} and \emph{\coqe{t : T}}, then \emph{\coqe{u :
    U}} for some \coqe{U} such that \coqe{T precise U}.
\end{definition}
The SGG captures the intuition that ``sprinkling \? over a term'' maintains its typeability. As such, the notion of precision \coqe{precise} used to formulate the SGG is inherently syntactic, over as-yet-untyped terms: typeability is the {\em consequence} of the SGG theorem.

The {\em dynamic gradual guarantee} (DGG) is the key result that
bridges the syntactic notion of precision to reduction: if \coqe{t precise u}
and \coqe{t} reduces to some value \coqe{v}, then
\coqe{u} reduces to some value \coqe{v'} such that \coqe{v precise v'}; and if \coqe{t} diverges, then so does \coqe{u}.
This property entails that \coqe{t precise u} means that \coqe{t} may error more than \coqe{u}, but otherwise they should behave the same.
Instead of the original formulation of the DGG by
\citet{siekAl:snapl2015}, \citet{newAhmed:icfp2018} appeal to the
semantic notion of {\em observational error-approximation} to capture
the relation between two terms that are contextually equivalent except
that the left-hand side term may fail more:\footnote{Observational
  error-approximation does not mention the case where $\mathcal{C}[t]$
  reduces to $\mathtt{true}$ or $\mathtt{false}$ but the quantification
  over all contexts ensures that, in that case, $\mathcal{C}[u]$ must
  reduce to the same value.}

\begin{definition}[Observational error-approximation]
\label{def:obsapprox}
  A term $\Gamma \vdash t : A$ observationally error-approximates
  a term $\Gamma \vdash u : A$, noted $ t
  \obsApprox u$, if for all boolean-valued observation contexts
  $\mathcal{C} : (\Gamma \vdash A) \Rightarrow (\vdash \bool{})$ closing over all
  free variables, either
  \begin{itemize}
  \item $\mathcal{C}[t]$ and $\mathcal{C}[u]$ both diverge. 
  \item Otherwise if $\mathcal{C}[u] \rtred \err_{\bool}$, then $\mathcal{C}[t] \rtred \err_{\bool}$.
  \end{itemize}
\end{definition}

Using this semantic notion, the DGG simply states that term 
precision implies observational error-approximation:

\begin{definition}[DGG]
If \coqe{t precise u} then \coqe{t obsApprox u}.
\end{definition}

While often implicit, it is important to highlight that the DGG is relative to both the notion of precision \coqe{precise} and the notion of observations \coqe{obsApprox}. Indeed, it is possible to study alternative notions of precisions beyond the natural definition stated by \citet{siekAl:snapl2015}. For instance, following the Abstracting Gradual Typing methodology~\cite{garciaAl:popl2016}, precision follows from the definition of gradual types as a concretization to sets of static types. This opens the door to justifying alternative precisions, \eg~by considering that the unknown type only stands for specific static types, such as base types. Additionally, variants of precision have been studied in more challenging typing disciplines where the natural definition seems incompatible with the DGG, see \eg~\cite{igarashiAl:icfp2017}. As we will soon see below, it can also be necessary in certain situations to consider another notion of observations.

\paragraph{Graduality} As we have seen, the DGG is relative to a notion of precision, but what should this relation be? To go beyond a syntactic axiomatic definition of precision, \citet{newAhmed:icfp2018} 
characterize the good dynamic behavior of precision: the runtime checking mechanism used to define a gradual language, such as casting, should only perform typechecking, and not otherwise affect behavior.
Specifically, they mandate that precision gives rise
to {\em embedding-projection pairs} (ep-pairs):
the cast induced by two types related by precision forms an adjunction,
which induces a retraction.
In particular, going to a less precise type and back is the identity:  
for any term \coqe{a} of type \coqe{A}, and given \coqe{A precise B}, then 
\coqe{a::B::A} should be observationally equivalent to \coqe{a} (recall from \cref{ftn:ascription} 
that \coqe{::} is a type ascription).
For instance, \coqe{1::?::nat} should be equivalent to \coqe{1}. 
Dually, when gaining precision, there is the potential for errors:
given a term \coqe{b} of type \coqe{B}, \coqe{b::A::B} may fail. 
By considering error as the least precise term, this can be stated as 
\coqe{b::A::B precise b}.
For instance, with the imprecise successor function \coqe{f := fun n:? => (S n)::?} of type \coqe{?->?}, we have \coqe{f::nat->bool::?->? precise f}, because the ascribed function will fail when applied.

Technically, the adjunction part states that if we have
\coqe{A precise B}, a term \coqe{a} of type \coqe{A},
and a term \coqe{b} of type \coqe{B}, then
\coqe{a precise b::A <=> a::B precise b}.
The retraction part further states that \coqe{t} is not only more precise
than \coqe{t::B::A} (which is given by the unit of the adjunction) but
is \emph{equi-precise} to it, noted
\coqe{t equiprecise t::B::A}.
Because the DGG dictates that precision implies observational error-approximation,
equi-precision implies observational equivalence,
and so losing and recovering precision
must produce a term that is observationally
equivalent to the original one.

A couple of additional observations need to be made here, as they will play a major role in the development of this article:
\begin{itemize}
\item These two approaches to characterizing gradual typing highlight the need to distinguish syntactic from semantic notions of precision. Indeed, with the usual {\em syntactic} precision from \citet{siekAl:snapl2015}, one cannot derive the ep-pair property, in particular the equi-precision stated above.
This is why \citet{newAhmed:icfp2018} introduce a {\em semantic}
precision, defined on well-typed terms. This semantic precision serves
as a proxy between the syntactic precision and the desired
observational error-approximation.

\item A type-based semantic precision cannot be used for the SGG. Indeed, this theorem (not addressed by \citet{newAhmed:icfp2018}) requires a syntactic notion of precision that {\em predates} typing: well-typedness of the less precise term is the {\em consequence} of the theorem. 
Therefore a full study of a gradual language that covers SGG, DGG, and embedding-projection pairs needs to consider both syntactic and semantic notions of precision.

\item The embedding-projection property does not {\em per se} imply the DGG: one could pick precision to be the universal relation, which trivially induces ep-pairs, but does not imply observational error-approximation. It appears that, in the simply-typed setting considered in prior work, the DGG implies the embedding-projection property. In fact, \citet{newAhmed:icfp2018} essentially advocate ep-pairs as an elegant and compositional proof technique to establish the DGG. But as we uncover later in this article, it turns out that in certain settings---and in particular dependent types---the embedding-projection property imposes {\em more} desirable constraints on the behavior of casts than the DGG alone. 
\end{itemize}
In this paper, we use the term {\bf Graduality} (\pgrad) for the DGG established with respect to a notion of precision that also induces embedding-projection pairs.

\subsection{The Gradual Approach: Dependent Types}
\label{sec:graduality}

Extending the gradual approach to a setting with full dependent types requires reconsidering several aspects.

\paragraph{Newcomers: the unknown term and the error type}
In the simply-typed setting, there is a clear stratification: $\?$ is at the type level, \coqe{raise} is at the term level. Likewise, {\em type precision}, with \? as greatest element, is separate from {\em term precision}, with \coqe{raise} as least element.
In the absence of a type/term syntactic distinction as in \CIC, this stratification is untenable:
\begin{itemize}
\item Because types permeate terms, \? is no longer only the unknown type, but it also acts as the ``unknown term''. In particular, this makes it possible to consider unknown indices for types, as in \cref{ex:indices}. More precisely, there is a family of unknown terms \coqe{?_A}, indexed by their type \coqe{A}. The traditional unknown type is just \coqe{?_Type}, the unknown of the universe \coqe{Type}.
\item Dually, because terms permeate types, we also have the ``error type'', \coqe{raise_Type}.
We have to deal with errors in types.
\item Precision must be unified as a single preorder, with $\?$ at the top and \coqe{raise} at the bottom. The most imprecise term of all is $\mathtt{\?_{\?_\square}}$ ($\?$ for short)---more exactly, there is one such term per type universe. At the bottom, \coqe{raise_A} is the most precise term of type \coqe{A}.
\end{itemize}

\paragraph{Revisiting safety}

The notion of closed canonical forms used to characterize legitimate normal forms
via safety (\psafe) needs to be extended not only with errors as in the
simply-typed setting, but also with unknown terms.
Indeed, as there is an unknown term \coqe{?_A} inhabiting any type
\coqe{A}, we have one new canonical form for each type \coqe{A}. In particular,
\coqe{?_bool} cannot possibly reduce to either \coqe{true} or \coqe{false}
or \coqe{raise_bool}, because doing so would collapse the precision order.
Therefore, \coqe{?_A} should propagate computationally,
like \coqe{raise_A} (\cref{sec:extt}).

The difference between errors and unknown terms is rather on their
static interpretation.
In essence, the unknown term \coqe{?_A} is a dual form of exceptions: it
propagates, but is optimistically comparable, \ie consistent
with, any other term of type \coqe{A}. Conversely, \coqe{raise_A}
should not be consistent with any term of type \coqe{A}.
Going back to the issues we identified with the axiomatic (\cref{sec:axiom}) and exceptional
(\cref{sec:extt}) approaches when dealing with type-level computation, the term:
\begin{center}
\coqe{head nat ?_nat (filter nat 4 even [ 0 ; 1 ; 2 ; 3 ])}
\end{center}
now typechecks: \coqe{vec A ?_nat} can be deemed consistent with
\coqe{vec A (S ?_nat)}, because \coqe{S ?_nat} is consistent with \coqe{?_nat}.
This newly-brought flexibility is the key to support the different scenarios from the introduction.
So let us now turn to the question of how to integrate consistency in
a dependently-typed setting.

\paragraph{Relaxing conversion}
In the simply-typed setting, consistency is a relaxing of syntactic type equality to account for imprecision. In a dependent type theory, there is a more powerful notion than syntactic equality to compare types, namely {\em conversion} (\cref{sec:norm-canon-endang}): if \coqe{t:T} and \coqe{T==U}, then \coqe{t:U}.
For instance, a term of type \coqe{T} can be used as a function as soon as \coqe{T}
is {\em convertible} to the type \coqe{forall (a:A),B} for some types \coqe{A} and \coqe{B}.
The proper notion to relax in the gradual dependently-typed setting is therefore conversion, not syntactic equality.

\citet{garciaAl:popl2016} give a general framework for gradual typing that explains how to relax any static type predicate to account for imprecision: for a binary type predicate \coqe{P}, its consistent lifting \coqe{Q(A,B)} holds iff there exist static types \coqe{A'} and \coqe{B'} in the denotation (\emph{concretization} in abstract interpretation parlance) of \coqe{A} and \coqe{B}, respectively, such that \coqe{P(A',B')}. As observed by \citet{castagnaAl:popl2019}, when applied to equality, this defines consistency as a unification problem. Therefore, the consistent lifting of conversion ought to be that two terms \coqe{t} and \coqe{u} are consistently convertible iff they denote some static terms \coqe{t'} and \coqe{u'} such that \coqe{t' == u'}. This property is essentially higher-order unification, which is undecidable.

It is therefore necessary to adopt some approximation of consistent conversion (hereafter called consistency for short) in order to be able to implement a gradual dependent type theory.
And there lies a great challenge: because of the absence of stratification between typing and reduction, the static gradual guarantee (SGG) already demands monotonicity for conversion, a demand very close to that of the DGG.\footnote{In a dependently-typed programming language with separate typing and execution phases, this demand of the SGG is called the {\em normalization gradual guarantee} by \citet{eremondiAl:icfp2019}.}

\paragraph{Dealing with neutrals}
Prior work on gradual typing usually only considers reduction on closed terms in order to establish results about the dynamics, such as the DGG.
But in dependent type theory,
conversion must operate on {\em open} terms, yielding {\em neutral} terms such as \coqe{1::X::nat} where \coqe{X} is a type variable, or \coqe{x+1} where \coqe{x} is of type \coqe{nat} or \coqe{?_Type}.
Such neutral terms cannot reduce further, and can occur in both terms and types. Depending on the upcoming substitutions, neutrals can fail or not. For instance, in \coqe{1::X::nat}, if \coqe{?_Type} is substituted for \coqe{X}, the term reduces to \coqe{1}, but fails if \coqe{bool} is substituted instead.

Importantly, less precise variants of neutrals can reduce {\em more}. For instance, both 
\coqe{1::?_Type::nat} and \coqe{?_nat+1} are less precise than the neutrals above, but do evaluate further (typically, to \coqe{1} and to \coqe{?_nat}, respectively). This interaction between neutrals, reduction, and precision spices up the goal of establishing DGG and \pgrad. In particular, this re-enforces the need to consider semantic precision, because a syntactic precision is likely not to be stable by reduction: \coqe{1::X::nat precise 1::?::nat} is obvious syntactically, but \coqe{1::X::nat precise 1} is not.

\paragraph{DGG vs Graduality} In a dependently-typed setting, it is possible to satisfy the DGG while not satisfying the embedding-projection pairs requirement of \pgrad.
To see why, consider a system in which any term of type \coqe{A} that is not fully-precise immediately reduces to \coqe{?_A}. This system would satisfy \pcons, \psafe, \pnorm, and \ldots the DGG. Recall that the DGG only requires reduction to be monotone with respect to precision, so using the most imprecise term \coqe{?_A} as a universal redux is surely valid. This collapse of the DGG is impossible in the simply-typed setting because there is no unknown term: it is only possible when \coqe{?_A} exists {\em as a term}. It is therefore possible to satisfy the DGG while being useless when {\em computing} with imprecise terms.
Conversely, the degenerate system breaks the embedding-projection requirement of graduality stated by \citet{newAhmed:icfp2018}.
For instance, \coqe{1::?_Type::nat} would be convertible to \coqe{?_nat}, which is
{\em not} observationally equivalent to \coqe{1}.
Therefore, the embedding-projection requirement of graduality goes beyond the DGG in a way that is critical in a dependent type theory, where it captures both the smoothness of the static-to-dynamic checking spectrum, and the proper computational content of valid uses
of imprecision.

\paragraph{Observational refinement}

Let us come back to the notion of observational
error-approximation used in the simply-typed setting to state the
DGG.
\citet{newAhmed:icfp2018} justify this notion because in
``gradual typing we are not particularly interested in
when one program diverges more than another, but rather when it
produces more type errors.''
This point of view is adequate in the simply-typed setting because
the addition of casts may only produce more type errors; in particular, 
adding casts can never lead to divergence when the original term does not diverge
itself.
Therefore, in that setting, the definition of error-approximation
includes equi-divergence.
The situation in the dependent setting is however more
complicated, if the theory admits divergence. 
There exist non-gradual dependently-typed programming languages that admit divergence (\eg~Dependent Haskell~\cite{eisenberg2016dependent}, \Idris~\cite{brady:jfp2013}); we will also present one such theory in this article. 

In a gradual dependent type theory that admits divergence,
a {\em diverging term} is more precise than the
{\em unknown term} $\?$. Because the unknown term in itself does not diverge, this
breaks the left-to-right implication of
equi-divergence. Note that this argument does not rely on any specific definition of precision, 
just on the fact that the unknown term is the most imprecise term (at its type).
Additionally, an error at a {\em diverging type} $X$ may be ascribed to $\?_{[]}$
then back to $X$. Evaluating this
roundtrip requires evaluating $X$ itself, which makes the less
precise term diverge. This breaks the right-to-left implication of
equi-divergence.

To summarize,
the way to understand these counterexamples is that 
in a dependent and non-terminating setting, 
the motto of graduality ought to be adjusted: more precise programs produce
more type errors {\em or diverge more}. This leads to the following definition
of \emph{observational refinement}.

\begin{definition}[Observational refinement]
\label{def:obsref}
  A term $\Gamma \vdash t : A$ observationally refines a term $\Gamma \vdash u : A$, noted $ t
  \obsRef u$  if for all boolean-valued observation context
  $\mathcal{C} : (\Gamma \vdash A) \Rightarrow (\vdash \bool{})$ closing over all
  free variables, if $\mathcal{C}[u] \rtred \err_{\bool}$ or diverges, then either $\mathcal{C}[t] \rtred \err_{\bool}$ or $\mathcal{C}[t]$ diverges.
\end{definition}

In this definition, errors and divergence are collapsed. Thus, in a gradual dependent theory that admits divergence, equi-refinement does not imply observational equivalence, because one term might diverge while the other reduces to an error. Of course, if the gradual dependent theory is strongly normalizing, then both notions $\obsApprox$ (\cref{def:obsapprox}) and  $\obsRef$ (\cref{def:obsref})  coincide.

\subsection{The Fire Triangle of Graduality}
\label{sec:fire-triangle}

To sum up, we have seen four important properties that can be expected from a gradual type theory:
safety (\psafe), conservativity with respect to a theory $X$ (\pconst{X}), graduality (\pgrad),
and normalization (\pnorm). Any type theory ought to satisfy at least \psafe. Unfortunately, we now show that mixing the three other properties \pcons, \pgrad and \pnorm is impossible for \STLC, as well as for \CIC.

\paragraph{Preliminary: regular reduction.}
To derive this general impossibility result, by relying only on the properties and without committing to a specific language or theory, we need to assume that the reduction system used to decide conversion is regular, in that it only looks at the weak head normal form of subterms for reduction rules, and does not magically shortcut reduction, for instance based on the specific syntax of inner terms. As an example, $\beta$-reduction is not allowed to look into the body of the lambda term to decide how to proceed.

This property is satisfied in all actual systems we know of, but formally stating it in full generality, in particular without devoting
to a particular syntax, is beyond the scope of this paper. Fortunately, in the following, we need only rely on a much weaker hypothesis,
which is a slight strengthening of the retraction hypothesis of \pgrad.
Recall that retraction says that when \coqe{A precise B}, any term \coqe{t} of type \coqe{A} is equi-precise to \coqe{t::B::A}.
We additionally require that for any context \coqe{C}, if \coqe{C[t]}
reduces at least $k$ steps, then \coqe{C[t::B::A]} also reduces at
least $k$ steps.
Intuitively, this means that the reduction of \coqe{C[t::B::A]}, while free to decide when to get rid of the embedding-to-\coqe{B}-projection-to-\coqe{A}, cannot use it to avoid reducing \coqe{t}. This property is true in all gradual languages, where type information at runtime is used only as a monitor.

\paragraph{Gradualizing \STLC.}
Let us first consider the case of \STLC.
We show that $\Omega$ is {\em necessarily} a well-typed diverging term in any
  gradualization of \STLC that satisfies the other properties.

\begin{theorem}[Fire Triangle of Graduality for \STLC]\label{thm:triangle-STLC}

Suppose a gradual type theory that satisfies properties
\pconst{\STLC} and \pgrad.
Then \pnorm cannot hold.
\end{theorem}
\begin{proof}
  We pose $\Omega := \delta~(\asc{\delta}{\?})$ with
  $\delta := \lambda~ x : \?.~(\asc{x}{\? \to
    \?})~x$ and show that it must necessarily be a well-typed diverging
  term.
  Because the unknown type \? is consistent with any type (\cref{sec:grad-simple}) and $\? \to \?$ is a valid type (by \pconst{\STLC}), the self-applications in $\Omega$ are well-typed, $\delta$ has type $\? \to \?$, and $\Omega$ has type $\?$.
  Now, we remark that $\Omega = C[\delta]$ with $C[\cdot] = [\cdot]~(\asc{\delta}{\?})$.

  We show by induction on $k$ that $\Omega$ reduces at least $k$
  steps, the initial case being trivial.
  Suppose that $\Omega$ reduces at least $k$ steps.
  By maximality of $\?$ with respect to precision, we have that
  $\? \to \? \pre \?$, so we can apply the strengthening of \pgrad
  applied to $\delta$, which tells us that
  $C[\asc{\asc{\delta}{\?}}{\?\to\?}]$ reduces at least $k$ steps
  because $C[\delta]$ reduces at least $k$ steps.
  But by $\beta$-reduction, we have that $\Omega$ reduces in one step to
  $C[\asc{\asc{\delta}{\?}}{\?\to\?}]$. So $\Omega$ reduces at least
  $k+1$ steps.

  This means that $\Omega$ diverges, which is a violation of \pnorm.
\end{proof}

This result could be extended to all terms of the untyped lambda calculus, not only $\Omega$, in order to obtain the embedding theorem of \GTLC~\cite{siekAl:snapl2015}. Therefore, the embedding theorem is not an independent property, but rather a consequence of \pcons and \pgrad---that is why we have not included it as such in our overview of the gradual approach (\cref{sec:grad-simple}).

\paragraph{Gradualizing \CIC.}
We can now prove the same impossibility theorem for \CIC, by reducing
it to the case of \STLC. Therefore this theorem can be proven
for type theories others than \CIC, as soon as they faithfully embed \STLC.

\begin{theorem}[Fire Triangle of Graduality for \CIC]\label{thm:triangle}

  A gradual dependent type theory cannot simultaneously satisfy properties
  \pconst{\CIC}, \pgrad and \pnorm.
\end{theorem}
\begin{proof}
  We show that a gradual dependent type theory satisfying \pconst{\CIC} and \pgrad
  must contain a diverging term, thus contravening \pnorm.
  The typing rules of \CIC contain the typing rules of \STLC,
  using only one universe $[]_0$,
  where the function type is interpreted using the dependent
  product and the notions of reduction coincide, so \CIC embeds
  \STLC; a well-known result on PTS~\cite{barendregt:jfp1991}.
  This means that \pconst{\CIC} implies \pconst{\STLC}.
  Additionally, \pgrad can be specialized to the simply-typed fragment of the theory,
  by setting the unknown type $\?$ to be $\?_{[]_0}$. Therefore, we can apply
  \cref{thm:triangle-STLC} and we get a well-typed term that diverges, finishing the proof.
\end{proof}

\paragraph{The Fire Triangle in practice}
In non-dependent settings, all gradual languages where \? is universal admit non-termination and therefore compromise \pnorm. \citet{garciaTanter:wgt2020} discuss the possibility to gradualize \STLC without admitting non-termination, for instance by considering that \? is not universal and denotes only base types (in such a system, \coqe{? -> ? nprecise ?}, so the argument with $\Omega$ is invalid).
Without sacrificing the universal unknown type, one could design a variant of \GTLC that uses some mechanism to detect divergence, such as termination contracts~\cite{nguyenAl:pldi2019}. This would yield a language that certainly satisfies \pnorm, but it would break \pgrad. Indeed, because the contract system is necessarily over-approximating in order to be sound (and actually imply \pnorm), there are effectively-terminating programs with imprecise variants that yield termination contract errors.

To date, the only related work that considers the gradualization of full dependent types with \? as both a term and a type, is the work on GDTL \cite{eremondiAl:icfp2019}. GDTL is a programming language
with a clear separation between the typing and execution phases, like \Idris~\cite{brady:jfp2013}.
GDTL adopts a different strategy in each phase: for typing, it uses Approximate Normalization (AN), which always produces \coqe{?_A} as a result of going through imprecision and back. This means that conversion
is both total and decidable (satisfies \pnorm), but it breaks \pgrad for the same reason as the degenerate system we discussed in \cref{sec:graduality} (notice that the example uses a gain of precision from the unknown type to \coqe{nat}, so the example behaves just the same with AN).
In such a phased setting, the lack of computational content of AN is not critical, because it only means that typing becomes overly optimistic. To execute programs, GDTL relies on standard \GTLC-like reduction semantics, which is computationally precise, but does not satisfy \pnorm.

\section{\GCIC: Overall Approach, Main Challenges and Results}
\label{sec:gcic-overview}

Given the Fire Triangle of Graduality (\cref{thm:triangle}), we know that gradualizing \CIC implies making some compromise.
Instead of focusing on one possible compromise, this work develops
three novel solutions, each compromising one specific property (\pnorm, \pgrad, or \pconst{\CIC}), and does so in a common parametrized framework, \GCIC.

This section gives an informal, non-technical overview of our approach to gradualizing \CIC, highlighting the main challenges and results. As such, it serves as a gentle roadmap to the following sections, which are rather dense and technical.

\subsection{\GCIC: 3-in-1}
\label{sec:gcic:-3-1}

To explore the spectrum of possibilities enabled by the Fire Triangle of Graduality, we develop a general approach to gradualizing \CIC, and use it to define three theories, corresponding to different resolutions of the triangular tension between normalization (\pnorm), graduality (\pgrad) and conservativity with respect to \CIC (\pconst{\CIC}).

The crux of our approach is to recognize that, while there is not much to vary within \STLC itself to address the tension of the Fire Triangle of Graduality, there are several variants of \CIC
that can be considered by changing the hierarchy of universes and its
impact on typing---after all, \CIC is but a particular Pure Type System
(PTS)~\cite{barendregt:jfp1991}.

In particular, we consider a parametrized version of a gradual \CIC, called \GCIC, with two parameters (\cref{fig:ccic-typing}):
\begin{itemize}
\item The first
parameter characterizes how the universe level of a $\Pi$ type is determined
in typing rules: either as taking the {\em maximum} of the levels of the involved 
types, as in standard \CIC, or as the {\em successor} of that maximum. The latter option yields a variant
of \CIC that we call \CICs (read ``\CIC-shift''). \CICs is a subset of \CIC, with a stricter constraint on universe levels. In particular \CICs loses the closure of universes under dependent product that CIC enjoys. As a consequence, some well-typed \CIC terms are not well-typed in \CICs.\footnote{A minimal example of a well-typed \CIC term that is ill typed in \CICs is \texttt{narrow :} $\nat -> \square$, where \texttt{narrow n} is the type of functions that accept \texttt{n} arguments. Such dependent arities violate the universe constraint of \CICs.}
\item The second parameter is the dynamic counterpart of the first parameter: its role is to enforce that universe levels are coherent through type casts during the reduction of casts. 
Note that we only allow this reduction parameter to be loose (\ie~using maximum) if the typing parameter is also loose. Indeed, letting the typing parameter be strict (\ie~using successor) while the reduction parameter is loose breaks subject reduction, and hence \psafe.
\end{itemize}

Based on these parameters, this work develops the following three variants of \GCIC,
whose properties are summarized in \cref{tab:gcic} with pointers to the respective theorems---because \GCIC is one common parametrized framework, we are able to establish most properties for all variants at once:
\begin{enumerate}
  \item {\bf \GCICP: a theory that satisfies both \pconst{\CIC} and \pgrad, but sacrifices \pnorm.}
  This theory is a rather direct application of the
principles discussed in~\cref{sec:tradeoffs} by extending \CIC
with errors and unknown terms, and changing conversion with
consistency. This results in a theory that is not normalizing.

  \item {\bf \GCICs: a theory that satisfies both \pnorm and \pgrad, and supports \pcons with respect to \CICs.} This theory uses the universe hierarchy at the \emph{typing level} to detect
the potential non-termination induced by the use of consistency
instead of conversion. This theory simultaneously satisfies \pgrad, \pnorm and
\pconst{\CICs}.

  \item {\bf \GCICT: a theory that satisfies both \pconst{\CIC} and \pnorm, but does not fully validate \pgrad.}
  This theory uses the
universe hierarchy at the \emph{computational level} to detect
potential divergence. Such runtime check failures invalidate the DGG for some terms,
and hence \pgrad, as well as the SGG.
\end{enumerate}

\begin{table}[t]
\begin{tabular}{|l|l|l|l|l|l|l|}
\hline
 & \psafe & \pnorm  & \pconst{X} & \pgrad & SGG & DGG \\
\hline
\GCICP \rule{0pt}{3ex}
  & \cmark \footnotesize{(Th.~\labelcref{thm:ccic-psafe})}
  & \xmark
  & \CIC \footnotesize{(Th.~\labelcref{thm:conservativity})}
  & \cmark \footnotesize{(Th. ~\labelcref{thm:GCICP-graduality})}
  & \cmark \footnotesize{(Th.~\labelcref{thm:static-graduality})}
  & \cmark \footnotesize{(Th.~\labelcref{thm:dgg})}\\
\GCICs
  & \cmark \footnotesize{(idem)}
  & \cmark \footnotesize{(Th.~\labelcref{thm:ccic-pnorm} \& \labelcref{thm:discrete-model})}
  & \CICs  \footnotesize{(idem)}
  & \cmark  \footnotesize{(Th.~\labelcref{thm:graduality-gcics})}
  & \cmark \footnotesize{(idem)}
  & \cmark  \footnotesize{(Th.~\labelcref{thm:dgg})}\\
\GCICT
  & \cmark \footnotesize{(idem)}
  & \cmark \footnotesize{(idem)}
  & $\mathsf{CIC}$\phantom{$^{\uparrow}$}    \footnotesize{(idem)}
  & \xmark  
  & \xmark
  & \xmark\\
\hline
\end{tabular}\\
\begin{minipage}{\textwidth}
\footnotesize{
\begin{center}
\psafe: safety, \pnorm: normalization, \pconst{X}: conservativity wrt theory $X$,
\pgrad: graduality (DGG + ep-pairs),
\newline SGG: static gradual guarantee, DGG: dynamic gradual guarantee
\end{center}
}
\end{minipage}
\caption{\GCIC variants and their properties}
\label{tab:gcic}
\end{table}

\paragraph{Practical implications of \GCIC variants.} Regarding the examples from \cref{sec:intro}, all three variants of \GCIC support the exploration of the type-level precision spectrum for the functions described in \cref{ex:indices,ex:elim,ex:specif}.  In particular, we can define \coqe{filter} by giving it the imprecise type
\coqe{forall A n (f : A -> bool), vec A n -> vec A ?_nat} in order to bypass the difficulty of precisely characterizing the size of the output vector. Any invalid optimistic assumption is detected during reduction and reported as an error.

Unsurprisingly, the semantic differences between the three \GCIC variants crisply manifest in the treatment of potential non-termination (\cref{ex:rec}), more specifically, {\em self application}.
Let us come back to the term $\Omega$
used in the proof of~\cref{thm:triangle}.
In all three variants, this term is well-typed. In \GCICP, it
reduces forever, as it would in the untyped lambda calculus. In that
sense, \GCICP can embed the untyped lambda calculus just as
GTLC~\cite{siekAl:snapl2015}. In \GCICT, this term fails at runtime
because of the strict universe check in the reduction of casts, which
breaks graduality because $\coqe{?_Typei -> ? _Typei precise ?_Typei}$
tells us that the upcast-downcast coming from an \eppair should not fail.
A description of the reductions in \GCICP and in \GCICT is given in
full details in \cref{sec:back-to-omega}.
In \GCICs, \coqe{Omega} fails in the same way as in \GCICT, but this
does not break graduality because of the shifted universe level on
$\Pi$ types.
A consequence of this stricter typing rule is that in \GCICs, \coqe{?_Typei -> ?_Typei precise ?_Typej} for any $j > i$, but \coqe{?_Typei -> ?_Typei nprecise ?_Typei}.
Therefore, the casts performed in $\Omega$ do not come from an \eppair
anymore and can legitimately fail.

Another scenario where the differences in semantics manifest is functions with {\em dependent arities}.
For instance, the well-known C function \printf{} can be embedded in a well-typed fashion in
\CIC: it takes as first argument a format string and computes from it both the type and {\em number}
of later arguments.
This function brings out the limitation of \GCICs: since the format string can
specify an arbitrary number of arguments, we need as many $\to$, and
\printf{} cannot typecheck in a theory where universes are not closed under
function spaces.
In \GCICT{}, \printf{} typechecks but the same problem will appear dynamically
when casting \printf{} to $\?$ and back to its original type: the result will be
a function that works only on format strings specifying no more arguments than
the universe level at which it has been typechecked.
Note that this constitutes an example of violation of graduality for
\GCICT{}, even of the dynamic gradual guarantee.
Finally, in \GCICP{} the function can be gradualized as much as one wants, without surprises.

\paragraph{Which variant to pick?} As explained in the introduction, the aim of this paper is to shed light on the design space of gradual dependent type theories, not to advocate for one specific design.
We believe the appropriate choice depends on the specific goals of the language designer, or perhaps more pertinently, on the specific goals of a given project, at a specific point in time.

The key characteristics of each variant are:
\begin{itemize}
\item \GCICP favors flexibility over decidability of type-checking. While this might appear heretical in the context of proof assistants, this choice has been embraced by practical languages such as Dependent Haskell~\cite{eisenberg2016dependent}, a dependently-typed Haskell where both divergence and runtime errors can happen at the type level. The pragmatic argument is simplicity: by letting programmers be responsible, there is no need for termination checking techniques and other restrictions.

\item \GCICs is theoretically pleasing as it enjoys both normalization and graduality. In practice, though, the fact that it is not conservative wrt full \CIC means that one would not be able to simply import existing libraries as soon as they fall outside of the \CICs subset.
In \GCICs, the introduction of \? should be done with an appropriate understanding of universe levels. This might not be a problem for advanced programmers, but would surely be harder to grasp for beginners.

\item \GCICT is normalizing and able to import existing libraries without restrictions, at the expense of some surprises on the graduality front. Programmers would have to be willing to accept that they cannot just sprinkle \? as they see fit without further consideration, as any dangerous usage of imprecision will be flagged during conversion.
\end{itemize}

In the same way that systems like \Coq, \Agda or \Idris support
different ways to customize their semantics (such as allowing
\texttt{Type-in-Type}, switching off termination checking, using the \texttt{partial}/\texttt{total}
compiler flags)---and of course, many programming languages implementations supporting some sort of customization, GHC being a salient representative---one can imagine a flexible realization of \GCIC that give users the control over the two parameters we identify in this work, and therefore have access to all three \GCIC variants.
Considering the inherent tension captured by the Fire Triangle of Graduality, such a pragmatic approach might be the most judicious choice, making it possible to gather experience and empirical evidence about the pros and cons of each in a variety of concrete scenarios.

\subsection{Typing, Cast Insertion, and Conversion}

As explained in \cref{sec:grad-simple},
in a gradual language, whenever we reclaim precision, we might be wrong and need to fail in order to preserve safety (\psafe).
In a simply-typed setting, the standard approach is to define typing on the
gradual source language, and then to translate terms via a type-directed cast insertion
to a target cast calculus, \ie~a language with explicit runtime type
checks, needed for a well-behaved reduction~\cite{siekTaha:sfp2006} . For instance, in a call-by-value language, the upcast (loss of precision) $\ascdom{10}{\nat}{\?}$ is considered a (tagged) value, and the downcast (gain of precision) $\ascdom{v}{\?}{\nat}$ reduces successfully if $v$ is such a tagged natural number, or to an error otherwise.

We follow a similar approach for \GCIC, which is
elaborated in a type-directed manner to a second calculus,
named \CCIC (\cref{sec:cast-calculus}).
The interplay between typing and cast insertion is however more subtle in the
context of a dependent type theory. Because typing needs computation, and
reduction is only meaningful in the target language, \CCIC is used
 {\em as part of the typed elaboration} in order to compare types (\cref{sec:elaboration}).
This means that \GCIC has no typing on its own, independent of its
elaboration to the cast calculus.%
\footnote{This is similar to what happens in practice in proof assistants such as \Coq \cite[Core language]{Coq:manual}, where terms input by the user in the Gallina language are first elaborated in order to add implicit arguments, coercions, etc. The computation steps required by conversion are
performed on the elaborated terms, never on the raw input syntax.}

In order to satisfy conservativity with respect to \CIC (\pconst{\CIC}), ascriptions in \GCIC are required to satisfy consistency: for instance, \coqe{true::?::nat} is well-typed by consistency (twice), but \coqe{true::nat} is ill typed. Such ascriptions in \CCIC are realized by casts. For instance
$\asc{\asc{0}{\?}}{\bB}$
in \GCIC elaborates (modulo sugar and reduction) to $\ascdom{\ascdom{0}{\nat}{\?_{[]}}}{\?_{[]}}{\bool}$ in \CCIC.
A major difference between ascriptions in \GCIC and casts in \CCIC is
that casts are not required to satisfy consistency: a cast between any
two types is well-typed, although of course it might produce an
error.

Finally, standard presentations of \CIC use a standalone conversion rule, as usual in declarative presentations of type systems. To gradualize \CIC, we have to move to a more algorithmic presentation in order to forbid transitivity, otherwise all terms would be well-typed by way of a transitive step through \coqe{?}. But \pconst{\CIC} demands that only terms with explicitly-ascribed imprecision enjoy its flexibility.
This observation is standard in the gradual typing literature~\cite{siekTaha:sfp2006,siekTaha:ecoop2007,garciaAl:popl2016}. As in prior work on gradual dependent types~\cite{eremondiAl:icfp2019}, we adopt a bidirectional presentation of typing for \CIC (\cref{sec:bidirectional-cic}), which allows us to avoid accidental transitivity and directly derive a deterministic typing algorithm for \GCIC.

\subsection{Realizing a Dependent Cast Calculus: \CCIC}

To inform the design and justify the reduction rules provided for
\CCIC, we build a
syntactic model of \CCIC by translation to \CIC augmented with
induction-recursion~\cite{Martin-Lof1996,DybjerS03,GhaniMF15} (\cref{sec:bare-model}). From a type
theory point of view, what makes \CCIC peculiar is first of all the
possibility of having {\em errors} (both ``pessimistic'' as \coqe{raise} and
``optimistic'' as \coqe{?}), and the necessity to do {\em intensional type
  analysis} in order to resolve casts. For the former, we build upon the work of
\citet{pedrotTabareau:esop2018} on the exceptional type theory \ExTT.
For the latter, we reuse the technique of \citet{BoulierPT17} to account for
\coqe{typerec}, an elimination principle for the universe \coqe{Type}, which
requires induction-recursion to be implemented.%

We call the syntactic model of \CCIC the {\em discrete model}, in contrast with
a semantic model motivated in the next subsection. The discrete model of \CCIC
captures the intuition that the unknown type is
inhabited by ``hiding'' the underlying type of the injected term. In other words,
\coqe{?_Typei} behaves as a
dependent sum \coqe{Sigma A:Typei. A}.
 Projecting out of the unknown type is realized
through type analysis (\coqe{typerec}), and may fail (with an error in the \ExTT sense).
Note that here, we provide a particular interpretation of the unknown
term in the universe, which is legitimized by an observation
made by \citet{pedrotTabareau:esop2018}: \ExTT does not constrain in
any way the definition of exceptions in the universe.
The syntactic model of \CCIC allows us to establish that the reduction semantics enjoys strong normalization (\pnorm), for the two variants \CCICT and \CCICs. Together with
safety (\psafe), this gives us weak logical consistency for \CCICT and \CCICs.

\subsection{Precisions and Properties}
\label{sec:precision-graduality}

As explained earlier (\cref{sec:graduality}), we need two different notions of
precision to deal with SGG and \pgrad.
At the source level (\GCIC), we introduce a notion of {\em syntactic precision} that captures the
intuition of a more imprecise term as ``the same term with subterms and/or annotated types replaced by \?'', and is defined without any assumption of typing.
% , and hence conversion, because the latter is not defined in \GCIC proper.
In \CCIC, we define a notion of {\em structural precision}, which is mostly syntactic except that, in order to account for cast insertion during elaboration, it tolerates precision-preserving casts (for instance, $\ascdom{t}{A}{A}$ is related to $t$ by structural precision).
Armed with these two notions of precision, we prove
% We prove that \GCIC satisfies static graduality.\km{forward ref?}
% However, because \GCIC does not have a type system, we rather prove
{\em elaboration graduality} (\cref{thm:static-graduality}), which is
the equivalent of SGG in our setting: if a term $t$ of \GCIC elaborates to a
term $t'$ of \CCIC, then a term $u$ less syntactically precise than $t$ in \GCIC elaborates to
a term $u'$ less structurally precise than $t'$ in \CCIC.

Because DGG is about the behavior of terms, it is technically stated and
proven for \CCIC. We show in \cref{sec:gcic-theorems2} that DGG can be proven
for \CCIC (in its variants \CCICP and \CCICs) on the structural
precision.
However, as explained in \cref{sec:grad-simple}, we cannot expect to prove \pgrad for these \CCIC variants with respect to structural precision directly.
In order to overcome this problem, we build an alternative model of
\CCIC 
called the {\em monotone model} (\cref{sec:poset-model-dtt,sec:realizing-unknown-type,sec:monotone-universe,sec:monotone-model}).
This model endows types with the structure of an ordered set, or poset. In
the monotone model, we can reason about the semantic notion of {\em propositional
  precision} and prove that it gives rise to embedding-projection
  pairs~\cite{newAhmed:icfp2018}, thereby establishing
\pgrad for \CCICs (\cref{thm:graduality-gcics}).
The monotone model only works for a normalizing gradual type theory,
thus we then establish \pgrad for \CCICP using a variant of the
monotone model based on Scott's
model~\cite{scott76} of the untyped $\lambda$-calculus using $\omega$-complete
partial orders (\cref{sec:grad-non-term}).

\section{Preliminaries: Bidirectional \CIC}
\label{sec:bidirectional-cic}

We develop \GCIC on top of a bidirectional version of \CIC, whose
presentation was folklore among type
theory specialists~\cite{McBride2019}, and that has recently been
studied in details by \citet{LennonBertrand2021}.
As explained before, this bidirectional
presentation is mainly useful to avoid multiple uses of a standalone conversion rule
during typing, which becomes crucial to preserve \pconst{\CIC} in a gradual setting
where conversion is replaced by consistency, which is not transitive.
We give here a comprehensive summary of the bidirectional version of
\CIC that will help the reader follow the presentation of \GCIC in \cref{sec:gcic}.

\paragraph{Syntax}
Our syntax for \CIC terms, featuring a predicative universe hierarchy $[]_i$, is the following (in Backus-Naur form):
\begin{align}
  \label{fig:syntax-cic}
\terms_{\CIC} \ni t ::= x \mid \square{}_i \mid t~t \mid \l x : t . t \mid \P x : t . t \mid I\ulev{i}(\orr{t}) \mid c\ulev{i}(\orr{t},\orr{t}) \mid \match{I}{t}{z.t}{f.\orr{y}.\orr{t}} \tag{Syntax of \CIC}
\end{align}
We reserve letters $x,y,z$ to denote variables. Other lower-case and upper-case Roman letters are used to represent terms, with the latter
used to emphasize that the considered terms should be thought of as types (although the difference does not occur at a syntactic level in
this presentation). Finally Greek capital letters are for contexts
(lists of declarations of the form $x : T$).
We also use bold letters $\mathbf{X}$ to denote sequences of objects
$X_1,\ldots,X_n$
and $t\subs{\mathbf{a}}{\mathbf{y}}$ for the simultaneous substitution of $\orr{a}$ for
$\orr{y}$.
We present generic inductive types $I$ with constructors $c$, although we restrict to well-formed (and in particular, strictly positive) ones
to preserve normalization, following~\cite{DBLP:conf/icalp/Gimenez98}.
 At this point we consider only inductive types without indices; we consider indexed inductive types in \cref{sec:giit}.
Inductive types are formally annotated with a universe level $\ulev{i}$, controlling the level of its parameters:
for instance $\operatorname{List}\ulev{i}(A)$ expects $A$ to be a type in $[]_i$. This level is omitted when inessential.
\label{params}
An inductive type at level $i$ with parameters $\orr{a}$ is noted $I\ulev{i}(\orr{a})$, and we use $\pars(I,i)$ to denote the types of those parameters.
The well-formedness condition on inductives in particular enforces that the $k$-th parameter
$\pars_k(I,i)$ only contains $k-1$ variables, corresponding to the previous $k-1$ parameters.
Thus if $\orr{a}$ is a list of terms of the same length as $\pars(I,i)$ we denote as
$\pars(I,i)\parsub{\orr{a}}$ the list where in parameter type $\pars_k(I,i)$, the $k-1$ first
elements of $\orr{a}$ have been substituted for the $k-1$ free variables.
Similarly $c_k^I\ulev{i}(\orr{a},\orr{b})$ denotes the $k$-th constructor of the
inductive $I$, taking parameters $\orr{a}$ and arguments $\orr{b}$. Again, the type
of parameters is denoted $\pars(I,i)$, and the type of the arguments $\args(I,i,c_k)$.
Similarly as for parameters, we also use $\args(I,i,c_k)\parsub{\orr{a},\orr{b}}$ for
the list where in the $m$-th argument type
$\orr{a}$ have been substituted for parameter variables, and
the first $m-1$ elements of $\orr{b}$ for argument variables.

The inductive eliminator $\match{I}{s}{z.P}{f.\orr{y.t}}$ corresponds to a
fixpoint immediately followed by a match. In \Coq, one would write it
\begin{coq}
  fix $f$ $s$ := match $s$ as $z$ return $P$ with | $c_1~\orr{y}$ => $t_1$ ... | $c_n~\orr{y}$ => $t_n$ end
\end{coq}
In particular, the return predicate $P$ has access to an extra bound variable $z$ for the scrutinee, and similarly
the branches $t_k$ are given access to variables $f$ and $\orr{y}$,
corresponding respectively to the recursive function and the
arguments of the corresponding constructor.
Describing the exact guard condition to ensure termination is
outside the scope of this presentation, again see~\cite{DBLP:conf/icalp/Gimenez98}.
We implicitly assume in the rest of this paper that every fixpoint is guarded.

\paragraph{Bidirectional Typing}
\begin{figure}
	\boxedrule{$\vdash \Gamma$}
	\begin{mathpar}
		\inferrule{ }{\vdash \cdot}[Empty] \ilabel{infrule:cic-axiom} \and
		\inferrule{\vdash \Gamma \\ \Gamma \vdash T
                  \pcheckty{\square{}} \square{}_{i}}{\vdash \Gamma, x
                  : T}[Concat] \ilabel{infrule:cic-concat} \\
	\end{mathpar}
	\boxedrule{$\Gamma \vdash t \inferty T$} \vspace{1em}
	\begin{mathpar}
		\inferrule{ }{\Gamma \vdash \square{}_{i} \inferty \square{}_{i + 1}}[Univ] \ilabel{infrule:cic-univ} \and
		\inferrule{(x : T) \in \Gamma}{\Gamma \vdash x
                  \inferty T}[Var] \ilabel{infrule:cic-var} \and
		\inferrule{\Gamma \vdash A \pcheckty{\square{}} \square{}_{j} \\ \Gamma, x :
      A \vdash B \pcheckty{\square{}} \square{}_{i}}{\Gamma \vdash  \P x : A . B
      \inferty \square{}_{\max(i,j)}}[Prod] \ilabel{infrule:cic-prod} \\
		\inferrule{\Gamma \vdash A \pcheckty{\square{}} \square{}_{i} \\ \Gamma, x : A \vdash t \inferty B}{\Gamma \vdash \l x : A . t \inferty  \P x : A . B}[Abs] \ilabel{infrule:cic-abs} \and
		\inferrule{\Gamma \vdash t \pcheckty{\P}  \P x : A . B \\ \Gamma \vdash u \checkty A}{\Gamma \vdash t~u \inferty B\subs{x}{u}}[App] \ilabel{infrule:cic-app}
    \\

		\inferrule
    {\Gamma \vdash a_k \checkty \pars_k(I,i)\parsub{\orr{a}}}
    {\Gamma \vdash I\ulev{i}(\orr{a}) \inferty \square{}_i}[Ind] \ilabel{infrule:cic-ind} \and
		\inferrule
    {\Gamma \vdash a_k \checkty \pars_k(I,i)\parsub{\orr{a}} \\
			\Gamma \vdash b_m \checkty \args_m(I,i,c)\parsub{\orr{a}, \orr{b}}}
    {\Gamma \vdash c^I\ulev{i}(\orr{a},\orr{b})
      \inferty I\ulev{i}(\orr{a})}[Cons] \ilabel{infrule:cic-cons}
    \\
		\inferrule{
			\Gamma \vdash s \pcheckty{I} I\ulev{i}(\orr{a}) \\
			\Gamma, z : I(\orr{a}) \vdash P \pcheckty{\square{}} \square{}_j \\\\
			\Gamma, f : (\P z : I\ulev{i}(\orr{a}). P), \orr{y} : \args(I,i,c_k)\parsub{\orr{a},\orr{y}} \vdash t_k \checkty P\subs{c^I_k\ulev{i}(\orr{a},\orr{y})}{z}
		}{\Gamma \vdash \match{I}{s}{z.P}{f.\orr{y.t}} \inferty P\subs{s}{z}}[Fix]
    \ilabel{infrule:cic-fix}
  \end{mathpar}

              \vspace{1em}
	\boxedrule{$\Gamma \vdash t \checkty T$}
	\begin{mathpar}
		\inferrule{\Gamma \vdash t \inferty T' \\ T' \conv
                  T}{\Gamma \vdash t \checkty T}[Check] \ilabel{infrule:cic-check}
	\end{mathpar}
	\boxedrule{$\Gamma \vdash t \pcheckty{\bullet} T$} \vspace{1em}
	\begin{mathpar}
		\inferrule{\Gamma \vdash t \inferty T \\ T \rtred \P x : A. B}{\Gamma \vdash t \pcheckty{\P} \P x : A . B}[Prod-Inf] \ilabel{infrule:cic-prod-inf} \and
		\inferrule{\Gamma \vdash t \inferty T \\ T \rtred I\ulev{i}(\orr{a})}{\Gamma \vdash t \pcheckty{I} I\ulev{i}(\orr{a})}[Ind-Inf] \ilabel{infrule:cic-ind-inf} \and
		\inferrule{\Gamma \vdash t \inferty T \\ T \rtred \square{}_{i}}{\Gamma \vdash t \pcheckty{\square{}} \square{}_{i}}[Univ-Inf] \ilabel{infrule:cic-univ-inf}
	\end{mathpar}

	\boxedrule{$t \redCCIC u$} \text{(congruence rules omitted)}

  \begin{mathpar}
		(\l x: A . t)~u \redCCIC t \subs{u}{x} \and
    \match{I}{c_k(\mathbf{a}, \mathbf{b})}{z.P}{f.\orr{y.t}}
		\redCCIC t_k \subs{\l x: I(\mathbf{a}).
      \match{I}{x}{z.P}{f.\orr{y.t}}
      }{f} \subs{\mathbf{b}}{\mathbf{y}} \and
	\end{mathpar}

	\boxedrule{$t \conv u$}

	$$t \conv u \quad:=\quad \exists v~v', t \rtred v \wedge u \rtred v'
        \wedge v =_\alpha v'$$\\[0.5em]

        \begin{center}
          where $=_\alpha$ denotes syntactic equality up-to renaming
        \end{center}
	\caption{\CIC: Bidirectional typing}
	\label{fig:bidir}
\end{figure}

In the usual, declarative, presentation of \CIC, conversion between types is
allowed at any stage of a typing derivation through a free-standing conversion
rule.
However, when conversion is replaced by a non-transitive relation of consistency, this free-standing rule is much too permissive and would violate \pconst{\CIC}.
 Indeed, as every type should be consistent with the unknown type
 $\?_{\square{}}$, using such a rule twice in a row makes it possible to change the type of a
 typable term to any arbitrary type: if $\Gamma \vdash t : T$, because $T \cons \?_{[]}$
 and $\?_{[]} \cons S$, we could derive $\Gamma \vdash t : S$. This in turn would allow typeability of any term, including fully-precise terms, which is in contradiction with \pconst{\CIC}.

Thus, we rely on a bidirectional presentation of \CIC typing, presented in \cref{fig:bidir}, where the usual judgment $\Gamma \vdash t : T$ is decomposed into several mutually-defined judgments. The difference between the judgments lies in the role of the type:
in the \emph{inference} judgment $\Gamma \vdash t \inferty T$, the type is considered an output, whereas in the
\emph{checking} judgment $\Gamma \vdash t \checkty T$, the type is instead seen as an input.
Conversion can then be restricted to specific positions, namely to mediate between
inference and checking judgments (see \nameref{infrule:cic-check}), and can thus never appear
twice in a row.

Additionally, in the framework of an elaboration procedure, it is interesting to make a clear distinction between the subject of the rule (\ie the object that is to be elaborated), inputs that can be used for this elaboration, and outputs that must be constructed during the elaboration.
In the context checking judgment $\vdash \Gamma$, $\Gamma$ is the subject of the judgment. In all the other judgments, the subject is the term, the context is an input, and the type is either an input or an output, as we just explained.

An important discipline, that goes with this distinction, is that judgments should ensure that outputs are well-formed, under the hypothesis that the inputs are. All rules are built to ensure this invariant. This distinction between inputs, subject and output, and the associated discipline, are inspired by \citet{McBride2018,McBride2019}.
This is also the reason why no rule for term elaboration re-checks the context, as it is an input that is assumed to be well-formed. Hence, most properties we state in an open context involve an explicit hypothesis that the involved context is well-formed.

\paragraph{Constrained Inference}
Apart from inference and checking, we also use a set of {\em constrained inference} judgments $\Gamma \vdash t \pcheckty{\bullet} T$, with the same modes as inference. These judgments infer the type $T$ but under some constraint $\bullet$: for instance that it should be a universe at some level ($\bullet=\square{}$), a $\Pi$-type ($\bullet=\P$), or an instance of an inductive $I$ ($\bullet=\text{I}$).
Constrained inference judgments come from a close analysis of typing algorithms, such as the one of \Coq, where in some places, an intermediate judgment between inference and checking happens: inference is performed, but then the type is reduced to expose its head constructor, which is imposed to be a specific one.
A stereotypical example is \nameref{infrule:cic-app}: one starts by inferring a type for $t$, but want it to be a $\Pi$-type so that its domain can be used to check $u$.
To the best of our knowledge, these judgments have never been formally described elsewhere.
Instead, in the rare bidirectional presentations of \CIC, they are inlined in some way, as they only amount to some reduction.
However, this is no longer true in a gradual setting: $\?$ introduces an alternative, valid solution to the constrained inference, as a term of type $\?$ can be used where a term with a $\Pi$-type is expected. Thus, we will need multiple rules for constrained inference, which is why we make it explicit already at this stage.

\paragraph{Reduction}
From here on, we impose no reduction strategy by default, and use $\redCCIC$ and the unqualified word "reduction" for \emph{full} reduction, \ie reduction that can be performed at an arbitrary place in a term, and $\rtred$ for its reflexive, transitive closure. Most of the properties would however carry over if we fixed \emph{weak-head} reduction instead, and we sketch at the end of some proofs how they would carry over to such a fixed strategy. As uniqueness of inferred types and elaborated terms becomes stronger with a deterministic reduction strategy, we discuss weak-head reduction specifically in that case.

Finally, we observe that the equivalence of this bidirectional
formulation with standard \CIC relies on the transitivity of
conversion; this has been very recently spelled out in details and formalized by
\citet{LennonBertrand2021}.
However, in the gradual setting, this property does not hold. This is precisely the point of
using a bidirectional formulation: since
consistency is not a transitive relation, a standard
presentation of typing is not appropriate. 

\section{From \GCIC to \CCIC}
\label{sec:gcic}\label{sec:gcic-to-ccic}

We now present the elaboration from the source gradual system \GCIC to the cast calculus \CCIC.
We start with \CCIC, describing its typing, reduction
and metatheoretical properties (\cref{sec:cast-calculus}).
We next describe \GCIC and its elaboration to
\CCIC, along with few direct properties (\cref{sec:elaboration}). This elaboration is mainly an extension of the bidirectional \CIC presented in the previous section.
We illustrate the semantics of the different \GCIC variants by considering the $\Omega$ term (\cref{sec:back-to-omega}).
We finally expose technical properties
of the reduction of \CCIC (\cref{sec:gcic-simulation}) used to prove the most
important theorems on elaboration: conservativity over \CIC or \CICs, as well as
the gradual guarantees (\cref{sec:gcic-theorems2}).

\subsection{\CCIC}
\label{sec:cast-calculus}

\paragraph{Syntax}
The syntax of \CCIC\footnote{Written using a \targetcolor{blue color}.}
extends that of \CIC (\cref{sec:bidirectional-cic})
with three new term constructors: the unknown term $\tcol{\?_T}$ and dynamic
error $\tcol{\err_T}$ of type $\tcol{T}$, as well as the cast $\tcol{\cast{S}{T}{t}}$ of a term $\tcol{t}$ of type $\tcol{S}$ to type $\tcol{T}$
\begin{align}
  \label{fig:syntax-castcic}
  \tcol{\terms_{\CCIC}} \ni \tcol{t} ::= \dots \mid \tcol{\?_t} \mid \tcol{\err_t} \mid \tcol{\cast{t}{t}{t}} \tag{Syntax of \CCIC}
\end{align}
with casts associating to the right: $\tcol{\cast{S}{S'}{\cast{T'}{T}{t}}}$ is $\tcol{\cast{S}{S'}{\left( \cast{T}{T'}{t}\right)}}$. We also compress successive ones in the following way: $\tcol{\cast{T}{T'' \Leftarrow T'}{t}}$ is shorthand for $\tcol{\cast{T'}{T''}{\cast{T}{T'}{t}}}$.
The unknown term and dynamic error both behave as exceptions as
defined in \ExTT~\cite{pedrotTabareau:esop2018}.
Casts keep track of the use of consistency during elaboration, implementing
a form of runtime type-checking, raising the error $\tcol{\err_T}$ in case of a type mismatch.
We call \emph{static} the terms of \CCIC that do not use any of these new constructors---static \CCIC terms correspond to \CIC terms.
\paragraph{Universe parameters}
\begin{figure}
\begin{align}
\qquad \sortOfPi{i}{j} &:= \max(i,j) &   \castOfPi{i} &:= i \label{eq:spi-cpi-gcicp}\tag{\text{\GCICP-\CCICP}}\\[-0.5em]
\sortOfPi{i}{j} &:= \max(i,j) & \castOfPi{i} &:= i-1 \label{eq:spi-cpi-gcict}\tag{\text{\GCICT-\CCICT}}\\[-0.5em]
\sortOfPi{i}{j} &:= \max(i,j)+1 &  \castOfPi{i} &:= i-1 \label{eq:spi-cpi-gcics}\tag{\text{\GCICs-\CCICs}}
\end{align}
	\caption{Universe parameters}
	\label{fig:univ-param}
\end{figure}

\CCIC is parametrized by two functions, described in \cref{fig:univ-param}, to account for the three
different variants of \GCIC we consider (\cref{sec:gcic:-3-1}).
The first function $\sortOfPiName$ computes
the level of the universe of a dependent product, given
the levels of its domain and codomain (see the updated \nameref{infrule:ccic-prod} rule in \cref{fig:ccic-typing}). The second function
$\castOfPiName$ controls the universe level in the reduction of a cast between
$\? \rightarrow \?$ and $\?$ (see \cref{fig:CCIC-reduction}).

\paragraph{Typing}

\cref{fig:ccic-typing} gives the typing rules for the three new primitives of \CCIC. Apart from the modified \nameref{infrule:ccic-prod} rule, which uses the $\sortOfPiName$ parameter, all other typing rules are exactly the same as in \CIC. When disambiguation is needed, we note this typing judgment as $\caty$.
The typing rules \nameref{infrule:ccic-unk} and \nameref{infrule:ccic-err}
say that both $\tcol{\?_T}$ and $\tcol{\err_T}$ infer $\tcol{T}$
when $\tcol{T}$ is a type.
Note that in \CCIC, as is sometimes the case in cast calculi~\cite{siekAl:popl10,newAhmed:icfp2018}, no consistency premise is required for a cast to be well-typed. 
Here, consistency only plays a role in \GCIC, but disappears after
elaboration. 
Instead, we rely on the usual conversion, defined as
in \CIC as the existence of $\alpha$-equal reducts for the reduction
described hereafter.
The \nameref{infrule:ccic-cast} rule only ensures that both the source and target of the cast are indeed types, and that the casted term indeed has the source type.

\begin{figure}
		\boxedrule{$\tcol{\Gamma} \vdash t \inferty T$}
		\begin{mathpar}
			\qquad \dots \and
			\inferrule{
        \tcol{\Gamma} \vdash \tcol{A} \pcheckty{[]} \tcol{[]_{j}} \\
        \tcol{\Gamma, x : A} \vdash \tcol{B} \pcheckty{[]} \tcol{[]_{i}}}
			{\tcol{\Gamma} \vdash \tcol{\P x : A . B} \inferty
                          \tcol{[]_{\sortOfPi{i}{j}}}}[Prod] \ilabel{infrule:ccic-prod}
                        \and \dots \\
			\inferrule{
        \tcol{\Gamma} \vdash \tcol{T} \pcheckty{[]} \tcol{[]_i}}
      {\tcol{\Gamma} \vdash \tcol{\?_T} \inferty \tcol{T}}[Unk]
        \ilabel{infrule:ccic-unk} \and
      \inferrule{
        \tcol{\Gamma} \vdash \tcol{T} \pcheckty{[]} \tcol{[]_i}}
        {\tcol{\Gamma} \vdash \tcol{\err_T} \inferty \tcol{T}}[Err]
        \ilabel{infrule:ccic-err} \and
			\inferrule{
        \tcol{\Gamma} \vdash \tcol{A} \pcheckty{[]} \tcol{[]_i} \\
        \tcol{\Gamma} \vdash \tcol{B} \pcheckty{[]} \tcol{[]_j} \\ 
        \tcol{\Gamma} \vdash \tcol{t} \checkty \tcol{A}}
        {\tcol{\Gamma} \vdash \tcol{\cast{A}{B}{t}} \inferty \tcol{B}}[Cast]
        \ilabel{infrule:ccic-cast} \and
		\end{mathpar}

		\caption{\CCIC: Bidirectional typing (extending \CIC \cref{fig:bidir},
replacing \nameref{infrule:cic-prod})}
		\label{fig:ccic-typing}
\end{figure}

\paragraph{Reduction}
\begin{figure}
  \boxedrule{$\H$, $\hd : \types_{\CCIC} \to \H$ and $\stalk : \H \to \types_{\CCIC}$}

	\begin{mathpar}
    ~\\

    \H \ni h ::= \square_i \mid \Pi \mid I \\

	\hd(\tcol{\P x : A . B}) := \Pi  \and
	\hd(\tcol{\square_i}) := \square_i  \and
	\hd(\tcol{I(\mathbf{a})}) := I	 \\\\

	\stalkCIC{i}{[]_j}
	:= \left\{ \begin{array}{lr}
	\tcol{[]_j} & \text{ if $j < i$}
	\\
	\tcol{\err_{[]_i}} & \text{if $j \geq i$}
	\end{array} \right. \and
	\stalkCIC{i}{I}
	:= \tcol{I(\Unk{\pars(I,i)})}  \and
	\stalkCIC{i}{\Pi}
	:= \left\{ \begin{array}{lr}
	\tcol{\Unk{[]_{\castOfPi{i}}} \rightarrow
	\Unk{[]_{\castOfPi{i}}}} & \text{if $\castOfPi{i}
	\geq 0$}
	\\
	\tcol{\err_{[]_i}} & \text{if $\castOfPi{i}
	< 0$}
	\end{array} \right.
	\end{mathpar}

	\caption{Head constructor and germ}
	\label{fig:head-germ}
\end{figure}

The typing rules provide little insight on the new primitives; the interesting
part really lie in their reduction behavior. The reduction rules of \CCIC are given in
\cref{fig:CCIC-reduction} (congruence rules omitted).
Reduction relies on two auxiliary functions relating head constructors $h \in \H$ (\cref{fig:head-germ}) to those terms that start with either $\Pi$, $[]$ or $I$, the set of which we call $\types_{\CCIC}$.
The first is the function $\hd$, which returns the head constructor of a type.
In the other direction, the germ%
\footnote{
  The germ function corresponds to an abstraction function as in
  AGT~\cite{garciaAl:popl2016}, if one interprets the head $h$ as the set of all
  types whose head type constructor is $h$. \citet{wadlerFindler:esop2009}
  christened the corresponding notion a \emph{ground type}, later reused in the
  gradual typing literature. This terminology however clashes with
  its prior use in denotational semantics~\cite{Levy2004}: there a ground type
  is a first-order datatype. Note that \citet{siekTaha:sfp2006} also call ground types the base types of the language, such as $\bool$ and $\nat$. We therefore prefer the less overloaded term {\em germ}, used by analogy with the
  geometrical notion of the \emph{germ of a section}~\cite{maclane-moerdijk}: the
  germ of a head constructor represents an equivalence class of types that are
  locally the same.
} function
$\stalkCIC{i}{h}$ constructs the least
precise type with head $h$ at level $i$.
In the case where no such type exists (\eg when $\castOfPi{i} < 0$),
this least precise type is the error.

\begin{figure}
  \boxedrule{$\tcol{t} \redCCIC \tcol{t}$}
  \begin{mathpar}
    \text{\textbf{Propagation rules for $\tcol{\?}$ and $\tcol{\err}$}} \\
    \redrule{
      \tcol{\?_{\P(x:A).B}}
    }{
      \tcol{\l (x : A). \?_{B}}
    \hfill}[Prod-Unk]
    \ilabel{redrule:prod-unk} \\
    \redrule{
      \tcol{\err_{\P(x:A).B}}
    }{
      \tcol{\l (x : A). \err_{B}}
    \hfill}[Prod-Err]
    \ilabel{redrule:prod-err} \\
    \redrule{
      \tcol{\match{I}{\?_{I(\mathbf{a})}}{z.P}{f.\mathbf{y.t}}}
    }{
      \tcol{\?_{P\subs{\?_{I(\orr{a})}}{z}}}
    \hfill}[Match-Unk] \ilabel{redrule:match-unk} \\
    \redrule{
      \tcol{\match{I}{\err_{I(\mathbf{a})}}{z.P}{f.\mathbf{y.t}}}
    }{
      \tcol{\err_{P\subs{\err_{I(\orr{a})}}{z}}}
    \hfill}[Match-Err] \ilabel{redrule:match-err} \\
    \redrule{
      \tcol{\ascdom{\?_{I(\orr{a})}}{I(\orr{a''})}{I(\orr{a'})}}
    }{
      \tcol{\?_{I(\orr{a'})}}
    \hfill}[Ind-Unk] \ilabel{redrule:ind-unk} \\
    \redrule{
      \tcol{\ascdom{\err_{I(\orr{a})}}{I(\orr{a''})}{I(\orr{a'})}}
    }{
      \tcol{\err_{I(\orr{a'})}}
    \hfill}[Ind-Err] \ilabel{redrule:ind-err} \\
    \redrule{
      \tcol{\ascdom{\?_{\?_{\square}}}{\?_{\square}}{X}}
    }{
      \tcol{\?_{X}}
    \hfill}[Down-Unk] \ilabel{redrule:down-unk} \\
  \redrule{
    \tcol{\ascdom{\err_{\?_{\square}}}{\?_{\square}}{X}}
  }{
    \tcol{\err_{X}}
  \hfill}[Down-Err] \ilabel{redrule:down-err} \\
  
  \text{\textbf{Reduction rules for cast}} \\
  
  \redrule{
    	\tcol{\ascdom{(\l x : A . t)}{\P (x: A_1). B_1}{\P (y:A_2).B_2}}
    }{ 
      \myflushright
    }[Prod-Prod] \ilabel{redrule:prod-prod} \\
  \myflushright \tcol{ \l y : A_2. \ascdom{(t\subs{\cast{A_2}{A}{y}}{x})}{B_1\subs{\cast{A_2}{A_1}{y}}{x}}{B_2} } \\
  \redrule{
    \tcol{\ascdom{A}{[]_i}{[]_i}}
  }{\tcol{A} \hfill}[Univ-Univ] \ilabel{redrule:univ-univ} \\
  \redrule{
      \tcol{\ascdom{c(\orr{a},b_1, \dots, b_n)}{I(\orr{a_1})}{I(\orr{a_2})}}
    }{
      \tcol{ c(\orr{a'}, \orr{b'_1}, \dots, \orr{b'_n})}
    \hfill}[Ind-Ind] \label{redrule:ind-ind} \\
    \myflushright 
    \text{with $\tcol{\orr{b_k'}} :=
      \tcol{\ascdom{b_k}{\args_k(I,i,c)\parsub{\orr{a},\orr{b}}}{\args_k(I,i,c)\parsub{\orr{a'},\orr{b'}}}}$} \\
	\redrule{
    \tcol{\ascdom{t}{T}{T'}}
  }{
    \tcol{\err_{T'}}
  }[Head-Err] \ilabel{redrule:head-err}
	  \hfill \text{when $\tcol{T}, \tcol{T'} \in \types_{\CCIC}$ and $\hd \tcol{T} \neq \hd \tcol{T'}$} 
     \\
  \redrule{
    \tcol{\ascdom{t}{\err_{\square}}{T}}
  }{
    \tcol{\err_T}
  \hfill}[Dom-Err] \ilabel{redrule:err-dom} \\
  \redrule{
    \tcol{\ascdom{t}{T}{\err_{\square}}}
  }{
    \tcol{\err_{\err_\square}}
  }[Codom-Err] \label{redrule:err-codom}
  \hfill \text{when $T \in \types_{\CCIC}$} \\
  \redrule{
    \tcol{\ascdom{f}{\P x:A.B}{\?_{\square_i}}}
  }{
    \tcol{\cast{\P x:A.B}{\?_{[]_i} \Leftarrow \stalkCIC{i}{\Pi}}{f}}
  \hfill}[Prod-Germ] \label{redrule:prod-germ} \\
	\myflushright \text{when $\tcol{\P x : A . B} \neq \tcol{\stalkCIC{j}{\Pi}}$ for $j \geq i$} \\
	\redrule{
		\tcol{\ascdom{t}{I(\orr{a})}{\?_{\square_i}}}
	}{
		\tcol{\cast{I(\orr{a})}{\?_{[]_i} \Leftarrow \stalkCIC{i}{I}}{t}}
	}[Ind-Germ] \label{redrule:ind-germ}
	\hfill \text{when $\tcol{I(\orr{a})} \neq \tcol{\stalkCIC{j}{I}}$ for $j\geq i$} \\
  \redrule{
		\tcol{\cast{\stalkCIC{i}{h}}{X \Leftarrow \?_{\square_i}}{t}}
	}{
		\tcol{\ascdom{t}{\stalkCIC{i}{h}}{X}}
  }[Up-Down] \label{redrule:up-down}
	\hfill \text{when $\tcol{\stalkCIC{i}{h}} \neq \tcol{\err_{[]_i}}$} \\
  \redrule{
    \tcol{\ascdom{t}{A}{\?_{\square_i}}}
  }{
    \tcol{\err_{\?_{\square_i}}}
  }[Size-Err] \label{redrule:size-err}
  \hfill \text{when $\min \{j \mathop{|} \exists h \in \H, \tcol{\stalkCIC{j}{h}} = A\} > i$}
  \end{mathpar}
  \caption{\CCIC: Reduction rules (extending \cref{fig:bidir},
          congruence rules omitted)}
  \label{fig:CCIC-reduction}

\end{figure}

The design of the reduction rules is mostly dictated by the discrete
and monotone models of \CCIC presented later
in~\cref{sec:realizing-cast-calculus}.
Nevertheless, we now provide some intuition about their meaning.
Let us start with rules \nameref{redrule:prod-unk}, \nameref{redrule:prod-err}, \nameref{redrule:match-unk} and \nameref{redrule:match-err}. These rules specify the exception-like propagation behavior of both $\tcol{\?}$
and $\tcol{\err}$ at product and inductive types.
Rules \nameref{redrule:ind-unk} and \nameref{redrule:ind-err}
similarly propagate $\tcol{\?}$ and $\tcol{\err}$ when cast between the same inductive
type, and rules \nameref{redrule:down-unk} and \nameref{redrule:down-err} do the same
from the unknown type to any type $\tcol{X}$.

Next are rules \nameref{redrule:prod-prod}, \nameref{redrule:ind-ind} and \nameref{redrule:univ-univ},
which correspond to success cases of dynamic checks, where the cast is between types with the same head.
In that case, casts are either completely
erased when possible, or propagated. As usual in gradual typing, directly inspired by higher-order contracts~\cite{findlerFelleisen:icfp2002}, \nameref{redrule:prod-prod} distributes the function cast in two casts, one for the argument and one for the body; note the substitution in the source codomain in order to account for dependency.
Also, because constructors and inductives are fully-applied,
this \nameref{redrule:prod-prod} rule cannot be blocked on a partially-applied constructor or inductive.
Regarding inductive types, the restriction to reduce only on constructors means that a cast between $\tcol{\nat}$ and $\tcol{\nat}$ is blocked until its argument term is a constructor, rather than disappearing right away as for the universe. 
We follow this somewhat non-optimal strategy to be consistent
between inductive types, because for more complex inductive types such as lists, the propagation of casts on subterms cannot be avoided.

On the contrary, rule \nameref{redrule:head-err} specifies failure of a dynamic check when the considered types have different heads. Similarly, rules \nameref{redrule:err-dom}, \nameref{redrule:err-codom} specify that cast to or from the error type is always an error.

Finally, there are specific rules pertaining to casts to and from $\tcol{\?}$, showcasing its behaviour as a universal type. Rules \nameref{redrule:prod-germ} and \nameref{redrule:ind-germ} decompose an upcast into $\tcol{\?}$ as an upcast to a germ followed by an upcast from the germ to $\tcol{\?}$. This decomposition of an upcast to $\tcol{\?}$ into a series of "atomic" upcasts from a germ to $\tcol{\?}$ is a consequence of the way the cast operation is implemented in \cref{sec:realizing-ccic}, but similar decompositions appear e.g. in \citet{siekAl:snapl2015}, where the equivalent of our germs are called ground types.
The side conditions guarantee that this rule is used when no other applies.
Rule \nameref{redrule:up-down} erases the succession of an upcast to $\tcol{\?}$ and a downcast from it. Note that in this rule the upcast $\tcol{\ascdom{t}{\stalkCIC{h}{i}}{\?_{[]_i}}}$ works like a constructor for $\tcol{\?_{[]_i}}$ and $\tcol{\ascdom{}{\?_{[]_i}}{X}}$ as a destructor---a view reflected by the canonical and neutral forms of \cref{fig:CCIC-canonical} for $\tcol{\?_{[]}}$.%
\footnote{
  In a simply-typed language such as GTLC~\cite{siekAl:snapl2015},
  where there are no neutrals at the type level, 
  casts from a germ/ground type 
  to the unknown type are usually interpreted as tagged values~\cite{siekTaha:sfp2006}. 
  Here, these correspond exactly to the canonical forms of $\tcol{\?_{[]}}$, but we also have to account for the many neutral forms that appear in open contexts.
}
Finally, rule \nameref{redrule:size-err} corresponds to a peculiar kind of error, which only happens due to the presence of a type hierarchy: 
$\tcol{?_{[]_i}}$ is only universal with respect to types at level $i$, and so a type might
be of a level too high to fit into it. To detect such a case, we check whether $A$ is a germ for a level that is below $i$, and when not throw an error.

\paragraph{Meta-Theoretical Properties}
The typing and reduction rules just given ensure two of the
meta-theoretical properties introduced in \cref{sec:tradeoffs}:
\psafe for the three variants of \CCIC, as well as \pnorm for \CCICT and
\CCICs.
Before turning to these properties, let us establish a crucial lemma, namely the
confluence of the rewriting system induced by reduction.
\begin{lemma}[Confluence of \CCIC]
  \label{lem:confluence}
  If $\tcol{t}$ and $\tcol{u}$ are related by the symmetric, reflexive, transitive closure
  of $\redCCIC$, then there exists $s$ such that $\tcol{t} \rtred \tcol{s}$ and $\tcol{u} \rtred \tcol{s}$.
\end{lemma}

\begin{proof}
  We extend the notion of parallel reduction ($\paraRed$) for \CIC
  from~\cite{SozeauBFTW20} to account for our additional reduction
  rules and show that the triangle property---the existence, for any
  term $t$, of an optimal reduced term $\tcol{\optRed{t}}$ in one step (\cref{fig:triangle})---still 
  holds. From the triangle property, it is
  easy to deduce confluence of parallel reduction in one step
  (\cref{fig:triangle-confluence}), which implies confluence
  because parallel reduction is between one-step reduction and
  iterated reductions.
  This proof method is basically an extension of the Tait-Martin Löf criterion on
  parallel reduction~\citep{barendregt:84,TAKAHASHI1995120}.
  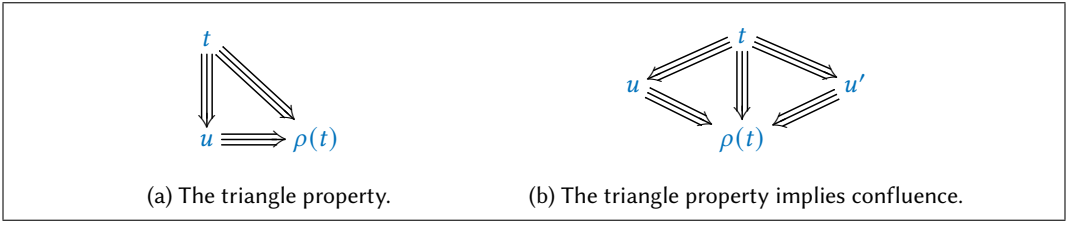
\begin{figure}
  \begin{subfigure}[b]{.4\textwidth}
  $$
  \xymatrix{\tcol{t} \ar@3{->}[d] \ar@3{->}[rd] & \\
  \tcol{u} \ar@3{->}[r] & \tcol{\optRed{t}}}
  $$
  \caption{The triangle property.}
  \label{fig:triangle}
  \end{subfigure}
  \begin{subfigure}[b]{.5\textwidth}
  $$
   \xymatrix@R=0.5pc{& \tcol{t} \ar@3{->}[ld] \ar@3{->}[dd] \ar@3{->}[rd] &  \\
     \tcol{u} \ar@3{->}[rd] & &  \tcol{u'} \ar@3{->}[ld] \\
     & \tcol{\optRed{t}}& }
  $$
  \caption{The triangle property implies confluence.}
  \label{fig:triangle-confluence}
  \end{subfigure}
  \caption{Representation of the triangle property (left) and its
    consequence on confluence (right).}
\end{figure}
\end{proof}

Let us now turn to \psafe, which we prove using the standard
progress and subject reduction properties~\cite{wrightFelleisen1994}.
Progress describes a set of canonical forms, asserting that all terms that do not belong to
such canonical forms are not in normal form, \ie can take at least
one reduction step. \cref{fig:CCIC-canonical} provides the definition
of canonical forms, considering head reduction.

As standard in dependent type theories, 
we distinguish between canonical forms and neutral
terms. Neutral terms
correspond to (blocked) destructors, waiting for a substitution to
happen, while other canonical forms correspond to constructors.
Additionally, the notion of neutral terms naturally induces a weak-head
reduction strategy that consists in either applying a top-level
reduction or reducing the (only) argument of the
top-level destructor that is in a neutral position.

The canonical forms for plain \CIC are given by the first three lines of \cref{fig:CCIC-canonical}.
The added rules deal with
errors, unknown terms and casts.
First, an error $\tcol{\err_t}$ or an unknown term $\tcol{\?_t}$ is neutral when
$\tcol{t}$ is neutral, and is canonical only when $\tcol{t}$ is $\tcol{[]}$ or $\tcol{I(\orr{a})}$,
but not a $\Pi$-type. This is because exception-like terms reduce on 
$\Pi$-types~\cite{pedrotTabareau:esop2018}.
Second, there is an additional specific form of canonical
inhabitants of $\tcol{\?_{[]}}$: these are upcasts from a germ, which can be
seen as a term tagged with the head constructor of its type, in a
matter reminiscent of actual implementations of dynamic typing using
type tags. As we explained when presenting \cref{fig:CCIC-reduction},
these canonical forms work as constructors for $\tcol{\?_{[]}}$.
Finally, the cast operation behaves as a destructor on the
universe $\tcol{[]}$---as if it were an inductive type of usual \CIC.
This destructor first scrutinizes the source type of
the cast. This is why the cast is neutral as soon as its source type
is neutral. When the source type reduces to a head constructor, there are two
possibilities. Either that constructor is $\tcol{\?_{[]}}$, in which case the cast
scrutinizes its argument to be a canonical form $\tcol{\cast{t}{\?_{[]}}{\stalkCIC{i}{h}}}$ and is neutral when this is not the case.
In all other cases, it first scrutinizes the target type, so the cast is neutral when the target type is neutral.
Finally, when both types have head constructors, the cast
might still need its argument to be either a $\lambda$-abstraction or an inductive
constructor to reduce.

\begin{figure}
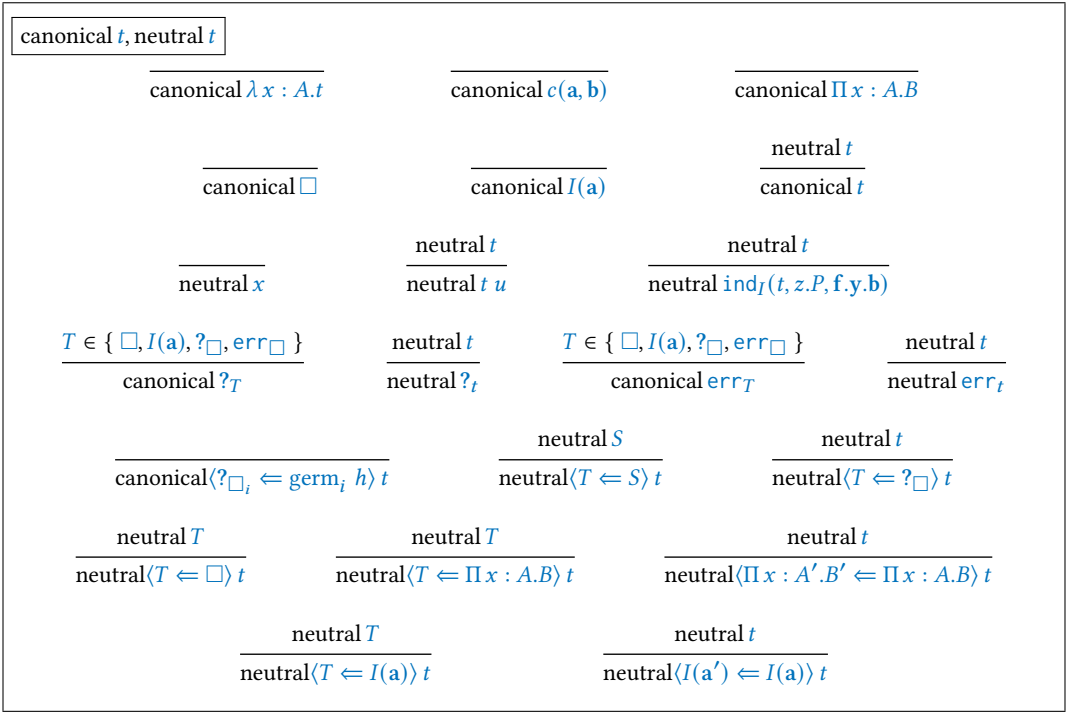

\begin{small}
\boxedrule{$\can{\tcol{t}}, \neu{\tcol{t}}$}\\
	\begin{mathpar}
		\inferrule{ }{\can{\tcol{\l x : A . t}}} \and
		\inferrule{ }{\can{\tcol{c(\orr{a},\orr{b})}}} \and
		\inferrule{ }{\can{\tcol{\P x : A . B}}} \\
		\inferrule{ }{\can{\tcol{[]}}} \and
		\inferrule{ }{\can{\tcol{I(\orr{a})}}} \and
    \inferrule{\neu{\tcol{t}}}{\can{\tcol{t}}} \\
		\inferrule{ }{\neu{\tcol{x}}} \and
		\inferrule{\neu{\tcol{t}}}{\neu{\tcol{t~u}}} \and
		\inferrule{\neu{\tcol{t}}}{\neu{\tcol{\match{I}{t}{z.P}{\orr{f.y.b}}}}} \\
    \inferrule{ \tcol{T} \in \{ ~\tcol{[]},\tcol{I(\orr{a})}, \tcol{\?_{[]}}, \tcol{\err_{[]}}~\} }{\can{\tcol{\?_{T}}}} \and
		\inferrule{\neu{\tcol{t}}}{\neu{\tcol{\?_{t}}}} \and %\\
    \inferrule{ \tcol{T} \in \{ ~\tcol{[]},\tcol{I(\orr{a})}, \tcol{\?_{[]}}, \tcol{\err_{[]}}~\} }{\can{\tcol{\err_{T}}}} \and
    \inferrule{\neu{\tcol{t}}}{\neu{\tcol{\err_{t}}}} \\
		\inferrule{ }{\can{\tcol{\cast{\stalkCIC{i}{h}}{\?_{[]_i}}{t}}}} \and
		\inferrule{\neu{\tcol{S}}}{\neu{\tcol{\cast{S}{T}{t}}}} \and
		\inferrule{\neu{\tcol{t}}}{\neu{\tcol{\cast{\?_{[]}}{T}{t}}}} \\
		\inferrule{\neu{\tcol{T}}}{\neu{\tcol{\cast{[]}{T}{t}}}} \and
		\inferrule{\neu{\tcol{T}}}{\neu{\tcol{\cast{\P x : A. B}{T}{t}}}} \and
		\inferrule{\neu{\tcol{t}}}{\neu{\tcol{\cast{\P x : A. B}{\P x : A'. B'}{t}}}} \and
		\inferrule{\neu{\tcol{T}}}{\neu{\tcol{\cast{I(\orr{a})}{T}{t}}}} \and
    \inferrule{\neu{\tcol{t}}}{\neu{\tcol{\cast{I(\orr{a})}{I(\orr{a'})}{t}}}} \and
	\end{mathpar}
\end{small}
	\caption{Head neutral and canonical forms for \CCIC}
	\label{fig:CCIC-canonical}
\end{figure}

Equipped with the notion of canonical forms, we can state \psafe for \CCIC:

\begin{theorem}[Safety of the three variants of \CCIC (\psafe)]
  \label{thm:ccic-psafe}
  \CCIC enjoys:
  \begin{description}
  \item[Progress:]
    if $\tcol{t}$ is a well-typed term of \CCIC, then either $\can{\tcol{t}}$ or there is some $\tcol{t'}$ such that $\tcol{t} \redCCIC \tcol{t'}$.
  \item[Subject reduction:] if $\tcol{\Gamma} \caty \tcol{t} \inferty \tcol{A}$ and $\tcol{t} \redCCIC{} \tcol{t'}$ then $\tcol{\Gamma} \caty \tcol{t'} \checkty \tcol{A}$.
  \end{description}
  Thus \CCIC enjoys \psafe.
\end{theorem}
\begin{proof}

  \begin{description}
  \item[Progress:]
	The proof is by induction on the typing derivation of $\tcol{t}$.
  As standard, we show that in all cases,
  either a reduction on a subterm happens, $\tcol{t}$ itself reduces
  because some canonical form was not neutral and creates a redex, or $\tcol{t}$ is neutral.

  \item[Subject reduction:]
  Subject reduction can be derived from the injectivity of type
  constructors, which is a direct consequence of confluence. See
  \cite{SozeauBFTW20} for a detailed account of this result in the
  simpler setting of \CIC.
  \end{description}
\end{proof}

We now establish normalization of \CCICT and \CCICs, although the proof below
relies on the discrete model defined in \cref{sec:bare-model}.

\begin{theorem}[Normalization of \CCICT and \CCICs (\pnorm)]
  \label{thm:ccic-pnorm}
  Every reduction path for a well-typed term in \CCICT or \CCICs is finite.
\end{theorem}

\begin{proof}
  The translation induced by the discrete model presented
  in~\cref{sec:bare-model} maps each reduction step to at least one
  step (\cref{thm:discrete-model}). So strong normalization holds
  because the target calculus of the translation is normalizing.
\end{proof}

\subsection{Elaboration from \GCIC to \CCIC}
\label{sec:elaboration}
Now that \CCIC has been described, we move on to \GCIC. The typing judgment of \GCIC is {\em defined} by an elaboration judgment from \GCIC to \CCIC, based upon \cref{fig:bidir}, augmenting all judgments with an extra output: the elaborated \CCIC term. This definition of typing using elaboration is required because of the intricate interdependency between typing and reduction exposed in \cref{sec:gcic-overview}.

\paragraph{Syntax}
\label{gcic-syntax}
The syntax of \GCIC\footnote{We use \sourcecolor{green} for terms of
  \GCIC. To maintain a distinction in the absence of colors, we also
  use tildes ($\sourcecolor{\tilde{t}}$) for terms in \GCIC in expressions mixing both source and target terms.} extends that of \CIC
with a single new term constructor $\scol{\?\ulev{i}}$, where $i$ is a universe level. From a user perspective, one is not given direct access to the failure and cast primitives, those only arise through uses of $\scol{\?}$.

\paragraph{Consistent conversion}

\begin{figure}
  \begin{small}
  \boxedrule{$\tcol{t} \acons \tcol{t}$}
  \begin{mathpar}
    \inferrule{ }{\tcol{x} \acons \tcol{x}} \and
    \inferrule{ }{\tcol{\square{}_{i}} \acons \tcol{\square{}_{i}}} \and
    \inferrule{\tcol{A} \acons \tcol{A'} \\ \tcol{t} \acons \tcol{t'}}
    {\tcol{\l x : A . t} \acons \tcol{\l x : A' . t'}} \and
    \inferrule{\tcol{A} \acons \tcol{A'} \\ \tcol{B} \acons \tcol{B'}}
    {\tcol{\P x : A . B} \acons \tcol{\P x : A'.B'}} \and
    \inferrule{\tcol{t} \acons \tcol{t'} \\ \tcol{u} \acons \tcol{u'}}{\tcol{t~u} \acons \tcol{t'~u'}} \and
    \inferrule{\tcol{\orr{a}} \acons \tcol{\orr{a'}}}{\tcol{I(\orr{a})} \acons \tcol{I(\orr{a'})}} \and
    \inferrule{\tcol{\orr{a}} \acons \tcol{\orr{a'}} \\ \tcol{\orr{b}} \acons \tcol{\orr{b'}}}
    {\tcol{c_k(\orr{a},\orr{b})} \acons \tcol{c_k(\orr{a'},\orr{b'})}} \and
    \inferrule{\tcol{s} \acons \tcol{s'} \\ \tcol{P} \acons \tcol{P'} \\ \tcol{\orr{t}} \acons \tcol{\orr{t'}}}
    {\tcol{\match{I}{s}{z.P}{f.\orr{y.t}}} \acons \tcol{\match{I}{s'}{z.P'}{f.\orr{y.t'}}}} \\
    \inferrule{\tcol{t} \acons \tcol{t'}}
    {\tcol{t} \acons \tcol{\cast{A'}{B'}{t'}}} \and
    \inferrule{\tcol{t} \acons \tcol{t'}}
    {\tcol{\cast{A}{B}{t}} \acons \tcol{t'} } \and
    \inferrule{ }{\tcol{t} \acons \tcol{\?_{T'}}} \and
    \inferrule{ }{\tcol{\?_T} \acons \tcol{t}}
  \end{mathpar}
  \end{small}
  \caption{\CCIC: $\alpha$-consistency}
  \label{fig:acons}
\end{figure}

Before we can describe typing, we should focus on conversion. Indeed, to account for the imprecision introduced by $\tcol{\?}$, elaboration employs {\em consistent conversion} to compare \CCIC terms rather than usual conversion relation. 

\begin{definition}[Consistent conversion]
	\label{def:cons}
	Two \CCIC terms are \emph{$\alpha$-consistent}, written $\acons$, if they are in the relation defined by the inductive rules of \cref{fig:acons}.

	Two terms are \emph{consistently convertible}, or simply \emph{consistent}, noted $\tcol{s} \cons \tcol{t}$, if and only if there exists $\tcol{s'}$ and $\tcol{t'}$ such that $\tcol{s} \rtred \tcol{s'}$, $\tcol{t} \rtred \tcol{t'}$ and $\tcol{s'} \acons \tcol{t'}$.
\end{definition}

Thus $\alpha$-consistency is an extension of $\alpha$-equality that takes imprecision into account. Apart from the standard rules making $\tcol{\?}$ consistent with any term, $\alpha$-consistency optimistically ignores casts, and does not consider errors to be consistent with themselves. 
The first point is to prevent casts inserted by the elaboration from disrupting valid conversions, typically between static terms. The second is guided by the idea that if errors are encountered at elaboration already, the term cannot be well behaved, so it must be rejected as early as possible and we should avoid typing it.
The consistency relation is then built upon $\alpha$-consistency in a way totally similar to how conversion in \cref{fig:bidir,fig:CCIC-reduction} is built upon $\alpha$-equality. Also note that this formulation of consistent conversion makes no assumption of
normalization, and is therefore usable as such in the non-normalizing \GCICP.

An important property of consistent conversion,
and a necessary condition for the conservativity of \GCIC with respect to\ \CIC (\pconst{\CIC}),
is that it corresponds to conversion on static terms.

\begin{proposition}[Properties of consistent conversion]~
	\label{prop:cons-static}
  \begin{enumerate}
  \item Two static terms are consistently convertible if and only if they are convertible in \CIC.
  \item If $\tcol{s}$ and $\tcol{t}$ have a normal form, then $\tcol{s} \cons \tcol{t}$ is decidable.
  \end{enumerate}
\end{proposition}

\begin{proof}
	\emph{(1)} First remark that $\alpha$-consistency between static terms corresponds to $\alpha$-equality of terms. Thus, and because the reduction of static terms in \CCIC is the same as the reduction of \CIC, two consistent static terms must reduce to $\alpha$-equal terms, which in turn implies that they are convertible.
	Conversely two convertible terms of \CIC have a common reduct,
        which is $\alpha$-consistent with itself. % making the two terms consistent.

  \emph{(2)} If $\tcol{s}$ and $\tcol{t}$ are normalizing, they have a finite number
  of reducts, thus to decide their consistency it is sufficient to check each pair of reducts for the decidable
  $\alpha$-consistency. Comparing normal forms is not enough, because a term $\tcol{t}$might be stuck because of a cast while another one $\tcol{s}$ can be $\alpha$-consistent with it and reduce further, so that the normal form of $\tcol{t}$ and $\tcol{s}$ are not $\alpha$-consistent while $\tcol{t}$ and $\tcol{s}$ are consistent.
\end{proof}

\paragraph{Elaboration}

\begin{figure}
	\begin{small}

  \boxedrule{$\sinferelab{\Gamma}{t}{T}{\Gamma}{t}{T}$}
  \begin{mathpar}
		\inferrule{\tcol{(x : T) \in \Gamma}}{\sinferelab{}{x}{}{\Gamma}{x}{T}}[Var] \ilabel{infrule:gcic-var} \and
		\inferrule{ }{\sinferelab{}{\square{}_{i}}{}{\Gamma}{\square{}_{i}}{\square{}_{i+1}}}[Univ] \ilabel{infrule:gcic-univ} \\

		\inferrule{
			\spcheckelab{\square{}}{}{\tilde{A}}{}{\Gamma}{A}{\square{}_i} \\
			\spcheckelab{\square{}}{}{\tilde{B}}{\square{}_{j}}{\Gamma, x : A}{B}{\square{}_j}}
		{\sinferelab{}{ \P x : \tilde{A} . \tilde{B}}{}{\Gamma}{\P x : A.B}{\square{}_{\sortOfPi{i}{j}}}}[Prod] \ilabel{infrule:gcic-prod} \and
		\inferrule{
			\spcheckelab{\square{}}{}{\tilde{A}}{}{\Gamma}{A}{\square{}_i} \\
			\sinferelab{}{\tilde{t}}{}{\Gamma, x : A}{t}{B}}
		{\sinferelab{}{\l x : \tilde{A} . \tilde{t}}{}{\Gamma}{\l x : A. t}{ \P x : A . B}}[Abs] \ilabel{infrule:gcic-abs} \and
		\inferrule{
			\spcheckelab{\Pi}{}{\tilde{t}}{}{\Gamma}{t}{\P x : A . B} \\
			\scheckelab{}{\tilde{u}}{}{\Gamma}{u}{A}}
		{\sinferelab{}{\tilde{t}~\tilde{u}}{}{\Gamma}{t~u}{B\subs{u}{x}}}[App] \ilabel{infrule:gcic-app}

		\inferrule{ }
		{\sinferelab{}{\?\ulev{i}}{}{\Gamma}{\?_{\?_{[]_i}}}{\?_{[]_i}}}[Unk] \ilabel{infrule:gcic-unk}

                \\
		\inferrule{\scheckelab{}{\tilde{a}_k}{}{\Gamma}{a_k}{\pars_k(I,i)\parsub{\orr{a}}}}
		{\sinferelab{}{I\ulev{i}(\orr{\tilde{a}})}{}{\Gamma}{I\ulev{i}(\orr{a})}{[]_i}}[Ind] \ilabel{infrule:gcic-ind}
                \and
		\inferrule{
			\scheckelab{}{\tilde{a}_k}{}{\Gamma}{a_k}{\pars_k(I,i)\parsub{\orr{a}}} \\
			\scheckelab{}{\tilde{b}_m}{}{\Gamma}{b_m}{\args_m(I,i,c)\parsub{\orr{a},\orr{b}}}}
		{\sinferelab{}{c_k\ulev{i}(\orr{\tilde{a}},\orr{\tilde{b}})}{}{\Gamma}{c(\orr{a},\orr{b})}{I(\orr{a})}}[Cons] \ilabel{infrule:gcic-cons} \\
		\inferrule{
			\spcheckelab{I}{}{\tilde{s}}{}{\Gamma}{s}{I(\orr{a})} \\
			\spcheckelab{\square{}}{}{\tilde{P}}{}{\Gamma, z : I(\orr{a})}{P}{\square{}_i} \\\\
			\scheckelab{}{\tilde{t}_k}{}{\Gamma, f : (\P z : I(\orr{a}), P), \orr{y} : \args(I,i,c_k)\parsub{\orr{a},\orr{y}}}{t_k}{P\subs{c_k(\orr{a},\orr{y})}{z}}
		}{
    \sinferelab{}{\match{I}{\tilde{s}}{z.\tilde{P}}{f.\orr{y.\tilde{t}}}}{}{\Gamma}{\match{I}{s}{z.P}{\orr{f.y.t}}}{P\subs{s}{z}}}[Fix]
	\ilabel{infrule:gcic-fix}
	\end{mathpar}

	\boxedrule{$\scheckelab{\Gamma}{t}{T}{\Gamma}{t}{T}$}
    \begin{mathpar}
		\inferrule{
			\sinferelab{}{\tilde{t}}{}{\Gamma}{t}{T} \\ \targetcolor{T \cons S}}
		{\scheckelab{}{\tilde{t}}{}{\Gamma}{\cast{T}{S}{t}}{S}}[Check] \ilabel{infrule:gcic-check}
    \end{mathpar}

	\boxedrule{$\spcheckelab{\bullet}{\Gamma}{t}{T}{\Gamma}{t}{T}$} \vspace{1em}
    \begin{mathpar}
    	\inferrule{
    		\sinferelab{}{\tilde{t}}{}{\Gamma}{t}{T} \\ \targetcolor{T \rtred \square{}_{i}}}
		{\spcheckelab{\square{}}{}{\tilde{t}}{}{\Gamma}{t}{\square{}_{i}}}[Inf-Unk] \ilabel{infrule:gcic-inf-univ} \and
		 \inferrule{
		 	\sinferelab{}{\tilde{t}}{}{\Gamma}{t}{T} \\ \targetcolor{T \rtred \?_{[]_{i+1}}}}
		 {\spcheckelab{\square}{}{\tilde{t}}{}{\Gamma}{\cast{T}{[]_i}{t}}{[]_i}}[Inf-Univ?] \ilabel{infrule:gcic-univ-unk} \\
		\inferrule{
			\sinferelab{}{\tilde{t}}{}{\Gamma}{t}{T} \\ \targetcolor{T \rtred \P x : A. B}}
		{\spcheckelab{\Pi}{}{\tilde{t}}{}{\Gamma}{t}{\P x : A. B}}[Inf-Prod] \ilabel{infrule:gcic-inf-prod} \and
		 \inferrule{
		 	\sinferelab{}{\tilde{t}}{}{\Gamma}{t}{T} \\ \targetcolor{T \rtred \?_{[]_i}} \\ \castOfPi{i} \geq 0}
		 {\spcheckelab{\Pi}{}{\tilde{t}}{}{\Gamma}{\cast{T}{\stalkCIC{i}{\Pi}}{t}}{\stalkCIC{i}{\Pi}}}[Inf-Prod?] \ilabel{infrule:gcic-prod-unk} \\
		\inferrule{\sinferelab{}{\tilde{t}}{}{\Gamma}{t}{T} \\ \targetcolor{T \rtred I(\orr{a})}}
		{\spcheckelab{I}{}{\tilde{t}}{}{\Gamma}{t}{I(\orr{a})}}[Inf-Ind] \ilabel{infrule:gcic-inf-ind} \and
		\inferrule{
			\sinferelab{}{\tilde{t}}{}{\Gamma}{t}{T} \\ \targetcolor{T \rtred \?_{[]_i}}}
		{\spcheckelab{I}{}{\tilde{t}}{}{\Gamma}{\cast{T}{\stalkCIC{i}{I}}{t}}{\stalkCIC{i}{I}}}[Inf-Ind?] \ilabel{infrule:gcic-ind-unk}
	\end{mathpar}
  \end{small}
  \caption{Type-directed elaboration from \GCIC to \CCIC}
	\label{fig:elaboration}
\end{figure}

Elaboration from \GCIC to \CCIC is given in \cref{fig:elaboration}, closely following the bidirectional presentation of \CIC (\cref{fig:bidir}) for most rules, simply carrying around the extra elaborated terms. Note that only the subject of the judgment is a \sourcecolor{source term} in \GCIC;
other inputs (that have already been elaborated), as well as outputs (that are to be constructed), are \targetcolor{target terms} in \CCIC. Let us comment a bit on the specific modifications and additions compared to \cref{fig:bidir}.

The most salient feature of elaboration is the insertion of casts that mediate
between merely consistent but not convertible types.
They of course are needed in the rule \nameref{infrule:gcic-check} where the terms are compared using consistency. But this is not enough: casts also appear in the newly-introduced rules \nameref{infrule:gcic-univ-unk} \nameref{infrule:gcic-prod-unk} and \nameref{infrule:gcic-ind-unk} for constrained inference, where the type $\tcol{?_{[]_i}}$ is replaced by the least precise type of the appropriate universe level
having the constrained head constructor, which is exactly what the $\stalk$ function gives us. Note that in the case of \nameref{infrule:gcic-univ-unk} we could have replaced $\tcol{[]_i}$ with $\tcol{\stalkCIC{i+1}{[]_i}}$ to make for a presentation similar to the other two rules. The role of these three rules is to ensure that a term of type $\tcol{\?_{[]_i}}$ can be used as a function, or as a scrutinee of a match, by giving a way to derive constrained inference for such a term.

It is interesting to observe that the rules for constrained elaboration in a gradual setting bear a close resemblance with those described by \citet[Section 3.3]{ciminiSiek:popl2016}, where a matching operator is introduced to verify that an output type can fit into a certain type constructor---either by having that type constructor as head symbol or by virtue of being $\?$. Such a form of matching was already present in our static, bidirectional system, because of the presence of reduction in types. In a way, both \citet{ciminiSiek:popl2016} and \citet{LennonBertrand2021} have the same need of separating the inferred type from operations on it to recover its head constructor, and our mixing of both computation and gradual typing makes that need even clearer.

Rule \nameref{infrule:gcic-unk} also deserves some explanation: $\scol{\?\ulev{i}}$ is elaborated to
$\tcol{\?_{\?_{[]_i}}}$, the least precise term of the least precise type of the whole universe $\tcol{[]_i}$. This avoids unneeded type annotations on $\scol{\?}$ in \GCIC. Instead, the context is responsible for
inserting the appropriate cast, \eg $\scol{\asc{\?}{T}}$ elaborates to a term reducing to $\tcol{\?_T}$.
We do not drop annotations altogether because of an important property on which
bidirectional \CIC is built: any well-formed term should \emph{infer} a type, not just check.
Thus, we must be able to infer a type for $\scol{\?}$. The obvious choice is to have $\scol{\?}$ infer $\tcol{\?}$, but this $\tcol{\?}$ is a term of \CCIC, and thus needs a type index. 
Because this $\tcol{\?}$ is used as a type, this index must be $\tcol{[]}$, and the universe level of the source $\scol{\?}$ is there to give us the level of this $\tcol{[]}$. 
In a real system, this should be handled by \emph{typical ambiguity},\footnote{Typical ambiguity~\cite{Harper1991} is the possibility to avoid giving explicit
universe levels, letting the system decide whether a consistent assignment
of levels can be found. In \Coq, for instance, one almost never has to be explicit about universe levels when writing \texttt{Type}.\label{note:typ-amb}}
alleviating the user from the need to give any annotations when using $\scol{\?}$.

\paragraph{Direct properties}

As the elaboration rules are completely syntax-directed, they
immediately translate to
an algorithm for elaboration. Coupled with decidability of consistency (\cref{prop:cons-static}),
this makes elaboration decidable whenever $\rtred$ is normalizing; 
when $\rtred$ is not normalizing, the elaboration algorithm might diverge, resulting in only semi-decidability of typing (as in, for instance, Dependent Haskell~\cite{eisenberg2016dependent}).

\begin{theorem}[Decidability of elaboration]
	\label{thm:decidability}
	The relations of inference, checking and partial inference of \cref{fig:elaboration} are decidable in \GCICT and \GCICs. They are semi-decidable in \GCICP.
\end{theorem}

Let us now establish two important properties of elaboration that we can
prove at this stage: elaboration is {\em correct}, insofar as it produces well-typed
\CCIC terms, and functional, in the sense that a given \GCIC term can
be elaborated to at most one \CCIC term up to conversion.

\begin{theorem}[Correctness of elaboration]
  \label{thm:correction}
	The elaboration produces well-typed terms in a well-formed context. Namely, given $\tcol{\Gamma}$ such that $\caty \tcol{\Gamma}$, we have that:
	\begin{itemize}
		\item if $\inferelab{}{\tilde{t}}{}{\Gamma}{t}{T}$, then $\tcol{\Gamma} \caty \tcol{t} \inferty \tcol{T}$;
		\item if $\pcheckelab{\bullet}{}{\tilde{t}}{}{\Gamma}{t}{T}$ then $\tcol{\Gamma} \caty \tcol{t} \pcheckty{\bullet} \tcol{T}$ (with $\bullet$ denoting
		the same index in both derivations);
		\item if $\checkelab{}{\tilde{t}}{}{\Gamma}{t}{T}$ and $\tcol{\Gamma} \caty \tcol{T} \pcheckty{[]} \tcol{[]_i}$, then $\tcol{\Gamma} \caty \tcol{t} \checkty \tcol{T}$.
	\end{itemize}
\end{theorem}

\begin{proof}
  The proof is by induction on the elaboration derivation,
  mutually with similar properties for all typing judgments.
  In particular, for checking, we have an extra hypothesis that the
  given type is well-formed, as it is an input that should already
  have been typed.

  Because the bidirectional typing rules of \CIC are very similar to the \GCIC-to-\CCIC elaboration rules, the induction is mostly routine. Let us point however that the careful design of the bidirectional rules already in \CIC regarding the input/output separation is important here. Indeed, we have that inputs to the successive premises of a rule are always well-formed, either as inputs to the conclusion, or thanks to previous premises. In particular, all context extensions
    are valid, \ie $\tcol{\Gamma, x : A}$ is used only when $\tcol{\Gamma} \vdash \tcol{A} \pcheckty{[]} \tcol{[]_i}$, and similarly only well-formed types are used for checking. This ensures that we can always use the induction hypothesis.

    The only novel points to consider are the rules where a cast is inserted. For these, we rely on the validity property (an inferred type is always well-typed itself) to ensure that the domain of inserted casts is well-typed, and thus that the casts can be typed.
\end{proof}

Because of the absence of a fixed, deterministic reduction strategy, the elaborated term is not unique. Indeed, since a type can be reduced to multiple product types in rule \cref{infrule:gcic-prod}, a term can infer multiple, different types, and since those appear later on in casts, the elaborated terms can differ by having different, albeit convertible, types in their casts. We thus state two theorems: one is uniqueness up to conversion, in case full reduction is used. The second is a strengthening if a weak-head reduction strategy is imposed for reduction.

\begin{theorem}[Uniqueness of elaboration---Full reduction]
	\label{thm:uniqueness}
	Elaborated terms are convertible:
	\begin{itemize}
		\item if $\inferelab{}{\tilde{t}}{}{\Gamma}{t}{T}$ and $\inferelab{}{\tilde{t}}{}{\Gamma}{t'}{T'}$, then $\tcol{t} \conv \tcol{t'}$ and $\tcol{T} \conv \tcol{T'}$;
		\item if $\pcheckelab{\bullet}{}{\tilde{t}}{}{\Gamma}{t}{T}$ and $\pcheckelab{\bullet}{}{\tilde{t}}{}{\Gamma}{t'}{T'}$ then $\tcol{t} \conv \tcol{t'}$ and $\tcol{T} \conv \tcol{T'}$;
		\item if $\checkelab{}{\tilde{t}}{}{\Gamma}{t}{T}$ and $\checkelab{}{\tilde{t}}{}{\Gamma}{t'}{T}$ then $\tcol{t} \conv \tcol{t'}$.
	\end{itemize}
\end{theorem}
\noindent 
(Recall that conversion $\conv$ in \CCIC is defined (similarly as in \CIC) as the existence of $\alpha$-equal reducts for the reduction given in \cref{fig:CCIC-reduction}.)

\begin{theorem}[Uniqueness of elaboration---Weak-head reduction]
	\label{thm:uniqueness-wh}
	If in \cref{fig:elaboration}, $\rtred$ is replaced by weak-head reduction, then elaborated terms are unique:
	\begin{itemize}
		\item given $\tcol{\Gamma}$ and $\scol{\tilde{t}}$, there is at most one $\tcol{t}$ and one $\tcol{T}$ such that $\inferelab{}{\tilde{t}}{}{\Gamma}{t}{T}$;
		\item given $\tcol{\Gamma}$ and $\scol{\tilde{t}}$, there is at most one $\tcol{t}$ and one $\tcol{T}$ such that $\pcheckelab{\bullet}{}{\tilde{t}}{}{\Gamma}{t}{T}$;
		\item given $\tcol{\Gamma}$, $\scol{\tilde{t}}$ and $\tcol{T}$, there is at most one $\tcol{t}$ such that $\checkelab{}{\tilde{t}}{}{\Gamma}{t}{T}$.
	\end{itemize}\end{theorem}

\begin{proof}
    Like for \cref{thm:correction}, those are proven mutually by induction on the typing derivation.

    The main argument is that there is always at most one rule that can apply to get a typing conclusion for a given term. This is true for all inference statements because there is exactly one inference rule for each term constructor, and for checking because there is only one rule to derive checking. In those cases simply combining the hypothesis of uniqueness is enough.

    For $\pcheckty{\Pi}$, by confluence of \CCIC the inferred type cannot at the same time reduce to $\tcol{\?_{\square{}}}$ and $\tcol{\P x : A . B}$, because those do not have a common reduct. Thus, only one of the two rules \nameref{infrule:gcic-inf-prod} and \nameref{infrule:gcic-prod-unk} can apply. It is enough to conclude for \cref{thm:uniqueness}, because reducts of convertible types are still convertible. For \cref{thm:uniqueness-wh} the deterministic reduction strategy ensures that the inferred type is indeed unique, rather than unique up to conversion. The reasoning is similar for the other constrained inference judgments.
\end{proof}

\subsection{Illustration: Back to Omega}
\label{sec:back-to-omega}

Now that \GCIC has been entirely presented, let us come back to the important example of $\scol{\Omega}$, and explain in detail the behavior described in \cref{sec:gcic:-3-1} for the three \GCIC variants. 

Recall that $\scol{\Omega}$ is the
term $\scol{\delta~\delta}$, with $\scol{\delta} := \scol{\l x : \?\ulev{i+1} .
  x~x}$. We leave out the casts present in
\cref{sec:tradeoffs,sec:gcic-overview}, knowing that they will be introduced by
elaboration. We also use $\scol{\?}$ at level $i + 1$, because
$\scol{\?\ulev{i+1}}$, when elaborated as a type, becomes $\tcol{T} :=
\tcol{\cast{\?_{[]_{i+1}}}{[]_i}{\?_{\?_{[]_{i+1}}}}}$, such that $\tcol{T} \rtred \tcol{\?_{[]_i}}$.
For the rest of this section, we write $\tcol{\?_{j}}$ instead of $\tcol{\?_{[]_j}}$
to avoid stacked indices and ease readability.

If $i = 0$ the elaboration of $\scol{\delta}$ (and thus of $\scol{\Omega}$)
fails in \GCICs and \GCICT,
because the inferred type for $x$ is $\tcol{T}$, which reduces to $\tcol{\?_0}$. 
Then, because $\castOfPi{0} = -1 < 0$ in both \GCICs and \GCICT, rule \nameref{infrule:gcic-prod-unk} does not apply and $\scol{\delta}$ is deemed ill-typed, as is $\scol{\Omega}$.

Otherwise, if $i > 0$ or we are considering \GCICP, $\scol{\delta}$ can be elaborated, and we have
\[ \inferelab{}{\delta}{}{\cdot}{\tcol{\l x : T . \left(\cast{T}{\stalkCIC{i}{\Pi}}{x}
\right)~\left(\cast{T}{\?_{\castOfPi{i}}}{x} \right)}}{T \to \?_{\castOfPi{i}}} \]
From this, we get that $\scol{\Omega}$ also elaborates, namely 
(with $\tcol{\delta'}$ the elaboration of $\scol{\delta}$ above)
\[ \inferelab{}{\Omega}{}{\cdot}{\delta'~\left(\cast{T \to \?_{\castOfPi{i}}}{T}{\delta'}\right)}{\?_{\castOfPi{i}}} \]
Let us now look at the reduction behavior of this elaborated term $\tcol{\Omega'}$ in
the three systems: it reduces seamlessly when $\castOfPi{i} = i$ (\GCICP/\CCICP),
while having $\castOfPi{i} < i$ makes it fail (\GCICs/\CCICs and \GCICT/\CCICT).
The reduction of $\tcol{\Omega'}$ in \CCICP is as follows:
$$\begin{array}{rcl}
  \tcol{\Omega'} & \rtred & \tcol{
    \left(\l x : \?_{i} . \left(\cast{T}{\?_i \to \?_i}{x}\right)~\left(\cast{T}{\?_i}{x}\right)\right)~\left(\cast{T \to \?_i}{T}{\delta'}\right)
  } \\
   & \rtred & \tcol{\left(\l x : \?_{i} . \left(\cast{\?_i}{\?_i \to \?_i}{x}\right)~\left(\cast{\?_i}{\?_i}{x}\right)\right)\left(\cast{\?_i \to \?_i}{\?_i}{\delta'}\right) } \\
  & \rtred & \tcol{\left(\cast{\?_i \to \?_i}{\?_i \to \?_i \Leftarrow \?_i}{\delta'}\right)~\left(\cast{\?_i \to \?_i}{\?_i \Leftarrow \?_i}{\delta'}\right)} \\
  & \rtred & \tcol{\left(\cast{\?_i \to \?_i}{\?_i \to
             \?_i}{\delta'}\right)~\left(\cast{\?_i \to
             \?_i}{\?_i}{\delta'}\right)} \\
                  & \rtred & \tcol{\left(\l x : \?_i . \castbg{pastelgray}{\?_i}{\?_i}{\left(\left(\cast{\?_i}{\?_i \to \?_i}{x}\right)~(\castbg{pastelgray}{\?_i}{\?_i}{x})\right)}\right)~\left(\cast{\?_i \to \?_i}{\?_i}{\delta'}\right)} \\
  \end{array}
  $$
The first step is the identity, simply replacing $\tcol{\Omega'}, \castOfPi{i}$ and the first occurrence of $\tcol{\delta'}$ by their definitions.
The second reduces $\tcol{T}$ to $\tcol{\?_i}$. In the third, the casted $\tcol{\delta{}'}$ is substituted for $\tcol{x}$ by a $\beta$ step. Casts are finally simplified using \nameref{redrule:up-down} and \nameref{redrule:prod-prod}.
At that point, the reduction has almost looped back to the second step, apart from the casts $\tcol{\castbg{pastelgray}{\?_i}{\?_i}{}}$ in the first occurrence of $\tcol{\delta'}$, which will simply accumulate through reduction, but without hindering divergence.

On the contrary, the normalizing variants have $\castOfPi{i} < i$, 
and thus share the following reduction path:
$$\begin{array}{rcl}
 \tcol{\Omega'} 
  &\rtred & \tcol{\left(\cast{\?_i \to \?_{i-1}}{\?_{i-1} \to \?_{i-1} \Leftarrow \?_i}{\delta'}\right)~\left(\castbg{pastelgray}{\?_i \to \?_{i-1}}{\?_{i-1} \Leftarrow \?_i}{\delta'}\right)} \\
  &\rtred & \tcol{\left(\cast{\?_i \to \?_{i-1}}{\?_{i-1} \to \?_{i-1} \Leftarrow \?_i}{\delta'}\right)
  \left(\cast{\?_i \to \?_{i-1}}{\?_{i-1} \Leftarrow \?_{i-1} \to \?_{i-1}}{\delta'}\right)} \\
  &\rtred & \tcol{\left(\cast{\?_i \to \?_{i-1}}{\?_{i-1} \to \?_{i-1} \Leftarrow \?_i}{\delta'}\right)~\err_{\?_{i-1}}} \\
  &\rtred & \tcol{\left(\cast{\?_i \to \?_{i-1}}{\?_{i-1} \to \?_{i-1} \Leftarrow \?_{i-1} \to \?_{i-1}}{\delta'}\right)~\err_{\?_{i-1}}} \\
  &\rtred & \tcol{\left(\l x : \?_{i-1} . \cast{\?_{i-1}}{\?_{i-1} \Leftarrow \?_{i-1}}{\left(\left(\cast{\?_i}{\?_{i-1} \to \?_{i-1}}{x'}\right)~\left(\cast{\?_i}{\?_{i-1}}{x'}\right)\right)}\right)~\err_{\?_{i-1}}} \\
  && \text{ where $\tcol{x'}$ is $\tcol{\cast{\?_{i-1}}{\?_i \Leftarrow \?_{i-1}}{x}}$} \\
                &\rtred & \tcol{\cast{\?_{i-1}}{\?_{i-1} \Leftarrow \?_{i-1}}{(\err_{\?_{i-1} \to \?_{i-1}}~\err_{\?_{i-1}})}} \\
  &\rtred & \tcol{\err_{\?_{i-1}}}
  \end{array}
$$
The first step corresponds to the first three above, the only difference being the value of $\castOfPi{i}$. The reductions however differ in the next step because $\tcol{\?_i \to \?_{i-1}} \neq \tcol{\stalkCIC{i}{\Pi}}$, so \nameref{redrule:prod-germ} applies before \nameref{redrule:up-down}.
For the third step, note that $\tcol{\?_{i-1} \to \?_{i-1}} = \tcol{\stalkCIC{i}{\Pi}}$,
so that \nameref{redrule:down-err} applies in the rightmost sequence of casts.
The last three steps of reduction then propagate the error by first using \nameref{redrule:prod-germ}, \nameref{redrule:up-down} and \nameref{redrule:prod-prod}, then the $\beta$ rule, and finally \nameref{redrule:down-err}, \nameref{redrule:prod-err} and a last $\beta$ step.
At a high-level, the error can be seen as a dynamic universe inconsistency, triggered by the invalid downcast $\castbg{pastelgray}{\?_i}{\?_{i-1}}{}$ highlighted on the first line.

\subsection{Precision is a simulation for reduction}
\label{sec:gcic-simulation}

Establishing the graduality of elaboration---the formulation of
the static gradual guarantee (SGG) in our setting---is no small feat, as it
requires properties about
computations in \CCIC that amount to the {\em dynamic} gradual guarantee (DGG).
Indeed, to handle
the typing rules for checking and constrained inference, it is
necessary to know how consistency and reduction evolve as a type
becomes less precise.
As already explained in \cref{sec:precision-graduality}, we
cannot directly prove graduality for a syntactic notion of
precision. However, we can still show that this relation is a simulation for reduction. While weaker than graduality, this property implies the DGG and suffices to conclude that graduality of elaboration holds.
The purpose of this section is to establish it.
Our proof is partly inspired by the proof of DGG by
\citet{siekAl:snapl2015}.\footnote{
Lemma 7 in \citet{siekAl:snapl2015} is similar to our \cref{thm:simulation}, and \cref{fig:apre-ccic} draws from their Fig. 9, especially for \nameref{infrule:capre-castr} and \nameref{infrule:capre-castl}. Also, while we do not make them explicit here, Lemmas 8, 10 and 11 also appear in our proofs.
}
We however had to adapt to the much higher complexity of \CIC compared to \STLC.
In particular, the presence of computation in the domain and codomain of casts is quite subtle to tame, as we must in general reduce types in a cast before we can reduce the cast itself.\footnote{
  Thus, while \cref{lem:catchup-lambda,lem:catchup-cons} correspond roughly to Lemma 9 in \citet{siekAl:snapl2015}, \cref{lem:catchup-univ,lem:catchup-type} are completely novel.
}

Technically, we need to distinguish between two notions of precision, one for \GCIC and one for \CCIC:
(i) {\em syntactic precision} on terms in \GCIC, which corresponds to the
usual syntactic precision of gradual typing~\cite{siekAl:snapl2015}, (ii)
{\em structural precision} on terms in \CCIC, which corresponds
to syntactic precision together with a proper account of casts.
In this section, we concentrate on properties of structural precision in \CCIC. 
We only state and discuss the various lemmas and
theorems on a high level, and refer the reader to \cref{sec:precision-reduction} for
the detailed proofs.

\paragraph{Structural precision for \CCIC} As emphasized already, the key property we want to establish is that precision is a simulation for reduction,
\ie that less precise terms reduce at least as well as more precise ones. This property guides the quite involved definition we are about to give for structural precision: it is rigid enough to give the induction hypotheses needed to prove simulation, while being lax enough to be a consequence of syntactic precision after elaboration, which is the key point to establish elaboration graduality (\cref{thm:static-graduality}), our equivalent of the static gradual guarantee.

\begin{figure}
	\boxedrule{$\tcol{\GG} \vdash \tcol{t} \capre \tcol{t'}$}
	\begin{mathpar}
		\inferrule{ }{\tcol{\GG} \vdash \tcol{[]_{i}} \capre \tcol{[]_{i}}}[Diag-Univ] \ilabel{infrule:capre-univ} \and
		\inferrule{\tcol{\GG} \vdash \tcol{A} \capre \tcol{A'} \\ \tcol{\GG, x : A \mid A'} \vdash \tcol{B} \capre \tcol{B'}}
		{\tcol{\GG} \vdash \tcol{\P x : A . B} \capre \tcol{\P x : A'.B'}}[Diag-Prod] \ilabel{infrule:capre-prod} \and
		\inferrule{\tcol{\GG} \vdash \tcol{A} \cdpre \tcol{A'} \\ \tcol{\GG, x : A \mid A'} \vdash \tcol{t} \capre \tcol{t'}}
		{\tcol{\GG} \vdash \tcol{\l x : A . t} \capre \tcol{\l x : A' . t'}}[Diag-Abs] \ilabel{infrule:capre-abs} \and
		\inferrule{\tcol{\GG} \vdash \tcol{t} \capre \tcol{t'} \\ \tcol{\GG} \vdash \tcol{u} \capre \tcol{u'}}
		{\tcol{\GG} \vdash \tcol{t~u} \capre
                  \tcol{t'~u'}}[Diag-App] \ilabel{infrule:capre-app}
                \and
		\inferrule{ }{\tcol{\GG} \vdash \tcol{x} \capre \tcol{x}}[Diag-Var] \ilabel{infrule:capre-var} \and
		\inferrule{\tcol{\GG} \vdash \tcol{A} \capre \tcol{A'} \\ \tcol{\GG} \vdash \tcol{B} \capre \tcol{B'} \\ \tcol{\GG} \vdash \tcol{t} \capre \tcol{t'} }
		{\tcol{\GG} \vdash \tcol{\cast{A}{B}{t}} \capre \tcol{\cast{A'}{B'}{t'}}}[Diag-Cast] \ilabel{infrule:capre-diag-cast} \and
		\inferrule{\tcol{\GG} \vdash \tcol{\orr{a}} \capre \tcol{\orr{a'}} \\ i = i'}
		{\tcol{\GG} \vdash \tcol{I\ulev{i}(\orr{a})} \capre \tcol{I\ulev{i'}(\orr{a'})}}[Diag-Ind] \ilabel{infrule:capre-ind} \and
		\inferrule{\tcol{\GG} \vdash \tcol{\orr{a}} \capre \tcol{\orr{a'}} \\ \tcol{\GG} \vdash \tcol{\orr{b}} \capre \tcol{\orr{b'}} \\ i = i'}
		{\tcol{\GG} \vdash \tcol{c\ulev{i}(\orr{a},\orr{b})} \capre \tcol{c\ulev{i'}(\orr{a},\orr{b})}}[Diag-Cons] \ilabel{infrule:capre-cons} \and
		\inferrule{
			\tcol{\GG} \vdash \tcol{s} \capre \tcol{s'} \\
			\tcol{\fs{\GG}} \vdash \tcol{s} \pcheckty{I} \tcol{I(\orr{a})} \\
			\tcol{\sn{\GG}} \vdash \tcol{s'} \pcheckty{I} \tcol{I(\orr{a'})} \\
			\tcol{\GG, z : I(\orr{a}) \mid I(\orr{a'})} \vdash \tcol{P} \capre \tcol{P'} \\
			\tcol{\GG, f : (\P z : I(\orr{a}), P) \mid (\P z : I(\orr{a'}), P'), \orr{y} : \orr{Y_k}\subs{\orr{a}}{\orr{x}} \mid \orr{Y_k}\subs{\orr{a'}}{\orr{x}}} \vdash \tcol{t_k} \capre \tcol{t_k'}}
		{\tcol{\GG} \vdash \tcol{\match{I}{s}{z.P}{\orr{f.y.t}}} \capre \tcol{\match{I}{s'}{z.P'}{\orr{f.y.t'}}}}[Diag-Fix] \ilabel{infrule:capre-fix} \\
		\inferrule{
			\tcol{\fs{\GG}} \vdash \tcol{t} \inferty \tcol{T} \\
			\tcol{\GG} \vdash \tcol{T} \cdpre \tcol{A'} \\
			\tcol{\GG} \vdash \tcol{T} \cdpre \tcol{B'} \\
			\tcol{\GG} \vdash \tcol{t} \capre \tcol{t'}}
		{\tcol{\GG} \vdash \tcol{t} \capre \tcol{\cast{A'}{B'}{t'}}}[Cast-R] \ilabel{infrule:capre-castr} \\
		\inferrule{
			\tcol{\sn{\GG}} \vdash \tcol{t'} \inferty \tcol{T'} \\
			\tcol{\GG} \vdash \tcol{A} \cdpre \tcol{T'} \\
			\tcol{\GG} \vdash \tcol{B} \cdpre \tcol{T'} \\
			\tcol{\GG} \vdash \tcol{t} \capre \tcol{t'}}
		{\tcol{\GG} \vdash \tcol{\cast{A}{B}{t}} \capre \tcol{t'}}[Cast-L] \ilabel{infrule:capre-castl} \\
		\inferrule{\tcol{\fs{\GG}} \vdash \tcol{t} \inferty \tcol{T} \\ \tcol{\GG} \vdash \tcol{T} \cdpre \tcol{T'}}
		{\tcol{\GG} \vdash \tcol{t} \capre \tcol{\?_{T'}}}[Unk] \ilabel{infrule:capre-unk} \and
		\inferrule{\tcol{\fs{\GG}} \vdash \tcol{A} \pcheckty{[]} \tcol{[]_i} \\ i \leq j}
		{\tcol{\GG} \vdash \tcol{A} \capre \tcol{\?_{[]_j}}}[Unk-Univ] \ilabel{infrule:capre-unk-univ} \and
    \inferrule{\tcol{\sn{\GG}} \vdash \tcol{t'} \inferty \tcol{T'} \\ \tcol{\GG} \vdash \tcol{T} \cdpre \tcol{T'}}
    {\tcol{\GG} \vdash \tcol{\err_{T}} \capre \tcol{t'}}[Err] \ilabel{infrule:capre-err} \and
    \inferrule{\tcol{\fs{\GG}} \vdash \tcol{t'} \pcheckty{\Pi} \tcol{\P x : A' . B'} \\ \tcol{\GG} \vdash \tcol{\P x : A . B} \cdpre \tcol{\P x : A' . B'}}
    {\tcol{\GG} \vdash \tcol{ \l x : A . \err_{B}} \capre \tcol{t'}}[Err-Lambda] \ilabel{infrule:capre-err-lambda}
	\end{mathpar}
	\boxedrule{$\tcol{\GG} \vdash \tcol{t} \cdpre \tcol{t'}$}
	\begin{mathpar}
		\hspace{6em}
		\inferrule{\tcol{\GG} \vdash \tcol{t} \capre \tcol{t'}}{\tcol{\GG} \vdash \tcol{t} \cdpre \tcol{t'}} \and
		\inferrule{\tcol{\GG} \vdash \tcol{s} \cdpre \tcol{t'} \\ \tcol{t} \redCCIC \tcol{s}}{\tcol{\GG} \vdash \tcol{t} \cdpre \tcol{t'}} \and
		\inferrule{\tcol{\GG} \vdash \tcol{t} \cdpre \tcol{s'} \\ \tcol{t'} \redCCIC \tcol{s'}}{\tcol{\GG} \vdash \tcol{t} \cdpre \tcol{t'}}
	\end{mathpar}

	\caption{Structural precision in \CCIC}
	\label{fig:apre-ccic}
\end{figure}

Similarly to $\acons$, precision can ignore some casts, in order to handle casts that might appear or disappear in one term but not the other
during reduction. But in order to control what
casts can be ignored, we impose some restriction on the types
involved. In particular, we want to ensure that ignored casts
would not have raised an error: \eg we want to prevent $\tcol{0} \capre \tcol{\cast{\nat}{\bool}{0}}$.
Thus the definition of structural precision relies on typing, and to do this we need to record the contexts of the two compared terms. We do so by using double-struck letters to denote contexts where each variable is given two types, writing $\tcol{\GG , x : A \mid  A'}$ for context extensions. We use $\tcol{\GG_i}$ for projections, \ie $\tcol{\fs{(\GG, x : A \mid A')}} := \tcol{\fs{\GG}, x : A}$, and write $\tcol{\Gamma \mid \Gamma'}$ for the converse pairing operation.

\begin{definition}[Structural and definitional precision in \CCIC]~

\emph{Structural precision}, denoted $\tcol{\GG} \vdash \tcol{t}
\capre \tcol{t'}$, is defined in \cref{fig:apre-ccic}, mutually with
\emph{definitional precision},
denoted $\tcol{\GG} \vdash \tcol{t} \cdpre \tcol{t'}$,
which is its closure by reduction.
We write $\tcol{\Gamma} \capre \tcol{\Gamma'}$ and $\tcol{\Gamma} \cdpre \tcol{\Gamma'}$ for the pointwise extensions of those to contexts.
\end{definition}
Although $\tcol{\GG} \vdash \tcol{t}\cdpre \tcol{t'}$ is defined in a stepwise way, it is equivalent to the existence of $\tcol{s}$ and $\tcol{s'}$
such that $\tcol{t} \rtred \tcol{s}$, $\tcol{t'} \rtred \tcol{s'}$ and
$\tcol{\GG} \vdash \tcol{s} \capre \tcol{s'}$.
The situation is the same as for consistency (\resp conversion), which
is the closure by reduction of $\alpha$-consistency (\resp $\alpha$-equality).
However, here definitional precision is also used in
the definition of structural precision, in order to permit
computation in types---recall that in a dependently-typed setting 
the two types involved in a cast may need to reduce before the cast itself can
reduce---and thus the two notions are mutually defined.

Let us now explain the rules defining structural precision. Diagonal rules are completely structural, apart from the \nameref{infrule:capre-fix} rule, where typing assumptions provide us with the contexts needed to compare the predicates. More interesting are the non-diagonal rules.
First, $\tcol{\?_{T}}$ is greater than any term of the ''right type''. This incorporates loss of precision (rule \nameref{infrule:capre-unk}), and accommodates for a small bit of cumulativity (rule \nameref{infrule:capre-unk-univ}). This is needed because of technical reasons linked with possibility to form products between types at different levels.
On the contrary, the error is smaller than any term (rule \nameref{infrule:capre-err}), even in its extended form on $\Pi$-types (rule \nameref{infrule:capre-err-lambda}), with a typing premise similar to that of rule \nameref{infrule:capre-unk}.
Finally, casts on the right-hand side can be ignored as long as they
are performed on types that are \emph{less} precise than the type of the
term on the left (rule \nameref{infrule:capre-castr}). Dually, casts on the
left-hand side can be ignored as long as they are performed on types
that are \emph{more} precise than the type of the term on the right
(rule \nameref{infrule:capre-castl}).

\paragraph{Catch-up lemmas}
The fact that structural precision is a simulation relies on a series of lemmas that all have the same form: under the assumption that a term $\tcol{t'}$ is less precise than a term $\tcol{t}$ with a known head ($\tcol{[]}$, $\tcol{\P}$, $\tcol{I}$, $\tcol{\lambda}$ or $\tcol{c}$), the term $\tcol{t'}$ can be reduced to a term that either has the same head, or is some $\tcol{\?}$. We call these {\em catch-up} lemmas, as they enable the less precise term to catch up to the more precise one whose head is already known. Their aim is to ensure that casts appearing in a less precise term never block reduction, as they can always be reduced away.

The lemmas are established in a descending fashion: first, on the
universe in \cref{lem:catchup-univ}, then on other types in
\cref{lem:catchup-type}, and finally on terms, namely on $\lambda$-abstractions in \cref{lem:catchup-lambda} and inductive constructors in \cref{lem:catchup-cons}. Each time, the previously proven catch-up lemmas are used to reduce types in casts appearing in the less precise term, apart from \cref{lem:catchup-univ}, where the induction hypothesis of the lemma being proven is used instead.

\begin{lemma}[Universe catch-up]~\\
	\label{lem:catchup-univ}
	Under the hypothesis that $\tcol{\fs{\GG}} \capre \tcol{\sn{\GG}}$, if $\tcol{\GG} \vdash \tcol{[]_i} \cdpre \tcol{T'}$ and $\tcol{\sn{\GG}} \vdash \tcol{T'} \pcheckty{[]} \tcol{[]_{j}}$, either $\tcol{T'} \rtred \tcol{\?_{[]_{j}}}$ with $i < j$, or $\tcol{T'} \rtred \tcol{[]_i}$.
\end{lemma}

\begin{lemma}[Types catchup]
	\label{lem:catchup-type}
	Under the hypothesis that $\tcol{\fs{\GG}} \capre \tcol{\sn{\GG}}$, we have the following:
	\begin{itemize}
		\item if $\tcol{\GG} \vdash \tcol{\?_{[]_i}} \capre \tcol{T'}$ and $\tcol{\sn{\GG}} \vdash \tcol{T'} \pcheckty{[]} \tcol{[]_{j}}$, then $\tcol{T'} \rtred \tcol{\?_{[]_{j}}}$ and $i \leq j$;
		\item if $\tcol{\GG} \vdash \tcol{\P x : A.B} \capre \tcol{T'}$, $\tcol{\fs{\GG}} \vdash \tcol{\P x : A. B} \inferty \tcol{[]_i}$ and $\tcol{\sn{\GG}} \vdash \tcol{T'} \pcheckty{[]} \tcol{[]_j}$ then either $\tcol{T'} \rtred \tcol{\?_{[]_j}}$ and $i \leq j$, or $\tcol{T'} \rtred \tcol{\P x : A' . B'}$ for some $\tcol{A'}$ and $\tcol{B'}$ such that $\tcol{\GG} \vdash \tcol{\P x : A . B} \capre \tcol{\P x :A' . B'}$;
		\item if $\tcol{\GG} \vdash \tcol{I(\orr{a})} \capre \tcol{T'}$, $\tcol{\fs{\GG}} \vdash \tcol{I(\orr{a})} \inferty \tcol{[]_i}$ and $\tcol{\sn{\GG}} \vdash \tcol{T'} \pcheckty{[]} \tcol{[]_j}$ then either $\tcol{T'} \rtred \tcol{\?_{[]_j}}$ and $i \leq j$, or $\tcol{T'} \rtred \tcol{I(\orr{a'})}$ for some $\tcol{a'}$ such that $\tcol{\GG} \vdash \tcol{I(\orr{a})} \capre \tcol{I(\orr{a'})}$.
	\end{itemize}
\end{lemma}

\begin{lemma}[$\lambda$-abstraction catch-up]~\\
	\label{lem:catchup-lambda}
	If $\tcol{\GG} \vdash \tcol{\l x : A . t} \capre \tcol{s'}$, where $\tcol{t}$ is not an error, $\tcol{\fs{\GG}} \vdash \tcol{\l x : A. t} \inferty \tcol{\P x : A.B}$ and $\tcol{\sn{\GG}} \vdash \tcol{s'} \pcheckty{\Pi} \tcol{\P x : A'. B'}$, then $\tcol{s'} \rtred \tcol{\l x : A'.t'}$ with $\tcol{\GG} \vdash \tcol{\l x : A . t} \capre \tcol{\l x : A'. t'}$.

  This holds in \CCICP, \CCICs, and for terms without $\tcol{\?}$ in \CCICT.
\end{lemma}

\begin{lemma}[Constructors and inductive error catch-up]~\\
  \label{lem:catchup-cons}
  If $\tcol{\GG} \vdash \tcol{c(\orr{a},\orr{b})} \capre \tcol{s'}$, $\tcol{\fs{\GG}} \vdash \tcol{c(\orr{a},\orr{b})} \inferty \tcol{I(\orr{a})}$ and $\tcol{\sn{\GG}} \vdash \tcol{s'} \pcheckty{I} \tcol{I(\orr{a'})}$, then either $\tcol{s'} \rtred \tcol{\?_{I(\orr{a'})}}$ or $\tcol{s'} \rtred \tcol{c(\orr{a'},\orr{b'})}$ with $\tcol{\GG} \vdash \tcol{c(\orr{a},\orr{b})} \capre \tcol{c(\orr{a'},\orr{b'})}$.

\noindent Similarly, if $\tcol{\GG} \vdash \tcol{\?_{I(\orr{a})}} \capre \tcol{s'}$, $\tcol{\fs{\GG}} \vdash \tcol{\?_{I(\orr{a})}} \inferty \tcol{I(\orr{a})}$ and $\tcol{\sn{\GG}} \vdash \tcol{s'} \pcheckty{I} \tcol{I(\orr{a'})}$, then $\tcol{s'} \rtred \tcol{\?_{I(\orr{a'})}}$ with 
\mbox{$\tcol{\GG} \vdash \tcol{I(\orr{a})} \cdpre \tcol{I(\orr{a'})}$}.
\end{lemma}

Note that for \cref{lem:catchup-cons}, we need to deal with unknown terms
specifically, which is not necessary for \cref{lem:catchup-lambda}
because the unknown term in a $\Pi$-type reduces to a $\lambda$-abstraction.

\cref{lem:catchup-lambda} deserves a more extensive discussion, because it is the critical point where the difference between the three variants of \CCIC manifests. In fact, it does not hold in full generality for \CCICT. Indeed, the fact that $i \leq \castOfPi{\sortOfPi{i}{j}}$ and $j \leq \castOfPi{\sortOfPi{i}{j}}$ is used crucially to ensure that casting from a $\Pi$-type into $\tcol{\?}$ and back does not reduce to an error, given the restrictions on types in \nameref{infrule:capre-castr}. This is the manifestation in the reduction of the embedding-projection property~\cite{newAhmed:icfp2018}. In \CCICT it holds only if one restricts to terms without $\tcol{\?}$, where such casts never happen. This is important with regard to conservativity, as elaboration produces terms with casts but without $\tcol{\?}$, and \cref{lem:catchup-lambda} ensures that for those precision is still a simulation, even in \CCICT.

\begin{example}[Catch-up of $\lambda$-abstraction]
  \ilabel{ex:lambda-abstr-catch-up}
  The following term $\tcol{t_i}$ illustrates these differences
\[\tcol{t_i} := \tcol{\cast{\nat \to \nat}{\nat \to \nat \Leftarrow \?_{[]_i}}{\l x : \nat . \natsuc(x)}}\]
where $\tcol{\nat}$ is taken at the lowest level, \ie to mean $\tcol{\nat\ulev{0}}$. Such terms appear naturally whenever a loss of precision happens on a function, for instance when elaborating a term like $\scol{\left(\asc{(\l x : \nat . \natsuc(x))}{\?}\right)~0}$. Now this term $\tcol{t_i}$ always reduces to
\[\tcol{\cast{\stalkCIC{i}{\Pi} \Leftarrow \?_{[]_i} \Leftarrow \stalkCIC{i}{\Pi} \Leftarrow \nat \to \nat}{\nat \to \nat}{\l x : \nat . \natsuc(x)}}\]
and at this point the difference kicks in: if $\tcol{\stalkCIC{i}{\Pi}}$ is $\tcol{\err_{\?_{[]_i}}}$ (\ie if $\castOfPi{i} < 0$) then the whole term reduces to $\tcol{\err_{\nat \to \nat}}$. Otherwise, further reductions finally give
\[\tcol{\l x : \nat . \natsuc\left(\cast{\nat \Leftarrow \nat}{\nat}{x}\right) }\]
Although the body is blocked by the variable $\tcol{x}$, applying the function to $\tcol{0}$ would reduce to $\tcol{1}$ as expected. Let us compare what happens in the three systems.

In all of them, if $i \geq 1$, we have $\vdash \tcol{\l x : \nat . \natsuc(x)} \capre \tcol{t_i}$ via repeated uses of \nameref{infrule:capre-castr} since $\inferelab{}{\nat}{}{}{\nat}{[]_{\sortOfPi{0}{0}}}$ and $\sortOfPi{0}{0} \leq 1 \leq i$.
Moreover, also $0 \leq i - 1 \leq \castOfPi{i}$ and so the reduction is errorless.
Thus \cref{lem:catchup-lambda} holds in all three systems when $i \geq 1$.

The difference appears in the specific case where $i = 0$.
In \CCICP and \CCICT, we still have $\vdash \tcol{\l x : \nat . \natsuc(x)} \capre \tcol{t_0}$, since $\sortOfPi{0}{0} = 0 \leq i$. In the former, $\castOfPi{0} = 0$ so $\tcol{t_0}$ reduces safely and \cref{lem:catchup-lambda} holds. In the latter, however, $\castOfPi{0} = -1$,
and so $\tcol{t_0}$ errors even if it is less precise than an errorless term---\cref{lem:catchup-lambda} does not hold in that case.
Finally, in \CCICs, $\tcol{t_0}$ errors since again $\castOfPi{0} = -1$.
However, because $s_{\Pi}(0,0) = 1$, $\tcol{t_0}$ is not less precise than $\tcol{\l x : \nat. \natsuc(x)}$ thanks to the typing restriction in \nameref{infrule:capre-castr}, so this error does not contradict \cref{lem:catchup-lambda}.

\end{example}

Note that in an actual implementation with typical ambiguity (\cref{note:typ-amb}),
the case where $i = 0$ would most likely not manifest: elaborating
$\scol{\left(\asc{(\l x : \nat . \natsuc(x))}{\?}\right)~0}$ would produce a fresh
level that could be chosen high enough so as to prevent the error we just described.
Only more involved situations like that of $\scol{\Omega}$ (\cref{sec:back-to-omega}) would actually exhibit
failures due to universe levels, which are precisely those unavoidable to ensure normalization.

\paragraph{Simulation}
We finally come to the main property of this section, the advertised
simulation.
Remark that the simulation property needs to be stated (and proven) mutually for structural and definitional precision, but it is really informative only for
structural precision (definitional precision is somehow a
simulation by construction).

\begin{theorem}[Precision is a simulation for reduction]~\\
	\label{thm:simulation}
  Let $\tcol{\fs{\GG}} \cdpre \tcol{\sn{\GG}}$, $\tcol{\fs{\GG}} \vdash \tcol{t}
  \inferty \tcol{T}$, $\tcol{\sn{\GG}} \vdash \tcol{u} \inferty
  \tcol{U}$ and $\tcol{t} \rtred \tcol{t'}$. Then
  \begin{itemize}
  \item
    if $\tcol{\GG} \vdash \tcol{t} \capre \tcol{u}$ then there exists $\tcol{u'}$ such that $\tcol{u} \rtred \tcol{u'}$ and $\tcol{\GG} \vdash \tcol{t'} \capre \tcol{u'}$;
  \item
    if $\tcol{\GG} \vdash \tcol{t} \cdpre \tcol{u}$ then
    $\tcol{\GG} \vdash \tcol{t'} \cdpre \tcol{u}$.
  \end{itemize}
  This holds in \CCICP,
\CCICs and for terms without $\tcol{\?}$ in \CCICT.
\end{theorem}

\begin{proof}[Proof sketch]
  The case of definitional precision follows by confluence of
  reduction.
	For the case of structural precision, the hardest point is to simulate $\beta$ and $\iota$ redexes---terms of the shape $\match{I}{c(\orr{a})}{z.P}{f.\orr{y}.\orr{t}}$. This is where we use \cref{lem:catchup-lambda,lem:catchup-cons}, to show that similar reductions can also happen in $\tcol{t'}$.
	We must also put some care into handling the premises of precision where typing is involved. In particular, subject reduction is needed to relate the types inferred after reduction to the type inferred before, and the mutual induction hypothesis on $\cdpre$ is used to conclude that the premises holding on $\tcol{t}$ still hold on $\tcol{t'}$.
	Finally, the restriction to terms without $\tcol{\?}$ in \CCICT similar to \cref{lem:catchup-lambda} appears again when treating \nameref{redrule:up-down}, where having $\castOfPi{\sortOfPi{i}{i}} = i$ is required.
\end{proof}

From this theorem, we get as direct corollaries the following properties, that are required to handle reduction (\cref{cor:red-types}) and consistency (\cref{cor:mon-cons}) in elaboration. Again those corollaries hold in \GCICP, \GCICs and for terms in \GCICT containing no $\tcol{\?}$.

\begin{corollary}[Monotonicity of reduction to type constructor]~\\
	\label{cor:red-types}
	Let $\tcol{\GG}$, $\tcol{T}$ and $\tcol{T'}$ be such that $\tcol{\fs{\GG}} \vdash \tcol{T}  \pcheckty{[]} \tcol{\square{}_i}$, $\tcol{\sn{\GG}} \vdash \tcol{T'} \pcheckty{[]} \tcol{\square{}_j}$, $\tcol{\GG} \vdash \tcol{T} \capre \tcol{T'}$. Then
	\begin{itemize}
		\item if $\tcol{T} \rtred \tcol{\?_{[]_i}}$ then $\tcol{T'} \rtred \tcol{\?_{[]_j}}$ with $i \leq j$;
		\item if $\tcol{T} \rtred \tcol{\square{}_{i-1}}$ then either $\tcol{T'} \rtred \tcol{\?_{\square{}_j}}$ with $i \leq j$, or $\tcol{T'} \rtred \tcol{\square{}_{i-1}}$;
		\item if $\tcol{T} \rtred \tcol{\P x : A. B}$ then either $\tcol{T'} \rtred \tcol{\?_{[]_j}}$ with $i \leq j$, or $\tcol{T'} \rtred \tcol{\P x : A'. B'}$ and\\ 
    \mbox{$\GG \vdash \tcol{\P x : A.B} \capre \tcol{\P x : A'.B'}$};
		\item if $\tcol{T} \rtred \tcol{I(\orr{a})}$ then either $\tcol{T'} \rtred \tcol{\?_{[]_j}}$ with $i \leq j$, or $\tcol{T'} \rtred \tcol{I(\orr{a'})}$ and $\GG \vdash \tcol{I(\orr{a})} \capre \tcol{I(\orr{a'})}$.
	\end{itemize}
\end{corollary}

 \begin{proof}
 	It suffices to simulate the reductions of $\tcol{T}$ by using \cref{thm:simulation}, and then use \cref{lem:catchup-type,lem:catchup-univ} to conclude. Note that head reductions are simulated using head reductions in \cref{thm:simulation}, and the reductions of \cref{lem:catchup-type,lem:catchup-univ} are also head reductions. Thus the corollary still holds when fixing weak-head reduction as a reduction strategy.
 \end{proof}

\begin{corollary}[Monotonicity of consistency]
	\label{cor:mon-cons}
	If $\tcol{\GG} \vdash \tcol{T} \capre \tcol{T'}$, $\tcol{\GG} \vdash \tcol{S} \capre \tcol{S'}$ and $\tcol{T} \cons \tcol{S}$ then $\tcol{T'} \cons \tcol{S'}$.
\end{corollary}

\begin{proof}
	By definition of $\cons$, we get some $\tcol{U}$ and $\tcol{V}$ such that $\tcol{T} \rtred \tcol{U}$ and $\tcol{S} \rtred
 \tcol{V}$, and $\tcol{U} \acons \tcol{V}$. By \cref{thm:simulation},
 we can simulate these reductions to get some $\tcol{U'}$ and $\tcol{V'}$ such that $\tcol{T'} \rtred \tcol{U'}$ and $\tcol{S'} \rtred \tcol{V'}$, and also $\tcol{\fs{\GG}} \vdash \tcol{U} \capre \tcol{U'}$ and $\tcol{\fs{\GG}} \vdash \tcol{V} \capre \tcol{V'}$. Thus we only need to show that $\alpha$-consistency is monotone with respect to structural precision, which is direct by induction on structural precision.
\end{proof}

\subsection{Properties of \GCIC}
\label{sec:gcic-theorems2}

We now have enough technical tools to prove most of the properties of \GCIC.
We state those theorems in an empty context in this section to make
them more readable, but they are of course corollaries of similar
statements including contexts, proven by mutual induction. The complete
statements and proofs can be found in \cref{sec:properties-gcic}.

\paragraph{Conservativity with respect to \CIC}
Elaboration systematically inserts casts during checking, thus even
static terms are not elaborated to themselves. Therefore we use a
(partial) erasure function $\eras$ that translates terms of \CCIC to terms of \CIC by erasing all casts. We also introduce the notion of erasability, characterizing terms that contain ``harmless'' casts, such that in particular the elaboration of a static term is always erasable.

\begin{definition}[Equiprecision]
  \label{def:equipre}
	Two terms $\tcol{s}$ and $\tcol{t}$ are equiprecise in a context $\tcol{\GG}$, denoted $\tcol{\GG} \vdash \tcol{s} \caequipre \tcol{t}$ if both $\tcol{\GG} \vdash \tcol{s} \capre \tcol{t}$ and $\tcol{\GG} \vdash \tcol{t} \capre \tcol{s}$.
\end{definition}

\begin{definition}[Erasure, erasability]
  \label{def:erasure}
	Erasure $\eras$ is a partial function from the syntax of \CCIC to the syntax of \CIC, which is undefined on $\tcol{\?}$ and $\tcol{\err}$, is such that $\eras(\tcol{\cast{A}{B}{t}}) = \eras(\tcol{t})$, and is a congruence for all other term constructors.

	Given a context $\tcol{\GG}$ we say that a term well-typed in $\tcol{\GG_1}$ $\tcol{t}$ is erasable if $\eras(\tcol{t})$ is defined, well-typed in $\tcol{\GG_2}$, and equiprecise to $\tcol{t}$ in $\tcol{\GG}$. Similarly a context $\tcol{\Gamma}$ is called erasable if it is pointwise erasable. When $\tcol{\Gamma}$ is erasable, we say that a term $\tcol{t}$ is erasable in $\tcol{\Gamma}$ to mean that it is erasable in $\tcol{\Gamma} \mid \eras(\tcol{\Gamma})$.
\end{definition}

Conservativity holds in all three systems, typeability being of course taken into the corresponding variant of \CIC: full \CIC for \GCICP and \GCICT, and \CICs for \GCICs.

\begin{theorem}[Conservativity]
	\label{thm:conservativity}
	Let $\scol{\tilde{t}}$ be a static term (\ie is a term of \GCIC that is also a term of \CIC).
  If $\cicty \scol{\tilde{t}} \inferty T$ for some type $T$,
  then there exists $\tcol{t}$ and $\tcol{T'}$ such that
  $\inferelab{}{\tilde{t}}{}{}{t}{T'}$, and
  moreover $\eras(\tcol{t}) = \scol{\tilde{t}}$ and $\eras(\tcol{T'}) = T$. Conversely if $\inferelab{}{\tilde{t}}{}{}{t}{T}$ for some $\tcol{t}$ and $\tcol{T}$, then $\cicty \scol{\tilde{t}} \inferty \eras(\tcol{T})$.
\end{theorem}

\begin{proof}[Proof sketch]

	Because $\scol{t}$ is static, its typing derivation in \GCIC can only use rules that have a counterpart in \CIC, and conversely all rules of \CIC have a counterpart in \GCIC. The only difference is about the reduction/conversion side conditions, which are used on elaborated types in \GCIC, rather than their non-elaborated counterparts in \CIC.

	Thus, the main difficulty is to ensure that the extra casts
  inserted by elaboration do not alter reduction. For this we
  maintain the property that all terms $\tcol{t}$ considered in
  \CCIC are erasable, and in particular that any static term $\scol{{t}}$
  that elaborates to some $\tcol{t}$ is such that
  $\eras(\tcol{t}) = \scol{{t}}$. From the simulation property of
  structural precision (\cref{thm:simulation}), we get that an
  erasable term $\tcol{t}$ has the same reduction behavior as
  its erasure, \ie if $\tcol{t} \rtred \tcol{s}$ then
  $\eras(\tcol{t}) \rtred s'$ with $s'$ and $\tcol{s}$
  equiprecise, and conversely if $\eras(\tcol{t}) \rtred s'$
  then $\tcol{t} \rtred \tcol{s}$ with $s'$ and $\tcol{s}$
  equiprecise. Using that property, we  prove that
  constraint reductions ($\pcheckty{\Pi}$, $\pcheckty{[]}$ and
  $\pcheckty{I}$) in \CCIC and \CIC behave the same on static terms.
\end{proof}

\paragraph{Elaboration Graduality}

\begin{figure}
	\begin{mathpar}
		\inferrule{ }{\scol{x} \apre \scol{x}} \and
		\inferrule{ }{\scol{[]_i} \apre \scol{[]_i}} \and
		\inferrule{\scol{A} \apre \scol{A'} \\ \scol{B} \apre \scol{B'}}
		{\scol{\P x : A . B} \apre \scol{\P x : A'.B'}} \and
		\inferrule{\scol{A} \apre \scol{A'} \\ \scol{t} \apre \scol{t'}}
		{\scol{\l x : A . t} \apre \scol{\l x : A . t}} \and
		\inferrule{\scol{t} \apre \scol{t'} \\ \scol{u} \apre \scol{u'}}{\scol{t~u} \apre \scol{t'~u'}} \and
		\inferrule{\scol{\orr{a}} \apre \scol{\orr{a'}}}{\scol{I(\orr{a})} \apre \scol{I(\orr{a'})}} \and
		\inferrule{\scol{\orr{a}} \apre \scol{\orr{a'}} \\ \scol{\orr{\scol{b}}} \apre \scol{\scol{\orr{b'}}}}{\scol{c(\orr{a},\orr{b})} \apre \scol{c(\orr{a'},\orr{b'})}} \and
		\inferrule{\scol{a} \apre \scol{a'} \\ \scol{P} \apre \scol{P'} \\ \scol{\orr{t}} \apre \scol{\orr{t'}}}
		{\scol{\match{I}{a}{z.P}{\orr{f.y.t}}} \apre \scol{\match{I}{a'}{z.P'}{\orr{f.y.t'}}}} \and
		\inferrule{ }
		{\scol{t} \apre \scol{\?}}
	\end{mathpar}
	\caption{Syntactic precision for \GCIC}
	\label{fig:apre-gcic}
\end{figure}

Next, we turn to elaboration graduality, the equivalent of the static gradual
guarantee (SGG) of \citet{siekAl:snapl2015} in our setting.
We state it with respect to a notion of precision for terms in \GCIC,
\emph{syntactic precision} $\apre$, defined in \cref{fig:apre-gcic}.
Syntactic precision is the usual and expected source-level notion of precision in gradual languages: it is generated by a single non-trivial rule $\scol{t} \apre
\scol{\?\ulev{i}}$, and congruence rules for all term formers.

In contrast with the simply-typed setting, the presence of multiple unknown types $\scol{\?}$,
one for each universe level $\scol{i}$, requires an additional hypothesis relating elaboration
and precision.
We say that two judgments $\scol{\tilde{t}} \apre \scol{\?\ulev{i}}$ and
$\sinferelab{}{\tilde{t}}{}{\Gamma}{t}{T}$ are \emph{universe adequate} if the
universe level $j$ given by the well-formedness judgment $\tcol{\Gamma} \vdash \tcol{T}
\pcheckty{\square} \tcol{\square_j}$ induced by correction of the elaboration satisfies $i = j$.
More generally, $\scol{\tilde{t}} \apre \scol{\tilde{s}}$ and
$\sinferelab{}{\tilde{t}}{}{}{t}{T}$ are \emph{universe adequate} if for any
subterm $\scol{\tilde{t}_0}$ of $\scol{\tilde{t}}$ inducing judgments $\scol{\tilde{t}_0}
\apre \scol{\?\ulev{i}}$ and $\sinferelab{}{\tilde{t}_0}{}{\Gamma_0}{t}{T}$, those are
universe adequate.
Note that this extraneous technical assumption on universe levels is not needed if we use typical ambiguity (\cref{note:typ-amb}),
since universe levels are not given explicitly.

\begin{theorem}[Elaboration Graduality / Static Gradual Guarantee]
	\label{thm:static-graduality}
  In \GCICP and \GCICs, if $\scol{\tilde{t}} \apre \scol{\tilde{s}}$ and
  $\sinferelab{}{\tilde{t}}{}{}{t}{T}$ are universe adequate, then $\sinferelab{}{\tilde{s}}{}{}{s}{S}$ for some $\tcol{s}$ and $\tcol{S}$ such that $\vdash \tcol{t} \capre \tcol{s}$ and $\vdash \tcol{T} \capre \tcol{S}$.
\end{theorem}

\begin{proof}[Proof sketch]
	The proof is by induction on the elaboration derivation for $\scol{\tilde{t}}$. All cases for inference consist in a straightforward combination of the hypotheses, with the universe adequacy hypothesis used in the case where $\scol{\tilde{s}}$ is $\scol{\?\ulev{i}}$.
	Here again the technical difficulties arise in the rules involving reduction. This is where \cref{cor:red-types} is useful, proving that the less structurally precise term obtained by induction in a constrained inference reduces to a less precise type. Thus either the same rule can still be used, or one has to trade a \nameref{infrule:gcic-inf-univ}, \nameref{infrule:gcic-inf-prod} or \nameref{infrule:gcic-inf-ind} rule respectively for a \nameref{infrule:gcic-univ-unk}, \nameref{infrule:gcic-prod-unk} or \nameref{infrule:gcic-ind-unk} rule in case the less precise type is some $\tcol{\?_{[]_i}}$ and the more precise type is not. Similarly, \cref{cor:mon-cons} proves that in the checking rule the less precise types are still consistent.
  Note that again, because \cref{cor:red-types} holds when restricted to weak-head reduction, elaboration graduality also holds when fixing a weak-head strategy in \cref{fig:elaboration}.
\end{proof}

\paragraph{Dynamic Gradual Guarantee}

Following \citet{siekAl:snapl2015}, using the fact that structural precision is
a simulation (\cref{thm:simulation}), we can prove the DGG for \CCICP
and \CCICs (stated using the notion of observational refinement
$\obsRef$ from \cref{def:obsref}).

\begin{theorem}[Dynamic Gradual Guarantee for \CCICP and \CCICs]
  \label{thm:dgg}
  Suppose that $\tcol{\Gamma} \vdash \tcol{t}
  \inferty \tcol{A}$ and $\tcol{\Gamma} \vdash \tcol{u} \inferty
  \tcol{A}$.
  If moreover $\tcol{\Gamma \mid \Gamma} \vdash \tcol{t} \capre \tcol{u}$ then $\tcol{t} \obsRef \tcol{u}$.
\end{theorem}
\begin{proof}
  Let $\tcol{\mathcal{C}} : (\tcol{\Gamma} \vdash \tcol{A}) \Rightarrow (\vdash \tcol{\bool})$
  closing over all free variables.
  By the diagonal rules of structural precision, we have
  $\tcol{\Gamma \mid \Gamma} \vdash \tcol{\mathcal{C}[t]} \capre \tcol{\mathcal{C}[u]}$.
  By progress (\cref{thm:ccic-psafe}), $\tcol{\mathcal{C}[t]}$ either reduces
  to $\tcol{\btrue}$, $\tcol{\bfalse}$, $\tcol{\?_{\bool}}$, $\tcol{\err_{\bool}}$ or diverges, and similarly for $\tcol{\mathcal{C}[u]}$.
  If $\tcol{\mathcal{C}[t]}$ diverges or reduces to $\tcol{\err_\bool}$, we are done.
  If it reduces to either $\tcol{\btrue}$, $\tcol{\bfalse}$ or $\tcol{\?_\bool}$, then by the catch-up \cref{lem:catchup-cons},
  $\tcol{\mathcal{C}[u]}$ either reduces to the same value, or
  to $\tcol{\?_\bool}$. In particular, it cannot diverge or reduce to
  an error.
\end{proof}

Note that \cref{ex:lambda-abstr-catch-up} provides a counter-example to this theorem for \CCICT, by choosing the context $\tcol{\match{\nat}{\bullet~0}{z.\bool}{f.\btrue,f.n.\btrue}}$, because in that context the function $\tcol{\l x : \nat . \suc(x)}$ reduces to $\tcol{\btrue}$ while the less precise casted function reduces to $\tcol{\err_{\bool}}$.

As observed in \cref{sec:grad-simple}, 
graduality---and in particular the fact that precision induces \eppairs{}---is inherently semantic, and thus cannot rely on the syntactic precision $\apre$ introduced in this section.
Therefore, we defer the proof of \pgrad for \CCICs and \CCICP to the next
section, where the semantic notion of {\em propositional precision} is introduced.

\section{Realizing \CCIC and Graduality}
\label{sec:realizing-cast-calculus}
\label{sec:realizing-ccic}
To prove normalization of \CCICT and \CCICs, we now build a model of both
theories with a simple implementation of casts using case-analysis on types as well
as exceptions, yielding the \emph{\bareModel}, allowing us to reduce the
normalization of both theories to the normalization of the target theory
(\cref{sec:bare-model}).

Then, to prove graduality of \CCICs,
we build a more elaborate \emph{monotone model} inducing a
precision relation well-behaved with respect to conversion.
Following generalities about the interpretation of \CIC's types as posets
in~\cref{sec:poset-model-dtt}, we describe the construction of a monotone
unknown type $\unkMon$ in~\cref{sec:realizing-unknown-type} and a hierarchy of
universes in~\cref{sec:monotone-universe} and put these pieces together
in~\cref{sec:monotone-model}, culminating in a proof of graduality
for \CCICs (\cref{sec:grad-term}).
In both the discrete and monotone case, the parameters $\castOfPi{-}$ and
$\sortOfPi{-}{-}$
appear when building the hierarchy of universes and tying the knot with the
unknown type.

Finally, to deduce graduality for the non-terminating variant, \CCICP,
we describe at the end of this section a model based on
$\omega$-complete partial orders, extending the seminal
model of \citet{scott76} for $\lambda$-calulus to \CCICP (\cref{sec:grad-non-term}).

The discrete model embeds into a variant of \CIC extended with
induction-recursion~\cite{DybjerS03}, noted \CICIR, and the monotone model into a variant that additionally features quotients (and hence also function extensionality~\cite{ShulmanIntervalFunext}), noted \CICIRQ.

\paragraph{Formalization in \Agda}
We use \Agda{}~\cite{norell:afp2008} as a practical tool to typecheck
the components of the models and assume that \Agda{} satisfies
standard metatheoretical properties, namely subject reduction and
strong normalization.

The correspondence between the notions developed in the following
sections and the formal development in \Agda{}~\cite{GCICAgdaCode} is as follows.
The formalization covers most component of the discrete
(\texttt{DiscreteModelPartial.agda}) and monotone model
(\texttt{UnivPartial.agda}) in a partial (non-normalizing) setting and only
the discrete model is proved to be normalizing assuming normalization of the
type theory implemented by Agda (no escape hatch to termination checking is
used in \texttt{DiscreteModelTotal}).
The main definitions surrounding posets can be found in \texttt{Poset.agda}:
top and bottom elements (called \texttt{Initial} and \texttt{Final} in the
formalization), embedding-projection pairs (called \texttt{Distr}) as well as
the notions corresponding to indexed families of posets
(\texttt{IndexedPoset}, together with \texttt{IndexedDistr}). It is then
proved that we endow can the translation of each type
formers from \CCIC{} with a poset structure: natural numbers in \texttt{nat.agda}, booleans in
\texttt{bool.agda}, dependent product in \texttt{pi.agda}. The definition of
the monotone unknown type \unkMon{} is defined in the subdirectory
\texttt{Unknown/}. It is more involved since we need to use a
quotient (that we axiomatize together with a rewriting rule in
\texttt{Unknown/Quotient.agda}).
Finally, all these building blocks are put together when assembling the
inductive-recursive hierarchies of universes (\texttt{UnivPartial.agda},
\texttt{DiscreteModelPartial.agda} and \texttt{DiscreteModelTotal.agda}).

\subsection{Discrete Model of \CCIC}
\label{sec:bare-model}

The discrete model explains away the new term formers of
\CCIC~(\ref{fig:syntax-castcic}) by a translation into \CIC{} using two important ingredients from the literature:
\begin{itemize}
\item Exceptions, following the approach of \ExTT~\cite{pedrotTabareau:esop2018}: each inductive type is extended with two new constructors, one for $\?$ and one for $\err$. 
As alluded to early on (\cref{sec:graduality}), both $\?$ and $\err$ are exceptional terms in their propagation semantics, and only differ in their static interpretation: \coqe{?_A} is consistent with any other term of type \coqe{A}, while \coqe{err_A} is not consistent with any such term.

\item Case analysis on types~\cite{BoulierPT17} to define the cast operator. The
  essence of the translation is to interpret types as {\em codes} when they are
  seen as terms, and as {\em the semantics of those codes} when they are seen as
  types.
  This allows us to get the standard interpretation for a term inhabiting
  a type, but at the same time, it allows functions taking terms in the universe
  $\U{}_i$ to perform a case analysis on the code of the type, because this
  time, the type is seen as a term in $\U{}_i$.
\end{itemize}

The latter ingredient for intensional type
analysis requires the target theory of the translation to be an extension of
\CIC with induction-recursion~\cite{DybjerS03}, noted \CICIR. We write $\redIR$
and $\irty$ to denote the reduction and typing judgments of \CICIR,
respectively.

\paragraph{Inductive types}
Following the general pattern of \ExTT, we interpret each inductive type $I$ by
an inductive type $\liftErr{I}$ featuring all constructors of $I$
and extended with two new constructors
$\lifttop{\liftErr{I}}$ and $\liftbot{\liftErr{I}}$, corresponding respectively to
$\?_{I}$ and $\err_I$ of \CCIC.
The constructors $\lifttop{\liftErr{I}}$ and $\liftbot{\liftErr{I}}$ of
$\liftErr{I}$ are called \emph{exceptional} by opposition to the other
constructors that we call \emph{non-exceptional}.
For instance, the inductive type used to interpret natural numbers,
$\liftNat{}$, thus has 4 constructors: the non-exceptional constructors $0$
and $\suc$, and the exceptional 
constructors $\lifttop{\liftNat}$, $\liftbot{\liftNat}$.
In the rest of this section, we only illustrate inductive types on natural numbers.

\paragraph{Universe and type-case}
Case analysis on types is obtained through an explicit inductive-recursive description of
the universes \cite{itt,McBride10} to build a type of codes $\U_i$
described in \cref{fig:discrete-univ}.
Codes are noted with $\code{\cdot}$ and the universe type contains codes for
dependent product ($\PiU$), universes ($\uU_j$), inductive types (\eg $\natU$) as well as $\unkU$ for the unknown type and $\errU$ for the error type. The main subtlety here is that the code $\PiU\,A\,B$ is at level $\sortOfPi{i}{j}$
when $A$ is at $i$ and $B$ is a family at $j$, emulating the rule of
\cref{fig:ccic-typing}.
Accompanying the inductive definition of $\U_i$, the recursively defined
decoding function $\El$ provides a semantics for these codes.
The semantics of $\PiU$ is given by the dependent product in the target theory,
applying $\El$ on the domain and the codomain of the code.  
The semantics of $\uU_j$ is precisely the type of codes $\U_j$.
The semantics of $\natU$ is given by the extended natural numbers $\liftNat$, explained above.

Intuitively, the semantics of $\unkU[i]$ is that an inhabitant of 
the unknown type corresponds to a pair of a type and an inhabitant of
that type. More precisely, we first define a notion of germ for codes where we stratify the head constructors $\H$ (see~\cref{fig:head-germ}) according to the
universe level $i$, e.g.
$\stalkCode{i}{\Pi} := \PiU\,\uU_{\castOfPi{i}}\,(\lambda (x :
\U_{\castOfPi{i}}).\uU_{\castOfPi{i}})$ when $\castOfPi{i} \geq 0$,
and its decoding to types
$\stalkCIC{i}{h} := \El\,(\stalkCode{i}{h})$.
The unknown type $\unkU[i]$ is then decoded to the extended dependent sum
$\liftSigma \, \C_i \, \stalkCIC{i}{}$ whose elements are
either:
\begin{itemize}
\item one of the two freely added constructors $\lifttop{\liftSigma},
  \liftbot{\liftSigma}$ following the interpretation scheme of inductive types;
\item or a dependent pair $\pairSigma{h}{t}$ of a head constructor $h \in \H_i$
  together with an element $t \in \stalkCIC{i}{h}$.
  \end{itemize}
Finally, the error type $\errU[i]$ is decoded to the unit type $\unit$ containing a unique element
\unitK. 
\begin{figure}
  \begin{small}
    \begin{mathpar}
      \inferrule{A \in \U_i\\ B \in \El\, A \to \U_j}{\PiU\,A\,B
        \in \U_{\sortOfPi{i}{j}}} \and
      \inferrule{j < i}{\uU_j \in \U_i} \and
      \natU \in \U_i \and
      \unkU[i] \in \U_i \and
      \errU[i] \in \U_i
    \end{mathpar}
    \begin{align*}
      \El\,(\PiU\,A\,B) &= \Pi(a : \El\,A)~\El (B\,a)
      & \El\,\uU_j &= \U_j
      & \El\,\natU &= \liftNat{}
      & \El\,\unkU[i] &= \liftSigma \, \C_i \, \stalkCIC{i}{}
      & \El\,\errU[i] &= \unit
    \end{align*}
  \end{small}
  \caption[IR universe & exceptions]{Inductive-recursive encoding of the
    discrete universe hierarchy}
  \label{fig:discrete-univ}
\end{figure}
\paragraph{Variants of \CCIC}
Crucially, the code for $\Pi$-types (\cref{fig:discrete-univ}) depends on the choice made for $\sortOfPi{i}{j}$.
Observe that for the choice of parameters corresponding to \CCICP, the
inductive-recursive definition of $\U_i$ is ill-founded since
$\castOfPi{\sortOfPi{i}{i}} = \sortOfPi{i}{i}$. We can thus inject
$\stalk{}_{\sortOfPi{i}{i}}\Pi = \El\,\unkU{\sortOfPi{i}{i}} \to
\El\,\unkU{\sortOfPi{i}{i}}$ into $\El\,\unkU{\sortOfPi{i}{i}}$ and project
back in the other direction, exhibiting an embedding-retraction
suitable to interpret the untyped $\lambda$-calculus and hence $\Omega$.%
\footnote{In the \Agda{} implementation, we deactivate the termination checker
on the definition of the universe for the model interpreting \CCICP,
thus effectively working in a partial, inconsistent type theory.}
In order to maintain normalization, the construction of the unknown type and the
universe therefore needs to be stratified, which is possible when
$\castOfPi{\sortOfPi{i}{i}} < \sortOfPi{i}{i}$.
This strict inequality occurs for both \CCICT and \CCICs.
We then proceed by strong induction on the universe level, and note that thanks to
the level gap, the decoding $\El\,\unkU[i]$ of the unknown type at a level $i$
can be defined solely from the data of smaller universes available by inductive
hypothesis, without any reference to $\U_i$.
We can then define the rest of the universe $\U_i$ and the decoding function
$\El$ at level $i$ in a well-founded manner, validating the strict positivity
criterion of \Agda's termination checker.
\begin{figure}
  \begin{small}
    \begin{align*}
      \?_{\PiU{}\,A\,B} &:= \l x : \El\,A. \?_{\El\,(B\,x)} 
      & \?_{\uU_j} &:= \unkU[j] 
      &\?_{\natU} &:= \lifttop{\liftNat} 
                        & \?_{\unkU[j]} &:= \lifttop{\El\,\uU_j}
                   & \?_{\errU[j]} &:= \unitK 
      \\
      \err_{\PiU{}\,A\,B} &:= \l x : \El\,A. \err_{\El\,(B\,x)} 
      & \err_{\uU_j} &:= \errU[j] 
      &\err_{\natU} &:= \liftbot{\liftNat}
                        &\err_{\unkU[j]} &:= \liftbot{\El\,\uU_j}
                   & \err_{\errU[j]} &:= \unitK 
    \end{align*}
  \end{small}
	\caption{Realization of exceptions}
	\label{fig:exceptions}
\end{figure}
\paragraph{Exceptions}
The definition of exceptions $\?_{A}, \err_{A} : \El~A$ at an arbitrary code $A$ then
follows by case analysis on the code, as shown in \cref{fig:exceptions}.
On the code for the universe, $\uU_j$, we directly use the code for the unknown and the error
types respectively.
On codes that have an inductive interpretation---$\natU$,
$\unkU[i]$---we use the two added constructors.
On the code for dependent functions, exceptions are defined by re-raising the exception at the
codomain in a pointwise fashion.
Finally, on the error type $\errU$, exceptions are degenerated and
forced to take the only value $\unitK:\unit$\footnote{This definition is indeed uniform if
  $\unit$ is seen as the record type with no projection.} of its interpretation as a
type.
\begin{figure}
  \[
    \begin{array}[t]{lcl}
      \cas~(\PiU\,A^\pidom\,A^\picod)~(\PiU\,B^\pidom\,B^\picod)~f & := & \l b :
                                            \El\,B^\pidom.~\letin{a}{\cas~B^\pidom~A^\pidom~b}{\cas\,
                                            (A^\picod\,a)~(B^\picod\,b)~(f\,a)} \\
      \cas~(\PiU\,A^\pidom\,A^\picod)~\unkU[i]~f & := & (\Pi ;
                                                      \cas~(\PiU\,A^\pidom\,A^\picod)~(\stalkCode{i}{\Pi})~f)\qquad
                                                      \text{if }
                                                       \stalkCode{i}{\Pi} \neq \errU\\
      \cas~(\PiU\,A^\pidom\,A^\picod)~X~f & := & \err_X \qquad \text{otherwise}
    \end{array}
  \]

    \[
      \hspace{1cm}\begin{array}[t]{lcl}
        \cas~\natU~\natU~n &:=& n\\
        \cas~\natU~\unkU[i]~n &:=& (\nat; n)\\
        \cas~\natU~X~n &:=& \err_X\\[0.3cm]

        \cas~\errU[i]~Z~\unitK &:=& \err_Z\\
        \\
        \\
      \end{array}%
      \hspace{1cm}%
      \begin{array}[t]{lcl}
        \cas~\uU_j~\uU_j~A & := & A \\
        \cas~\uU_j~\unkU[i]~A & := & (\square_j ; A) \qquad \text{if } j < i\\
        \cas~\uU_j~X~A & := & \err_X \qquad \text{otherwise}\\[0.3cm]
        \cas~\unkU[i]~Z~(c ; x) & := & \cas~(\stalkCode{i}{c})~Z~x\\
        \cas~\unkU[i]~Z~\lifttop{\El\,\unkU[i]} & := & \?_Z\\
        \cas~\unkU[i]~Z~\liftbot{\El\,\unkU[i]} & := & \err_Z\\
      \end{array}
    \]
  \caption{Definition of $\cas$ (discrete model)}
  \label{fig:cast-implem-discrete}
\end{figure}

\paragraph{Casts}
Equipped with exceptions and type analysis, we define
\mbox{$\cas : \P (A : \U_i)(B : \U_j). A \to B$} by
induction on the universe levels and case analysis on the codes of the types
$A$ and $B$ (\cref{fig:cast-implem-discrete}).
In the total setting (when $\castOfPi{\sortOfPi{i}{i}} < \sortOfPi{i}{i}$), the definition of $\cas$ is well-founded: each recursive
call happens either at a strictly smaller universe (the two cases for
$\PiU$) or on a strict subterm of the term being cast (case of
inductives, \ie~$\natU$ and $\unkU$).
Note that each of the defining equations of $\cas$ corresponds straightforwardly to a
reduction rule of~\cref{fig:CCIC-reduction}.

\paragraph{Discrete translation}
We can finally define the discrete syntactic model of \CCIC in \CICIR (\cref{fig:discrete-translation}).
The translations $\bareTm{-}$ and $\bareTy{-}$ are defined by induction on the syntax of
terms and types.
A type $A$ is translated to its corresponding code $\bareTm{A}$
in $\U_i$ when seen as a term, and is translated to the interpretation
of this code $\bareTy{A} := \El\,\bareTm{A}$ when seen as a type.
$\?_A$ and $\err_A$ are directly translated using the exceptions
defined in \cref{fig:exceptions}.
The following theorem shows that the translation is a syntactic model in the sense
of~\citet{BoulierPT17}.
\begin{figure}
  \[
    \begin{array}[t]{lcl}
      \bareTy{\cdot} & := & \cdot\\[0.6em]
      \bareTy{A} & := & \El\,\bareTm{A} \\
      \\
      {\bareTm{ x }} & := & x\\
      {\bareTm{\square{}_i}} & := & \uU_i \\[0.6em]
      {\bareTm{\P x : A . B}} & := & \PiU~\bareTm{ A }~(\l x : \bareTy{A} . \bareTm{ B }) \\
      {\bareTm{t~u}} & := & \bareTm{t}~\bareTm{u}\\
      {\bareTm{\l x : A . t}} & := & \l x : \bareTy{ A } . \bareTm{t}
    \end{array}%
    \begin{array}[t]{lcl}
      \bareTy{\Gamma, x : A} & := & \bareTy{\Gamma}, x : \bareTy{A}\\[0.6em]
      {\bareTm{\nat}} & := & \natU \\
      {\bareTm{0}} &:= & 0\\
      {\bareTm{\suc}} & := & \suc \\
      \bareTm{\ind_{\nat}}~ P~h_0\,h_{\suc} & :=& \ind_{\liftNat{}}~P~h_0\,h_{\suc}\,\?_{P\,\?_{\liftNat{}}}\,\err_{(P\err_{\liftNat{}})} \\[0.6em]
      {\bareTm{\?_A}} & := & \?_{\bareTm{A}}\\
      {\bareTm{\err_A}} &:= & \err_{\bareTm{A}} \\
      {\bareTm{ \ascdom{t}{A}{B} }} &:= & \cas\,\bareTm{A}\,\bareTm{B}\,\bareTm{t}
    \end{array}
  \]
  \caption{Discrete translation from \CCIC to \CICIR}
  \label{fig:discrete-translation}
\end{figure}
\begin{theorem}[Discrete syntactic model]
  \label{thm:discrete-model}
  The translation defined in \cref{fig:discrete-translation} 
  preserves conversion and typing derivations:
  \begin{enumerate}
  \item if $\tcol{\Gamma} \caty \tcol{t} \redCCIC{} \tcol{u}$ then $\bareTy{\Gamma} \irty \bareTm{t}
    \redIRn \bareTm{u}$, in particular $\bareTy{\Gamma} \irty \bareTm{t}
    \equiv \bareTm{u}$,
  \item if $\tcol{\Gamma} \caty \tcol{t} : \tcol{A}$ then $\bareTy{\Gamma} \irty \bareTm{t} : \bareTy{A}$.
  \end{enumerate}
\end{theorem}
\begin{proof}
  (1) All reduction rules from \CIC are preserved without a
  change so that we only need to be concerned with the reduction rules involving
  exceptions or casts.
  A careful inspection shows that these reductions are preserved too once we
  observe that the terms of the shape
  $\tcol{\ascdom{t}{\stalkCIC{i}{h}}{\?_{\square_i}}}$ that are stuck in \CCIC
  are in one-to-one correspondence with the one-step reduced form of its
  translation $\pairSigma{h}{\bareTm{t}} : \liftSigma \, \H_i \,
  \stalkCIC{i}{}$.
  (2) Proved by a direct induction on the typing derivation of
  $\catyJ{\Gamma}{t}{A}$, using the fact that exceptions and casts are well-typed---that is $\irty \? : \P(A : \U_i)\,\El\,A$ , $\irty \err : \P(A : \U_i)\,\El\,A$, and $\irty \cas : \P(A : \U_i)(B : \U_i)
    \El\,A {\to} \El\,B$---and relying on assertion $(1)$ to handle the
  conversion rule.
\end{proof}

As explained in \cref{thm:ccic-pnorm}, \cref{thm:discrete-model}
implies in particular that \CCICs and \CCICT are strongly
normalizing.

\subsection{Poset-Based Models of Dependent Type Theory}
\label{sec:poset-model-dtt}

The simplicity of the discrete model comes at the price of an inherent inability to
characterize which casts are guaranteed to succeed, \ie a graduality theorem.
To overcome this limitation, we develop a monotone model on top of the discrete
model where, by construction, each type $A$ comes equipped with an order
structure $\sqsubseteq^A$---a reflexive, transitive, antisymmetric and
proof-irrelevant relation---modelling precision between terms.
In particular, the exceptions $\err_A$ and $\?_A$ correspond respectively to the smallest
and greatest element of $A$ for this order.
We note \poset{} for a universe of types equipped with the structure of a poset
together with smallest and greatest elements.
Each term and type constructor is enforced to be monotone with respect to these orders,
providing a strong form of graduality.
This implies in particular that such a model cannot be defined for \CCICT
because this type theory lacks graduality, as shown by~\cref{ex:lambda-abstr-catch-up}.

As an illustration, the order on extended natural numbers~(\cref{fig:nat-numbers}) makes
$\liftbot{\liftNat}$ the smallest element and $\lifttop{\liftNat}$ the biggest element.\footnote{We abusively note $\liftNat$ for both
  the poset and its carrier to avoid introducing too many notations.}
The standard natural numbers -- $\natzero$ or $\natsuc~n$ for a standard natural
number $n$ -- then stand between failure and indeterminacy,
but are never related to each other by precision. Indeed, in order to ensure conservativity with respect to \CIC, $\sqsubseteq^{\natU{}}$ must coincide with \CIC{}'s conversion on static closed natural numbers.
\begin{figure}
  \begin{minipage}{0.35\linewidth}
    \begin{align*}
      0 &\sqsubseteq^{\natU} 0
      &\liftbot{\liftNat} &\sqsubseteq^{\natU} n\\
      0 &\sqsubseteq ^{{\natU}} \lifttop{\liftNat}
       &\lifttop{\liftNat} &\sqsubseteq^{{\natU}} \lifttop{\liftNat}
    \end{align*}
  \end{minipage}%
  \begin{minipage}{0.40\linewidth}
    \begin{mathpar}
      \inferrule{n \sqsubseteq^{{\natU}} m}{\suc\,n \sqsubseteq^{{\natU}} \suc\,m} \and
      \inferrule{n \sqsubseteq^{{\natU}} \lifttop{\liftNat}}{\suc\,n \sqsubseteq^{{\natU}} \lifttop{\liftNat}}
    \end{mathpar}
  \end{minipage}
  \caption{Order structure on extended natural numbers}
  \label{fig:nat-numbers}
\end{figure}
Beyond the precision order on types, the nature of dependency forces us to spell
out what the precision between types entails.
Following the analysis of~\citet{newAhmed:icfp2018},
a relation $A \sqsubseteq B$ between types should
induce an embedding-projection pair (\eppair{}): a pair of an \emph{upcast}
${\upcast{} : A {\to} B}$ and a \emph{downcast} ${\downcast{} : B {\to} A}$
satisfying a handful of properties with gradual guarantees as a
corollary.
\begin{definition}[Embedding-projection pairs]
  \label{def:ep-pair}
  An \eppair{} $d : A \distr{} B$ between posets $A,B : \poset$ consists of
  \begin{itemize}
  \item an underlying relation $d \subseteq A \times B$ such that
    \[a' \sqsubseteq^A a \>\wedge\> d(a,b) \>\wedge\> b \sqsubseteq^B b' \quad\implies\quad d(a',b')\]
  \item that is bi-represented by $\upcast_d~ : A \to B$, $\downcast{}_d : B \to
    A$, \ie
    \[\upcast_d a \sqsubseteq^B b \quad\Leftrightarrow\quad d(a,b) \quad\Leftrightarrow\quad a \sqsubseteq^A \;\downcast_d b,\]
  \item such that the equality $\downcast_d \circ \upcast_d~ = \id_A$ holds.
  \end{itemize}
\end{definition}
Note that here equiprecision of the retraction becomes an equality
because of antisymmetry.
Under these conditions, $\upcast_d : A \hookrightarrow B$ is injective,
$\downcast_d : B \twoheadrightarrow A$ is surjective and both preserve bottom
elements, explaining that we call $d : A \distr{} B$ an embedding-projection
pair.
The definition of \eppairs{} is based on a relation rather than just its
pair of representing functions to highlight the connection between \eppairs{} and
parametricity~\cite{newAl:popl2020}.
Assuming function extensionality, being an \eppair{} is a
property of the underlying relation: there is at most one pair $(\upcast{}_d,
\downcast_d)$ representing the underlying relation of $d$.
An \eppair{} straightforwardly induces the following
relations that will be used in later proofs.
\begin{lemma}[Properties of \eppairs{}]
  \label{lem:adj}
  Let $d : A \distr{} B$ be an \eppair{} between posets.
  \begin{enumerate}
  \item If $a : A$ then $d\,(a, \upcast_d a)$ and $a \sqsubseteq^A
    \downcast_d\upcast_d a$.
  \item If $b : B$ then $d\,(\downcast_d b, b)$ and $\upcast_d\downcast_d b \sqsubseteq^B
    b$.
  \end{enumerate}
\end{lemma}

\paragraph{Posetal families}
By monotonicity, a family $B : A \to \poset$ over a poset $A$ gives rise not
only to a poset $B\,a$ for each $a \in A$, but also to \eppairs{} $B_{a,a'} :
B\,a \distr{} B\,a'$ for each $a \sqsubseteq^A a'$.
These \eppairs{} need to satisfy functoriality conditions:
\[B_{a,a} = {\sqsubseteq^{B\,a}} \qquad\text{and}\qquad B_{a,a''} = B_{a',a''}
  \circ B_{a,a'}\quad \text{whenever}\quad a \sqsubseteq^A a' \sqsubseteq^A
  a''.\] In particular, this ensures that heterogeneous transitivity is well
defined:
\[B_{a,a'}(b,b') \wedge B_{a',a''}(b',b'') \Rightarrow
  B_{a,a''}(b,b'').\]

\paragraph{Dependent products}
Given a poset $A$ and a posetal family $B$ over $A$, we can form the poset
$\Pmon{}\,A\,B$ of \emph{monotone} dependent functions from $A$ to $B$,
equipped with the pointwise order.
Its inhabitants are dependent functions $f : \P (a : A). B\,a$ such that $a \sqsubseteq^A
a' {\Rightarrow} B_{a,a'}\,(f\,a, f\,a')$.
Moreover, given \eppairs{} $d_A : A \distr{} A'$ and $d_B : B \distr{} B'$, we
can build an induced \eppair{} $d_\Pi : \Pmon{}\,A\,B \distr{} \Pmon{}\,A'\,B'$
with underlying relation
\begin{align*}
  d_\Pi(f,f') &:= d_A(a,a') \Rightarrow d_B(f\,a,f'\,a'),\\
  \upcast{}_{d_\Pi}f := {\upcast_{d_B}}\circ f \circ \downcast{}_{d_A}
              &\qquad\text{ and }\qquad
  \downcast{}_{d_\Pi}f := {\downcast_{d_B}} \circ f \circ \upcast{}_{d_A}.
\end{align*}
The general case where $B$ and $B'$ actually depend on $A, A'$ is obtained with
similar formulas, but a larger amount of data is required to handle
the dependency: we refer to the accompanying \Agda{} development for details.

\paragraph{Inductive types}
Generalizing the case of natural numbers, the order on an arbitrary extended
inductive type $\liftErr{I}$ uses the following scheme:
\begin{enumerate}
\item $\liftbot{\liftErr{I}}$ is the least element
\item $\lifttop{\liftErr{I}} \sqsubseteq^{\liftErr{I}} \lifttop{\liftErr{I}}$
\item $c~\mathbf{t} \sqsubseteq^{\liftErr{I}} \lifttop{\liftErr{I}}$ whenever $t_i \sqsubseteq^{X_i} \lifttop{X_i}$ for
  all $i$
\item each constructor $c$ is monotone with respect to the order on its arguments
\end{enumerate}
The precondition on subterms in the third case is unnecessary in simple cases
and is kept to be uniform with definition of order on the monotone unknown type
in the following section.
Similarly to dependent product, an \eppair{} $\mathbf{X} \distr{} \mathbf{X'}$
between the parameters of an extended inductive type $\liftErr{I}$ induces an \eppair{}
$\liftErr{I}\,\mathbf{X} \distr{} \liftErr{I}\,\mathbf{X'}$.
For instance, \eppairs{} $d_A : A \distr{} A'$ and $d_B : B \distr{} B'$ induce
an \eppair{} $d_{\liftSigma} : \liftSigma\,A\,B \distr{} \liftSigma\,A'\,B'$
defined by $d_{\liftSigma}((a,b),(a',b')) := d_A(a,a') \wedge d_B(b,b')$.

\subsection{Microcosm: the Monotone Unknown Type \unkMon}
\label{sec:realizing-unknown-type}

The interpretation \unkMon[i] of the unknown type in the monotone model should
morally group together approximations of every type at the same universe
level.
Working in (bi)pointed orders, \unkMon[i] can be realized as a coalesced sum
~\cite[section 3.2.3]{AbramskyJung} of
the family $\stalkCIC{i}{h}$ indexed by head constructors $h \in \H_i$.
A concrete presentation of \unkMon[i] is obtained as
the quotient of $\liftSigma\, \H_i \, \stalkCIC{i}{}$
identifying $\liftbot{\liftSigma\, \H_i \, \stalkCIC{i}{}}$ with any pair
$\pairSigma{h}{\err_{\stalkCode{i}{h}}}$.
The equivalence classes of $\pairSigma{h}{x}$ is noted as $\unkInj{h}{x}$,
$\liftbot{\liftSigma\, \H_i \, \stalkCIC{i}{}}$ as $\liftbot{\unkMon[i]}$ and
$\lifttop{\liftSigma\, \H_i \, \stalkCIC{i}{}}$ as $\lifttop{\unkMon[i]}$.
The obtained type \unkMon[i] is then equipped
with a precision relation defined by the rules:
\begin{align}
  \label{eq:unknown-prec}
  \liftbot{\unkMon[i]} &\sqsubseteq^{\unkU[i]} z
  & \lifttop{\unkMon[i]} &\sqsubseteq^{\unkU[i]} \lifttop{\unkMon[i]}
  & \inferrule{x \sqsubseteq^{\stalkCode{i}{h}} x'}{ \unkInj{h}{x} \sqsubseteq^{\unkU[i]} \unkInj{h}{x'}} &
  & \inferrule{}%{x \sqsubseteq^{\stalkCode{i}{h}} \?_{\stalkCode{i}{h}}}
    {\unkInj{h}{x} \sqsubseteq^{\unkU[i]} \lifttop{\unkMon[i]}}
\end{align}
These rules ensure that the exceptions
$\liftbot{\unkMon[i]}$ and $\lifttop{\unkMon[i]}$ are
respectively the smallest and biggest elements of $\unkMon[i]$.
Non-exceptional elements
are comparable only if they have the same head constructor $h$ and if so are
compared according to the interpretation of that head constructor as an ordered
type $\stalkCIC{i}{h}$ .
Because of the quotient, it is not immediate that this presentation of
$\sqsubseteq^{\unkMon[i]}$ is independent of the choice of representatives in
equivalence classes and that it forms a proof-irrelevant relation.
In the formal development, we define the relation by quotient-induction on each argument, thus verifying that it respects the quotient, and also show that
it is irrelevant. This relies crucially on equality being decidable
on head constructors when comparing $\unkInj{h}{x}$ and $\unkInj{h'}{x'}$.
In order to globally satisfy $\pgrad$, \unkMon[i] should admit an \eppair{} ${d_h
  : \stalkCIC{i}{h} \distr{} \unkMon[i]}$ whenever we have a head constructor $h
\in \H_i$ such that $\stalkCode{i}{h} \sqsubseteq \unkU[i]$ (we return to that
point in the next section~\cref{sec:monotone-universe}).
Embedding an element $x \in \stalkCIC{i}{h}$ by $\upcast_{d_h}x = \unkInj{h}{x}$
and projecting out of $\stalkCIC{i}{h}$ by the following equations form a reasonable
candidate:
\begin{align*}
  \downcast_{d_h}\unkInj{h'}{x} &=
  \begin{dcases*}
    x & if $h = h'$ \\
    \err_{\stalkCIC{i}{h}} & otherwise
  \end{dcases*}&
  \downcast_{d_h} \lifttop{\unkMon[i]} &= \?_{\stalkCIC{i}{h}} &
  \downcast_{d_h} \liftbot{\unkMon[i]} &= \err_{\stalkCIC{i}{h}}
\end{align*}
Note that we rely again on $\H$ having decidable equality to compute the
$\downcast{}_{d_h}$.
Moreover $\upcast_{d_h} \dashv \downcast_{d_h}$ should be adjoints; in particular,
 the following precision relation needs to hold:
\begin{align*}
  \err_{\stalkCIC{i}{h}} \sqsubseteq^{\stalkCIC{i}{h}} {\downcast{}_{d_h}\liftbot{\unkMon[i]}}\quad \iff \quad
  \unkInj{h}{\err_{\stalkCIC{i}{h}}} = {\upcast_{d_h}\err_{\stalkCIC{i}{h}}} \sqsubseteq^{\unkMon[i]} \liftbot{\unkMon[i]}
\end{align*}
Since $\sqsubseteq^{\unkMon[i]}$ should be antisymmetric, this is possible only
if $\unkInj{h}{\err_{\stalkCIC{i}{h}}}$ and $\liftbot{\unkMon[i]}$ are
identified in $\unkMon[i]$, explaining why we have to quotient in the first place.

\subsection{Realization of the Monotone Universe Hierarchy}
\label{sec:monotone-universe}

\begin{figure}
\begin{small}
  \begin{center}
    \textbf{Monotone universes} $\U_i$
    \textbf{and decoding function}  $\El\, : \U_i \to \poset$
    \text{(cases distinct from \cref{fig:discrete-univ})}
  \end{center}
  \vspace{0.5em}
  \begin{minipage}{0.40\linewidth}
    \[
      \inferrule{A \in \U_i\\ B \in \Pmon(a :\El\, A).\,\U_j}{\PiU\,A\,B \in \U_{\sortOfPi{i}{j}}}
    \]
  \end{minipage}%
  \begin{minipage}{0.60\linewidth}
    \begin{align*}
      \El\,(\PiU\,A\,B) &= \Pmon (a : \El\,A).\, \El (B\,a)
      & \El\,\unkU[i] &= \unkMon[i]
    \end{align*}
  \end{minipage}
  \vspace{1em}
  \begin{center}
   \textbf{Precision order} $\sqsubseteq$ \textbf{on the universes}
   (where $i \le j$)
  \end{center}
  \vspace{0.5em}
  \begin{mathpar}
    \inferrule[$\errU$-$\sqsubseteq$]{A : \U_j }{\Ue{\errU[i]}{A}}
    \and
    \inferrule[$\natU$-$\sqsubseteq{}$]{ }{\Ue{\natU}{\natU}}
    \and
    \inferrule[$\uU$-$\sqsubseteq$]{ }{\Ue{\uU_i}{\uU_{i}}}
    \and
    \inferrule[$\unkU$-$\sqsubseteq$]{ }{\Ue{\unkU[i]}{~\unkU[j]}}
    \and
    \inferrule[$\PiU$-$\sqsubseteq$]{\Ue{A}{A'}\\ a:\El\,A,a':\El\,A', a_\epsilon : a
      \precision{A}{A'} a' \vdash \Ue{B\,a}{B'\,a'}}
    {\Ue{\PiU\,A\,B}{\PiU\,A'\,B'}}
    \and
    \inferrule[$\C$-$\sqsubseteq$]{h = \hd{A} \in \C_i\\\Ue{A}{\stalkCode{j}{h}}}{\Ue{A}{~\unkU[j]}}
  \end{mathpar}
  \vspace{1em}
  \begin{center}
   \textbf{Precision on terms} $\precision{A}{B} := \ElRel(A{\sqsubseteq}B)$ 
  \end{center}
  \vspace{0.5em}
  \begin{mathpar}
    % rule corresponding to [$\errU$-$\sqsubseteq$]
    \inferrule[$\ElRel(\errU\text{-}{\sqsubseteq})$]{a:\El~A}{\unitK \precision{\errU}{A} a}
    \and
    % rule corresponding to [$\natU$-$\sqsubseteq{}$] and [$\uU$-$\sqsubseteq$]
    \inferrule[$\ElRel(\natU\text{-}{\sqsubseteq}),\ElRel(\uU\text{-}\sqsubseteq)$]
    {A = \natU, \uU_{i}\\ x \sqsubseteq^{\El~A} y}{x \precision{A}{A} y}
    \and
    % rules corresponding to [$\unkU$-$\sqsubseteq$]
    % there are some redondances, the first rule is derivable from the 2 others
    % (and not uniquely)
    \inferrule[$\ElRel(\unkU\text{-}{\sqsubseteq})$]{z : \unkMon[j]}{\err_{\unkU[i]} \precision{\unkU[i]}{\unkU[j]} z}
    \and
    \inferrule{z : \unkMon[i]}{ z \precision{\unkU[i]}{\unkU[j]} \?_{\unkU[j]}}
    \and
    \inferrule{x \precision{\stalkCode{i}{h}}{\stalkCode{j}{h}} x'}{\unkInj{h}{x} \precision{\unkU[i]}{\unkU[j]} \unkInj{h}{x'}}
    \and
    % rule corresponding to [$\PiU$-$\sqsubseteq$]
    \inferrule[$\ElRel(\PiU\text{-}{\sqsubseteq})$]{a:\El\,A,a':\El\,A', a_\epsilon : a \precision{A}{A'} a' \vdash
      f~a \precision{B\,a}{B'\,a'}f'~a'}
    {f \precision{\PiU\,A\,B}{\PiU\,A'\,B'} f'}
    \and
    % rules corresponding to [$\C$-$\sqsubseteq$]
    \inferrule[$\ElRel(\C\text{-}{\sqsubseteq})$]{a \precision{A}{\stalkCode{j}{(\hd{A})}} x}{a \precision{A}{\unkU[j]} \unkInj{\hd{A}}{x}}
    \and
    \inferrule{a:\El~A}{a \precision{A}{\unkU[j]} \?_{\unkU[j]}}
    % \and
    % \inferrule{z: \unkMon[j]}{\err_A \precision{A}{\unkU[j]} z}
    % This rule is not needed (it can be derived by err_A \sqsubseteq [head A;
    % err_A ] = err_? \sqsubseteq z )
  \end{mathpar}
  \end{small}
  \caption{Monotone universe of codes and precision}
  \label{fig:IRUniverse}
\end{figure}

Following the discrete model, the monotone universe hierarchy is
also implemented through an inductive-recursive datatype of codes $\U_i$
together with a decoding function $\El : \U_i \to \square$, both presented in
\cref{fig:IRUniverse}.
The precision relation ${\sqsubseteq} : \U_i \to \U_j \to \square$
presented below is an
order (\cref{thm:universe-hierarchy}) on this universe hierarchy.
The ``diagonal'' inference rules, providing evidence for relating type
constructors from \CIC, coincide with those of binary
parametricity~\cite{bernardyAl:jfp2012}.
Outside the diagonal, $\errU$ is placed at the bottom.
More interestingly, the derivation of a precision proof $A \sqsubseteq
\unkU$ provides a unique decomposition of $A$ through iterated germs
directed by the relevant head constructors.
For instance, in the gradual systems \CCICP and \CCICs where the equation $\castOfPi{\sortOfPi{i}{j}} =
\max(i,j)$ holds for any universe levels $i,j$, the derivation of $(\natU{} {\to} \natU{}) {\to} \natU{}
\sqsubseteq \unkU$ canonically decomposes as:
\[ (\natU{} {\to} \natU{}) {\to} \natU{} \quad \sqsubseteq
  \quad (\unkU {\to} \unkU) {\to} \natU{} \quad \sqsubseteq
  \quad \unkU {\to} \unkU \quad \sqsubseteq\quad \unkU\]

This unique decomposition is at the heart of the reduction of the cast
operator given in~\cref{fig:CCIC-reduction}, and it can be described informally
as taking the path of maximal length between two related types.\footnote{This
  decomposition is already present in~\cite{newAhmed:icfp2018} and
  to be contrasted with the AGT approach~\cite{garciaAl:popl2016}, which tends 
  to pair a value with the most precise witness of its type, \ie~canonical path of minimal length.}
Such a derivation of precision $A \sqsubseteq B$ gives rise through decoding to
\eppairs{} $\ElRel\,(A {\sqsubseteq}B) : \El\,A \distr{}\El\,B$, with underlying relation
noted ${\precision{A}{B}} : \El\,A \to \El\,B \to \square$.
This decoding function $\ElRel$ is described on generators of $\sqsubseteq$ at
the bottom of~\cref{fig:IRUniverse}.
$\ElRel(\errU\text{-}{\sqsubseteq})$ states that the unique value $\unitK$ of
$\El\,\errU = \unit$ is smaller than any other value.
The diagonal cases
$\ElRel(\natU\text{-}{\sqsubseteq})$ and $\ElRel(\uU\text{-}\sqsubseteq)$ reuse
the order specified on the carrier.
The \eppair $\ElRel(\unkU\text{-}{\sqsubseteq})$ between two unknown types
$\unkMon[i] \distr{} \unkMon[j]$ at potentially distinct universe levels $i \leq
j$ stipulate that $\err_{\unkU[i]}$ and $\?_{\unkU[j]}$ are respectively smaller
and greater than any other value, and that the comparison between two injected
terms with same head is induced by their second component.
Note that these rules are redundant since $\unkMon$ is obtained through a quotient.
Functions $f,f'$ are related by $\ElRel(\PiU\text{-}{\sqsubseteq})$ when they map
related elements $a \precision{A}{A'} a'$ to related elements $f\,a \precision{B\,a}{B'\,a'}f'\,a'$.
Finally, $\ElRel(\C\text{-}{\sqsubseteq})$ embeds a type
$A$ into $\unkMon$ through its $\hd$.
It is interesting to observe what happens in \CCICT, 
 where
$\castOfPi{\sortOfPi{i}{j}} \neq \max(i,j)$, for instance on the previous example:
\[\natU{} {\to} \natU{} \quad \not\sqsubseteq \quad \errU =
  \stalkCode{0}{\Pi} \quad \sqsubseteq \quad \unkU \]
So $\natU{} {\to} \natU{}$ is not lower than $\unkU$ in that setting.
One crucial point of the monotone model is the mutual definition of codes $\U_i$
together with the precision relation, particularly salient on codes for
$\Pi$-types: in $\PiU\,A\,B$, $B : \El\,A \to \U_i$ is a monotone function with
respect to the order on $\El\,A$ and the precision on $\U_i$.
This intertwining happens because the order is required to be reflexive, a fact
observed previously by~\citet{AtkeyGJ14} in the similar setting of reflexive
graphs.
Indeed, a dependent function $ f : \Pi (a : \El\,A).\,  \El\,(B\,a)$ is related to
itself $f \precision{\PiU{}\,A\,B}{\PiU{}\,A\,B} f$ if and only if $f$ is
\emph{monotone}.
\begin{theorem}[Properties of the universe hierarchy]\text{}
  \label{thm:universe-hierarchy}
  \begin{enumerate}
  \item $\sqsubseteq$ is reflexive, transitive, antisymmetric and irrelevant so
    that $(\U_i, \sqsubseteq)$ is a poset.
  \item $\U_i$ has a bottom element $\errU[i]$ and a top element $\unkU[i]$; in
    particular, $A \sqsubseteq \unkU[i]$ for any $A : \U_i$.
  \item $\El : \U_i \to \square$ is a family of posets over $\U_i$ with underlying relation
    $\precision{A}{B}$ whenever $A \sqsubseteq B$.
  \item $\U_i$ and $\El\;A$ for any $A : \U_i$ verify UIP\footnote{Uniqueness of
    Identity Proofs; in HoTT parlance, $\U_i$ and $\El\,A$ are hSets.}: the equality on these
    types is irrelevant.
  \end{enumerate}
\end{theorem}
\begin{proof}[Proof sketch]
  All these properties are proved mutually, first by strong induction on the
  universe levels, then by induction on the codes of the universe or the
  derivation of precision. Here, we only sketch the proof of point $(1)$ and refer to the Agda
  development (\cf \texttt{UnivPartial.agda}) for detailed formal proofs.

  For reflexivity, all cases are immediate but for $\PiU\,A\,B$: the induction
  hypothesis provides $A \sqsubseteq A$ and by point $(3)$ $\ElRel{}(A
  \sqsubseteq A) = {\sqsubseteq^A}$ so we can apply the monotonicity of $B$.
  For anti-symmetry, assuming $A \sqsubseteq B$ and $B \sqsubseteq A$, we prove
  by induction on the derivation of $A \sqsubseteq B$ and case analysis on the other
  derivation that $A \equiv B$. Note that we never need to consider the rule
  $\C$-$\sqsubseteq$. The case $\Pi$-$\sqsubseteq$ holds by induction hypothesis
  and because the relation ${}_A\sqsubseteq_A$ is reflexive. All the other cases
  follow from antisymmetry of the order on universe levels.
  For transitivity, assuming $AB : A \sqsubseteq B$ and $BC : B
  \sqsubseteq C$, we prove by induction on the (lexicographic) pair $(AB,BC)$
  that $A \sqsubseteq C$:
  \begin{description}
  \item[Case $AB = \unkU\dash{\sqsubseteq}$,] necessarily $BC =
    \unkU\dash{\sqsubseteq}$, we conclude by $\unkU\dash{\sqsubseteq}$.
  \item[Case $AB = \C\dash{\sqsubseteq}$,] necessarily $BC =
    \unkU\dash{\sqsubseteq} , \unkU[j] \sqsubseteq \unkU[j']$, we can
    thus apply the inductive hypothesis to
    ${A \sqsubseteq \stalkCode{j}{(\hd{A})}}$ and
    ${\stalkCode{j}{(\hd{A})} \sqsubseteq \stalkCode{j}{(\hd{A})}}$
      in order to conclude with $\C\dash{\sqsubseteq}$.
  \item[Case $AB = \errU\dash{\sqsubseteq}$,] we conclude immediately by $\errU\dash{\sqsubseteq}$.
  \item[Case $AB = \natU\dash{\sqsubseteq}, BC = \natU\dash{\sqsubseteq}$] we
    conclude with $\natU\dash{\sqsubseteq}$.
  \item[Case $AB = \uU\dash{\sqsubseteq}, BC = \uU\dash{\sqsubseteq}$] immediate
    by $\uU\dash{\sqsubseteq}$.
  \item[Case $AB = \PiU\dash{\sqsubseteq}, BC = \PiU\dash{\sqsubseteq}$]
    by hypothesis we have
    \begin{mathpar}
      A = \PiU\,A^\pidom\,A^\picod \and
      B = \PiU\,B^\pidom\,B^\picod \and
      C = \PiU\,C^\pidom\,C^\picod \and
      A^\pidom \sqsubseteq B^\pidom \and
      B^\pidom \sqsubseteq C^\pidom \and
      AB^\picod : \forall a\,b, a \precision{A^\pidom}{B^\pidom} b \to A^\picod\,a \sqsubseteq B^\picod\,b \and
      BC^\picod : \forall b\,c, b \precision{B^\pidom}{C^\pidom} c \to B^\picod\,b \sqsubseteq C^\picod\,c
    \end{mathpar}
    By induction hypothesis applied to $A^\pidom \sqsubseteq B^\pidom$ and
    $B^\pidom \sqsubseteq C^\pidom$, the domains of the dependent product are related $A^\pidom \sqsubseteq C^\pidom$.
    For the codomains, we need to show that for any $a : A^\pidom, c : C^\pidom$
    such that $a \precision{A^\pidom}{C^\pidom} c$ we have
    $A^\picod\,a \sqsubseteq C^\picod\,c$. By induction hypothesis, it is enough to prove that
    $A^\picod\,a \sqsubseteq B^\picod\,(\upcast_{A^\pidom\sqsubseteq B^\pidom}a)$
    and $B^\picod\,(\upcast_{A^\pidom\sqsubseteq B^\pidom}a) \sqsubseteq
    C^\picod\,c$.
    The former follows from $AB^\picod$ applied to
    $ a \precision{A^\pidom}{B^\pidom} \upcast_{A^\pidom\sqsubseteq B^\pidom} a
    \Leftrightarrow a \sqsubseteq^{A^\pidom} { \downcast{} }{ \upcast{} } a \Leftrightarrow a \sqsubseteq^{A^\pidom} a$ which holds by reflexivity, and the latter
    follows from $BC^\picod$ applied to $\upcast_{A^\pidom\sqsubseteq B^\pidom}
    a \precision{B^\pidom}{C^\pidom} c \Leftrightarrow a \precision{A^\pidom}{C^\pidom} c$.
  \item[Otherwise,]
    we are left with the cases where $AB = \natU\dash{\sqsubseteq},
    \PiU\dash{\sqsubseteq}$ or $\uU\dash{\sqsubseteq}$ and
    $BC = \C\dash{\sqsubseteq}$, we apply the inductive hypothesis to
    $AB$ and $B \sqsubseteq \stalkCode{j}{(\hd{B})}$ in order to conclude
    with $\C\dash{\sqsubseteq}$.
  \end{description}
  % So $\sqsubseteq$ is a reflexive, transitive and antisymmetric relation, 
  Finally, we show proof-irrelevance, \ie~that for any $A,B$ there is at most one
  derivation of $A \sqsubseteq B$. Since the conclusions of the rules do not
  overlap, we only have to prove that the premises of each rules are uniquely
  determined by the conclusion. This is immediate for $\PiU\dash{\sqsubseteq}$.
  For $\C\dash{\sqsubseteq}$, $h = \hd{A} \in \C_i$ with $i = \pred\,j$ are uniquely
  determined by the conclusion so it holds too.
\end{proof}

\subsection{Monotone Model of \CCICs}
\label{sec:monotone-model}

The monotone translation $\monTm{-}$ presented in \cref{fig:monotone-translation} brings together the monotone
interpretation of inductive types (e.g. $\liftNat$), dependent products, the unknown type
$\unkMon$ as well as the universe hierarchy.
Following the approach of~\citet{newAhmed:icfp2018}, casts are derived out of the
canonical decomposition through the unknown type using the property $(2)$
from~\cref{thm:universe-hierarchy}:
\[\monTm{\ascdom{t}{A}{B}} \quad:= \quad\downcast{}_{\raisebox{-2pt}{\tiny$\ElRel{\monTm{B}{\sqsubseteq}\unkU}$}}\upcast_{\raisebox{-2pt}{\tiny$\ElRel{\monTm{A}{\sqsubseteq}\unkU}$}}\monTm{t} \]
Note that this definition formally depends on a chosen universe level $j$ for
$\unkU : \U_{j}$, but the resulting operation is independent of this choice thanks to the
section-retraction properties of \eppairs{}.
The difficult part of the model, the monotonicity of $\cas$, thus holds by design.
However, the translation of some terms do not reduce as in \CCIC{}: $\cas$ can
get stuck on type variables eagerly, \eg on a \nameref{redrule:down-err}
step.\footnote{An analysis of the correspondence between the discrete and
  monotone models can be found in \cref{sec:logical-relation}.}
These reduction rules still hold propositionally though so that we have at least
a model in an extensional variant of the target theory (\ie~in which two terms
are \emph{definitionally} equal whenever they are \emph{propositionally} so).
\begin{lemma}
  If $\tcol{\Gamma} \caty \tcol{t} \redCCIC{} \tcol{u}$ then there exists a \CICIRQ
  term $e$ such that
  $\monTm{\Gamma} \irty e : \monTm{t} = \monTm{u}$.
\end{lemma}
We can further enhance this result using the fact that we assume
functional extensionality in our target and can prove that the translation of
all our types satisfy \UIP{}.
Under these assumptions, the conservativity results
of~\citet{Hofmann95} and \citet{WinterhalterST19} apply, so we can recover a translation
targeting \CICIRQ.
\begin{theorem}[Monotone model]
  \label{thm:monotone-model}
 The translation $\monTm{-}$ of \cref{fig:monotone-translation} extends to a
 model of \CCIC into \CIC extended with induction-recursion and functional
 extensionality: if $\catyJ{\Gamma}{t}{A}$ then $\monTy{\Gamma}
 \irty \monTm{t} : \monTy{A}$.
\end{theorem}
It is unlikely that the principle that we demand in the target calculus $\CICIRQ$ are optimal.
We conjecture that a variation of the translation described here could be
developed in \CIC{} extended only with induction-induction to describe the
intensional content of the codes $\U$ in the universe, and strict
propositions~\cite{GilbertCST19} following the construction of the setoid models of
type theory~\cite{Altenkirch99,AltenkirchBKT19,AltenkirchBKSS21}.
\begin{figure}
\begin{small}
  \begin{center}
    \textbf{Monotone translation of contexts}
  \end{center}
  \[
    \vspace{0.5em}
    \begin{array}[t]{lcl}
      \monTy{\cdot} & := & \cdot\\
      \monTyRel{\cdot} & := & \cdot\\
    \end{array}\qquad%
    \begin{array}[t]{lcl}
      \monTy{\Gamma, x : A} & := & \monTy{\Gamma}, x : \monTy{A}\\
      \monTyRel{\Gamma, x : A} & := & \monTyRel{\Gamma}
                                      , x_0 : {\monTy{A}}_{0}
                                      , x_1 : {\monTy{A}}_{1}
                                      , x_\varepsilon : \monTyRel{A}\,x_0\,x_1
    \end{array}
  \]

  \begin{center}
    \textbf{Monotone translation on terms and types}
  \end{center}
  \vspace{0.5em}
  \[
    \begin{array}[t]{lcl}
      \monTy{A} & := & \El\,\monTm{A} : \poset \\
      \\
      {\monTm{ x }} & := & x\\
      {\monTm{\square{}_i}} & := & \uU_i \\[0.6em]
      {\monTm{\P x : A . B}} & := & \PiU~\monTm{ A }~\monTmRel{\l x : A . B} \\
      {\monTm{t~u}} & := & \monTm{t}~\monTm{u}\\
      {\monTm{\l x : A . t}} & := & \l x : \monTy{ A } . \monTm{t} \\
      {\monTm{\nat}} & := & \natU \\
      {\monTm{\?_A}} & := & \?_{\monTm{A}}\\
      {\monTm{\err_A}} &:= & \err_{\monTm{A}} \\
      \monTm{\ascdom{t}{A}{B}} & :=&
                                     \downcast{}_{\raisebox{-2pt}{\tiny$\ElRel{\monTm{B}{\sqsubseteq}\unkU}$}}\upcast_{\raisebox{-2pt}{\tiny$\ElRel{\monTm{A}{\sqsubseteq}\unkU}$}}\monTm{t}
    \end{array}%
    \begin{array}[t]{lcl}
      \monTyRel{A} &:=& \ElRel\,\monTmRel{A} : \monTy{A} \distr \monTy{A}\\
      \\
      {\monTmRel{ x }} & := & x_\varepsilon\\
      {\monTmRel{\square{}_i}} & := & \uU\dash{\sqsubseteq}_{i} \\[0.6em]
      {\monTmRel{\P x : A . B}} & := & \PiU\dash{\sqsubseteq}~\monTmRel{ A
                                       }~\monTmRel{\l x : A. B}\\
      {\monTmRel{t~u}} & := & \monTmRel{t}~{\monTm{u}}_0~{\monTm{u}}_1~\monTmRel{u}\\
      {\monTmRel{\l x : A .~t}} &:= &\l (x_0\,x_1 : \monTy{A})(x_\varepsilon :
                                      \monTyRel{A}\,x_0\,x_1).~\monTmRel{t} \\
      {\monTmRel{\nat}} & := & \natU\dash{\sqsubseteq} \\
      {\monTmRel{\?_A}} & := & \texttt{refl}~\monTy{A}~\?_{\monTm{A}}\\
      {\monTmRel{\err_A}} &:= & \texttt{refl}~\monTy{A}~\err_{\monTm{A}} \\
      {\monTmRel{ \ascdom{t}{A}{B} }} &:= & \downcast{}_{\raisebox{-2pt}{\tiny$\ElRel{\monTm{B}{\sqsubseteq}\unkU}$}}\dash\texttt{mon}~\upcast_{\raisebox{-2pt}{\tiny$\ElRel{\monTm{A}{\sqsubseteq}\unkU}$}}\!\!\dash\texttt{mon}~\monTmRel{t}
    \end{array}
  \]
  \vspace{0.5em}

  ${ \monTm{-} }_\alpha$ and ${ \monTy{-} }_\alpha$ where $\alpha \in \{0,1\}$
  stand for the variable-renaming
  counterparts of $\monTm{-}$ and $\monTy{-}$.
  \caption{Translation of the monotone model}
  \label{fig:monotone-translation}
\end{small}
\end{figure}

\subsection{Back to Graduality}
\label{sec:grad-term}
The precision order equipping each types of the monotone model can be reflected back to
\CCIC, giving rise to the {\em propositional precision} judgment:
\begin{align}
  \label{def:propositional-precision}
  \tcol{\Gamma} \caty \tcol{t} \precision{\tcol{T}}{\tcol{U}} \tcol{u} \qquad := \qquad \exists e, \quad \monTyRel{\Gamma} \irty e : \monTm{t}
  \precision{\monTm{T}}{\monTm{U}} \monTm{u}.
\end{align}
By the properties of the monotone model~(\cref{thm:universe-hierarchy}), there is at most one witness
up to propositional equality in the target that this judgment holds.
This precision relation bears a similar relationship to the structural precision
$\capre$ as propositional equality with
definitional equality in \CIC{}.
On the one hand, propositional precision can be used to prove precision statements
inside the target type theory, for instance we can show by a straightforward
case analysis on $\tcol{b} : \tcol{\bool}$ that $\tcol{b} : \tcol{\bool} \caty \tcol{\texttt{if } b \texttt{ then } A \texttt{ else } A}
\precision{\tcol{\square}}{\tcol{\square}} \tcol{A}$, a judgment that does not hold for syntactic precision.
In particular, propositional precision is compatible with propositional
equality, and a fortiori it is invariant by conversion in \CCIC: if
$\tcol{t} \equiv \tcol{t'}$, $\tcol{u} \equiv \tcol{u'}$ and
$\tcol{\Gamma} \caty \tcol{t} \precision{\tcol{T}}{\tcol{U}} \tcol{u}$
then
$\tcol{\Gamma} \caty \tcol{t'} \precision{\tcol{T}}{\tcol{U}}
\tcol{u'}$.
On the other hand, propositional precision is not decidable, thus
not suited for typechecking, where structural precision has to be used instead.
\newcommand{\precBool}{\precision{\tcol{\bool}}{\tcol{\bool}}}

\begin{lemma}[Compatibility of structural and propositional precision]
  \label{thm:prop-prop-prec}\text{}
  \begin{enumerate}
  \item If $\catyJ{}{t}{T}$, $\catyJ{}{u}{U}$ and $\vdash \tcol{t} \capre \tcol{u}$ then $\caty \tcol{t} \precision{\tcol{T}}{\tcol{U}} \tcol{u}$.
  \item Conversely, if the target of the translation $\CICIRQ$ is logically
    consistent and
    $\caty \tcol{v_1} \precision{\tcol{\bool}}{\tcol{\bool}}
    \tcol{v_2}$ for normal forms $\tcol{v_1}, \tcol{v_2}$, then
    $\vdash \tcol{v_1} \capre \tcol{v_2}$.
  \end{enumerate}
\end{lemma}
\begin{proof}
  For the first statement, we strengthen the inductive hypothesis,
  proving by induction on the derivation of structural precision the stronger
  statement:
  \begin{center}
    If $\tcol{\GG} \vdash \tcol{t} \capre \tcol{u}$, $\tcol{\GG_1} \caty \tcol{t} : T$ and $\tcol{\GG_2} \caty \tcol{u} : \tcol{U}$
    then there exists a term $e$ such that $\monTy{\GG} \irty e : \monTm{t} \precision{\monTm{T}}{\monTm{U}} \monTm{u}$.
  \end{center}
  The cases for variables (\nameref{infrule:capre-var}) and universes (\nameref{infrule:capre-univ}) hold by reflexivity.
  The cases involving $\?$ (\nameref{infrule:capre-unk},
  \nameref{infrule:capre-unk-univ}) and $\err$ (\nameref{infrule:capre-err}, \nameref{infrule:capre-err-lambda}) amount to $\monTm{\?}$ and
  $\monTm{\err}$ being respectively interpreted as top and bottom elements at
  each type.
  For \nameref{infrule:capre-castr}, we have $\tcol{u} =
  \tcol{\cast{A'}{B'}{t'}}$, $\tcol{B'} = \tcol{U}$, and
  by induction hypothesis $\monTy{\GG} \vdash e : \monTm{t}
  \precision{\monTm{{T}}}{\monTm{A'}} \monTm{t'}$
  and $\monTy{\GG} \vdash \monTm{T} \sqsubseteq \monTm{B'}$.
  Let $j$ be a universe level such that $\monTm{A'} \sqsubseteq \unkU[j]$,
  $\monTm{B'} \sqsubseteq \unkU[j]$.
  By (heterogeneous) transitivity of precision applied to $e$ and a witness of
  $\monTy{\GG} \vdash \monTm{t'} \precision{\monTm{A'}}{\unkU[j]}
  \upcast_{\monTm{A'} \sqsubseteq \unkU^j} \monTm{t'}$ (\cref{lem:adj}), we obtain a proof $e'$ of
  $\monTy{\GG} \vdash e' : \monTm{t} \precision{\monTm{{T}}}{\monTm{B'}}
  \upcast_{\monTm{A'} \sqsubseteq \unkU^j} \monTm{t'}$ and finally, using the
adjunction property,
  a proof $e''$ of
  \[\monTy{\GG} \vdash e : \monTm{t}
  \precision{\monTm{{T}}}{\monTm{B'}}
  \downcast_{\monTm{B'} \sqsubseteq \unkU^j}
  \upcast_{\monTm{A'} \sqsubseteq \unkU^j} \monTm{t'} \equiv
  \monTm{\cast{A'}{B'}{t'}} \equiv \monTm{u}.\]
  The case \nameref{infrule:capre-castl} proceeds in an entirely symmetric fashion
  since we only use the adjunction laws.
  All the other cases, being congruence rules with respect to some term
  constructor, are consequences of the monotonicity of said term constructor with a
  direct application of the inductive hypothesis and inversion of the typing
  judgments.
  For the second statement, by progress (\cref{thm:ccic-psafe}), both
  $\tcol{v_1}$ and $\tcol{v_2}$ are canonical booleans, so we can proceed by case analysis on the canonical forms $\tcol{v_1}$ and $\tcol{v_2}$ that are either
  $\tcol{\btrue}, \tcol{\bfalse}, \tcol{\err_\bool}$ or $\tcol{\?_\bool}$, ruling out the impossible cases by inversion of the
  premise $\caty \tcol{v_1} \precBool \tcol{v_2}$ and logical consistency of $\irty$. Out of the $16$
   cases, we obtain that only the following $9$ cases are possible:
  $$
  \begin{array}{lrclrrclrrcl}
    \caty & \tcol{\err_\bool} & \hspace{-1em} \precBool \hspace{-1em}  & \tcol{\err_\bool}
    & \qquad \caty & \tcol{\err_\bool} &\hspace{-1em} \precBool \hspace{-1em} &  \tcol{\btrue}
    & \qquad \caty & \tcol{\err_\bool} &\hspace{-1em} \precBool \hspace{-1em}  & \tcol{\bfalse} \\
    \caty & \tcol{\err_\bool} & \hspace{-1em} \precBool \hspace{-1em}  & \tcol{\?_\bool}
    & \caty & \tcol{\btrue} &\hspace{-1em} \precBool \hspace{-1em}  & \tcol{\btrue}
    & \caty &  \tcol{\btrue} &\hspace{-1em} \precBool \hspace{-1em}  & \tcol{\?_\bool} \\
    \caty & \hspace{-1em} \tcol{\bfalse} &\hspace{-1em} \precBool \hspace{-1em}  & \tcol{\bfalse}
    & \caty & \hspace{-1em} \tcol{\bfalse} &\hspace{-1em} \precBool \hspace{-1em}   &\tcol{\?_\bool}
    &\caty & \tcol{\?_\bool} &\hspace{-1em} \precBool \hspace{-1em}   & \tcol{\?_\bool}
  \end{array}
  $$
  For each case, a corresponding rule exists for the structural precision,
  proving that $\vdash v_1 \capre v_2$.
\end{proof}

With a similar method, we show that 
\CCICs satisfies graduality, which is the key missing point of
\cref{sec:gcic} and the \emph{raison d'etre} of the monotone model.

\begin{theorem}[Graduality for \CCICs]
  \label{thm:graduality-gcics}
  % \GCICP and
  For $\tcol{\Gamma }\caty \tcol{t} : \tcol{T}, \tcol{\Gamma} \caty \tcol{t'} : \tcol{T}$ and $\tcol{\Gamma} \caty \tcol{u} : \tcol{U}$, we have
  
  \begin{itemize}
  \item (DGG) If
    $\tcol{\Gamma} \caty \tcol{t} \precision{\tcol{T}}{\tcol{T}} \tcol{t'}$
    then $ \tcol{t} \obsRef \tcol{t'}$;
  \item (Ep-pairs)
    If $\tcol{\Gamma} \caty \tcol{T} \precision{\tcol{\square}}{\tcol{\square}} \tcol{U}$ then
    \[ \tcol{\Gamma} \caty \tcol{\cast{T}{U}{t}} \precision{\tcol{U}}{\tcol{U}} \tcol{u}
      \Leftrightarrow  \tcol{\Gamma} \caty \tcol{t} \precision{\tcol{T}}{\tcol{U}} \tcol{u}
      \Leftrightarrow  \tcol{\Gamma} \caty \tcol{t}
      \precision{\tcol{T}}{\tcol{T}} \tcol{\cast{U}{T} u}, \]
    Furthermore, $\tcol{\Gamma} \caty \tcol{{\cast{U}{T}{\cast{T}{U}{t}}} \equiprecise \tcol{t}}$.
  \end{itemize}

\end{theorem}
\begin{proof}\mbox{}
  \begin{itemize}
  \item (DGG)
  Let $\tcol{\mathcal{C}[ - ]} : (\tcol{\Gamma} \vdash \tcol{T}) \Rightarrow
  (\vdash \tcol{\bool})$ be an observation context, by monotonicity of the
  translation
  $\caty \tcol{\mathcal{C}[t]} \precBool \tcol{\mathcal{C}[t']}$.
  By progress, subject reduction (\cref{thm:ccic-psafe}) and strong normalization
  (\cref{thm:ccic-pnorm}) of $\CCICs$, there exists canonical forms
  $\caty \tcol{v},\tcol{v'} : \tcol{\bool}$ such that $\tcol{\mathcal{C}[t]}
  \conv \tcol{v}$ and $\tcol{\mathcal{C}[t']} \conv \tcol{v'}$.
  Since propositional precision is stable by conversion in \CCIC, $\caty \tcol{v'}
  \equiprecise \tcol{v}$.
  Finally, we conclude that $\tcol{v} \conv \tcol{v'}$ by a case analysis on the boolean
  normal forms $\tcol{v}$ and $\tcol{v'}$, that are either $\tcol{\btrue}, \tcol{\bfalse},
  \tcol{\err_\bool}$ or $\tcol{\?_\bool}$: if $\tcol{v}$ and $\tcol{v'}$ are distinct normal
  forms then $\caty \tcol{v'} \equiprecise \tcol{v}$ is a closed proof of an empty type,
  contradicting the consistency of the target.
  \item (Ep-pairs)
    The fact that propositional precision induces an adjunction
    is a direct reformulation of the fact that the relation
    $\precision{\monTm{T}}{\monTm{U}}$ underlies an \eppair{}
    (\cref{thm:universe-hierarchy}.(3)), using the fact that there is at most one
  upcast and downcast between two types.
  Similarly, the equi-precision statement is an application of the first point to the
  proofs
  $$
  \left\{
  \begin{array}{l}
    {\monTy{\Gamma} \irty \_ : \monTm{t} \precision{\monTm{T}}{\monTm{T}}
  \monTm{\cast{T}{U}{\cast{U}{T}{t}}}} \\
    {\monTy{\Gamma} \irty \_ : \monTm{\cast{U}{T}{\cast{T}{U}{t}}} \precision{\monTm{T}}{\monTm{T}}
    \monTm{t}}
  \end{array} \right .
$$
  which holds because $\downcast{}_{\monTm{T}\sqsubseteq\monTm{U}} \circ
  \upcast{}_{\monTm{T}\sqsubseteq\monTm{U}} = \id{}$ in the monotone
  model.
\end{itemize}
\end{proof}

We conjecture that the target $\CICIRQ$ mentioned in the above theorem and
propositions is consistent relative to a strong enough metatheory,\footnote{For
  instance, $\mathcal{ZFC}$ + the existence of Mahlo cardinals~\cite{Setzer00,Forsberg13,DybjerS03}.} that is the
assumed inductive-recursive definition for the universe does not endanger
consistency.
As can be seen from the proof, this hypothesis allows to move from a
contradiction internal to $\CICIRQ$ to a contradiction in the ambient metatheory.

\subsection{Graduality of \CCICP}
\label{sec:grad-non-term}

To prove graduality of \CCICP, we need to provide a model accounting
for both monotony and non-termination. 
The monotone model presented in the previous sections, which gives us
graduality for \CCICs and can be related to the pointed model of
\citet[Section 6.1]{NewLicatalmcs2020}, only accounts for terminating
functions. 
In order to capture also non-termination, we can adapt the Scott model of
\citet[Section 6.2]{NewLicatalmcs2020} based on pointed $\omega$-cpo to our
setting.
We now explain the construction of the main type formers, overloading the
notations from the previous sections.

Types are interpreted as bipointed $\omega$-cpos, that is as orders $(A,
\sqsubseteq)$ equipped with a smallest element $\err_A \in A$, a largest element
$\?_A \in A$ and an operation $\sup_{i}\,a_i$ computing the suprema of countable ascending chains, \ie sequences $(a_i)_{i\in \omega} \in A^\omega$ indexed by the ordinal $\omega =
\{0 < 1 < \ldots \}$ such that $a_i \sqsubseteq^A a_j$ whenever $i < j$.
A monotone function $f : A \to B$ between $\omega$-cpos is called
\emph{$\omega$-continuous} if for any ascending chain $(a_k)_{k\in\omega}$, $\sup_k f\,a_k
= f (\sup_k a_k)$; we write $d : A \distrcont{} B$ for an \eppair{} between
$\omega$-cpos where $\upcast_d$ preserves suprema (the left adjoint $\downcast_d$ automatically preserves suprema).
A type-theoretical construction of the free $\omega$-cpo on a set is described
in~\cite{BidlingmaierFS19,ChapmanUV19} using quotient-inductive-inductive
types (QIIT)~\cite{AltenkirchCDKF18,KaposiKA19}.
We can adapt this technique to provide an interpretation for inductive types,
and in particular natural numbers, throwing in freely a new constructor $\sup$
denoting the suprema of any chain of elements and quotienting by the appropriate
(in)equations: the suprema of a chain is greater than any of its parts $a_i
\sqsubseteq \sup\,a_i$ and an element $b$ that is greater than a chain $\forall
i, a_i \sqsubseteq b$ is greater than its suprema $\sup a_i \sqsubseteq b$.
Functions $f : A \to B$ between types are interpreted as continuous monotone
maps, and the type $A \tocont B$ of continuous monotone
functions is an $\omega$-cpo with suprema computed pointwise.
As in~\cref{sec:realizing-unknown-type}, the construction of the
$\omega$-cpo corresponding to the unknown type is intertwined with the
universe hierarchy.
Assuming by induction that we have $\omega$-cpos $\U_0$, \ldots, $\U_i$ for
universes at level lower than $i$, we follow the seminal work
of~\citet{scott76} on domains and we take the unknown type $\unkMon[i+1]$ to be
a solution to the recursive equation:
\[\unkMon[i+1] \quad\cong\quad \liftNat + (\unkMon[i+1] \tocont \unkMon[i+1]) + \U_0 + \ldots + \U_i \]
The key techniques to build a solution to this equation in the setting of
$\omega$-cpos are detailed in~\cite{SmythP77,Wand79}.
In a nutshell, this construction amounts to iterate the assignment $F(X) :=
\liftNat + (X \tocont X) + \U_0 + \ldots + \U_i$
starting from the initial bipointed $\omega$-cpo
$\liftZero$---the free bipointed $\omega$-cpo on an empty
  type, consisting just of $\liftbot{\liftZero} \sqsubseteq
  \lifttop{\liftZero}$---and to take the colimit of the induced sequence:
\begin{align}
  \label{eq:colimUnkMon}
  \liftZero \distrcont F(\liftZero) \distrcont \ldots\> F^k(\liftZero)\> \ldots\>
  \distrcont   \mathrm{colim}_k\,F^k(\liftZero) =: \unkMon[i+1] 
\end{align}
For this construction to succeed, $F$ should extend to an $\omega$-continuous
functor on the category of $\omega$-cpos and $\omega$-continuous \eppairs{} that
moreover preserves countable sequential colimits as above so that the following hold:
\[\unkMon[i+1] = \mathrm{colim}_k\,F^k(\liftZero) \cong
  \mathrm{colim}_k\,F^{k+1}(\liftZero) \cong F(\mathrm{colim}_k\,F^k(\liftZero)) = F(\unkMon[i+1]).\]
The construction of $\unkMon[i+1]$ as a fixpoint for $F$ should be contrasted
with the construction from~\cref{sec:realizing-unknown-type} where we
essentially describe explicitly a construction of $\unkMon[i+1] :=
F(\unkMon[i])$ in the setting of bipointed posets.
The existence of countable sequential colimits in the category of $\omega$-cpos
and \eppairs{}, as employed in \cref{eq:colimUnkMon}, is an interesting fact
proved in \cite[Theorem 3.1]{Wand79}, which we also use to equip the next universe
of codes $\U_{i+1}$ with an $\omega$-cpo structure.
In brief, we adapt the inductive description of the universe of codes $\U_i$
given in~\cref{fig:IRUniverse} with an additional code $\sup\,(A_k)_{k \in
  \omega} : \U_i$ for suprema of chains of codes $A_k : \U_i$, and
decode it with the function $\El$ satisfying $\El\,(\sup\,(A_k)_{k \in
  \omega}) \cong \mathrm{colim}_k\, (\El~A_k)$.
However, the isomorphism above cannot be used as a definition because
the definition of $\El$ has to respect the quotiented nature of
$\U_i$. In particular, when the chain is the constant chain $(\natU)_{k \in \omega}$, $\sup\,(\natU)_{k \in \omega} = \natU$ and thus
$\El\,(\sup\,(\natU)_{k \in \omega})$
must also be equal to $\liftNat$, which is different from $\mathrm{colim}_k\, \liftNat$.
Technically, we define $\El\,A$ as its isomorphic image onto $\unkMon[i]$,
recovering a canonical choice for the inhabitants of $\El\,(\sup\,(A_k)_{k \in \omega})$.

Note that in contrast with the construction from~\cref{sec:realizing-unknown-type}
that depended on $\stalkCIC{i}{}$ and hence $\El_{i}$, the present construction of
$\unkMon[i]$ does not depend on the construction of $\U_{i}$ and $\El_i$,
cutting the non-wellfounded loop observed in~\cref{sec:bare-model}.
The components that we describe assemble as a model of \CCICP into
\CICIRQ. In order to be able to prove DGG for \CCICP, we first need to
characterize the semantic interpretation of diverging terms of type $\bool$ in the
model.  
\begin{lemma}
  \label{lem:diverge-err}
  If $\tcol{\Gamma} \caty \tcol{t} \inferty \tcol{\bool}$ and $\tcol{t}$ has no weak head
  normal form, then $\monTm{\tcol{t}} = \err_{\monTm{{\tcol{\bool}}}}$.
\end{lemma}
\begin{proof}
  The proof of this lemma is based on the definition of a logical
  relation which is shown to relate $\tcol{t}$ to its translation
  $\monTm{\tcol{t}}$ in the model (a.k.a. the fundamental lemma).
  The precise definition of the logical and proof of the fundamental
  lemma is given in~\cref{sec:divergence-omega-cpo-model}.
\end{proof}

Relativizing the notion of precision $\tcol{\Gamma} \caty \tcol{t}
\precision{\tcol{T}}{\tcol{S}} \tcol{u}$ of the monotone model to use the order
induced by this $\omega$-cpo model instead of the monotone model,
we can replay the steps of
\cref{thm:graduality-gcics} and derive graduality for \CCICP.

\begin{theorem}[Graduality for \CCICP]
  \label{thm:GCICP-graduality}
  For $\tcol{\Gamma }\caty \tcol{t} : \tcol{T}, \tcol{\Gamma} \caty \tcol{t'} : \tcol{T}$ and $\tcol{\Gamma} \caty \tcol{u} : \tcol{U}$, we have
  
  \begin{itemize}
  \item (DGG) If $\tcol{\Gamma} \caty \tcol{t} \precision{\tcol{T}}{\tcol{T}} \tcol{t'}$
    then $ \tcol{t} \obsRef \tcol{t'}$;
  \item (Ep-pairs)
    If $\tcol{\Gamma} \caty \tcol{T} \precision{\tcol{\square}}{\tcol{\square}} \tcol{U}$ then
    \[ \tcol{\Gamma} \caty \tcol{\cast{T}{U}{t}} \precision{\tcol{U}}{\tcol{U}} \tcol{u}
      \Leftrightarrow  \tcol{\Gamma} \caty \tcol{t} \precision{\tcol{T}}{\tcol{U}} \tcol{u}
      \Leftrightarrow  \tcol{\Gamma} \caty \tcol{t}
      \precision{\tcol{T}}{\tcol{T}} \tcol{\cast{U}{T} u}, \]
    Furthermore, $\tcol{\Gamma} \caty \tcol{{\cast{U}{T}{\cast{T}{U}{t}}} \equiprecise \tcol{t}}$.
  \end{itemize}

\end{theorem}

\begin{proof}\mbox{ }
  \begin{itemize}
  \item (DGG)
    Similarly to the proof of \cref{thm:graduality-gcics}, we consider
    a context $\tcol{\mathcal{C}[ - ]} : (\tcol{\Gamma} \vdash \tcol{T}) \Rightarrow
    (\vdash \tcol{\bool})$. We know  by monotonicity of the
    translation that $\caty \tcol{\mathcal{C}[t]} \precBool \tcol{\mathcal{C}[t']}$.
    We need to distinguish whether the evaluation of $\tcol{\mathcal{C}[t]}$ and
    $\tcol{\mathcal{C}[t']}$ terminates or not.
    If $\tcol{\mathcal{C}[t]}$ diverges, we are done. If $\tcol{\mathcal{C}[t]}$
    terminates and $\tcol{\mathcal{C}[t']}$ diverges, by progress,
    $\tcol{\mathcal{C}[t]}$ reduces to a value $\tcol{v}$ and by
    \cref{lem:diverge-err}, $\monTm{\tcol{\mathcal{C}[t']}} =
    \err_{\monTm{{\tcol{\bool}}}}$.
    This means that $\monTm{\tcol{\mathcal{C}[t]}} =
    \err_{\monTm{{\tcol{\bool}}}}$ because $\err_{\monTm{{\tcol{\bool}}}}$ is
    the smallest element of $\liftBool$.
    Since $\monTm{-}$ is stable by conversion, $\monTm{\tcol{v}} =
    \monTm{\tcol{\mathcal{C}[t]}} = \err_{\monTm{{\tcol{\bool}}}}$, and so
    $\tcol{v} = \err_{\tcol{\bool}}$ by case analysis of the possible values for
    $\tcol{v}$.
    If both terminates, then the reasoning is the same as in the proof of
    \cref{thm:graduality-gcics}.
    
  \item (Ep-pairs)
    As for \cref{thm:graduality-gcics}, this fact derives directly from the
    interpretation of the precision order as \eppairs in the $\omega$-cpo model.
  \end{itemize}
\end{proof}

\section{Gradual Indexed Inductive Types}
\label{sec:giit}

We now explore how indexed inductive types, as used in the introduction (\cref{ex:indices}), can be handled in \GCIC.
Recall the definition of \coqe{vec}:
\begin{coq}
Inductive vec (A : Type) : nat -> Type :=
| nil  : vec A 0
| cons : A -> forall n : nat, vec A n -> vec A (S n).
\end{coq}
and recall the difference between {\em parameters} (here, \coqe{A}), which are common to all constructors, and {\em indices} (here, \coqe{n}), which can differ between constructors.
Also recall from \cref{sec:bidirectional-cic} that our formal development does not consider indexed inductive types, only parametrized ones.

This section first explains two alternatives to indexed inductive types that can directly be expressed in \GCIC (\cref{sec:alt-indices}). We then describe how these alternatives actually behave in the gradual setting (\cref{sec:index-induct-type,sec:decid-equal-more}). Finally, we present an extension of \CCIC to directly support indexed inductive types, focusing on the specific case of vectors (\cref{sec:case-vectors}), showing that it combines the advantages of the other approaches. \cref{sec:giit-summary} summarizes our findings.

\subsection{Alternatives to indexed inductive types}
\label{sec:alt-indices}

Indexed inductive types make it possible to define structures that are intrinsically characterized by some property, which holds by construction, as opposed to extrinsically establishing such properties after the fact. There are two well-known alternatives to indexed inductive types for capturing properties intrinsically: type-level fixpoints, and ``forded'' inductive types.

\begin{description}
\item[Type-level fixpoint.] The vector can be defined as a recursive function on the index, at the type level. For instance, the following formulation represents sized lists as nested pairs:

\begin{coq}
  Fixpoint FVec (A : Type) (n : nat) :Type := match n with 0 => unit | S n => A * FVec A n end.
\end{coq}
  Type-level fixpoints can be used as soon as the indices are \emph{concretely
  forceable}~\cite{bradyEtAl:types2004}. Intuitively, concretely forceable indices are those that can be matched upon (like \coqe{n} in this example definition). See \citet{GilbertCST19} for a description of a general translation.

\item[Forded inductive type.] Instead of using an indexed inductive type, one can use a parametrized inductive type, with {\em explicit} equalities as arguments to constructors.% 
\footnote{This technique has reportedly been coined ``fording'' by \citet[\textsection 3.5]{mcbride99}. Fording is in allusion to the Henry Ford quote {\em ``Any customer can have a car painted any color that he wants, so long as it is black.''}}
For instance, vectors can be defined in this style as follows:
\begin{coq}
  Inductive vec_eqdec (A : Type) (n : nat) : Type :=
  | nil_eqdec : eq_nat 0 n -> vec_eqdec A n
  | cons_eqdec : A -> forall m : nat, eq_nat (S m) n -> vec_eqdec A m -> vec_eqdec A n.
\end{coq}
Note that this definition uses \coqe{eq_nat}, the type of decidable equality proofs over natural numbers, for expressing the constraints on \coqe{n} instead of propositional equalities (\eg~\coqe{0=n}), because propositional equality is not available in \GCIC (\cref{sec:limit-propeq}).
\end{description}

In \CIC, these two alternative presentations of an indexed inductive type can be shown internally to be
equivalent. But each of these presentations has advantages and drawbacks depending on the considered system and scenarios of use, so practitioners have different preferences in that respect. More important to us here, these presentations are {\em not} equivalent in \GCIC.

\subsection{Type-level fixpoints}
\label{sec:index-induct-type}

\paragraph{Constructors.}
The definition of \coqe{FVec} above can directly be written in \GCIC, as it
uses only inductive types with parameters (here the unit and product
types and natural numbers). 
The vector constructors can be defined as:
\begin{coq}
  Definition Fnil (A:Type) : FVec A 0 := tt. 
\end{coq}
\begin{coq}
  Definition Fcons (A:Type) (a:A) (n:nat) (v:FVec A n) : FVec A (S n) := (a , v). 
\end{coq}
whose definitions typecheck because \coqe{FVec} computes on its
indices.

\paragraph{Behavior.}
Let us now look at the type computed at $\?_\nat$. Because $\?_\nat$ is an
exceptional term, the fixpoint has to return unknown in the universe: 
\coqe{FVec A ?_nat ==* ?_Type}. This means that the mechanism for
casting a vector into a vector with the unknown index is directly
inherited from the generic mechanism for casting to the unknown type.
Therefore, we get for free the following computation rules, because
they involve embedding-projection pairs:
\begin{coq}
  Fnil A::FVec A ?_nat::FVec A 0 ==* Fnil A
\end{coq}
\begin{coq}
  Fnil A::FVec A ?_nat::FVec A 1 ==* raise
\end{coq}

Similarly, the eliminator \coqe{FVec_rect} can be defined by first
matching on the index, and then on the vector and satisfies the
computation rule of vectors when the index is non-exceptional. 
The only drawback of this encoding is that the behavior of the
eliminator is not satisfactory when the index is unknown.
Consider for instance the following term from \cref{ex:indices}, which unfortunately reduces to 
\coqe{?_nat}:
\begin{coq}
  Fhead ?_nat (Ffilter nat 4 even [ 0 ; 1 ; 2 ; 3 ]) ==* ?_nat
\end{coq}
\noindent
This behavior occurs because the eliminator starts by
matching on the index, which is unknown, and thus has to return the
unknown itself.

\subsection{Fording with decidable equalities}
\label{sec:decid-equal-more}

\paragraph{Constructors}

With the definition of the forded inductive type \coqe{vec_eqdec}, 
the \coqe{nil_eqdec} constructor can legitimately be used to inhabit
\coqe{vec_eqdec A ?_nat}, provided we have an inhabitant (possibly \?)
of \coqe{eq_nat 0 n}.

Note that we can provide the same vector interface as that of the indexed inductive type by
defining the following constructor wrappers, 
using the term \coqe{refl n} of reflexivity on \coqe{eq_nat}:
\begin{coq}
  Definition nil' (A:Type) : vec_eqdec A 0 := nil_eqdec A (refl 0). 
\end{coq}
\begin{coq}
  Definition cons' A a n (v: vec_eqdec A n) : vec_eqdec A (S n) := cons_eqdec A a n (refl n) v. 
\end{coq}
and define the corresponding eliminator \coqe{vec_rect'} accordingly.

\paragraph{Behavior.}  

The computational content of the eliminator on
\coqe{vec_eqdec A ?_nat} is more precise than with \coqe{FVec}: the
eliminator never matches on the proof of equality to produce a term,
but only to guarantee that a branch is not accessible. 
Concretely, this means that we observe the expected reduction:
\begin{coq}
  nil_eqdec A e::vec_eqdec A ?_nat::vec_eqdec A 0 ==* nil_eqdec A e
\end{coq}
Again, the fact that upcasting to \coqe{vec_eqdec A ?_nat} and then
downcasting back is the identity relies on the \CCIC mechanism on the unknown
for the universe, but this time only for the type representing the
decidable equality.
Likewise, the example of \coqe{filter} (\cref{ex:indices}) computes as expected:
\begin{coq}
  Dhead ?_nat (Dfilter nat 4 even [ 0 ; 1 ; 2 ; 3 ]) ==* 0
\end{coq}

On the other hand, an invalid assertion does not produce an error, but a term with an error in place of the equality proof:
\begin{coq}
  nil_eqdec A e::vec_eqdec A ?_nat::vec_eqdec A 1 ==* nil_eqdec A raise
\end{coq}
where \coqe{raise} is at type \coqe{eq_nat 1 0}.
Consequently, we have \coqe{Dhead ?_nat (Dfilter nat 2 even [ 1 ; 3 ]) ==* raise}, 
because the branch of \coqe{Dhead} that deals with the \coqe{nil} case matches 
on the (erroneous) equality proof.
Invalid assertions are therefore very lazily observed, if at all, which is not satisfactory.

Finally, there is a drawback of using decidable equalities, which only manifests when working with the original vector interface (\coqe{nil'}/\coqe{cons'}/\coqe{vec_rect'}). 
In that case, the eliminator does not enjoy the expected computational rule 
on the constructor \coqe{cons'}.
Because the eliminator is defined by induction on natural numbers, therefore it only reduces 
when the index is a concrete natural number, not a variable.

\subsection{Direct support for indexed inductive types: the case of vectors}
\label{sec:case-vectors}

\begin{figure}
  \small
\begin{mathpar}
  \inferrule{  }{\can{\GnilU{A}}}
  \and
  \inferrule{ }{\can{(\GconsU{A}{a}{n}{v})}}
  \\
    \redrule{
    \ascdom{(\Gcons{A}{a}{k}{v})}{\Gvect{A}{(S~n)}}{\Gvect{B}{\?_\nat}} 
  }
  {
    \GconsU{B}{(\ascdom{a}{A}{B})}{n}{(\ascdom{v}{\Gvect{A}{k}}{\Gvect{B}{n}})}
    \hfill
  }[V-cons-\?]\label{red:v-S?}

  \\
  \redrule{
    \ascdom{(\Gcons{A}{a}{k}{v})}{\Gvect{A}{(S~n)}}{\Gvect{B}{(S~m)}} 
  }
  {\Gcons{B}{(\ascdom{a}{A}{B})}{m}{(\ascdom{v}{\Gvect{A}{k}}{\Gvect{B}{m}})}
    \hfill}[V-cons]\label{red:v-S}
  \\
  \redrule{
    \ascdom{(\Gcons{A}{a}{k}{v})}{\Gvect{A}{(S~n)}}{\Gvect{B}{0}} 
  }
  {
    \err_{\Gvect{B}{0}} \myflushright
  }[V-cons-nil]\label{red:v-S0}
  \\

  \redrule{
    \GvectRec{P}{P_{nil}}{P_{cons}}{\GnilU{A}}
  }
  {\ascdom{(\GvectRec{P}{P_{nil}}{P_{cons}}{\Gnil{A}})}{P~0}{P~\?_\nat} \hfill
  }[V-rect-nilu]\label{red:v-rect-unk-0}
  \\
  \redrule{
    \GvectRec{P}{P_{nil}}{P_{cons}}{(\GconsU{A}{a}{n}{v})}
  }
  {\myflushright}[V-rect-consu]\label{red:v-rect-unk-S}
  \\
  \myflushright \ascdom{(\GvectRec{P}{P_{nil}}{P_{cons}}{(\Gcons{A}{a}{n}{v})})}{P~(S~n)}{P~\?_\nat}

\end{mathpar}
\caption{New canonical forms and reduction rules for vectors in \CCIC (excerpt).}
\label{fig:vectors-excerpt}
\end{figure}

Extending \GCIC/\CCIC with direct support for indexed inductive types can provide a fully satisfactory solution, in contrast to the two previously-exposed encodings that both have serious shortcomings. 
The idea is to reason about indices directly in the reduction of casts. 
Here, we expose this approach for the specific case of length-indexed vectors and leave a generalization to future work.
\cref{sec:case-vectors-appendix} describes the extension for vectors in full details; here, 
we only present selected rules (\cref{fig:vectors-excerpt}) and illustrate how reduction works.

\paragraph{Constructors}
We add two new canonical forms,
corresponding to the casts of \coqe{nil} and \coqe{cons} to
$\Gvect{A}{\?_{\nat}}$: namely, $\GnilU{A}$ and
$\GconsU{A}{a}{n}{v}$ (\cref{fig:vectors-excerpt}).
Note that we cannot simply declare casts such as 
$\ascdom{t}{\Gvect{A}{n}}{\Gvect{A}{\?_{\nat}}}$ to be canonical, 
because they involve non-linear occurrences of types (here, $A$).

\paragraph{Reduction rules}
We add reduction rules to conduct casts between vectors
in canonical forms. \cref{fig:vectors-excerpt} presents these rules
when the argument of the cast is a \coqe{cons}.
Rule~\nameref{red:v-S?} propagates the cast on the arguments,
but using the newly-introduced $\GconsUName$, effectively converting precise
information to less precise information.
Rule~\nameref{red:v-S} applies when both source and target indices are successors,
and propagates the cast of the arguments, just like the
standard rule for casting a constructor.
As expected, Rule~\nameref{red:v-S0} raises an error when the indices do not match.

For the eliminator, there are two new computation rules, one for each
 new constructor: \nameref{red:v-rect-unk-0} and
\nameref{red:v-rect-unk-S}. They both apply the eliminator to the 
underlying non-exceptional constructor, and then cast the result back
to $P~\?_\nat$.
Intuitively, these rules transfer the cast on vectors to a cast on
the returned type of the predicate.

\paragraph{Behavior}
Given these rules, we can actually realize the behavior described in
\cref{ex:indices}. For instance, we have both
\begin{coq}
  nil A::vec A ?_nat::vec A 0 ==* nil A
\end{coq}
\begin{coq}
  nil A::vec A ?_nat::vec A 1 ==* raise_A
\end{coq}
and coming back to \cref{ex:indices}, in all three \GCIC variants the term:

\begin{center}
\coqe{head ?_nat (filter nat 4 even [ 0 ; 1 ; 2 ; 3 ])}
\end{center}
\noindent typechecks and reduces to \coqe{0}. Additionally, as expected:

\begin{center}
\coqe{head ?_nat (filter nat 2 even [1 ; 3])}
\end{center}
\noindent typechecks and fails at runtime. And similarly for \cref{ex:specif}.

Note that to be able to define the action of casts on vectors, we have
crucially used the fact that it is possible to discriminate between
\coqe{0}, \coqe{S n} and $\?_\nat$ in the reduction rule.

\subsection{Summary}
\label{sec:giit-summary}
To summarize, the different approaches to define structures with intrinsic properties in \GCIC compare as follows:
\begin{itemize}
\item The type-level fixpoint coincides with the indexed inductive presentation 
on non-exceptional terms, but is extremely imprecise in presence of unknown indices.
\item The forded inductive is more accurate when dealing with unknown indices, but is arguably too permissive with invalid index assertions. 
\item The direct support of the indexed inductive type with additional constructors and reduction rules yields a satisfactory solution. We conjecture that this presentation can be generalized to support arbitrary indexed inductive types as long as they have 
concretely forceable indices; we leave such a general construction for future work. 
\end{itemize}
Recall that fording is only an option in \GCIC when the indices pertain to a type with decidable equality; properly handling general propositional equality in a gradual type theory 
is an open question (\cref{sec:limit-propeq}). 
The constraint of indices being concretely forceable (for type-level fixpoints, direct support)  
are intuitively understandable and expected: gradual typing requires synthesizing dynamic checks, therefore these checks need to be somehow computable.

\section{Limitations and Perspectives}
\label{sec:future-cic}

Up to now, we have left aside three important aspects of \CIC, namely,
impredicativity, $\eta$-equality and propositional equality.
This section explains the challenges induced by each feature, and possibly,
venues to explore.

\subsection{Impredicativity}

In this work, we do not deal with the impredicative sort
\texttt{Prop}, for multiple reasons.
The models used in \cref{sec:realizing-ccic} to justify termination and
graduality crucially rely on the predicativity of the universe hierarchy for the
inductive-recursive definition of codes to be well-founded.
Moreover, the results of \citet[Theorem 6.1]{Palmgren98} show that it is not
possible to endow an impredicative universe with an inductive-recursive
structure in a consistent and strongly-normalizing theory,
hinting that it may be difficult to devise an inductively-defined cast
function between types that belong to an impredicative universe.
Additionally, it seems difficult to avoid the
divergence of $\Omega$ with an impredicative sort, as no universe levels can
be used to prevent a self-application from
being well-typed.

\subsection{$\eta$-equality}

In most presentations of \CIC, and in particular its \Coq
implementation, conversion satisfies an additional rule, called
$\eta$-equality, which corresponds to an extensional property for
functions:
$$
\Gamma \vdash f \conv \l x : A . f~x \quad \mbox{when} \quad \Gamma \vdash f:\Pi x:A . B.
$$
The difficulty of integrating $\eta$-equality in the setting of \GCIC is that the conversion we consider in \CCIC is entirely induced
by a notion of reduction: two
terms are convertible exactly when they have a common reduct up to
$\alpha$-equivalence.
It is well-known that $\eta$-equality cannot be easily modeled using a
rewrite rule, as both $\eta$-expansion and $\eta$-reduction have significant
drawbacks~\cite{10.1145/1047659.1040312}, and so we would have to consider another
approach to the one we took if we were to integrate $\eta$-equality.
The most prominent alternative way is to define conversion
as an alternation of reduction steps (for instance using a weak-head
reduction strategy) not containing $\eta$ and comparison of terms up to congruence and $\eta$-equality.

This approach has been recently formalized by \citet{Abel:POPL2018} in a fully-typed setting. That is, types participate crucially
in the conversion relation: they are maintained during conversion, so that for instance comparison of terms at a $\Pi$-type systematically $\eta$-expands them
before recursively calling conversion at the domain types.
Defining a gradual variant of such a typed conversion might be quite interesting,
but would require a significant amount of work.

On the contrary, a precise, formalized, account is still missing for $\eta$-equality for an untyped conversion as used in practice in the \Coq proof assistant and in \GCIC.
The MetaCoq project, which aims at such a formalized account, leaves the treatment of $\eta$-equality to future work~\cite{SozeauBFTW20}.
While we envision no specific issues to the adaptation to this approach
to gradual typing once a clear and precise solution for
\CIC itself has been reached, solving the issue in a satisfactory way for
\CIC is obviously out of scope for this article.
Thus, while it should in principle be possible to add $\eta$-equality to \GCIC,
either via typed or untyped conversion,
we leave this for future work.

\subsection{Propositional equality}
\label{sec:limit-propeq}

In \CIC, propositional equality \coqe{eq A x y}, corresponds to the Martin-L{\"o}f
identity type~\cite{Martin-Lof-1973}, with a single constructor \coqe{refl} for reflexivity, and the
elimination principle known as \coqe{J}:
\begin{coq}
  Inductive eq (A : Type) (x : A) : A -> Type := refl : eq A x x.
\end{coq}
\begin{coq}
  J : forall (A : Type) (P : A -> Type) (x : A) (t : P x) (y : A) (e : eq A x y), P y
\end{coq}
together with the conversion rule:

\begin{center}
 \coqe{J A P x t x (refl A x) equiv t}
\end{center}

For the sake of exposing the problem, suppose that we can define this identity type in \GCIC, 
while still satisfying canonicity, conservativity with respect to \CIC and graduality.
This means that for an equality \coqe{t = u} involving closed terms
\coqe{t} and \coqe{u} of \CIC, there
should only be three possible canonical forms: \coqe{refl A t} whenever \coqe{t} and \coqe{u} are convertible terms (of type \coqe{A}), as well as \coqe{raise} and \coqe{?}.

Just under these assumptions, 
we can show that there exist two functions that are pointwise equal in \CIC, and hence equal by extensionality, but are no longer equivalent in \GCIC/\CCIC. 
Consider the two functions \coqe{id_nat} and \coqe{add0} below:
\begin{coq}
     id_nat := fun n : nat => n         add0 := fun n : nat => n + 0
\end{coq}
In \CIC, these functions are not convertible, but they are observationally equivalent.
However, they would not be observationally equivalent in \GCIC.
To see why, consider the following term:

\begin{coq}
  test := fun f => J (nat -> nat) (fun _ => bool) id_nat true f (refl id_nat::?_Type::id_nat = f)
\end{coq}
We have \coqe{test id_nat ==* true} because, by \pgrad,
\coqe{refl id_nat::?_Type::id_nat = id_nat ==* refl id_nat}.
However, because \coqe{add0} is {\em not} convertible to \coqe{id_nat},
\coqe{refl id_nat::id_nat = add0} cannot possibly reduce to \coqe{refl}, and
thus would need to reduce either to $\err$ or \?; and so does \coqe{test add0}.

This means that a model for such a gradual type theory would need to be intensional, conversely to the extensional models usually used to justify type theories. Studying such a model as well as exploring alternatives approaches to propositional equality in a gradual type theory are interesting venues for future work.

\section{Related Work}
\label{sec:related}

\mparagraph{Bidirectional typing and unification}
Our framework uses a bidirectional version of the type system of
\CIC.
Although this presentation is folklore among type theory
specialists~\cite{McBride2019}, the type system of \CIC is rarely
presented in this way on paper and has been studied in details only
recently~\cite{LennonBertrand2021}.
However, the bidirectional approach becomes necessary when dealing
with unification and elaboration of implicit arguments.
Bidirectional elaboration is a common feature of proof assistant
implementations, for instance \cite{Asperti2012}, as it clearly delineates what
information is available to the elaboration system in the different typing
modes. In a context with missing information due to implicit arguments, those
implementations face the undecidable higher order unification \cite{Dowek2001}.
In this error-less context, the solution must be a form of under-approximation,
using complex heuristics \cite{Ziliani2017}. Deciding consistency is very close
to unification, as observed by \citet{castagnaAl:popl2019}, but our notion of
consistency over-approximates unification, making sure that unifiable terms are
always consistent, relying on errors to catch invalid over-approximations at
runtime.

\mparagraph{Dependent types with effects}
As explained in this paper, introducing the unknown type of gradual
typing also require, in a dependently-typed setting, to introduce
unknown terms at any type. This means that a gradual dependent type
theory naturally endorses an effectful mechanism which is similar to
having exceptions.
This connects \GCIC to the literature on dependent types and effects.
Several programming languages mix dependent types with effectful computation, either giving up on metatheoretical properties, such as Dependent Haskell~\cite{eisenberg2016dependent}, or by restricting the dependent fragment to pure expressions~\cite{xiPfenning:pldi98,swamyAl:popl2016}.
In the context of dependent type theories,
\citet{PedrotT17,pedrotTabareau:esop2018} have leveraged the monadic approach to
type theory, at the price of a weaker form of dependent large
elimination for inductive types.
The only way to recover full elimination is to accept a weaker form of
logical consistency, as crystallized by the fire triangle between
observable effects, substitution and logical
consistency~\cite{pedrotTabareau:popl2020}.

\mparagraph{Ordered and directed type theories}
The monotone model of \CCIC interpret types as posets in order to
give meaning to the notion of precision.
Interpretations of dependent type theories in ordered structures goes back to
various works on domain theoretic and realizability interpretations of (partial)
Martin-Löf Type Theory~\cite{PalmgrenS90,Ehrhard88}.
More recently, \citet{LicataH11} and \citet{North19} extend type theory with directed
structures corresponding to a categorical interpretation of types, a higher
version of the monotone model we consider.

\mparagraph{Hybrid typing} \cite{ouAl:tcs2004} present a programming language with separate dependently- and simply-typed fragments, using arbitrary runtime checks at the boundary. \citet{knowlesFlanagan:toplas2010} support runtime checking of refinements. In a similar manner, \cite{tanterTabareau:dls2015} introduce casts for subset types with decidable properties in \Coq. They use an axiom to denote failure, which breaks weak canonicity. Dependent interoperability \cite{oseraAl:plpv2012,dagandAl:jfp2018} supports the combination of dependent and non-dependent typing through deep conversions.
All these approaches are more intended as programming languages than as type theories, and none support the notion of (im)precision that is at the heart of gradual typing.

\mparagraph{Dependent contracts}
\cite{greenbergAl:popl2010} relates hybrid typing to dependent contracts, which are dynamically-checked assertions that can relate the result of a function application to its argument~\cite{findlerFelleisen:icfp2002}. The semantics of dependent contracts are subtle because contracts include arbitrary code, and in particular one must be careful not to violate the precondition on the argument in the definition of the postcondition contract~\cite{bloomMcAllester:jfp2006}.  Also, blame assignment when the result and/or argument are themselves higher-order is subtle. Different variants of dependent contracts have been studied in the literature, which differ in terms of the violations they report and the way they assign blame~\cite{greenbergAl:popl2010,dimoulasAl:popl2011}.
An in-depth exploration of blame assignment for gradual dependent type theories such as \GCIC is an important perspective for future work.

\mparagraph{Gradual typing}
The blame calculus of \citet{wadlerFindler:esop2009} considers subset types on base types, where the refinement is an arbitrary term, as in hybrid type checking \cite{knowlesFlanagan:toplas2010}. It however lacks the dependent function types found in other works.
\citet{lehmannTanter:popl2017} exploit the Abstracting Gradual Typing (AGT) methodology~\cite{garciaAl:popl2016} to design a language with imprecise formulas and implication. They support dependent function types, but gradual refinements are only on base types refined with decidable logical predicates.
\citet{eremondiAl:icfp2019} also use AGT to develop approximate normalization and GDTL. While being a clear initial inspiration for this work, the technique of approximate normalization cannot yield a computationally-relevant gradual type theory (nor was its intent, as clearly stated by the authors). We hope that the results in our work can prove useful in the design and formalization of such gradual dependently-typed programming languages. \citet{eremondiAl:icfp2019} study the dynamic gradual guarantee, but not its reformulation as graduality~\cite{newAhmed:icfp2018}, which as we explain is strictly stronger in the full dependent setting. Finally, while AGT provided valuable intuitions for this work,
graduality as embedding-projection pairs was the key technical driver in the design of \CCIC.

\section{Conclusion}
\label{sec:conclusion}

We have unveiled a fundamental tension in the design of gradual dependent type
theories between conservativity with respect to a dependent type theory such as \CIC,
normalization, and graduality.
We explore several resolutions of this Fire Triangle of Graduality, yielding
three different gradual counterparts of \CIC, each compromising with one
edge of the Triangle.
We develop the metatheory of all three variants of \GCIC thanks to a common
formalization, parametrized by two knobs controlling universe constraints on
dependent product types in typing and reduction.
This work opens a number of perspectives for future work, in addition to addressing the limitations discussed in \cref{sec:future-cic}.
The delicate interplay between universe levels and computational
behavior of casts begs for a more flexible approach to the
normalizing \GCICT, for instance using gradual universes.
The approach based on multiple universe hierarchies to support logically
consistent reasoning about exceptional programs~\cite{pedrotAl:icfp2019} could
be adapted to our setting in order to provide a seamless integration
inside a single theory of gradual features together with
standard \CIC without compromising normalization.
This could also open the door to supporting consistent reasoning about gradual
programs in the context of \GCIC.
On the more practical side, there is still a lot of challenges ahead in order to implement a gradual incarnation of \GCIC in \Coq or \Agda, possibly parametrized in order to support the three variants presented in this work.

\bibliography{strings,pleiad,bib,common,refs}

\ifappendix
\clearpage
\appendix

\section{Index of notations}
\label{sec:notations}
\begin{center}
  \begin{tabular}{|c|c|c|p{0.35\textwidth}|}
    \hline
    Description & Symbol & Ref & Remark\\
    \hline
    \multicolumn{4}{|c|}{Section \labelcref{sec:bidirectional-cic}} \\
    \hline
    Universe & $[]_i$ & \cpageref{fig:syntax-cic} & At level $i$\\
    Inductive type & $I\ulev{i}(\orr{a})$ & \cpageref{fig:syntax-cic} & At level $i$ with parameters $\orr{a}$ \\
    Inductive constructor & $ c_k^I\ulev{i}(\orr{a},\orr{b})$ & \cpageref{fig:syntax-cic} & $k$-th constructor of $I$ at level $i$ with parameters $\orr{a}$ and arguments $\orr{b}$ \\
    Inductive destructor & $\match{I}{s}{z.P}{f.\orr{y}.\orr{b}}$ & \cpageref{fig:syntax-cic} & corresponds to fix + match in Coq \\
    Substitution & $t\subs{u}{x}$ & \cpageref{fig:syntax-cic} & extended to parallel substitution \\
    Types of parameters & $\pars(I,i)$ & \cpageref{params} & of inductive $I$ at level $i$ \\
    Types of arguments & $\args(I,i,c_k)$ & \cpageref{params} & of constructor $k$ of inductive $I$ at level $i$ \\
    Substitution in parameters & $\pars(I,i)\parsub{\orr{a}}$ & \cpageref{params} & \\
    Substitution in arguments & $\args(I,i,c_k)\parsub{\orr{a}, \orr{b}}$ & \cpageref{params} & \\
    Context checking & $\vdash \Gamma$ & \cref{fig:bidir} & \\
    Type inference & $\Gamma \vdash t \inferty T$ & \cref{fig:bidir} & \\
    Type checking & $\Gamma \vdash t \checkty T$ & \cref{fig:bidir} & \\
    Constrained inference & $\Gamma \vdash t \pcheckty{\bullet}$ & \cref{fig:bidir} &
    $\bullet$ is either $\Pi$, $I$ or $[]$ \\
    One-step reduction & $\redCCIC$ & \cref{fig:bidir} & \emph{full}, \ie with all congruences \\
    Reduction & $\rtred$ & \cref{fig:bidir} & reflexive, transitive closure of $\redCCIC$ \\
    Conversion & $\conv$ & \cref{fig:bidir} & \\
    \hline
    \multicolumn{4}{|c|}{Section \labelcref{sec:gcic}} \\
    \hline
    Unknown type & $\tcol{\?_{T}}$ & \cpageref{fig:syntax-castcic} & in \CCIC \\
    Error & $\tcol{\err_{T}}$ & \cpageref{fig:syntax-castcic} & \\
    Cast & $\cast{\tcol{T}}{\tcol{T'}}{\tcol{t}}$ & \cpageref{fig:syntax-castcic} & \\
    Level of product type & $\sortOfPi{i}{j}$ & \cref{fig:univ-param} & \\
    Level of product germ & $\castOfPi{i}$ & \cref{fig:univ-param} & \\
    Type heads & $\H$ & \cref{fig:head-germ} & \\
    Head of a type & $\hd(\tcol{T})$ & \cref{fig:head-germ} & \\
    Germ & $\stalkCIC{i}{h}$ & \cref{fig:head-germ} & Least precise type with head $h$ at level $i$ \\
    Parallel reduction & $\paraRed$ & \cref{lem:confluence} & \\
    Canonical term & $\can{\tcol{t}}$ & \cref{fig:CCIC-canonical} & inductive caracterization \\
    Neutral term & $\neu{\tcol{t}}$ & \cref{fig:CCIC-canonical} & inductive caracterization \\
    $\alpha$-consistency & $\tcol{t} \acons \tcol{t'}$ & \cref{fig:acons} & \\
    Consistent conversion & $\tcol{t} \cons \tcol{t'}$ & \cref{def:cons} & Also called \emph{consistency} \\
    Unknown type & $\scol{\?\ulev{i}}$ & \cpageref{gcic-syntax} & in \GCIC, at level $i$ \\
    Elaboration (inference) & $\inferelab{}{t}{}{\Gamma}{t'}{T}$ & \cref{fig:elaboration} & \\
    Elaboration (checking) & $\checkelab{}{t}{}{\Gamma}{t'}{T}$ & \cref{fig:elaboration} & \\
    Elaboration (constrained) & $\pcheckelab{\bullet}{}{t}{}{\Gamma}{t'}{T}$ & \cref{fig:elaboration} & \\
    Structural precision & $\tcol{\GG} \vdash \tcol{t} \capre \tcol{t'}$ & \cref{fig:apre-ccic} & extended to contexts pointwise \\
    Definitional precision & $\tcol{\GG} \vdash \tcol{t} \cdpre \tcol{t'}$ & \cref{fig:apre-ccic} & extended to contexts pointwise \\
    Typing in \CIC/\CCIC & $\cicty$ / $\caty$ & \cref{sec:gcic-theorems2} & to differentiate between systems \\
    Equiprecision & $\tcol{\GG} \vdash \tcol{t} \caequipre \tcol{t'}$ & \cref{def:equipre} & \\
    Erasure & $\eras(\tcol{t})$ & \cref{def:erasure} & \\
    Syntactic precision & $\scol{t} \apre \scol{t'}$ & \cref{fig:apre-gcic} & \\
    \hline
  
  \end{tabular}
\end{center}
\begin{center}
  \begin{tabular}{|c|c|c|p{0.35\textwidth}|}
    \hline
    Description & Symbol & Ref & Remark\\
    \hline
    \multicolumn{4}{|c|}{Section \labelcref{sec:realizing-ccic}} \\
    \hline
    \CIC + Induction-Recursion & \CICIR &  \cref{sec:bare-model} & Target for the discrete model \\
     Judgements for \CICIR & $\irty$ & & \\
    \CICIR + quotients & \CICIRQ & & Target for the monotone models \\
    Universe of codes &$\U$ & \cref{fig:discrete-univ} & \\
    Bipointed poset on inductive I &$\liftErr{I}$ &  \cref{sec:bare-model} & \\
    Top element in $\liftErr{I}$ & $\lifttop{\liftErr{I}}$ & \cref{sec:bare-model}& \\
    Bottom element in $\liftErr{I}$ & $\liftbot{\liftErr{I}}$& \cref{sec:bare-model}& \\
    Bipointed poset on % natural numbers 
    $\nat$ & $\liftNat{}$ & \cref{sec:bare-model}& \\
    Bipointed poset on % dependent sum 
    $\Sigma$ & $\liftSigma{}$ & \cref{sec:bare-model}& \\
    % Code for a type & $\code{.}$& & \\
    Code for nat &$\natU$ & \cref{sec:bare-model}& \\
    Code for dependent product& $\PiU$ & \cref{sec:bare-model}& \\
    Code for universes &$\uU_i$ & \cref{sec:bare-model}& \\
    Code for unknown types &$\unkU[i]$ & \cref{sec:bare-model}& \\
    Code for error type &$\errU$ & \cref{sec:bare-model}& \\
    Decoding function to types &$\El$ & \cref{fig:discrete-univ,fig:IRUniverse}& $\El : \U \to \poset$ \\
    Type heads & $\H_i$ & \cref{fig:head-germ} & \\
    Head of a type & $\hd(T)$ & \cref{fig:head-germ} & \\
    Germ as a code &$\stalkCode{i}{h}$ & \cref{sec:bare-model} & \\
    Germ & $\stalkCIC{i}{h}$ & \cref{fig:head-germ} & Least precise type with head ${h \in \H_i}$ at level $i$ \\
    Cast in discrete model&$\cas$ & \cref{fig:cast-implem-discrete}&\\
    Discrete translation of types &$\bareTy{\cdot}$ & \cref{fig:discrete-translation} & \\
    Discrete translation of terms & $\bareTm{\cdot}$& \cref{fig:discrete-translation} & \\
    Order on type $A$ & $\sqsubseteq^{A}$& \cref{sec:poset-model-dtt} &\\
    Type of posets & \poset  & \cref{sec:poset-model-dtt} & \\
    Monotone dependent product & $ \Pmon{}\,A\,B$  & \cref{sec:poset-model-dtt}& \\
    Ep-pairs & $A \distr{} B$ & \cref{def:ep-pair} & \\
    Upcast & $\upcast_d$ & \cref{def:ep-pair}& Embedding part of an \eppair{} $d$ \\
    Downcast & $\downcast_d$ & \cref{def:ep-pair}& Projection part of an \eppair{} $d$ \\
    Monotone unknown type & $\unkMon[i]$ & \cref{sec:realizing-unknown-type} & \\
    Quotiented pairs in $\unkMon[i]$ & $\unkInj{h}{x}$ & \cref{sec:realizing-unknown-type}& \\
    Top element in $\unkMon[i]$ & $\lifttop{\unkMon[i]}$ & \cref{sec:realizing-unknown-type}& \\
    Bottom element in $\unkMon[i]$ & $\liftbot{\unkMon[i]}$& \cref{sec:realizing-unknown-type}& \\
    Decoding function to \eppairs &$\ElRel$ & \cref{fig:IRUniverse}& $\ElRel(A{\sqsubseteq}B) : \El\,A \distr{} \El\,B$ \\
    Precision on terms& $\precision{A}{B}$& \cref{fig:IRUniverse} &  \\
    Monotone translation of types &$\monTy{\cdot}$ & \cref{fig:monotone-translation} & \\
    Monotone translation of terms & $\monTm{\cdot}$& \cref{fig:monotone-translation} & \\
    Propositional precision & $\tcol{\Gamma} \caty \tcol{t} \precision{\tcol{T}}{\tcol{U}} \tcol{u}$&\cref{def:propositional-precision} & \\
    $\omega$-continuous maps & $A \tocont B$ & \cref{sec:grad-non-term} & $A,B$ $\omega$-cpos \\
    $\omega$-continuous \eppair & $A \distrcont{} B$ & \cref{sec:grad-non-term}& \\
    % & & & \\
    % & & & \\
    % & & & \\
    % & & & \\
    % & & & \\
    % & & & \\
    % & & & \\
    % & & & \\
    % & & & \\
    % & & & \\
    % & & & \\
    % & & & \\
    \hline
  \end{tabular}
\end{center}
\newpage

\section{Complements on Elaboration and \CCIC}

This section gives an extended account of \cref{sec:gcic-to-ccic}. The structure is the same, and we refer to the main section when things are already spelled out there.

\subsection{\CCIC}

We state and prove a handful of standard, technical properties of \CCIC, that are useful in the next sections. They should not be very surprising, the main specific point here is their formulation in the bidirectional setting.

\begin{property}[Weakening]
	If $\tcol{\Gamma} \vdash \tcol{t} \inferty \tcol{T}$ then $\tcol{\Gamma, \Delta} \vdash \tcol{t} \inferty \tcol{T}$, and similarly for the other typing judgments.
\end{property}

\begin{proof}
	We show by (mutual) induction on the typing derivation the more general statement that if $\tcol{\Gamma, \Gamma'} \vdash \tcol{t} \inferty \tcol{T}$ then $\tcol{\Gamma, \Delta, \Gamma'} \vdash \tcol{t} \inferty \tcol{T}$. It is true for the base cases (including the variable), and we can check that all rules preserve it.
\end{proof}

\begin{property}[Substitution]
	If $\tcol{\Gamma, x : A, \Delta} \vdash \tcol{t} \inferty \tcol{T}$ and $\tcol{\Gamma} \vdash \tcol{u} \checkty \tcol{A}$ then $\tcol{\Gamma, \Delta\subs{u}{x}} \vdash \tcol{t\subs{u}{x}} \inferty \tcol{S}$ with $\tcol{S} \conv \tcol{T\subs{u}{x}}$.
\end{property}

\begin{proof}
	Again, the proof is by mutual induction on the derivation. In the checking judgment, we use the transitivity of conversion to conclude. In the constrained inference, we need injectivity of type constructors, which is a consequence of confluence.
\end{proof}

\begin{property}[Validity]
	\label{prop:validity}
	If $\tcol{\Gamma} \vdash \tcol{t} \inferty \tcol{T}$ and $\vdash \tcol{\Gamma}$, then $\tcol{\Gamma} \vdash \tcol{T} \pcheckty{[]} \tcol{[]_i}$ for some $i$.
\end{property}

\begin{proof}
	Once again, this is a routine induction on the inference derivation, using subject reduction to handle the reductions in the constrained inference rules, to ensure that the reduced type is still well-formed. The hypothesis of context well-formedness is needed for the base case of a variable, to get that the type obtained from the context is indeed well-typed.
\end{proof}

\subsection{Precision and Reduction}
\label{sec:precision-reduction}

\paragraph{Structural lemmas} Let us start our lemmas by counterparts to the weakening and substitution lemmas for precision.

\begin{lemma}[Weakening of precision]
	If $\tcol{\GG} \vdash \tcol{t} \capre \tcol{t'}$, then $\tcol{\GG, \DD} \vdash \tcol{t} \capre \tcol{t'}$ for any $\tcol{\DD}$.
\end{lemma}

\begin{proof}
	This is by induction on the precision derivation, using weakening of \CCIC to handle the uses of typing.
\end{proof}

\begin{lemma}[Substitution and precision]
	If $\tcol{\GG, x : S \mid S', \DD} \vdash \tcol{t} \capre \tcol{t'}$, $\tcol{\GG} \vdash \tcol{u} \capre \tcol{u'}$, $\tcol{\fs{\GG}} \vdash \tcol{u} \checkty \tcol{S}$ and $\tcol{\sn{\GG}} \vdash \tcol{u'} \checkty \tcol{S'}$ then $\tcol{\GG, \DD\subs{u \mid u'}{x}} \vdash \tcol{t\subs{u}{x}} \capre \tcol{t'\subs{u'}{x}}$.
\end{lemma}

\begin{proof}
	The substitution property follows from weakening, again by induction on the precision derivation. Weakening is used in the variable case where $x$ is replaced by $u$ and $u'$, and the substitution property of \CCIC appears to handle the uses of typing.
\end{proof}

\paragraph{Catch-up lemmas}
With these structural lemmas at hand, let us turn to the proofs of the catch-up lemmas.

%\begin{lemma}[Universe catch-up]
%	\label{lem:catchup-univ-annex}
%	Under the hypothesis that $\tcol{\fs{\GG}} \capre \tcol{\sn{\GG}}$, if $\tcol{\GG} \vdash \tcol{[]_i} \cdpre T'$ and $\sn{\GG} \vdash \tcol{T'} \pcheckty{[]} \tcol{[]_{j}}$, then $\tcol{T'} \rtred \tcol{\?_{[]_{j}}}$ and $i + 1 \leq j$ or $\tcol{T'} \rtred \tcol{[]_i}$.
%\end{lemma}

\begin{proof}[Proof of \cref{lem:catchup-univ}]
	We want to prove the following: under the hypothesis that $\tcol{\fs{\GG}} \capre \tcol{\sn{\GG}}$, if $\tcol{\GG} \vdash \tcol{[]_i} \cdpre \tcol{T'}$ and $\tcol{\sn{\GG}} \vdash \tcol{T'} \pcheckty{[]} \tcol{[]_{j}}$, then either $\tcol{T'} \rtred \tcol{\?_{[]_{j}}}$ with $i + 1 \leq j$, or $\tcol{T'} \rtred \tcol{[]_i}$.

	The proof is by induction on the precision derivation, mutually with the same property where $\cdpre$ is replaced by $\capre$.

	Let us start with the proof for $\capre$. Using the precision derivation, we can decompose $\tcol{T'}$ into $\tcol{\cast{U_{n-1}}{S_n}{\dots \cast{U_1}{S_2}{T''}}}$, where the casts come from \nameref{infrule:capre-castr} rules, and $\tcol{T''}$ is either $\tcol{[]_i}$ (rule \nameref{infrule:capre-univ}) or $\tcol{?_{S}}$ for some $\tcol{S}$ (rule \nameref{infrule:capre-unk}), and we have $\tcol{\GG} \vdash \tcol{[]_{i+1}} \cdpre \tcol{S_k}$, $\tcol{\GG} \vdash \tcol{[]_{i+1}} \cdpre \tcol{T_k}$ and $\tcol{\GG} \vdash \tcol{[]_{i+1}} \cdpre \tcol{S}$.
	By induction hypothesis, all of $\tcol{S_k}$, $\tcol{T_k}$ and $\tcol{S}$ reduce either to $\tcol{[]_{i+1}}$ or some $\tcol{\?_{[]_l}}$ with $i + 1 \leq l$. Moreover, because $\tcol{T'}$ type-checks against $\tcol{[]_j}$, we must have $\tcol{S_n} \conv \tcol{[]_j}$. This implies that $\tcol{S_n}$ cannot reduce to $\tcol{?_{[]_{l}}}$ by confluence, and thus it must reduce to $\tcol{[]_{i+1}}$.

	Using that $i+1 \leq l$ and rules \nameref{redrule:down-unk}, \nameref{redrule:univ-univ} and \nameref{redrule:up-down} giving respectively
	\begin{align*}
	\tcol{\cast{?_{[]_l}}{X}{\?_{?_{[]_l}}}} &\redCCIC \tcol{\?_X} \\
	\tcol{\cast{[]_{i+1}}{[]_{i+1}}{t}} &\redCCIC \tcol{t} \\
	\tcol{\cast{?_{[]_{l}}}{X}{\cast{[]_{i+1}}{?_{[]_{l}}}{t}}} &\redCCIC \tcol{\cast{[]_{i+1}}{X}{t}}
	\end{align*}
	we can reduce away all casts. We thus get $\tcol{T'} \rtred \tcol{[]_i}$ or $\tcol{T'} \rtred \tcol{\?_{[]_{i+1}}}$, as expected.

	For $\cdpre$, if $\tcol{\GG} \vdash \tcol{[]_i} \cdpre \tcol{T'}$ then by decomposing the precision derivation there is an $\tcol{S'}$ such that $\tcol{T'} \rtred \tcol{S'}$, $\tcol{\GG} \vdash \tcol{[]_i} \capre \tcol{S'}$, and by subject reduction $\tcol{\fs{\GG}} \vdash \tcol{S'} \pcheckty{[]} \tcol{[]_j}$. By induction hypothesis, either $\tcol{S'} \rtred \tcol{[]_i}$ or $\tcol{S'} \rtred \tcol{?_{[]_{i+1}}}$, and composing both reductions we get the desired result.
\end{proof}

%\begin{lemma}[Type catch-up]
%	\label{lem:catchup-type-long}
%	We have the following property, under the hypothesis that $\tcol{\fs{\GG}} \capre \tcol{\sn{\GG}}$:
%	\begin{enumerate}
%		\item if $\GG \vdash \tcol{\?_{[]_i}} \cdpre T'$ and $\sn{\GG} \vdash \tcol{T'} \pcheckty{[]} \tcol{[]_{j}}$, then $\tcol{T'} \rtred \tcol{\?_{[]_{j}}}$ and $i \leq j$;
%		\item if $\GG \vdash \tcol{\P x : A.B} \cdpre \tcol{T'}$, $\tcol{\fs{\GG}} \vdash \tcol{\P x : A. B} \inferty \tcol{[]_i}$ and $\sn{\GG} \vdash \tcol{T'} \pcheckty{[]} \tcol{[]_j}$ then $\tcol{T'} \rtred \tcol{\?_{[]_j}}$ and $i \leq j$, or $\tcol{T'} \rtred \tcol{\P x : A' . B'}$ for some $\tcol{A'}$ and $\tcol{B'}$ such that $\GG \vdash \tcol{\P x : A . B} \cdpre \tcol{\P x :A' . B'}$;
%		\item if $\GG \vdash \tcol{I(\orr{a})} \cdpre \tcol{T'}$, $\tcol{\fs{\GG}} \vdash \tcol{I(\orr{a})} \inferty \tcol{[]_i}$ and $\sn{\GG} \vdash \tcol{T'} \pcheckty{[]} \tcol{[]_j}$ then $\tcol{T'} \rtred \tcol{\?_{[]_j}}$ and $i \leq j$, or $\tcol{T'} \rtred \tcol{I(\orr{a'})}$ for some $\tcol{a'}$ such that $\GG \vdash \tcol{I(\orr{a})} \cdpre \tcol{I(\orr{a'})}$.
%	\end{enumerate}
%\end{lemma}

\begin{proof}[Proof of \cref{lem:catchup-type}]
	The proof of those catch-up lemmas is very similar to the previous one for structural precision, but this time without the need for induction---we use \cref{lem:catchup-univ} instead. We show the one for product types, the others are identical.

	First, let us show the property for $\capre$. Decompose $\tcol{T'}$ into $\tcol{\cast{U_{n-1}}{S_n}{\dots \cast{U_1}{S_2}{T''}}}$, where $\tcol{T''}$ is not a cast, but either some $\tcol{\?_{S}}$ or a product type structurally less precise than $\tcol{\P x : A.B}$. Now by \cref{lem:catchup-univ}, $\tcol{U_k}$, $\tcol{T_k}$ and possibly $\tcol{S}$ all reduce to $\tcol{[]}$ or $\tcol{\?_{[]}}$. Using the same reduction rules as before, all casts can be reduced away, leaving us with either $\tcol{\?_{[]}}$ or a product type structurally less precise than $\tcol{\P x : A.B}$, as stated.

\end{proof}

\begin{proof}[Proof of \cref{lem:catchup-lambda}]
	The proof still follows the same idea: decompose the less precise term as a series of casts, and show that all those casts can be reduced, using \cref{lem:catchup-type} for product types. However it is somewhat more complex, because the reduction of a cast between product types does a substitution, which we need to handle using the previous substitution lemma for precision.

	Let us now detail the reasoning. First, decompose $\tcol{s'}$ into $\tcol{\cast{U_{n-1}}{S_n}{\dots \cast{U_1}{S_2}{u'}}}$, where $\tcol{u'}$ is either $\tcol{\l x : A''. t''}$ or $\tcol{\?_{S}}$ for some $\tcol{S}$. All of the $\tcol{S_k}$, $\tcol{U_k}$ and possibly $\tcol{S}$ are definitionally less precise than $\tcol{\P x : A. B}$. By definition of $\cdpre$ they all reduce to a term structurally less precise than a reduct of $\tcol{\P x : A . B}$, which must be a product type, and thus by \cref{lem:catchup-type} they all reduce to either some $\tcol{?_{[]_{j}}}$ or some product type. Moreover, given the typing hypothesis and confluence $\tcol{S_n}$ can only be in the second case. By rule \nameref{redrule:down-unk}, we get
			\[\tcol{\cast{?_{[]}}{X}{?_{?_{[]}}} \redCCIC ?_{X}}\]
	so if $\tcol{S}$ is $\tcol{\?_{[]}}$ we can reduce the innermost casts until it is (knowing that we will encounter one because $\tcol{S_n}$ is a product type), then use rule \nameref{redrule:prod-unk} on $\tcol{u'}$ if it applies, so that without loss of generality we can suppose that $\tcol{u'}$ is an abstraction.

	Now we show that all casts reduce, and that this reduction preserves precision, starting with the innermost one. There are three possibilities for that innermost cast.

	If it is $\tcol{\cast{\stalkCIC{j}{\Pi}}{?_{[]_j}}{u'}}$, then by typing this cannot be the outermost cast, and thus rule \nameref{redrule:up-down} applies to get
		\[\tcol{ \cast{\?_{[]_j}}{X}{\cast{\stalkCIC{j}{\Pi}}{\?_{[]_j}}{u'}} \redCCIC \cast{\stalkCIC{j}{\Pi}}{X}{u'} }\]

	In the second case, the cast is some $\tcol{\cast{\P x : A_1. B_1}{\P x : A_2. B_2}{\l x : A''. t''}}$, and rule \nameref{redrule:prod-prod} applies to give
        \[\begin{array}{r}
                  \tcol{\cast{\P x : A_1. B_1}{\P x : A_2. B_2}{\l x :
                  A''. t''}} \redCCIC \hspace{12em}  \\ \tcol{\l x : A''
                  . \cast{B_1\subs{\cast{A_2}{A_1}{x}}{x}}{B_2}{t''\subs{\cast{A_2}{A''}{x}}{x}}} \end{array}\]
	Moreover, using the precision hypothesis of \nameref{infrule:capre-castr}, we know that $\tcol{\GG} \vdash \tcol{\P x : A . B} \cdpre \tcol{\P x : A_1 . B_2}$ and $\tcol{\GG} \vdash \tcol{\P x : A . B} \cdpre \tcol{\P x : A_2 . B_2}$. From the first one, using substitution and rule \nameref{infrule:capre-castr}, we get that $\tcol{\GG, x : A \mid A_2} \vdash \tcol{B} \cdpre \tcol{B_1\subs{\cast{A_2}{A_1}{x}}{x}}$.
  The second gives in particular that $\tcol{\GG} \vdash \tcol{A} \cdpre \tcol{A_2}$. Finally, inverting the proof of $\tcol{\GG} \vdash \tcol{\l x : A . t} \capre \tcol{\l x : A'' . t''}$ we also have $\tcol{\GG} \vdash \tcol{A} \capre \tcol{A''}$ and $\tcol{\GG, x : A \mid A''} \vdash \tcol{t} \capre \tcol{t''}$. From this, again by substitution, we can derive $\tcol{\GG, x : A \mid A''} \vdash \tcol{t} \capre \tcol{t''\subs{\cast{A_2}{A''}{x}}{x}}$. Combining all of those, we can construct a derivation of
	\[ \tcol{\GG} \vdash \tcol{\l x : A . t} \capre \tcol{\l x : A_2 . \cast{B_1\subs{\cast{A_2}{A_1}{x}}{x}}{B_2}{t'\subs{\cast{A_2}{A''}{x}}{x}}} \]
	by a use of \nameref{infrule:capre-abs} followed by one of \nameref{infrule:capre-castr}.

	The last case corresponds to $\tcol{\cast{\P x : A'' . B''}{?_{[]_j}}{u'}}$ when $\tcol{\P x : A'' . B''}$ is not $\tcol{\stalkCIC{j}{h}}$, in which case the reduction that applies is \nameref{redrule:prod-germ}, giving
	\[\tcol{\cast{\P x : A'' . B''}{?_{[]_j}}{u'}} \redCCIC
		\tcol{\cast{\?_{[]_{\castOfPi{j}}} \to \?_{[]_{\castOfPi{j}}}}{?_{[]_j}}{\cast{\P x : A'' . B''}{\?_{[]_{\castOfPi{j}}} \to \?_{[]_{\castOfPi{j}}}}{u'}}}\]
	For this reduct to be less precise that $\tcol{\l x : A. t}$, we need that all types involved in the casts are definitionally precise than $\tcol{\P  x : A . B}$, as we already have that $\tcol{\GG} \vdash \tcol{\l x : A . t } \capre \tcol{u'}$. For $\tcol{\?_{[]_j}}$ and $\tcol{\P x : A''. B''}$ it is direct, as they were obtained using \cref{lem:catchup-type} with a reduct of $\tcol{\P x : A . B}$.
	Thus only the germ remains, for which it suffices to show that both $\tcol{A}$ and $\tcol{B}$ are less precise than $\tcol{?_{[]_{\castOfPi{j}}}}$. Because $\tcol{\P x : A. B}$ is typable and less precise than $\tcol{\?_{[]_j}}$, we know that $\tcol{\fs{\GG}} \vdash \tcol{A} \pcheckty{[]} \tcol{[]_k}$ and $\tcol{\fs{\GG}, x : A} \vdash \tcol{B} \pcheckty{[]} \tcol{[]_l}$ with $\sortOfPi{k}{l} \leq j$, thus $k \leq \castOfPi{j}$ and $l \leq \castOfPi{j}$. Therefore $\tcol{\GG} \vdash \tcol{A} \capre \tcol{\?_{[]_{\castOfPi{j}}}}$ using rule \nameref{infrule:capre-unk-univ}, and similarly for $\tcol{B}$.

	Note that this last reduction is the point where the system under consideration plays a role: in \CCICT, the reasoning does not hold. However, when considering only terms without $\tcol{\?}$, this case never happens, and thus the rest of the proof still applies.

	Thus, all casts must reduce, and each of those reductions preserves precision, so we end up with a term $\tcol{\l x : A' . t'}$ such that $\tcol{\GG} \vdash \tcol{\l x : A . t} \capre \tcol{\l x : A' . t'}$, as expected.
\end{proof}

%\begin{lemma}[Constructors and inductive error catch-up]
%	\label{lem:pre-constructor}
%	If $\tcol{\GG} \vdash \tcol{c(\orr{a},\orr{b})} \capre \tcol{s'}$, $\tcol{\fs{\GG}} \vdash \tcol{c(\orr{a},\orr{b})} \inferty \tcol{I(\orr{a})}$ and $\tcol{\sn{\GG}} \vdash \tcol{s'} \pcheckty{I} \tcol{I(\orr{a'})}$, then either $\tcol{s'} \rtred \tcol{\?_{I(\orr{a'})}}$ or $\tcol{s'} \rtred \tcol{c(\orr{a'},\orr{b'})}$ with $\tcol{\GG} \vdash \tcol{c(\orr{a},\orr{b})} \capre \tcol{c(\orr{a'},\orr{b'})}$.
%
%	Similarly, if $\tcol{\GG} \vdash \tcol{\?_{I(\orr{a})}} \capre \tcol{s'}$, $\tcol{\fs{\GG}} \vdash \tcol{\?_{I(\orr{a})}} \inferty \tcol{I(\orr{a})}$ and $\tcol{\sn{\GG}} \vdash \tcol{s'} \pcheckty{I} \tcol{I(\orr{a'})}$, then $\tcol{s'} \rtred \tcol{\?_{I(\orr{a'})}}$ with $\tcol{\GG} \vdash \tcol{I(\orr{a})} \capre \tcol{I(\orr{a'})}$.
%\end{lemma}

\begin{proof}[Proof of \cref{lem:catchup-cons}]
	We start by the proof of the second property. We have as hypothesis that $\tcol{\GG} \vdash \tcol{\?_{I(\orr{a})}} \capre \tcol{s'}$, $\tcol{\fs{\GG}} \vdash \tcol{\?_{I(\orr{a})}} \inferty \tcol{I(\orr{a})}$ and $\tcol{\sn{\GG}} \vdash \tcol{s'} \pcheckty{I} \tcol{I(\orr{a'})}$, and wish to prove that $\tcol{s'} \rtred \tcol{\?_{I(\orr{a'})}}$ with $\tcol{\GG} \vdash \tcol{I(\orr{a})} \capre \tcol{I(\orr{a'})}$.

	As previously, decompose $\tcol{s'}$ as $\tcol{\cast{U_{n-1}}{S_n}{\dots \cast{U_1}{S_2}{\?_{I(\orr{a''})}}}}$, where all $\tcol{U_k}$, $\tcol{S_k}$ and $\tcol{I(\orr{a''})}$ are definitionally less precise than $\tcol{I(\orr{a})}$, and thus reduce to either $\tcol{\?_{[]_l}}$ for some $l$, or $\tcol{I(\orr{c})}$ for some $\tcol{\orr{c}}$, and $\tcol{S_n}$ can only be the second by typing. Using the three rules \nameref{redrule:ind-unk}, \nameref{redrule:up-down} and \nameref{redrule:ind-germ}, we respectively get
  \begin{align*}
	\tcol{\cast{I(\orr{c})}{I(\orr{c'})}{\?_{I(\mathbf{c''})}}}
    &\redCCIC \tcol{\?_{I(\orr{c'})}} \\
	\tcol{ \cast{\?_{[]_j}}{X}{\cast{\stalkCIC{j}{I}}{\?_{[]_j}}{u'}}}
    &\redCCIC \tcol{\cast{\stalkCIC{j}{I}}{X}{u'}} \\
	\tcol{\cast{I(\orr{c})}{\?_{[]_j}}{u'}}
    & \redCCIC \tcol{\cast{\stalkCIC{j}{I}}{\?_{[]_j}}{\cast{I(\orr{c})}{\stalkCIC{j}{I}}{u'}}}
  \end{align*}
	we can reduce all casts: \nameref{redrule:up-down} (possibly using \cref{redrule:ind-germ} first) removes all casts through $\?_{[]}$; we can then use \nameref{redrule:ind-unk} to propagate $\?_{I(\orr{a''})}$ all the way through the casts, ending up with $\tcol{\?_{S_n}}$ which is the term we sought.

	For the first property, again decompose $\tcol{s'}$ as $\tcol{\cast{U_{n-1}}{S_n}{\dots \cast{U_1}{S_2}{u'}}}$ where $\tcol{u'}$ does not start with a cast. If $\tcol{u'}$ is some $\tcol{\?_{I(\orr{a''})}}$, we can re-use the proof above and are finished. Otherwise $\tcol{u'}$ must be of the form $\tcol{c(\orr{a''},\orr{b''})}$. Again we reduce the casts starting with the innermost, using rules \nameref{redrule:up-down} and \nameref{redrule:ind-germ} to remove the occurrences of $\tcol{\?_{[]}}$. The last case to handle is $\tcol{\cast{I(\orr{c})}{I(\orr{c'})}{c(\orr{a_3},\orr{b_3})}}$. Then rule \nameref{redrule:ind-ind} applies, and it preserves precision by repeated uses of the substitution property, and giving a term with $\tcol{c}$ as a head constructor. Thus, we get the desired term with $\tcol{c}$ as a head constructor and arguments less precise than $\tcol{\orr{a}}$ and $\tcol{\orr{b}}$, respectively.
\end{proof}

\paragraph{Simulation}

% \begin{theorem}[Simulation of reduction]
% 	The two following properties are true:
% 	\begin{enumerate}
% 		\item if $\cdpre \tcol{\GG}$, $\tcol{\fs{\GG}} \vdash \tcol{t} \inferty \tcol{T}$, $\tcol{\sn{\GG}} \vdash \tcol{t'} \inferty \tcol{T'}$, $\tcol{\GG} \vdash \tcol{t} \capre \tcol{t'}$ and $t \rtred s$ then there exists $\tcol{s'}$ such that $\tcol{t'} \rtred \tcol{s'}$ and $\tcol{\GG} \vdash \tcol{t'} \capre \tcol{s'}$,
% 		\item if $\cdpre \tcol{\GG}$, $\tcol{\fs{\GG}} \vdash \tcol{t} \inferty \tcol{T}$, $\tcol{\sn{\GG}} \vdash \tcol{t'} \inferty \tcol{T'}$, $\tcol{\GG} \vdash \tcol{t} \cdpre \tcol{t'}$ and $\tcol{t} \rtred \tcol{s}$ then $\tcol{\GG} \vdash \tcol{s} \cdpre \tcol{t'}$.
% 	\end{enumerate}
% \end{theorem}

\begin{proof}[Proof of \cref{thm:simulation}]
	Both are shown by mutual induction on the precision derivation. We use a stronger induction principle that the one given by the induction rules. Indeed, we need extra induction hypothesis on the inferred type for a term. Proving this stronger principle is done by making the proof of \cref{prop:validity} slightly more general: instead of proving that an inferred type is always well-formed, we prove that any property consequence of typing is true of all inferred types. Let us now detail the most important cases of the inductive proof.

	\paragraph{Definitional precision.}
	We start with the easier second point. The proof is summarized by the following diagram:\\
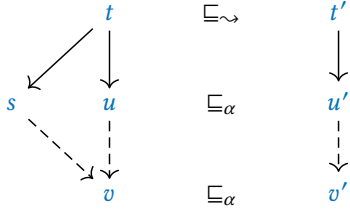
		\[\begin{tikzcd}[scale=0.6]
		& \tcol{t} \arrow[dl] \arrow[d] & \cdpre &\tcol{t'} \arrow[d] \\
		\tcol{s} \arrow[dr, dashed] & \tcol{u} \arrow[d, dashed] & \capre & \tcol{u'} \arrow[d, dashed] \\
		& \tcol{v} & \capre & \tcol{v'}
		\end{tikzcd}\]
	 By definition of $\cdpre$, there exists $\tcol{u}$ and $\tcol{u'}$, reducts respectively of $\tcol{t}$ and $\tcol{t'}$, and such that $\tcol{\GG} \vdash \tcol{u} \capre \tcol{u'}$. By confluence, there exists some $\tcol{v}$ that is a reduct of both $\tcol{u}$ and $\tcol{s}$. By subject reduction, $\tcol{u}$ and $\tcol{u'}$ are both well-typed, and thus by induction hypothesis there exists 
   $\tcol{v'}$ such that $\tcol{u'} \rtred \tcol{v'}$ and $\tcol{\GG} \vdash \tcol{v} \capre \tcol{v'}$. But then $\tcol{v}$ is a reduct of $\tcol{s}$ and $\tcol{v'}$ is a reduct of $\tcol{t'}$, and so $\tcol{\GG} \vdash \tcol{s} \cdpre \tcol{t'}$.

	This implies in particular that if $\tcol{\GG} \vdash \tcol{t} \inferty \tcol{T}$, $\tcol{\GG} \vdash \tcol{T} \cdpre \tcol{T'}$, $\tcol{t} \rtred \tcol{s}$ and $\tcol{\fs{\GG}} \vdash \tcol{s} \inferty \tcol{S}$, then $\tcol{\GG} \vdash \tcol{S} \cdpre \tcol{T'}$. Indeed $\tcol{\fs{\GG}} \vdash \tcol{s} \checkty \tcol{T}$ by subject reduction, thus $\tcol{S}$ and $\tcol{T}$ are convertible, and have a common reduct $\tcol{U}$ by confluence.	The property just stated then gives $\tcol{\GG} \vdash \tcol{U} \cdpre \tcol{T'}$, hence $\tcol{\GG} \vdash \tcol{S} \cdpre \tcol{T'}$.

	\paragraph{Syntactic precision---Non-diagonal precision rules.}
	Let us now turn to $\capre$. It is enough to show that one step of reduction can be simulated, by induction on the path $\tcol{t} \rtred \tcol{s}$.

	First, we get rid of most cases where the last rule used for $\tcol{\GG} \vdash \tcol{t} \capre \tcol{t'}$ is not a diagonal rule.
  For \nameref{infrule:capre-unk} we must handle the side-condition involving the type of $\tcol{t}$. However, by the previous property, the inferred type of $\tcol{s}$ is also definitionally less precise than $\tcol{T'}$. Thus the reduction in $\tcol{t}$ can be simulated by zero reduction steps. The reasoning for rules \nameref{infrule:capre-err} and \nameref{infrule:capre-err-lambda} is similar. As for rule \nameref{infrule:capre-univ}, subject reduction is enough to get what we seek, without even resorting to the previous property. Rule \nameref{infrule:capre-castr} is treated in the same way as \nameref{infrule:capre-unk}, as the typing side-conditions are similar.
	Thus the only non-diagonal rule left for $\capre$ is \nameref{infrule:capre-castl}.

	\paragraph{Syntactic precision---Non-top-level reduction.}

	Next, we can get rid of reductions that do not happen at top level. Indeed, if the last rule used was \nameref{infrule:capre-castl}, and the reduction happens in one of the types of the cast, the same reasoning as for \nameref{infrule:capre-castr} applies. If it happens in the term, we can use the induction hypothesis on this term to conclude.
  Also, if the last rule used was a diagonal rule, then the reduction in $\tcol{t}$ can be simulated by a similar congruence rules in $\tcol{t'}$.

  So we are left with the simulation of a reduction that happens at the top-level in $\tcol{t}$, and where the last precision rule used is either \nameref{infrule:capre-castl} or a diagonal one, and this is the real core of the proof.

	\paragraph{Syntactic precision---non-diagonal cast.}
  Let us first turn to the case where the last precision rule is \nameref{infrule:capre-castl}, and that cast reduces. More precisely, $\tcol{t}$ is some $\tcol{\cast{S}{T}{u}}$, with $\tcol{\GG} \vdash \tcol{u} \capre \tcol{t'}$. There are four possibilities for the reduction.
  \begin{itemize}

  \item The cast fails. When it does, whatever the rule, it always reduces to $\tcol{\err_{T}}$. But then we know that $\tcol{\sn{\GG}} \vdash \tcol{t'} \inferty \tcol{T'}$ and $\tcol{\GG} \vdash \tcol{T} \cdpre \tcol{T'}$. Thus $\tcol{\GG} \vdash \tcol{\err_{T}} \capre \tcol{t'}$ by rule \nameref{infrule:capre-err}, and the reduction is simulated by zero reductions.

	\item The cast disappears (\nameref{redrule:univ-univ}) or expands into two casts without changing $\tcol{u}$ (\nameref{redrule:ind-germ}, \nameref{redrule:prod-germ}). In those cases the reduct of $\tcol{t}$ is still smaller than $\tcol{t'}$. In the case of cast expansion, we must use \nameref{infrule:capre-castl} twice, and thus prove that the type of $\tcol{t'}$ is less precise than the introduced germ. But by the \nameref{infrule:capre-castl} rule that was used to prove $\tcol{\GG} \vdash \tcol{t} \capre \tcol{t'}$, we know that $\tcol{t'}$ infers a type $\tcol{T'}$ which is definitionally less precise than some $\tcol{\?_{[]_i}}$. Thus, $\tcol{T'}$ reduces to some $\tcol{S'}$ such that $\tcol{\GG} \vdash \tcol{\?_{[]_i}} \capre \tcol{S'}$, and this implies that also $\tcol{\GG} \vdash \tcol{\stalkCIC{i}{h}} \capre \tcol{S'}$, \ie what we sought.

	\item Both $\tcol{T}$ and $\tcol{S}$ are either product types or inductive types, and $\tcol{u}$ starts with an abstraction or an inductive constructor. In that case, by \cref{lem:catchup-lambda,lem:catchup-cons}, $\tcol{t'}$ reduces to a term $\tcol{u'}$ with the same head constructor as $\tcol{u}$
  or some $\tcol{\?_{I(\orr{a})}}$. In the first case, by the substitution property of precision we have $\tcol{\GG} \vdash \tcol{s} \capre \tcol{u'}$. In the second, we can use \nameref{infrule:capre-unk} to conclude.

	\item The reduction rule is \nameref{redrule:up-down}, that is $\tcol{t}$ is $\tcol{\cast{?_{[]_i}}{T}{\cast{\stalkCIC{i}{h}}{\?_{[]_i}}{u}}}$ which reduces to $\tcol{\cast{\stalkCIC{i}{h}}{T}{u}}$.
  If rule \nameref{infrule:capre-castl} was used twice in a row then we directly have $\tcol{\GG} \vdash \tcol{u} \capre \tcol{t'}$ and so $\tcol{\GG} \vdash \tcol{\cast{\stalkCIC{i}{h}}{X}{u}} \capre \tcol{t'}$. Otherwise, rule \nameref{infrule:capre-diag-cast} was used, $\tcol{t'}$ is some $\tcol{\cast{S'}{T'}{u'}}$ and we have $\tcol{\GG} \vdash \tcol{u} \capre \tcol{u'}$ and $\tcol{\fs{\GG}} \vdash \tcol{\stalkCIC{i}{h}} \cdpre \tcol{S'}$.
  Moreover, \nameref{infrule:capre-castl} also gives $\tcol{\fs{\GG}} \vdash \tcol{X} \cdpre \tcol{B'}$, since $\tcol{\sn{\GG}} \vdash \tcol{\cast{A'}{B'}{u'}} \inferty \tcol{B'}$. Thus $\tcol{\GG} \vdash \tcol{\cast{\stalkCIC{i}{h}}{X}{u}} \capre \tcol{\cast{A'}{B'}{u'}}$ by a use of \nameref{infrule:capre-diag-cast}.
	
  \end{itemize}

	\paragraph{Syntactic precision---$\beta$ redex.}
	Next we consider the case where $\tcol{t}$ is a $\beta$ redex $\tcol{(\l x : A . t_1)~t_2}$. Because the last applied precision rule is diagonal, $\tcol{t'}$ must also decompose as $\tcol{t_1''~t_2'}$. If $\tcol{t_1}$ is some $\tcol{\err_{T}}$, then the reduct is $\tcol{\err_{T}}$ and must be still smaller that $\tcol{t'}$. Otherwise, \cref{lem:catchup-lambda} applies, thus $\tcol{t_1''}$ reduces to some $\tcol{\l x : A' . t_1'}$ that is syntactically less precise than $\tcol{\l x : A . t_1}$. Then the $\beta$ reduction of $t$ can be simulated with a $\beta$ reduction in $\tcol{t'}$, and using the substitution property we conclude that the redexes are still related by precision.

  \paragraph{Syntactic precision---$\iota$ redex.}
	If $\tcol{t}$ is a $\iota$ redex $\tcol{\match{c(\orr{a},\orr{b})}{I}{z.P}{\orr{f.\orr{y}.t}}}$, the reasoning is similar. Because the last precision rule is diagonal, $\tcol{t'}$ must also be a fixpoint. We thus can use \cref{lem:catchup-cons} to ensure that its scrutinee reduces either to $\tcol{c(\orr{a'},\orr{b'})}$ or $\tcol{\?_{I(\orr{a'})}}$. In the first case, a $\iota$ reduction of $\tcol{t'}$ and the substitution property is enough to conclude.
	In the second case, $\tcol{t'}$ reduces to a term $\tcol{s'} := \tcol{\?_{P'\subs{\?_{I(\orr{a'})}}{z}}}$, and we must show this term to be less precise than $\tcol{s}$, which is $\tcol{t_k\subs{\l x : I(\orr{a}) . \match{I}{x}{z.P}{\orr{f.\orr{y}.t}}}{z}\subs{\orr{b}}{\orr{y}}}$.
	Let $\tcol{S}$ be the type inferred for $\tcol{s}$, by rule \nameref{infrule:capre-unk}, it is enough to show $\tcol{\GG} \vdash \tcol{S} \cdpre \tcol{P'\subs{\?_{I(\orr{a'})}}{z}}$. By subject reduction, $\tcol{S}$ and $\tcol{P\subs{c_k(\orr{a},\orr{b})}{z}}$ (the type of $\tcol{t}$) are convertible, thus they have a common reduct $\tcol{U}$. Now we also have by substitution that $\tcol{\GG} \vdash \tcol{P\subs{c_k(\orr{a},\orr{b})}{z}} \capre \tcol{P'\subs{\?_{I(\orr{a'})}}{z}}$. Because $\tcol{P\subs{c_k(\orr{a},\orr{b})}{z}}$ is the inferred type for $\tcol{t}$, the induction hypothesis applies to it, and thus there is some $\tcol{U'}$ such that $\tcol{P'\subs{\?_{I(\orr{a'})}{z}}} \rtred \tcol{U'}$ and also $\tcol{\GG} \vdash \tcol{U} \capre \tcol{U'}$.

	\paragraph{Syntactic precision---$\err$ and $\?$ reductions.}
	For reductions \nameref{redrule:prod-err}, \ie when $\tcol{\err_{\P x : A.B}} \redCCIC \tcol{\l x : A . \err_B}$, we can replace the use of \nameref{infrule:capre-err} by a use of \nameref{infrule:capre-err-lambda}. For reduction \nameref{redrule:ind-err}, \ie when $\tcol{t}$ is $\tcol{\match{I}{\err_{I(\orr{a})}}{z.P}{\orr{f.y.t}}}$ we distinguish three cases depending on $\tcol{t'}$. If $\tcol{t'}$ is $\tcol{\?_{T'}}$ (the precision rule between $\tcol{t}$ and $\tcol{t'}$ was \nameref{infrule:capre-unk}) or $\tcol{\cast{S'}{T'}{t'}}$, then $\tcol{\GG} \vdash \tcol{P\subs{\err_{I(\orr{a})}}{z} \cdpre T'}$, and thus $\tcol{\GG} \vdash \tcol{\err_{P\subs{\err{_{I(\orr{a})}}}{z}}} \capre \tcol{t'}$ by using \nameref{infrule:capre-err}. Otherwise, the last rule was \nameref{infrule:capre-fix}, and again we can conclude using \nameref{infrule:capre-err} and the substitution property of $\capre$.

	Conversely, let us consider the reduction rules for $\tcol{\?}$. If $\tcol{t}$ is $\tcol{\?_{\P x : A.B}}$ and reduces to $\tcol{\l x : A. \?_{B}}$, then $\tcol{t'}$ must be $\tcol{\?_{T}}$, possibly surrounded by casts.
  If there are casts, they can all be reduced away until we are left with $\tcol{\?_{T'}}$ for some $\tcol{T'}$ such that $\tcol{\GG} \vdash \tcol{\P x : A. B} \cdpre \tcol{T}$. By \cref{lem:catchup-type}, $\tcol{T} \rtred \tcol{\?_{\?_{[]}}}$ or $\tcol{T} \rtred \tcol{T_{\P x : A'. B'}}$. In the first case, $\tcol{\?_{\?_{[]}}}$ is still less precise than $\tcol{\l x : A. B}$, and in the second case, $\tcol{\?_{\P x : A'. B'}}$ can reduce to $\tcol{\l x : A' . \?_{B'}}$, which is less precise than $\tcol{s'}$.
  If $\tcol{t}$ is $\tcol{\match{I}{\?_{I(\orr{a})}}{P}{\orr{b}}}$, reducing to $\tcol{\?_{P\subs{\?_{I(\orr(a))}}{z}}}$, we use the second part of \cref{lem:catchup-cons} to conclude that also $\tcol{t'}$ reduces to some $\tcol{\match{I}{\?_{I(\orr{a'})}}{P'}{\orr{b'}}}$ that is less precise than $\tcol{t}$. From this, $\tcol{t'} \redCCIC \tcol{\?_{P'\subs{\?_{I(\orr(a'))}}{z}}}$, which is less precise than $\tcol{s}$.

	\paragraph{Syntactic precision---diagonal cast reduction.}
  This only leaves us with the reduction of a cast when the precision rule is \nameref{infrule:capre-diag-cast}: we have some $\tcol{\cast{S}{T}{u}}$ and $\tcol{\cast{S'}{T'}{u'}}$ that are pointwise related by precision, such that $\tcol{\cast{S}{T}{t}} \rtred \tcol{s}$ by a head reduction, and we must show that $\tcol{\cast{S}{T}{u}}$ simulates that reduction.

  First, if the reduction for $\tcol{\cast{S}{T}{t}}$ is any reduction to an error, then the reduct is $\tcol{\err_{T}}$, and since $\tcol{\sn{\GG}} \vdash \tcol{\cast{S'}{T'}{u'}} \inferty \tcol{T'}$ and $\tcol{\GG} \vdash \tcol{T} \cdpre \tcol{T'}$ we can use rule \nameref{infrule:capre-err} to conclude.

	Next, consider \nameref{redrule:prod-prod}. We are in the situation where $\tcol{t}$ is $\tcol{\cast{\P x : A_1 . B_1}{\P x : A_2.B_2}{\l x : A. v}}$.
	If $\tcol{v}$ is $\tcol{\err_{B_1}}$ then the reduct is more precise than any term.
	Otherwise, by \cref{lem:catchup-type}, $\tcol{S'}$ reduces either to $\tcol{\?_{[]}}$ or to a product type. In the first case, $\tcol{u'}$ must reduce to $\tcol{\?_{\?_{[]}}}$ by \cref{lem:catchup-lambda}, since it is less precise than $\tcol{\l x : A.v}$ and by typing it cannot start with a $\tcol{\lambda}$. In that case, $\tcol{\cast{S'}{T'}{u'}} \redCCIC \tcol{\?_{T'}}$, and since $\tcol{\GG} \vdash \tcol{\P x : A_2. B_2} \cdpre \tcol{T'}$, we have that $\tcol{\GG} \vdash \tcol{s} \capre \tcol{\?_{T'}}$.
	Otherwise $\tcol{S'}$ reduces to some $\tcol{\P x : A_1'. B_1'}$. By \cref{lem:catchup-lambda}, $\tcol{t'}$ reduces either to some $\tcol{\?}$ or to an abstraction. In the first case, the previous reasoning still applies. Otherwise, $\tcol{t'}$ reduces to some $\tcol{\l x : A' . v'}$. Again, by \cref{lem:catchup-type}, $\tcol{T'}$ reduces either to a product type or to $\tcol{\?}$.
  In the first case $\tcol{t'}$ can simply do the same cast reduction as $\tcol{t}$, and the substitution property of precision enables us to conclude. Thus, the only case left is that where $\tcol{t'}$ is $\tcol{\cast{\P x : A_1'. B_1'}{\?_{[]_i}}{\l x : A' . v'}}$. If $\tcol{\P x : A_1' . B_1'}$ is $\tcol{\stalkCIC{i}{\P}}$, then all of $\tcol{A}$, $\tcol{A_1}$, $\tcol{A_2}$, $\tcol{B_1}$ and $\tcol{B_2}$ are more precise than $\tcol{\?_{[]_{\castOfPi{i}}}}$, and this is enough to conclude that $\tcol{s}$ is less precise than $\tcol{\cast{\?_{[]_i}}{\stalkCIC{i}{\P}}{\l x : ?_{[]_{\castOfPi{i}}}.u'}}$, using the substitution property of precision to relate $\tcol{u'}$ with the substituted $\tcol{u}$, and the \nameref{infrule:capre-abs}, \nameref{infrule:capre-castl} and \nameref{infrule:capre-castr} rules.
	The last case is when $\tcol{\P x : A_1' . B_1'}$ is not a germ. Then the reduction of $\tcol{t'}$ first does a cast expansion through $\tcol{\stalkCIC{i}{\Pi}}$, followed by a reduction of the cast between $\tcol{\P x : A_1' . B_1'}$ and $\tcol{\stalkCIC{i}{\Pi}}$. The reasoning of the two previous cases can be used again to conclude.
	The proof is similar for rule \nameref{redrule:ind-ind}.

	Next, let us consider \nameref{redrule:prod-germ}, that is when $\tcol{t}$ is $\tcol{\cast{\P x : A_1 . B_1}{\?_{[]_i}}{f}}$. We have that $\tcol{T' \redCCIC \?_{[]_j}}$ by \cref{lem:catchup-type} with $i \leq j$, and thus $\tcol{\GG} \vdash \tcol{\stalkCIC{i}{\Pi}} \cdpre \tcol{T'}$. Thus, using \nameref{infrule:capre-diag-cast} for the innermost cast in $\tcol{s}$, and \nameref{infrule:capre-castl} for the outermost one, we conclude $\tcol{\GG} \vdash \tcol{s} \capre \tcol{\cast{S'}{T'}{u'}}$. Again, the reasoning is similar for \nameref{redrule:ind-germ}.

	As for \nameref{redrule:univ-univ}, $\tcol{t}$ is $\tcol{\cast{[]_i}{[]_i}{A}}$, and we can replace rule \nameref{infrule:capre-diag-cast} by rule \nameref{infrule:capre-castr}. Indeed $\tcol{\fs{\GG}} \vdash \tcol{A} \checkty \tcol{[]_i}$ by typing, thus $\tcol{\fs{\GG}} \vdash \tcol{A} \inferty \tcol{T}$ for some $\tcol{T}$ such that $\tcol{T} \redCCIC \tcol{[]_i}$. Therefore, since $\tcol{\GG} \vdash \tcol{[]_i} \cdpre \tcol{T'}$, we have $\tcol{\GG} \vdash \tcol{T} \cdpre \tcol{T'}$ and similarly $\tcol{\GG} \vdash \tcol{T} \cdpre \tcol{S'}$. Thus, rule \nameref{infrule:capre-castr} gives $\tcol{\GG} \vdash \tcol{A} \capre \tcol{t'}$.

	The last case left is the one of \nameref{redrule:up-down}, where $\tcol{t}$ is $\tcol{\cast{\?_{[]_i}}{X}{\cast{\stalkCIC{i}{h}}{\?_{[]_i}}{v}}}$. We distinguish on the rule used to prove $\tcol{\GG} \vdash \tcol{\cast{\stalkCIC{i}{h}}{\?_{[]_i}}{v}} \capre \tcol{u'}$. If it is \nameref{infrule:capre-castl}, then we simply have $\tcol{\GG} \vdash \tcol{\cast{\stalkCIC{i}{h}}{X}{t}} \capre \tcol{\cast{S'}{T'}{u'}}$ using rule \nameref{infrule:capre-diag-cast}, as $\tcol{\GG} \vdash \tcol{\stalkCIC{i}{h}} \cdpre \tcol{S'}$ since $\tcol{\GG} \vdash \tcol{\?_{[]_i}} \cdpre \tcol{S'}$. Otherwise the rule is \nameref{infrule:capre-diag-cast}, $\tcol{t'}$ reduces to $\tcol{\cast{\?_{[]_j}}{T'}{\cast{U'}{\?_{[]_j}}{u'}}}$, using \cref{lem:catchup-type} to reduce types less precise than $\tcol{\?_{[]_i}}$ to some $\tcol{\?_{[]_j}}$ with $i \leq j$. We can use \nameref{infrule:capre-diag-cast} on the outermost cast, and \nameref{infrule:capre-castr} on the innermost to prove that this term is less precise than $\tcol{s}$, as $\tcol{\GG} \vdash \tcol{\stalkCIC{i}{h}} \cdpre \tcol{\?_{[]_j}}$ since $i \leq j$.
\end{proof}

%\begin{corollary}[Reduction and types]
%	\label{cor:red-types-long}
%	If $\fs{\GG} \vdash \tcol{T}  \pcheckty{[]} \tcol{\square{}_i}$, $\sn{\GG} \vdash \tcol{T'} \pcheckty{[]} \tcol{\square{}_j}$, $\GG \vdash \tcol{T} \capre \tcol{T'}$ then
%	\begin{itemize}
%		\item if $\tcol{T} \rtred \tcol{\?_{[]_j}}$ then $\tcol{T'} \rtred \tcol{\?_{[]_j}}$ with $i \leq j$;
%		\item if $\tcol{T} \rtred \tcol{\square{}_{i-1}}$ then either $\tcol{T'} \rtred \tcol{\?_{\square{}_j}}$ with $i \leq j$, or $\tcol{T'} \rtred \tcol{\square{}_{i-1}}$;
%		\item if $\tcol{T} \rtred \tcol{\P x : A. B}$ then either $\tcol{T'} \rtred \tcol{\?_{[]_j}}$ with $i \leq j$, or $\tcol{T'} \rtred \tcol{\P x : A'. B'}$ and $\GG \vdash \tcol{\P x : A.B} \capre \tcol{\P x : A'.B'}$;
%		\item if $\tcol{T} \rtred \tcol{I(\orr{a})}$ then either $\tcol{T'} \rtred \tcol{\?_{[]_i}}$ with $i \leq j$, or $\tcol{T'} \rtred \tcol{I(\orr{a'})}$ and $\GG \vdash \tcol{I(\orr{a})} \capre \tcol{I(\orr{a'})}$.
%	\end{itemize}
%\end{corollary}

\subsection{Properties of \GCIC}
\label{sec:properties-gcic}

Conservativity is an equivalence, so to prove it we break it down into two implications. We now state and prove those in an open context and for the three different judgments.

\begin{theorem}[\GCIC is weaker than \CIC---Open context]~
	Let $\scol{\tilde{t}}$ be a static term and $\tcol{\Gamma}$ an erasable context. Then
	\begin{itemize}
		\item if $\eras(\tcol{\Gamma}) \cicty \tcol{t} \inferty T$ then $\inferelab{}{\tilde{t}}{}{\Gamma}{t}{T'}$ for some erasable $\tcol{t}$ and $\tcol{T'}$ such that $\eras(\tcol{t}) = \scol{\tilde{t}}$ and $\eras(\tcol{T'}) = T$;
		\item if $\tcol{T'}$ is an erasable term of \CCIC, and $\eras(\tcol{\Gamma}) \cicty \scol{\tilde{t}} \checkty \eras(\tcol{T'})$ then $\checkelab{}{\tilde{t}}{}{\Gamma}{t}{T'}$ for some erasable $\tcol{t}$ such that $\eras(\tcol{t}) = \scol{\tilde{t}}$;
		\item if $\eras(\tcol{\Gamma}) \cicty \scol{\tilde{t}} \pcheckty{h} T$ then $\pcheckelab{h}{}{\tilde{t}}{}{\Gamma}{t}{T'}$ for some erasable $\tcol{t}$ and $\tcol{T'}$ such that $\eras(\tcol{t}) = \scol{\tilde{t}}$ and $\eras(\tcol{T'}) = T$.
	\end{itemize}
\end{theorem}

\begin{proof}
	Once again, the proof is by mutual induction, on the typing derivation of $\scol{\tilde{t}}$ in \CIC.

	All inference rules are direct: one needs to combine the induction hypothesis together, using the substitution property of precision and the fact that erasure commutes with substitution to handle the cases of substitution in the inferred types.

	Let us consider the case of \nameref{infrule:cic-prod-inf} next. We are given $\tcol{\Gamma}$ erasable, and suppose $\eras(\tcol{\Gamma}) \cicty \scol{\tilde{t}} \inferty T$ and $T \rtred \P x : A .B$. By induction hypothesis there exists $\tcol{t}$ and $\tcol{T'}$ erasable such that $\inferelab{}{\ccol{t}}{}{\Gamma}{\scol{\tilde{t}}}{T'}$ and $\eras(\tcol{t}) = \scol{\tilde{t}}$, $\eras(\tcol{T'}) = T$. Because $\tcol{T'}$ is erasable, it is less precise than $T$. By \cref{cor:red-types}, it must reduce to either $\tcol{\?_{[]}}$ or a product type. The first case is impossible because $\tcol{T'}$ does not contain any $\tcol{\?}$ as it is erasable. Thus there are some $\tcol{A'}$ and $\tcol{B'}$ such that $\tcol{T'} \rtred \tcol{\P x : A' . B'}$ and $\tcol{\Gamma} \vdash \P x : A.B \capre \tcol{\P x : A'. B'}$. Since also $\tcol{\Gamma} \vdash \tcol{T'} \capre T$, by the same reasoning there are also some $A''$ and $B''$ such that $T \rtred \P x : A''. B''$ and $\tcol{\Gamma} \vdash \tcol{\P x : A' . B'} \capre \P x : A'' . B''$. Now because $T$ is static, so are $\P x : A . B$ and $\P x : A''. B''$, and because of the comparisons with $\P x : A' . B'$ we must have $\eras(\Gamma) \vdash \P x : A . B \capre \P x : A'' . B''$. Since both are static, this means they must be $\alpha$-equal, since no non-diagonal rule can be used on static terms. Hence, $\P x : A . B = \P x : A'' . B'' = \eras(\tcol{\P x : A'. B '})$, implying that $\tcol{\P x : A' . B'}$ is erasable.
  Thus, $\pcheckelab{\Pi}{}{\tilde{t}}{}{\Gamma}{t}{\P x : A'.B'}$, both $\tcol{t}$ and $\tcol{\P x : A' . B'}$ are erasable, and moreover $\eras(\tcol{t}) = \scol{\tilde{t}}$ and $\eras(\tcol{\P x : A' . B'}) = \P x : A.B$, which is what had to be proven.

	The other constrained inference rules being very similar, let us turn to \nameref{infrule:cic-check}. We are given $\tcol{\Gamma}$ and $\tcol{T'}$ erasable, and suppose that $\eras(\tcol{\Gamma}) \cicty \scol{\tilde{t}} \inferty S$ such that $S \conv \eras(\tcol{T'})$.
  By induction hypothesis, $\inferelab{}{\tilde{t}}{}{\Gamma'}{t}{S'}$ with $\tcol{t}$ and $\tcol{S'}$ erasable, $\eras(\tcol{t}) = \scol{\tilde{t}}$ and $\eras(\tcol{S'}) = S$. But convertibility implies consistency, so $S \cons \eras(\tcol{T'})$.
  By monotonicity of consistency, this implies $\tcol{S'} \cons \tcol{T'}$. Thus $\checkelab{}{\tilde{t}}{}{\Gamma}{\cast{S'}{T'}{t}}{T'}$. We have $\eras(\tcol{\cast{S'}{T'}{t}}) = \eras(\tcol{t}) = \scol{\tilde{t}}$, so we are left with showing that $\tcol{\Gamma} \vdash \tcol{\cast{S'}{T'}{t}} \caequipre \scol{\tilde{t}}$. Using rules \nameref{infrule:capre-castl} and \nameref{infrule:capre-castr}, and knowing already that $\tcol{\Gamma} \vdash \tcol{S'} \caequipre S$, it remains to show that $\tcol{\Gamma} \vdash \tcol{T'} \cdpre S$ and $\Gamma \vdash S \cdpre \tcol{T'}$. As $S$ and $\eras(\tcol{T'})$ are convertible, let $U$ be a common reduct. Using \cref{thm:simulation}, $\tcol{T'} \rtred \tcol{U'}$ with $\tcol{\Gamma} \vdash U \capre \tcol{U'}$. Simulating that reduction again, we get $\eras(\tcol{T'}) \rtred U''$ with $\tcol{\Gamma} \vdash U'' \capre \tcol{U'}$. As before, this implies $U = U'' = \eras(\tcol{U'})$. Thus, using the reduct $\tcol{U'}$ of $\tcol{T'}$ that is equiprecise with $\tcol{U}$, we can conclude $\tcol{\Gamma} \vdash S \cdpre \tcol{T'}$ and $\tcol{\Gamma} \vdash \tcol{T'} \cdpre S$.
\end{proof}

\begin{theorem}[\CIC is weaker than \GCIC---Open context]
	Let $\scol{\tilde{t}}$ be a static term and $\tcol{\Gamma}$ an erasable context of \CCIC. Then
	\begin{itemize}
		\item if $\inferelab{}{\tilde{t}}{}{\Gamma}{t}{T}$, then $\tcol{t}$ and $\tcol{T}$ are erasable, $\eras(\tcol{t'}) = \scol{\tilde{t}}$ and $\eras(\tcol{\Gamma}) \vdash \scol{\tilde{t}} \inferty \eras(\tcol{T'})$;
		\item if $\tcol{T'}$ is an erasable term of \CCIC such that $\checkelab{}{\tilde{t}}{}{\Gamma'}{t'}{T'}$, then $\tcol{t'}$ is erasable, $\eras(t') = \scol{\tilde{t}}$ and $\eras(\tcol{\Gamma}) \vdash \scol{\tilde{t}} \checkty \eras(\tcol{T'})$;
		\item if $\pcheckelab{h}{}{\tilde{t}}{}{\Gamma'}{t'}{T'}$, then $\tcol{t'}$ and $\tcol{T'}$ are erasable, $\eras(\tcol{t'}) = \scol{\tilde{t}}$ and $\eras(\tcol{\Gamma}) \vdash \scol{\tilde{t}} \pcheckty{h} \eras(\tcol{T'})$.
	\end{itemize}
\end{theorem}

\begin{proof}
	The proof is similar to the previous one. Again, the tricky part is to handle reduction steps, and we use equiprecision in the same way to conclude in those.
\end{proof}

As a direct corollary of those propositions in an empty context, we get conservativity \cref{thm:conservativity}.

\paragraph{Elaboration graduality}

Now for the elaboration graduality: again, we state it in an open context for all three typing judgments.

\begin{theorem}[Elaboration graduality---Open context]
	Let $\tcol{\GG}$ be a context such that $\tcol{\fs{\GG}} \capre \tcol{\sn{\GG}}$, and $\scol{\tilde{t}}$ and $\scol{\tilde{t}'}$ be two \GCIC terms such that $\scol{\tilde{t}} \apre \scol{\tilde{t}'}$. Then
	\begin{itemize}
		\item if $\inferelab{}{\tilde{t}}{}{\fs{\GG}}{t}{T}$ is universe adequate, then there exists $\tcol{t'}$ and $\tcol{T'}$ such that $\inferelab{}{\tilde{t}'}{}{\sn{\GG}}{t'}{T'}$, $\tcol{\GG} \vdash \tcol{t} \capre \tcol{t'}$ and $\tcol{\GG} \vdash \tcol{T} \capre \tcol{T'}$;
		\item If $\checkelab{}{t}{}{\fs{\tcol{\GG}}}{t'}{T}$ is universe adequate, then for all $\tcol{T'}$ such that $\tcol{\GG} \vdash \tcol{T} \capre \tcol{T'}$ there exists $\tcol{t'}$ such that $\checkelab{}{\tilde{t}'}{}{\sn{\GG}}{t'}{T'}$ and $\tcol{\GG} \vdash \tcol{t} \capre \tcol{t'}$;
		\item If $\pcheckelab{h}{}{t}{}{\fs{\tcol{\GG}}}{t'}{T}$ is universe adequate, then there exists $\tcol{t'}$ and $\tcol{T'}$ such that $\pcheckelab{h}{}{\tilde{t}'}{}{\sn{\GG}}{t'}{T'}$, $\tcol{\GG} \vdash \tcol{t} \capre \tcol{t'}$ and $\tcol{\GG} \vdash \tcol{T} \capre \tcol{T'}$.
	\end{itemize}
\end{theorem}

\begin{proof}
	Once again, we use our favorite tool: induction on the typing derivation of $\scol{\tilde{t}}$.

	\paragraph{Inference---Non-diagonal precision.}
	For inference, we have to make a distinction on the rule used to prove $\scol{tilde{t}} \apre \scol{\tilde{t}'}$: we have to handle specifically the non-diagonal one, where $\scol{\tilde{t}'}$ is some $\scol{\?}$. We start with this, and treat the ones where the rule is diagonal (\ie when $\scol{\tilde{t}}$ and $\scol{\tilde{t}'}$ have the same head) next.

	We have $\inferelab{}{\tilde{t}}{}{\fs{\GG}}{t'}{T'}$ and $\inferelab{}{\?\ulev{i}}{}{\sn{\GG}}{\?_{\?_{[]_i}}}{\?_{[]_i}}$. Correctness of elaboration gives $\tcol{\fs{\GG}} \vdash \tcol{t'} \inferty \tcol{T'}$, and by validity $\tcol{\fs{\Gamma}} \vdash \tcol{T'} \inferty \tcol{[]_i}$, universe adequacy ensuring us that this $i$ is the same as the one in $\scol{\tilde{t}'}$. Thus we have $\tcol{\GG} \vdash \tcol{T'} \capre \tcol{\?_{[]_i}}$ by rule \nameref{infrule:capre-unk}, and in turn $\tcol{\GG} \vdash \tcol{t'} \capre \tcol{\?_{\?_{[]_i}}}$ by a second use of the same rule, giving us the required conclusions.

	\paragraph{Inference---Variable.}
	Rule \nameref{infrule:gcic-var} gives us $(x : \tcol{T}) \in \tcol{\fs{\GG}}$. Because $\vdash \tcol{\fs{\GG}} \capre \tcol{\sn{\GG}}$, there exists some $\tcol{T'}$ such that $(x : \tcol{T'}) \in \tcol{\sn{\GG}}$, and $\tcol{\GG} \vdash \tcol{T} \capre \tcol{T'}$ using weakening. Thus, $\inferelab{}{x}{}{\sn{\GG}}{x}{T'}$, and of course $\tcol{\GG} \vdash \tcol{x} \capre \tcol{x}$.

%	\textbf{Inference---diagonal universe case}\\
%	By reflexivity.

	\paragraph{Inference---Product.}
  Premises of rule \nameref{infrule:gcic-prod} give $\pcheckelab{[]}{}{\tilde{A}}{}{\fs{\GG}}{A}{[]_i}$ and $\pcheckelab{[]}{}{\tilde{B}}{}{\fs{\GG}, x : A}{B}{[]_j}$, and the diagonal precision one gives $\scol{\tilde{A}} \apre \scol{\tilde{A}'}$ and $\scol{\tilde{B}} \apre \scol{\tilde{B}'}$. Applying the induction hypothesis, we get some $\tcol{A'}$ such that $\pcheckelab{[]}{}{\tilde{A}'}{}{\sn{\GG}}{A'}{[]_i}$ and $\tcol{\GG} \vdash \tcol{A} \capre \tcol{A'}$. The inferred type for $\scol{\tilde{A}'}$ must be $\tcol{[]_i}$ as it is some $\tcol{[]_j}$ because of the constrained elaboration, and it is less precise than $\scol{[]_i}$ by the induction hypothesis.
  From this, we also deduce that $\tcol{\fs{\GG}, x : A} \capre \tcol{\sn{\GG}, x : A'}$. Hence the induction hypothesis can be applied to $\scol{\tilde{B}}$, giving $\pcheckelab{[]}{}{\tilde{B}'}{}{\sn{\GG}}{B'}{[]_j}$. Combining this with the elaboration for $\tcol{\tilde{A}'}$, we obtain $\inferelab{}{\P x : \tilde{A}'. \tilde{B}'}{}{\sn{\GG}}{\P x : A' . B'}{[]_{\sortOfPi{i}{j}}}$. Moreover, $\tcol{\GG} \vdash \tcol{\P x : A. B} \capre \tcol{\P x : A' . B'}$ by combining the precision hypothesis on $A$ and $B$, and also $\tcol{\GG} \vdash \tcol{[]_{\sortOfPi{i}{j}}} \capre \tcol{[]_{\sortOfPi{i}{j}}}$.

	\paragraph{Inference---Application.}
	From rule \nameref{infrule:gcic-app}, we have $\pcheckelab{\Pi}{}{\tilde{t}}{}{\fs{\GG}}{t}{\P x : A . B}$ and $\checkelab{}{\tilde{u}}{}{\fs{\GG}}{u}{A}$, and the diagonal precision gives $\scol{\tilde{t}} \apre \scol{\tilde{t}'}$ and $\scol{\tilde{u}} \apre \scol{\tilde{u}'}$. By induction, we have $\pcheckelab{\Pi}{}{\tilde{t}'}{}{\fs{\GG}}{t'}{\P x : A' . B'}$ for some $\tcol{t'}$, $\tcol{A'}$ and $\tcol{B'}$ such that $\tcol{\GG} \vdash \tcol{t} \capre \tcol{t'}$, $\tcol{\GG} \vdash \tcol{A} \capre \tcol{A'}$ and $\tcol{\GG, x : A \mid A'} \vdash \tcol{B} \capre \tcol{B'}$. Using the induction hypothesis again with that precision property on $\tcol{A}$ and $\tcol{A'}$ gives $\checkelab{}{\tilde{u}'}{}{\sn{\GG}}{u'}{A'}$ with $\tcol{\GG} \vdash \tcol{u} \capre \tcol{u'}$.
  Therefore combining those we get $\inferelab{}{\tilde{t}'~\tilde{u}'}{}{\sn{\GG}}{t'~u'}{B'\subs{u'}{x}}$, $\tcol{\GG} \vdash \tcol{t~u} \capre \tcol{t'~u'}$ and, by substitution property of precision, $\tcol{\GG} \vdash \tcol{B\subs{u}{x}} \capre \tcol{B'\subs{u'}{x}}$.

	\paragraph{Inference---Other diagonal cases.}
	All other cases are similar to those: combining the induction hypothesis directly leads to the desired result, handling the binders in a similar way to that of products when needed.

	\paragraph{Checking.}
	For \nameref{infrule:gcic-check}, we have that $\inferelab{}{\tilde{t}}{}{\fs{\GG}}{t}{S}$, with $\tcol{S} \cons \tcol{T}$. By induction hypothesis, $\inferelab{}{\tilde{t}'}{}{\sn{\GG}}{\tilde{t}'}{S'}$ with $\tcol{\GG} \vdash \tcol{t} \capre \tcol{t'}$ and $\tcol{\GG} \vdash \tcol{S} \capre \tcol{S'}$. But we also have as an hypothesis that $\tcol{\GG} \vdash \tcol{T} \capre \tcol{T'}$. By monotonicity of consistency, we conclude that $\tcol{S'} \cons \tcol{T'}$, and thus $\checkelab{}{\tilde{t}'}{}{\sn{\GG}}{\cast{S'}{T'}{t'}}{T'}$. A use of \nameref{infrule:capre-diag-cast} then ensures that $\tcol{\GG} \vdash \tcol{\cast{S}{T}{t}} \capre \tcol{\cast{S'}{T'}{t'}}$, as desired.

	\paragraph{Constrained inference---\nameref{infrule:gcic-inf-prod} rule.}
	We are in the situation where $\inferelab{}{\tilde{t}}{}{\fs{\GG}}{t}{S}$ and $\tcol{S} \rtred \tcol{\P x : A. B}$. By induction hypothesis, $\inferelab{}{\tilde{t}'}{}{\sn{\GG}}{t'}{S'}$ with $\tcol{\GG} \vdash \tcol{S} \capre \tcol{S'}$. Using \cref{cor:red-types}, we get that $\tcol{S'} \rtred \tcol{\P x : A' . B'}$ such that $\tcol{\GG} \vdash \tcol{\P x : A . B} \capre \tcol{\P x : A' . B'}$, or $\tcol{S'} \rtred \tcol{\?_{[]_i}}$. In the first case, by rule \nameref{infrule:gcic-inf-prod} we get $\pcheckelab{\Pi}{}{\tilde{t}'}{}{\sn{\GG}}{t'}{\P x : A' . B'}$ together with the precision inequalities for $\tcol{t'}$ and $\tcol{\P x : A'. B'}$. In the second case, we can use rule \nameref{infrule:gcic-prod-unk} instead, and get $\pcheckelab{\Pi}{}{\tilde{t}'}{}{\sn{\GG}}{\cast{S'}{\stalkCIC{i}{\Pi}}{t'}}{\stalkCIC{i}{\Pi}}$, and $\castOfPi{i}$ is larger than the universe levels of both $\tcol{A'}$ and $\tcol{B'}$. A use of \nameref{infrule:capre-castr}, together with the fact that $\tcol{\GG} \vdash \tcol{A} \capre \tcol{\?_{[]_{\castOfPi{i}}}}$ by \nameref{infrule:capre-unk-univ} and similarly for $\tcol{B}$, gives that $\tcol{\GG} \vdash \tcol{t'} \capre \tcol{\cast{S'}{\stalkCIC{i}{\Pi}}{t'}}$, and the precision between types has been established already.

	\paragraph{Constrained inference---\nameref{infrule:gcic-prod-unk}.}
	This time, $\inferelab{}{\tilde{t}}{}{\fs{\GG}}{t}{S}$, but $\tcol{S} \rtred \tcol{\?_{[]_i}}$. By induction hypothesis, $\inferelab{}{\tilde{t}'}{}{\sn{\GG}}{t'}{S'}$ with $\tcol{\GG} \vdash \tcol{S} \capre \tcol{S'}$. By \cref{cor:red-types}, we get that $\tcol{S'} \rtred \tcol{\?_{[]_i}}$. Thus $\pcheckelab{\Pi}{}{\tilde{t}'}{}{\sn{\GG}}{\cast{S'}{\stalkCIC{i}{\Pi}}{t'}}{\stalkCIC{i}{\Pi}}$. A use of \nameref{infrule:capre-diag-cast} is enough to conclude.

	\paragraph{Constrained inference---Other rules.}
	All other cases are similar to the previous ones, albeit with a simpler handling of universe levels (since $\castOfPiName$ does not appear).

\end{proof}

\section{Connecting the discrete and monotone models}
\label{sec:logical-relation}

Comparing the discrete and the monotone translations, we can see that they
coincide on ground types such as $\nat$.
On functions over ground types, for instance $\nat {\to} \nat$, the
monotone interpretation is more conservative:
any monotone function $f : \monTy{\nat{\to}\nat}$ induces a function
$\tilde{f} : \bareTy{\nat{\to}\nat}$ by forgetting the monotonicity, but not all
functions from $\bareTy{\nat{\to}\nat}$ are monotone\footnote{For instance the
  function swapping $\err_{ \liftNat{} }$ and $\?_{ \liftNat{} }$ is not monotone.}.
\begin{figure}[h]
  \centering
  \textbf{Translation of contexts}
  \begin{align*}
    \binlogrelTy{\cdot}\quad &:=\quad {\cdot}
    &\binlogrelTy{\Gamma, x : A} \quad &{:=} \quad\binlogrelTy{\Gamma}, x_{\disc} : \bareTy{A}, x_{\mon} :
    \monTy{A}, x_\binlrel : \binlogrelTy{A}~x_{\disc}~x_{\mon}
  \end{align*}
  \smallskip

  \textbf{Logical relation on terms and types}

  \[
    \begin{array}[t]{lcl}
      \binlogrelTy{A} & := & \El_\binlrel\,\binlogrelTm{A} \\[0.7em]
      {\binlogrelTm{ x }} & := & x_\binlrel\\
      {\binlogrelTm{\square{}_i}} & := & \uU_{\binlrel,i} \\
      {\binlogrelTm{t~u}} & := & \binlogrelTm{t}~\bareTm{u}~\monTm{u}~\binlogrelTm{u}\\
      {\binlogrelTm{\l x : A . t}} & := & \l (x_{\disc} : \bareTy{A}) (x_{\mon} : \monTy{A}) (x_\binlrel : \binlogrelTy{A}~x_{\disc}~x_{\mon}) . \binlogrelTm{t} \\
      {\binlogrelTm{\P x : A . B}} & := & \PiU_\binlrel~\binlogrelTm{A}~(\l (x_{\disc} : \bareTy{A}) (x_{\mon} : \monTy{A}) (x_\binlrel : \binlogrelTy{A}~x_{\disc}~x_{\mon}). \binlogrelTy{ B }) \\
      {\binlogrelTm{\nat}} & := & \natU_\binlrel \\
      {\binlogrelTm{\?_A}} & := & \?_{\binlogrelTy{A}} : \binlogrelTy{A}~\?_{\bareTy{A}}~\?_{\monTm{A}}\\
      {\binlogrelTm{\err_A}} & := & \err_{\binlogrelTy{A}} : \binlogrelTy{A}~\err_{\bareTy{A}}~\err_{\monTm{A}}\\
      {\binlogrelTm{ \cas }} &:= & \cas_{\binlrel}
    \end{array}
  \]
  \medskip

  \begin{center}
    \textbf{Inductive-recursive relational universe $\U_{\binlrel} :
      \U^\disc{\to}\U^\mon{\to}\square$ }
  \end{center}
  \smallskip

  \begin{mathpar}
    \inferrule{A_\binlrel \in \U_{\binlrel, i}~A~A'\\
      B \in \P(a : A)(a' : A'). \El_\binlrel\,A_\binlrel\,a\,a' \to \U_{\binlrel,j}\,(B\,a)~(B'\,a')}
    {\PiU_\binlrel~A_\binlrel~B_\binlrel \in \U_{\binlrel,\sortOfPi{i}{j}}~(\PiU\,A\,B)~(\PiU{}\,A'\,B')} \and
    \inferrule{j < i}{\uU_{\binlrel,j} \in \U_{\binlrel, i}~\uU_j~\uU_j} \and
    \natU_\binlrel \in \U_{\binlrel,i}~\natU~\natU \and
    \unkU[\binlrel] \in \U_{\binlrel,i}~\unkU~\unkU \and
    \maltese_\binlrel \in \U_{\binlrel,i}~\maltese~\maltese
  \end{mathpar}
  \medskip
  \begin{center}
    \textbf{Decoding function $\El_\binlrel : \U_\binlrel\,A\,A' {\to} \El\,A {\to} \El\,A' {\to} \square{}$}
  \end{center}
  \smallskip
  \begin{align*}
    \El_\binlrel\,\uU_{\binlrel,j}~A~A' &:= \U_{\binlrel,j}~A~A'\\
    \El_\binlrel\,\natU_\binlrel~n~m &:= n = m\\
    \El_\binlrel\,\maltese_\binlrel~\unitK~\unitK &:= \unit\\
    \El_\binlrel\,\unkU[\binlrel]~(c;x)~y &:= \El_\binlrel\,(\stalkCIC{\binlrel}{c})~x~(\texttt{downcast}_{\unkU, \stalkCIC{}{c}}\,y) \\ % (c, x) ~ y when x ~ downcast ? (ElC c) y
    \El_\binlrel\,(\PiU_\binlrel~A_\binlrel~B_\binlrel)~f~f' &:= \P (a : \El\,A) (a' : \El\,A') (a_\binlrel : \El_\binlrel\,A_\binlrel~a~a').\\
    &\qquad\quad\El_\binlrel\,(B_\binlrel\,a\,a'\,a_\binlrel)~(f\,a)~(f'\,a')\\
  \end{align*}
  \caption{Logical relation between the discrete and monotone models}
  \label{fig:logrel-discrete-monotone}
\end{figure}
Extending the sketched correspondence at higher types, we obtain a (binary)
logical relation $\binlogrelTm{-}$ between terms of the discrete and monotone
translations described in \cref{fig:logrel-discrete-monotone}, that forgets the
monotonicity information on ground types.
More precisely we define for each types $A$ in the source a relation
$\binlogrelTy{A} : \bareTy{A} \to \monTy{A} \to \square{}$ and for each term $t : A$ a
witness $\binlogrelTm{t} : \binlogrelTy{A}~\bareTm{t}~\monTm{t}$.
The logical relation employs a an inductively defined relation $\U_{\binlrel,i}$
between $\U_i^\disc := \bareTy{\square_i}$ and $\U_i^\mon := \monTy{\square_i}$
whose constructors are relational codes relating codes of discrete and monotone
types.
These relational codes are then decoded to relations between the corresponding
decoded types thanks to ${ \El{} }_{\binlrel}$.
The main difficult case in establishing the logical relation lie in relating
the $\cas$s, since that's the main point of divergence of the two models.

\begin{lemma}[Basis lemma]\text{}
  \begin{enumerate}
  \item There exists a term $\cas_\binlrel : \binlogrelTy{\P (A\, B : \U).A \to B}~\bareTm{\cas}~\monTm{\cas}$.
  \item More generally, if $\Gamma \caty t : A$ then $\binlogrelTy{\Gamma} \irty \binlogrelTm{t} : \binlogrelTy{A}~\bareTm{t}~\monTm{t}$.
  \end{enumerate}
\end{lemma}
In particular \CCIC terms of ground types behave similarly in both models.

\begin{proof}

Expanding the type of $\cas_\binlrel$, we need to provide a
term
\[c_\binlrel = \cas_\binlrel\,A\,A'\,A_\binlrel\,B\,B'\,B_\binlrel\,a\,a'\,a_\binlrel
  : \El_\binlrel\,B_\binlrel~(\bareTm{\cas}\,A\,B\,a)~(\monTm{\cas}\,A'\,B'\,a')\]
where
\begin{align*}
  A &: \bareTy{\square{}_i},& A' &: \monTy{\square{}_i},& A_\binlrel &: \U_\binlrel\,A\,A',\\
  B &: \bareTy{\square{}_i},& B' &: \monTy{\square{}_i},& B_\binlrel &: \U_\binlrel\,B\,B',\\
  a &: \El\,A,& a' &: \El\,A',& a_\binlrel &: \El_\binlrel\,A_\binlrel~a~a'
\end{align*}
We proceed by induction on $A_\binlrel, B_\binlrel$, following the
defining cases for $\bareTm{\cas}$ (see \cref{fig:cast-implem-discrete}).
\noindent\textbf{Case $A_\binlrel =
  \PiU_\binlrel\,A^\pidom_\binlrel\,A^\picod_\binlrel$ and
  $B_\binlrel = \PiU_\binlrel\,B^\pidom_\binlrel\,B^\picod_\binlrel$:}
we pose ${A' = \PiU\,A'^\pidom\,A'^\picod}$ and ${B' = \PiU\,B'^\pidom\,B'^\picod}$
\newcommand{\ucast}[2]{\upcast^{#2}_{#1}}
\newcommand{\dcast}[2]{\downcast^{#1}_{#2}}
\begin{align*}
  \monTm{\cas}~A'~B'~f' &= ~\dcast{\unkU}{B'}(\ucast{A'}{\unkU} f') \tag{by definition of $\monTm{\cas}$}\\
                       &= ~\dcast{\unkU \to \unkU}{B'}\circ \dcast{\unkU}{\unkU \to \unkU} \circ \ucast{\unkU \to \unkU}{\unkU} \circ \ucast{A'}{\unkU \to \unkU} (f)\tag{by decomposition of $\PiU \sqsubseteq \unkU$}\\
                       &= ~\dcast{\unkU \to \unkU}{B'}\circ \ucast{A'}{\unkU \to \unkU} (f) \tag{by section-retraction identity}\\
                       &= \l (b' : \El\,A'^\pidom).
                         ~\letin{a'}{~\dcast{\unkU}{B'^\pidom}\circ\ucast{A'^\pidom}{\unkU} (b)}{}\tag{by def. of \eppair{} on $\Pi$}\\
                       &\hspace{2.5cm}\dcast{\unkU}{B'^\picod\,b'}\circ\ucast{A'^\picod\,a'}{\unkU}(f\,a')\\
                       &= \l (b' : \El\,A'^\pidom).
                         ~\letin{a'}{\monTm{\cas}~B'^\pidom~A'^\pidom~b'}{} \tag{by definition of $\monTm{\cas}$}\\
                       &\hspace{2.5cm}\monTm{\cas}\,(A'^\picod\,a')~(B'^\picod\,b')~(f\,a')
\end{align*}

For any $b : \El\,B^\pidom$ and $b' : \El\,B'^\pidom$, $b_\binlrel :
\El_\binlrel\,B^\pidom_\binlrel\,b\,b'$, we have by inductive hypothesis
\[a_\binlrel := \binlogrelTm{\cas}\,B^\pidom_\binlrel\,A^\pidom_\binlrel\,b_\binlrel : \El_\binlrel~{A_\binlrel}~(\bareTm{\cas}\,B^\pidom\,A^\pidom\,b)~(\monTm{\cas}\,B'^\pidom\,A'^\pidom\,b')\]
so that, posing $a = \bareTm{\cas}\,B^\pidom\,A^\pidom\,b$ and $a' =
\monTm{\cas}\,B'^\pidom\,A'^\pidom\,b'$,
\[f_\binlrel\,a\,a'\,a_\binlrel : \El_\binlrel~(A^\picod_\binlrel\,a\,a'\,a_\binlrel)~(f\,a)~(f'\,a')\]
and by another application of the inductive hypothesis
\[\binlogrelTm{\cas}~(B^\picod_\binlrel\,b\,b'\,b_\binlrel)~(A^\picod_\binlrel\,a\,a'\,a_\binlrel)~(f_\binlrel\,a\,a'\,a_\binlrel)
  :
  \binlogrelTy{B^\picod_\binlrel\,b\,b'\,b_\binlrel}~(\bareTm{\cas}~A~B~f~a)~(\monTm{\cas}~A'~B'~f'~a')\]
Packing these together, we obtain a term
\[\binlogrelTm{\cas}~A_\binlrel~B_\binlrel~f_\binlrel
  :
  \El_\binlrel~(\PiU~B^\pidom_\binlrel~B^\picod_\binlrel)~(\bareTm{\cas}~A~B~f)~(\monTm{\cas}~A'~B'~f').\]

\noindent\textbf{Case $A_\binlrel =
  \PiU_\binlrel\,A^\pidom_\binlrel\,A^\picod_\binlrel$ and
  $B_\binlrel = \unkU[\binlrel]$:}
By definition of the logical relation at $\unkU[\binlrel]$, we need to build a
witness of type
\[\El_\binlrel~(\unkU^{\castOfPi{i}} \to \unkU^{\castOfPi{i}})~(\bareTm{\cas}~A~(\unkU \to\unkU)~f)~(\dcast{\unkU}{\unkU \to\unkU}(\monTm{\cas}~A'~\unkU~f'))\]
We compute that
\[\dcast{\unkU}{\unkU \to\unkU} (\monTm{\cas}\,A'\,\unkU\,f') =
  ~\dcast{\unkU}{\unkU \to\unkU} \circ \dcast{\unkU}{\unkU} \circ \ucast{A'}{\unkU}f' =
  ~\dcast{\unkU}{\unkU \to\unkU} \circ \ucast{A'}{\unkU}f' =
  \monTm{\cas}\,A'\,(\unkU \to \unkU)\,f'
\]
So the result holds by induction hypothesis.

\noindent\textbf{Other cases with
  $A_\binlrel = \PiU_\binlrel\,A^\pidom_\binlrel\,A^\picod_\binlrel$:}
It is enough to show that $\monTm{\cas}\,A'\,B'\,f' = \maltese_{B'}$ when $B' =
\maltese$ (trivial) or $\hd\,B' \neq \texttt{pi}$. The latter case holds because
$\dcast{\unkU}{\stalkCIC{}{c}} \ucast{\stalkCIC{}{c'}}{\unkU} x = \maltese_{\El_{\C}\,c}$ whenever $c
\neq c'$ and downcasts preserve $\maltese$.

\noindent\textbf{Case $A_\binlrel = \unkU[\binlrel]$,
  $B_\binlrel = \PiU_\binlrel\,B^\pidom_\binlrel\,B^\picod_\binlrel$ and
  $a = (\texttt{pi} ; f)$:}
By hypothesis, $a_\binlrel : \El_\binlrel~(\unkU \to
\unkU)~f~(\dcast{\unkU}{\unkU \to \unkU}a')$ and
$\monTm{\cas}~\unkU~B'~a' = \monTm{\cas}~(\unkU \to
\unkU)~B'~(\dcast{\unkU}{\unkU \to \unkU}a')$ so by induction hypothesis
\[\binlogrelTm{\cas}~(\unkU[\binlrel] \to_\binlrel
  \unkU[\binlrel])~B_\binlrel~f~(\dcast{\unkU}{\unkU \to
    \unkU}a')~a_\binlrel : \El_\binlrel\,B_\binlrel~(\bareTm{\cas}\,\unkU\,B\,(\texttt{pi};f))~(\monTm{\cas}\,\unkU\,B'\,a')\]
\noindent The others cases with $A_\binlrel = \unkU[\binlrel]$ proceed in a
similarly fashion.
All cases with $A_\binlrel = \maltese_\binlrel{}$ are immediate since
$\maltese^\disc$ and $\maltese^\mon$ are related at any related types.
Finally, the cases with $A_\binlrel{} = \natU_\binlrel$ follow the same pattern
as for $\PiU_\binlrel$.
\end{proof}

\section{Diverging terms denote as errors in $\omega$-cpos}
\label{sec:divergence-omega-cpo-model}

In this section we define a logical relation between \CCICP and \CICIRQ and
prove a fundamental lemma, obtaining~\cref{lem:diverge-err} as a corollary.
The logical relation is presented in
\cref{fig:logrel-ccicp-omega-cpo,fig:logrel-ccicp-omega-cpo-cont,fig:logrel-ccicp-omega-cpo-cont-cont}
and relates types $A$ in \CCICP with sub-$\omega$-cpos of $\unkMon$,
following the description of $\El$ in that model.
A type $A$ related to an $\omega$-cpo $A'$ by the logical relation, noted $A
\sim A'$, induces a relation between terms of type $A$ and elements of $A'$.
We use variables with $\varepsilon$ subscript to name proof witnesses of
relatedness between two objects, for instance $A_\varepsilon : A \sim A'$, and
bold variables such as $\GG, \DD$ for the corresponding double contexts
consisting of variable bindings $\divbind{a}{a'}{A_\varepsilon}$.
The projections $\tcol{\fs{\GG}}$ and $\sn{\GG}$ are then respectively contexts
in \CCICP and \CICIRQ.
The logical relation uses weak head reduction to characterize divergence.
We note $\tcol{t} \whred \tcol{u}$ when a \CCIC term $\tcol{t}$ reduces to a
weak head normal form, that is a term $\tcol{u}$ such that $\can\tcol{u}$ hold (see~\cref{fig:CCIC-canonical}), using
only weak head reduction steps.
We note $\tcol{t} \not\whred$ when weak head reduction paths from $\tcol{t}$ never reach a
weak head normal form, that is $\tcol{t}$ is unsolvable.

\begin{figure}[h]
  
  \flushleft{}
  \fbox{$\divrelCtx{\GG}$} logical relation between \CCICT contexts and
  $\omega$-cpos.

  \begin{mathpar}
    \inferrule{}{\divrelCtx{\cdot}}
    \and
    \inferrule
    { \divrelCtx{\GG}\\
      A_\varepsilon : \divrelTy{\GG}{A}{A'}}
    { \divrelCtx{\GG, \divbind{a}{a'}{A_\varepsilon}}}
    \and
    {
      \begin{array}{l@{~=~}l}
        \fs\cdot& \tcol\cdot \\
        \fs{(\GG , \divbind{a}{a'}{A_\varepsilon})} & \tcol{\fs\GG , a : A}
      \end{array}
    }
  \end{mathpar}
  \bigskip

  \flushleft{}
  \fbox{$\divrelTy{\GG}{A}{B}$} logical relation between \CCICP types $\tcol{A}$ and
  $\omega$-cpos $B \subset \unkMon$.
  \begin{mathpar}
    \inferrule{\divrelTm{\GG}{T}{U}{\univDivrel}}
    {\divrelTy{\GG}{T}{\El\,U}}
  \end{mathpar}
  \bigskip

  \fbox{$\divrelTm{\GG}{t}{u}{\boolDivrel}$} logical relation between \CCICP terms
  of type $\tcol{\bool}$ and elements of $\liftBool$.
  \begin{mathpar}
    \inferrule
    { \tcol{\fs{\GG{}}} \vdash \tcol{t} \inferty \tcol{\bool} \\
      \tcol{t} \whred \tcol{\btrue} \\
      \divrelCtx{\GG}
    }{ \divrelTm{\GG}{t}{\btrue}{\boolDivrel} }

    \and

    \inferrule
    { \tcol{\fs{\GG{}}} \vdash \tcol{t} \inferty \tcol{\bool} \\
      \tcol{t} \whred \tcol{\bfalse} \\
      \divrelCtx{\GG}
    }{ \divrelTm{\GG}{t}{\bfalse}{\boolDivrel} }

    \and

    \inferrule
    { \tcol{\fs{\GG{}}} \vdash \tcol{t} \inferty \tcol{\bool} \\
      \tcol{t} \whred \tcol{\?_\bool} \\
      \divrelCtx{\GG}
    }{ \divrelTm{\GG}{t}{\lifttop{\liftBool}}{\boolDivrel} }

    \and

    \inferrule
    { \tcol{\fs{\GG{}}} \vdash \tcol{t} \inferty \tcol{\bool} \\
      \tcol{t} \whred \tcol{\err_\bool} \vee \tcol{t} \not\whred \\
      \divrelCtx{\GG}
    }{ \divrelTm{\GG}{t}{\liftbot{\liftBool}}{\boolDivrel} }

    \and
    
    \inferrule
    { \tcol{\fs{\GG{}}} \vdash \tcol{t} \inferty \tcol{\bool} \\
      \tcol{t}\whred \tcol{t'}\\
      \divrelNe{\GG}{t'}{u}{\boolDivrel}
    }{ \divrelTm{\GG}{t}{u}{\boolDivrel} }
  \end{mathpar}
  \bigskip

  \fbox{$\divrelTm{\GG}{t}{u}{\natDivrel}$} logical relation between \CCICP terms
  of type $\tcol{\nat}$ and elements of $\liftNat$.
  \begin{mathpar}
    \inferrule
    { \tcol{\fs{\GG{}}} \vdash \tcol{t} \inferty \tcol{\nat} \\
      \tcol{t} \whred \tcol{0} \\
      \divrelCtx{\GG}
    }{ \divrelTm{\GG}{t}{0}{\natDivrel} }

    \and

    \inferrule
    { \tcol{\fs{\GG{}}} \vdash \tcol{t} \inferty \tcol{\nat} \\
      \tcol{t} \whred \tcol{\suc\,t'} \\
      \divrelTm{\GG}{t'}{u'}{\natDivrel}
    }{ \divrelTm{\GG}{t}{\suc\,u'}{\natDivrel} }

    \and

    \inferrule
    { \tcol{\fs{\GG{}}} \vdash \tcol{t} \inferty \tcol{\nat} \\
      \tcol{t} \whred \tcol{\?_\nat} \\
      \divrelCtx{\GG}
    }{ \divrelTm{\GG}{t}{\lifttop{\liftNat}}{\natDivrel} }

    \and

    \inferrule
    { \tcol{\fs{\GG{}}} \vdash \tcol{t} \inferty \tcol{\nat} \\
      \tcol{t} \whred \tcol{\err_\nat} \vee \tcol{t} \not\whred \\
      \divrelCtx{\GG}
    }{ \divrelTm{\GG}{t}{\liftbot{\liftNat}}{\natDivrel} }

    \and

    \inferrule
    { \tcol{\fs{\GG{}}} \vdash \tcol{t} \inferty \tcol{\nat} \\
      \tcol{t}\whred \tcol{t'}\\
      \divrelNe{\GG}{t'}{u}{\natDivrel}
    }{ \divrelTm{\GG}{t}{u}{\natDivrel} }
  \end{mathpar}
  \bigskip

  \fbox{$\divrelTm{\GG}{t}{u}{\unkDivrel\,i}$} logical relation between \CCICP terms
  of type $\tcol{\?_{\square_i}}$ and elements of $\unkMon[i]$.

  \begin{mathpar}
    \inferrule
    { \tcol{\fs{\GG}} \vdash \tcol{t} \inferty \tcol{\?_\square} \\
      \tcol{t} \whred \tcol{\?_{\?_{\square_i}}}\\
      \divrelCtx{\GG}}
    {\divrelTm{\GG}{t}{\lifttop{\unkMon[i]}}{\unkDivrel\,i}}
    \and
    \inferrule
    { \tcol{\fs{\GG{}}} \vdash \tcol{t} \inferty \tcol{\?_\square} \\
      \tcol{t} \whred \tcol{\err_{\?_{\square_i}}} \vee \tcol{t} \not\whred\\
      \divrelCtx{\GG}}
    {\divrelTm{\GG}{t}{\liftbot{\unkMon[i]}}{\unkDivrel\,i}}
    \and
    \inferrule
    { \tcol{\fs{\GG}} \vdash \tcol{t} \inferty \tcol{\?_\square} \quad
      \tcol{t} \whred \tcol{\ascdom{t'}{A}{\?}} \quad
      \divrelTm{\GG}{t'}{u'}{A_\varepsilon}}
    {\divrelTm{\GG}{t}{\upcast^{\unkU}_{A'}u'}{\unkDivrel\,i}}

    \inferrule
    { \tcol{\fs{\GG{}}} \vdash \tcol{t} \inferty \tcol{\?_\square} \quad
      \tcol{t}\whred \tcol{t'}\quad
      \divrelNe{\GG}{t'}{u}{\unkDivrel\,i}
    }{ \divrelTm{\GG}{t}{u}{\unkDivrel\,i} }
  \end{mathpar}
  \caption{Logical relation between \CCICP and $\omega$-cpos}
  \label{fig:logrel-ccicp-omega-cpo}
\end{figure}

\begin{figure}[h]
  \flushleft
  \fbox{$\divrelTm{\GG}{t}{u}{\errDivrel\,\tcol{A}}$} logical relation between \CCICP terms
  of type $\tcol{A}$ and elements of $\unit$.

  \begin{mathpar}
    \inferrule
    {\tcol{\fs{\GG{}}} \vdash \tcol{t} \inferty \tcol{A}}
    {\divrelTm{\GG}{t}{\unitK}{\errDivrel\,\tcol{A}}}
  \end{mathpar}
  \bigskip

  \fbox{$\divrelTm{\GG}{t}{u}{\univDivrel\,i}$} logical relation between \CCICP terms
  of type $\tcol{\square_i}$ and elements of $\U_i$.

  \begin{mathpar}
    \inferrule[$\natDivrel$]
    { \tcol{\fs{\GG}}\vdash \tcol{T} \inferty \tcol{\square_i} \\
      \tcol{T} \whred \tcol{\nat} \\
      \divrelCtx{\GG} }
    { \divrelTm{\GG}{T}{\natU}{\univDivrel\,i}}
    \and
    \inferrule[$\unkDivrel$]{
      \tcol{\fs{\GG}}\vdash \tcol{T} \inferty \tcol{\square_i} \\
      \tcol{T} \whred \tcol{\?_{\square_i}} \\
      \divrelCtx{\GG} }
    { \divrelTm{\GG}{T}{\unkU[i]}{\univDivrel\,i} }
    \and

    \inferrule[$\errDivrel$]{
      \tcol{\fs{\GG}}\vdash \tcol{T} \inferty \tcol{\square_i} \\
      \tcol{T} \whred \tcol{\err_{\square_i}} \vee \tcol{T} \not \whred \\
      \divrelCtx{\GG}}
    {\divrelTm{\GG}{T}{\errU}{\univDivrel\,i} }

    \and

    \inferrule[$\univDivrel$]{
      \tcol{\fs{\GG}}\vdash \tcol{T} \inferty \tcol{\square_i} \\
      \tcol{T} \whred \tcol{\square_j} \\
      \divrelCtx{\GG} }
    {\divrelTm{\GG}{T}{\uU_j}{\univDivrel\,i}}

    \and

    \inferrule[$\piDivrel$]
    {
      \tcol{\fs{\GG}}\vdash \tcol{T} \inferty \tcol{\square_i} \\
      \tcol{T} \whred \tcol{\P(a : A)\,B}\\
      A_\varepsilon : \forall \rho : \DD \subseteq \GG, \divrelTm{\DD}{A[\fs\rho]}{A'[\sn\rho]}{\univDivrel\,i} \\
      \forall \rho : \DD \subseteq \GG,
      \divrelTm{\DD}{a}{a'}{A_\varepsilon\,\rho} \implies \divrelTm{\GG}{B[a]}{B'[a']}{\univDivrel\,i}}
    {\divrelTm{\GG}{T}{\PiU\,A'\,(\lambda(a' : A').B')}{\univDivrel\,i}}
    
    \and

    \inferrule[$\neDivrel$]
    { \tcol{\fs{\GG{}}} \vdash \tcol{t} \inferty \tcol{\square_i} \\
      \tcol{t}\whred \tcol{t'}\\
      \divrelNe{\GG}{t'}{u}{\univDivrel\,i}
    }{ \divrelTm{\GG}{t}{u}{\univDivrel\,i} }

  \end{mathpar}
  \bigskip

  \fbox{$\divrelTm{\GG}{t}{u}{\piDivrel\,A_\varepsilon\,B_\varepsilon}$} logical relation between \CCICP terms
  of  $\Pi$ type.
  \begin{mathpar}
    \inferrule
    {\forall (\rho : \DD \subseteq \GG)~(a_\varepsilon : \divrelTm{\DD}{a}{a'}{A_\varepsilon\,\rho}),
     ~\divrelTm{\DD}{t[\fs{\rho}]\,a}{u[\sn{\rho}]\,a'}{B_\varepsilon\,\rho~a_\varepsilon}}
    {\divrelTm{\GG}{t}{u}{\piDivrel\,A_\varepsilon\,B_\varepsilon}} 
  \end{mathpar}
  
  where
  \begin{itemize}
  \item $\tcol{\fs{\GG}} \vdash \tcol{t} \inferty \tcol{\Pi\,(a : A)\,B}$,
  \item $\sn{\GG} \irty u : \Picont\,(a' : A')\,B'$,
  \item $A_\varepsilon : \forall \rho : \DD \subseteq \GG, \divrelTm{\DD}{A[\fs{\rho}]}{A'[\sn{\rho}]}{\univDivrel\,i}$
  \item $B_\varepsilon : \forall \rho : \DD \subseteq \GG, \divrelTm{\DD}{a}{a'}{A_\varepsilon\,\rho} \implies \divrelTm{\DD}{B[a]}{B'[a']}{\univDivrel\,i}$
  \end{itemize}
  \caption{Logical relation between \CCICP and $\omega$-cpos}
  \label{fig:logrel-ccicp-omega-cpo-cont}
\end{figure}

\begin{figure}[h]
  \flushleft

  \fbox{$\divrelTm{\GG}{t}{u}{\neDivrel\,A_\varepsilon}$} logical relation on \CCICP terms
  of a neutral type $\tcol{A}$ ($A_\varepsilon : \divrelTy{\GG}{A}{A'}$).

  \begin{mathpar}
    \inferrule
    {\divrelNe{\GG}{t}{t'}{\neDivrel\,A_\varepsilon}}
    {\divrelTm{\GG}{t}{\unitK}{\neDivrel\,A_\varepsilon}}
  \end{mathpar}
  \bigskip

  \fbox{$\divrelNe{\GG}{t}{u}{A_\varepsilon}$} logical relation on neutral
  \CCICP terms (excerpt).

  \begin{mathpar}
    \inferrule{
      \divrelCtx{\GG} \\
      \divbind{a}{a'}{A_\varepsilon} \in \GG }
    {\divrelNe{\GG}{a}{a'}{A_\varepsilon}}
    
    \and
    
    \inferrule{
      \divrelNe{\GG}{f}{f'}{\piDivrel\,A_\varepsilon\,B_\varepsilon}\\
      t_\varepsilon : \divrelTm{\GG}{t}{t'}{A_\varepsilon}}
    {\divrelNe{\GG}{f~t}{g~u}{B_\varepsilon~t_\varepsilon}}

    \and

    \inferrule{
      b_\varepsilon : \divrelNe{\GG}{b}{b'}{\boolDivrel}\\
      P_\varepsilon : \divrelTy{\GG,\divbind{z}{z'}{\boolDivrel}}{P}{P'} \\
      \divrelTm{\GG}{t_\btrue}{t'_\btrue}{P_\varepsilon[\tcol{\btrue} \sim \btrue]}\\
      \divrelTm{\GG}{t_\bfalse}{t'_\bfalse}{P_\varepsilon[\tcol{\bfalse} \sim \bfalse]}
    }{\divrelNe{\GG}{\match{\bool}{b}{z.P}{(t_{\btrue}, t_{\bfalse})}}{\match{\bool}{b'}{z'.P'}{(t'_{\btrue}, t'_{\bfalse})}}{P_\varepsilon[b_\varepsilon]}}

    \and
    
    \inferrule{
      A_\varepsilon : \divrelNe{\GG}{A}{A'}{\univDivrel}\\
    }{\divrelNe{\GG}{\?_A}{\?_{A'}}{A_\varepsilon}}

    \and

    \inferrule{
      A_\varepsilon : \divrelNe{\GG}{A}{A'}{\univDivrel}\\
    }{\divrelNe{\GG}{\err_A}{\err_{A'}}{A_\varepsilon}}

    \and

    \inferrule
    { 
      A_\varepsilon : \divrelNe{\GG}{A}{A'}{\univDivrel}\\
      B_\varepsilon : \divrelTm{\GG}{B}{B'}{\univDivrel} \\
      \divrelTm{\GG}{t}{t'}{A_\varepsilon}
    }
    {\divrelNe{\GG}{\ascdom{t}{A}{B}}{\downcast^{\unkU}_{B'}\upcast^{\unkU}_{A'}t'}{B_\varepsilon}}
  \end{mathpar}

  \caption{Logical relation between \CCICP and $\omega$-cpos}
  \label{fig:logrel-ccicp-omega-cpo-cont-cont}
\end{figure}

We first state a lemma making explicit how divergence is accounted for by the
logical relation.

\begin{lemma}[Diverging terms relate to errors]
  \label{lem:div-to-error}\text{}
  \begin{enumerate}
  \item If $\tcol{\fs\GG} \vdash \tcol{t} : \tcol{T}$, $T_\varepsilon :
    \divrelTy{\GG}{T}{T'}$ and $\tcol{t}
    \not\whred$ then $\divrelTm{\GG}{t}{\err_{T'}}{T_\varepsilon}$.
  \item Conversely, if $\tcol{\Gamma} \vdash \tcol{t} : \tcol{T}$, $T_\varepsilon :
    \divrelTy{\GG}{T}{T'}$, $\divrelTm{\GG}{t}{t'}{T_\varepsilon}$ and $\tcol{t}
    \not\whred$ then $t' = \err_{T'}$.
  \end{enumerate}
\end{lemma}
\begin{proof}
  In the two parts of the lemma, we proceed by induction on $T_\varepsilon$.
  For the first part, the cases $T_\varepsilon = \univDivrel, \natDivrel, \boolDivrel$ and
  $\unkDivrel$ are immediate because in each case a rule apply for diverging terms.
  If $T_\varepsilon = \errDivrel$, then
  $\divrelTm{\GG}{t}{\unitK}{T_\varepsilon}$ which is enough because $\unitK =
  \err_{\unit} = \err_{\El\,\errU} = \err_{\El\,T'}$.
  Finally, if $T_\varepsilon = \piDivrel\,A_\varepsilon\,B_\varepsilon$, then
  for any $\rho : \DD \subset \GG$ and $a_\varepsilon :
  \divrelTm{\DD}{a}{a'}{A_\varepsilon\,\rho}$
  we have that $\tcol{\fs\GG} \vdash \tcol{t[\fs{\rho}]~a} : \tcol{B[\fs{\rho},a]}$,
  $B_\varepsilon\,\rho\,a_\varepsilon :
  \divrelTy{\DD}{B[\fs{\rho},a]}{B'[\sn{\rho}, a']}$ and
  $\tcol{t~a} \not\whred$, so by induction hypothesis
  $\divrelTm{\DD}{t[\fs{\rho}]~a}{\err_{B'[\sn{\rho},a']}}{B_\varepsilon\,\rho\,a_\varepsilon}$,
  hence $\divrelTm{\GG}{t}{\err_{\PiU\,A'\,B'}}{T_\varepsilon}$.
  We now turn to the second part of the lemma.
  When $T_\varepsilon = \univDivrel, \natDivrel, \boolDivrel, \unkDivrel$ and
  $\errDivrel$, there is exactly one rule that apply to relate to a term without
  weak head normal form $\tcol{t}$ so that necessarily $t' = \err_{T'}$.
  When $T_\varepsilon = \piDivrel\,A_\varepsilon\,B_\varepsilon$,  any $\rho
  : \DD \subset \GG$ and $a_\varepsilon:
  \divrelTm{\DD}{a}{a'}{A_\varepsilon\,\rho}$ we have that 
  $B_\varepsilon\,\rho\,a_\varepsilon : \divrelTy{\DD}{B[\fs{\rho},a]}{B'[\sn{\rho}, a']}$,
  $\divrelTm{\DD}{t[\fs{\rho}]~a}{t'[\sn{\rho}]~a'}{B_\varepsilon\,\rho\,a_\varepsilon}$ and 
  $\tcol{t~a} \not\whred$, so by induction hypothesis $t'\,a' =
  \err_{B'[\sn{\rho}, a']}$.
  Taking $\rho$ to be the weakening $\GG, \divbind{a}{a'}{A_\varepsilon} \subset
  \GG$, we have by function extensionality that $t' = \lambda(a':A'). \err_{B'}
  = \err_{\PiU\,A'\,B'}$.
\end{proof}

\begin{lemma}[Fundamental lemma]\text{}
  \begin{itemize}
  \item If $\tcol{\Gamma} \vdash$ then there exists $\GG$ such that
    $\divrelCtx{\GG}$, $\tcol{\fs\GG} = \tcol{\Gamma}$ and $\sn{\GG} =
    \monTy{\tcol{\Gamma}}$ ;
  \item If $\tcol{\Gamma} \vdash \tcol{T} \inferty \tcol{\square_i}$ there
    exists a derivation $T_\varepsilon : \divrelTy{\GG}{T}{\monTy{\tcol{T}}}$ 
  \item  If $\tcol{\Gamma} \vdash \tcol{t} \inferty \tcol{T}$ then there exists
    a derivation $ t_\varepsilon : \divrelTm{\GG}{t}{\monTm{\tcol{t}}}{T_\varepsilon}$
  \item If $\tcol{\Gamma} \vdash \tcol{T} \inferty \tcol{\square_i}$,
    $\tcol{\Gamma} \vdash \tcol{T'} \inferty \tcol{\square_i}$ and  $\tcol{T} \conv \tcol{T'}$ then $\monTy{\tcol{T}} =
    \monTy{\tcol{T}'}$ and $T_\varepsilon = T'_\varepsilon : \divrelTy{\GG}{T}{\monTy{\tcol{T}}}$.
  \end{itemize}
\end{lemma}

\begin{proof}
  Since the translation $\monTy{-}$ underly a model of \CCICT, it sends
  convertible types $\tcol{T}, \tcol{T'}$ in the source to provably equal types
  in the target  $\monTy{\tcol{T}} = \monTy{\tcol{T}'}$, proving the last claim.

  The three other claims are proved by mutual induction on the input derivation,
  assuming an undirected variant of the rules
  in~\cref{fig:ccic-typing,fig:bidir}, which is possible
  by~\cite{LennonBertrand2021}.
  Concretely, this modification means that we assume additional well-formedness
  premises in the derivations, \eg for contexts and types, and do not show that
  input well-formedness is preserved.
  Moreover the induction hypothesis needs to be strenghened to quantify over an
  arbitrary context $\DD$ with a substition $\sigma : \DD \to \GG$ whose
  components are related according to the logical relation.

  For contexts, if the derivation ends with a rule \nameref{infrule:cic-axiom}, it is enough take $\GG = \cdot$.
  If it ends with \nameref{infrule:cic-concat}, then by induction hypothesis
  there exists $\GG$ and $A_\varepsilon$ such that $\tcol{\fs{\GG}} =
  \tcol{\Gamma}$, $\sn{\GG} = \monTy{\tcol{\Gamma}}$, $\divrelCtx{\GG}$ and
  $A_\varepsilon : \divrelTy{\GG}{A}{\monTy{\tcol{A}}}$, so taking $\GG,
  \divbind{a}{a'}{A_\varepsilon}$ suffices.
  For \nameref{infrule:cic-univ} by induction hypothesis $\divrelCtx{\GG}$ with
  $\tcol{\fs{\GG}} = \tcol\Gamma$. Moreover, $\tcol\Gamma \vdash
  \tcol{\square_i} \inferty \tcol{\square_{i+1}}$ and $\monTm{\tcol{\square_i}}
  = \uU_i$ so
  $\divrelTm{\GG}{\square_i}{\monTm{\tcol{\square_i}}}{\univDivrel\,(i+1)}$ and
  $\divrelTy{\GG}{\square_i}{\monTm{\tcol{\square_i}}}$.
  The rules \nameref{infrule:cic-ind} (for $\nat, \bool$) and
  \nameref{infrule:cic-cons} (for $\natzero, \suc, \btrue, \bfalse$),
   introducing types and terms that are already
  in weak head normal form follow the same pattern as \nameref{infrule:cic-univ}.
  In the case of the rules \nameref{infrule:ccic-prod} and
  \nameref{infrule:cic-abs}, the context needs to be extended and we need to
  take advantage of the full induction hypothesis strenghened under arbitrary
  reducible substitutions.
  Dually, the rule \nameref{infrule:cic-app} is immediate by induction
  hypothesis and the definition the logical relation at function types.
  A bit more work is needed for the rule \nameref{infrule:cic-fix} for
  $\ind_\bool$ and $\ind_\nat$, doing a case analysis on the proof of
  relatedness of their main argument.
  If the main argument diverges, then the applied eliminator diverges too so it
  is related to $\err$ which is its translation because eliminators send errors
  at an inductive type to errors at the adequate type in $\omega$-cpos.
  Otherwise the main argument weak head reduces to a normal form and we can
  conclude by induction hypothesis and closure by anti-reduction.
  For the variable case, rule \nameref{infrule:cic-var}, we can show by
  induction on the proof of relatedness of its type that it is related to its
  $\eta$-expansion at $\Pi$ types an to itself at any other type using the rules
  for neutrals.
  We conclude by extensionality of the $\omega$-cpo model.
  Conversion rules \nameref{infrule:cic-check},
  \nameref{infrule:cic-prod-inf}, \nameref{infrule:cic-ind-inf} and
  \nameref{infrule:cic-univ-inf} satisfy the fundamental lemma because
  convertible types induce the same relation on their term.
  For \nameref{infrule:ccic-err}, we have by induction hypothesis that
  $T_\varepsilon : \divrelTm{\GG}{T}{\monTm{\tcol{T}}}{\univDivrel}$.
  By case analysis, $T_\varepsilon$ is necessarily one of $\errDivrel,\piDivrel, \unkDivrel, \boolDivrel,
  \natDivrel$ or $ \univDivrel$.
  If $T_\varepsilon = \errDivrel\,\tcol{T}$ then
  $\divrelTm{\GG}{\err_T}{\unitK}{\errDivrel\,\tcol{T}}$ since $\tcol{\fs\GG} \vdash
  \tcol{\err_T} \inferty{T}$, and we can conclude using extensionality of $\unit
  = \monTy{\tcol{T}}$, that is $\monTm{\tcol{\err_T}} = \unitK$.
  If $T_\varepsilon = \piDivrel\,A_\varepsilon\,B_\varepsilon$, then for any
  $\rho : \DD \subset \GG$ and $a_\varepsilon :
  \divrelTm{\DD}{a}{a'}{A_\varepsilon\,\rho}$ we have that
  $\tcol{T[\fs\rho]} \whred \tcol{\Pi(a:A)B}$ and $\tcol{\err_{T}~a} \whred
  \tcol{\err_{B[a]}}$, so we conclude this case by induction hypothesis
  $\divrelTm{\DD}{\err_{B[a]}}{\monTm{\tcol{\err_{B[a]}}}}{B_\varepsilon\,a_\varepsilon}$,
  closure by anti-reduction and the fact that $\monTm{\tcol{\err_{B[a]}}} =
  \err_{\monTm{B[a]}} = \err_{\monTm{B}[a']}$.
  
  In all the other cases $\tcol{T}$ weak head reduces to a type in weak head
  normal form $\tcol{\?_\square}, \tcol{\square_i}, \tcol{\nat}$ or
  $\tcol{\bool}$, and a corresponding rule is present in the logical relation to
  conclude directly.
  A similar proof apply for the rule \nameref{infrule:ccic-unk}.
  Finally, for the rule \nameref{infrule:ccic-cast} with conclusion
  $\tcol{\Gamma} \vdash \tcol{\ascdom{t}{A}{B}} \inferty \tcol{B}$, we have by
  induction hypothesis we have that $\divrelCtx{\GG}, {A_\varepsilon :
    \divrelTy{\GG}{A}{\monTy{\tcol{A}}}}, B_\varepsilon :
  \divrelTy{\GG}{B}{\monTy{\tcol{B}}}$ and
  $\divrelTm{\GG}{t}{\monTm{\tcol{t}}}{A_\varepsilon}$.
  By analysing all possible weak head reduction paths from
  $\tcol{\ascdom{t}{A}{B}}$, either:
  \begin{description}
  \item[(a)] $\tcol{\ascdom{t}{A}{B}} \whred \tcol{u}$
    such that $\divrelTm{\GG}{u}{\monTm{\tcol{u}}}{B_\varepsilon}$ using
    inversions on $A_\varepsilon, B_\varepsilon$ and $t_\varepsilon$, or
  \item[(b)] one of
    $\tcol{A}$, $\tcol{B}$ or $\tcol{t}$ never reduces to a weak head normal form.
  \end{description}
  In case \textbf{(a)}, we conclude that
  $\divrelTm{\GG}{\ascdom{t}{A}{B}}{\monTm{\tcol{\ascdom{t}{A}{B}}}}{B_\varepsilon}$
  by closure under anti-reduction and using the fact that
  $\monTm{\tcol{\ascdom{t}{A}{B}}} = \monTm{\tcol{u}}$ (because $\monTm{-}$ maps
  convertible terms to equal terms in the model).
  In case \textbf{(b)}, we have that
  $\divrelTm{\GG}{\ascdom{t}{A}{B}}{\err_{\monTm{\tcol{B}}}}{B_\varepsilon}$ by
  the first part of~\cref{lem:div-to-error} and the second part of that lemma
  ensures that one of $\monTm{\tcol{A}}$, $\monTm{\tcol{B}}$ or
  $\monTm{\tcol{t}}$ is an error at the adequate type so that
  $\monTm{\tcol{\ascdom{t}{A}{B}}} =
  \downcast^{\unkU}_{\monTm{\tcol{B}}}\upcast^{\unkU}_{\monTm{\tcol{A}}}\monTm{\tcol{t}}
  = \err_{\monTm{\tcol{B}}}$.
\end{proof}

\begin{corollary}
 If $\tcol{\Gamma} \vdash \tcol{t} \inferty \tcol{T}$ and $\tcol{t} \not\whred$
 then $\monTm{\tcol{t}} = \err_{\monTm{\tcol{T}}}$.
\end{corollary}
\begin{proof}
  By the fundamental lemma, $\divrelTm{\GG}{t}{\monTm{\tcol{t}}}{T_\varepsilon}$
  with $T_\varepsilon : \divrelTm{\GG}{T}{\monTm{\tcol{T}}}{\univDivrel}$
  and by the second part of \cref{lem:div-to-error}, $\monTm{\tcol{t}}
  = \err_{\monTm{\tcol{T}}}$.
\end{proof}

\section{A Direct Presentation of Vectors}
\label{sec:case-vectors-appendix}

\begin{figure}
  \small
\begin{mathpar}
  \inferrule{ }{\can{\Gnil}{A}} \and
  \inferrule{ }{\can{(\Gcons{A}{a}{n}{v})}} \and
  \inferrule{  }{\can{\GnilU{A}}}
  \and
  \inferrule{ }{\can{(\GconsU{A}{a}{n}{v})}}
  \\
  \inferrule{\neu{v} \quad \vee \quad \neu{n} \quad \vee \quad \neu{m} }{\neu{(\ascdom{v}{\Gvect{A}{n}}{\Gvect{B}{m}}})}
  \and
  \inferrule{\neu{v}}{\neu{(\GvectRec}{P}{P_{nil}}{P_{cons}}{v})}
  \\
  % \\
  \inferrule{\Gamma \vdash A \checkty \square{}_{i}}{\Gamma \vdash \GnilULevel{A}{i} \inferty
    \Gvect{A}{\?_\nat}}[nilu]
  \and
  \inferrule{\Gamma \vdash A \checkty \square{}_{i} \\ \Gamma \vdash
      a \checkty A \\ \Gamma \vdash
      n \checkty \nat \\ \Gamma \vdash v \checkty \Gvect{A}{n} 
    }{\Gamma \vdash \GconsULevel{A}{a}{n}{v}{i} \inferty
    \Gvect{A}{\?_\nat}}[consu]
  \\
  \\
  \redrule{
    \GvectRec{P}{P_{nil}}{P_{cons}}{\Gnil{A}}
  }
  {
    P_{nil}\myflushright
  }[v-rect-nil]\label{red:v-rect-nil}
  \and
  \redrule{
    \GvectRec{P}{P_{nil}}{P_{cons}}{(\Gcons{A}{a}{n}{v})}
  }
  {
    P_{cons}~a~n~(\GvectRec{P}{P_{nil}}{P_{cons}}{v}) \hfill
  }[v-rect-cons]\label{red:v-rect-cons}
  \and
  \redrule{
    \GvectRec{P}{P_{nil}}{P_{cons}}{\err_{\Gvect{A}{n}}}
  }
  {
    \err_{P~{\err_{\Gvect{A}{n}}}}\hfill
  }[v-rect-err]\label{red:v-rect-err}
  \and
  \redrule{
    \GvectRec{P}{P_{nil}}{P_{cons}}{\?_{\Gvect{A}{n}}}
  }
  {
    \?_{P~{\?_{\Gvect{A}{n}}}}\hfill
  }[v-rect-unk]\label{red:v-rect-unk}
  \\
  \redrule{
    \GvectRec{P}{P_{nil}}{P_{cons}}{\GnilU{A}}
  }
  {\ascdom{(\GvectRec{P}{P_{nil}}{P_{cons}}{\Gnil{A}})}{P~0}{P~\?_\nat} \hfill
  }[v-rect-nilu]\label{red:v-rect-unk-0}
  \\
  \redrule{
    \GvectRec{P}{P_{nil}}{P_{cons}}{(\GconsU{A}{a}{n}{v})}
  }
  {\myflushright}[v-rect-consu]\label{red:v-rect-unk-S}
  \\
  \myflushright \ascdom{(\GvectRec{P}{P_{nil}}{P_{cons}}{(\Gcons{A}{a}{n}{v})})}{P~(S~n)}{P~\?_\nat}
  \\
  \redrule{
    \ascdom{\Gnil{A}}{\Gvect{A}{0}}{\Gvect{B}{0}} 
  }
  {
    \Gnil{B} \hfill
  }[V-nil]\label{red:v-0}
  \and
  \redrule{
    \ascdom{\Gnil{A}}{\Gvect{A}{0}}{\Gvect{B}{(S~n)}} 
  }
  {
    \err_{\Gvect{B}{(S~n)}}\hfill
  }[V-nil-cons]\label{red:v-S0}
  \and
  \redrule{
    \ascdom{\Gnil{A}}{\Gvect{A}{0}}{\Gvect{B}{\?_\nat}} 
  }
  {
    \GnilU{B}
    \hfill
  }[V-nil-\?]\label{red:v-U0}
  \\
  \redrule{
    \ascdom{(\Gcons{A}{a}{k}{v})}{\Gvect{A}{(S~n)}}{\Gvect{B}{(S~m)}} 
  }
  {\Gcons{B}{(\ascdom{a}{A}{B})}{m}{(\ascdom{v}{\Gvect{A}{k}}{\Gvect{B}{m}})}
    \hfill}[V-cons]\label{red:v-S}
  \\
  \redrule{
    \ascdom{(\Gcons{A}{a}{k}{v})}{\Gvect{A}{(S~n)}}{\Gvect{B}{0}} 
  }
  {
    \err_{\Gvect{B}{0}} \myflushright
  }[V-cons-nil]\label{red:v-S0}
  \\
  \redrule{
    \ascdom{(\Gcons{A}{a}{k}{v})}{\Gvect{A}{(S~n)}}{\Gvect{B}{\?_\nat}} 
  }
  {
    \GconsU{B}{(\ascdom{a}{A}{B})}{n}{(\ascdom{v}{\Gvect{A}{k}}{\Gvect{B}{n}})}
    \hfill
  }[V-cons-\?]\label{red:v-S?}
  \\
  \redrule{
    \ascdom{\GnilU{A}}{\Gvect{A}{\?_{\nat}}}{\Gvect{B}{\?_\nat}}
  }
  { \GnilU{B} \hfill
  }[V-nilu]\label{red:v-unk-u}
 \\
  \redrule{
    \ascdom{\GnilU{A}}{\Gvect{A}{\?_{\nat}}}{\Gvect{B}{0}}
  }
  { \Gnil{B} \hfill
  }[V-nilu-nil]\label{red:v-unk-u}
  \\
  \redrule{
    \ascdom{\GnilU{A}}{\Gvect{A}{\?_{\nat}}}{\Gvect{B}{(S~n)}}
  }
  { \err_{\Gvect{B}{(S~n)}} \hfill
  }[V-nilu-cons]\label{red:v-unk-u}
  \\
  \redrule{
    \ascdom{(\GconsU{A}{a}{k}{v})}{\Gvect{A}{\?_{\nat}}}{\Gvect{B}{\?_{\nat}}}
  }
  {
    \GconsU{B}{(\ascdom{a}{A}{B})}{k}{(\ascdom{v}{\Gvect{A}{k}}{\Gvect{B}{k}})} \hfill
  }[V-consu]\label{red:v-unk-S}
   \\
  \redrule{
    \ascdom{(\GconsU{A}{a}{k}{v})}{\Gvect{A}{\?_{\nat}}}{\Gvect{B}{0}}
  }
  {
    \err_{\Gvect{B}{0}}  \hfill
  }[V-consu-nil]\label{red:v-unk-S-nil}
  \\
  \redrule{
    \ascdom{(\GconsU{A}{a}{k}{v})}{\Gvect{A}{\?_{\nat}}}{\Gvect{B}{(S~n)}}
  }
  { \myflushright
  }[V-consu-cons]\label{red:v-unk-S-cons}
  \\
  \myflushright
  \Gcons{B}{(\ascdom{a}{A}{B})}{n}{(\ascdom{v}{\Gvect{A}{k}}{\Gvect{B}{n}})} 
  \\
  \redrule{
    \ascdom{\?_{\Gvect{A}{n}}}{\Gvect{A}{n}}{\Gvect{B}{m}} 
  }
  {
    \?_{\Gvect{B}{m}}\myflushright m,n \in \{ 0, S~m,
      \?_\nat , \err_\nat \}
  }[V-unk]\label{red:v-unk}
  \and
  \redrule{
    \ascdom{\err_{\Gvect{A}{n}}}{\Gvect{A}{n}}{\Gvect{B}{m}} 
  }
  {
    \err_{\Gvect{B}{m}}\myflushright m,n \in \{ 0, S~m,
      \?_\nat , \err_\nat \}
  }[V-err]\label{red:v-err}
    \and
  \redrule{
    \ascdom{v}{\Gvect{A}{n}}{\Gvect{B}{\err_\nat}} 
  }
  {
    \err_{\Gvect{B}{\err_\nat}}\myflushright n \in \{ 0, S~m,
      \?_\nat , \err_\nat \} \mbox{ and } v \neq \?_{\Gvect{A}{n}}
  }[V-to-err]\label{red:v-to-err}

\end{mathpar}
\label{fig:vectors}
\caption{Canonical forms and reduction rule for vectors.}
\end{figure}

Vectors have two new normal forms, corresponding to cast of
\coqe{nil} and \coqe{cons} to $\Gvect{A}{\?_{\nat}}$. The difference
with the treatment of the universe is that the corresponding term,
for instance $\ascdom{t}{\Gvect{A}{n}}{\Gvect{A}{\?_{\nat}}}$ for the
case of \coqe{nil}, can not be considered as canonical form because
they involve a non-linear occurrence of $A$. 
To remedy to this issue, we add two new canonical forms ($\GnilU{A}$ and
$\GconsU{A}{a}{n}{v}$) to vectors
with introduction typing rules defined in \cref{fig:vectors}.

Regarding cast on vectors, it does not only compute in the argument of
the cast as it is the case for inductive types without indices, but it
also computes on the indices.
That is, a cast on vectors is neutral when either one of the indices
is neutral or the argument is neutral (see~\cref{fig:vectors}).
Other kind of neutral can be derived from the one of inductive types
without indices and are omitted here. 

Similarly, we do not detail the other typing rules for vectors as they
are similar to the one for inductive types wihtout indices, and focus
on explaining the new reduction rules, presented also in
\cref{fig:vectors}.

The two first reduction rules \nameref{red:v-rect-nil} and
\nameref{red:v-rect-cons} are standard reduction rules in \CIC for the
recursor \texttt{vect\_rect} on vectors.
The rules \nameref{red:v-rect-err}  and \nameref{red:v-rect-unk} are
the standard rules dealing with exceptions. 
Additionally, there are two computation rules for the eliminator on
the two new constructors \nameref{red:v-rect-unk-0} and
\nameref{red:v-rect-unk-S} which basically consist in the underlying
non-exceptional constructor to the eliminator and cast the result back
to $P~\?_\nat$.
This rule somehow transfers the cast on vectors to a cast on
the returned type of the predicate. 

Finally, there are rules to conduct casts between vectors
in canonical forms.
The last three rules (\nameref{red:v-unk}, \nameref{red:v-err} and
\nameref{red:v-to-err}) are simply propogation of errors. 
Then, there remains 12 rules, 3 by constructors of vectors. We just
explain the one on \coqe{cons}.
Rule~\nameref{red:v-S} applies when both indices of the form $S$
of somthing and progates the cast of the arguments, as does the
standard rule for casting a constructor.
Rule~\nameref{red:v-S0} detects that the indices do not match and
raise an error.
Finally, Rule~\nameref{red:v-S?} propagates the cast on the arguments,
but this time applied to $\GconsUName$, thus converting precise
information to a less precise information. 

\ifleftovers

\section{Leftovers}

\subsection{Optimized Cast}

The last definition corresponds to a usual optimization of elaboration to a cast
calculus: why insert a run-time check when you know at typing time already that
the check will always succeed? Thus, a $\cas$ operator should only be inserted
by elaboration when the types are not the same (\ie convertible). Although with
closed types this is merely for performance purposes, it becomes crucial for us,
as types can contain variables that block reduction of the $\cas$ operator:
$\cast{A}{A}{t}$ does not reduce, because $\cas$ waits for the head constructor
of $A$ to appear. Even though one could prove that $\cast{A}{A}{t} = t$ in the
models we consider, the fact that this equality is not definitional is enough to
threaten theoretical properties, such as conservativity, because it blocks reduction.

\begin{definition}[Optimized cast]
	Define the optimized cast $\ocast{A}{B}{t}$ to be $t$ if $A \conv B$, and $\cast{A}{B}{t}$ otherwise.
\end{definition}

Note that this cast can only be defined in a setting where conversion is decidable. If both $A$ and $B$ are well-typed, this depends on the normalization properties of the system: if the system normalizes, conversion is decidable (using normal forms), but if it does not then it is not. Thus, when considering the non-normalizing variants of \CCIC, we restrict this optimized cast to the \CIC fragment, that still normalizes. In other words, a cast is always inserted between gradual types, even if it is not needed.

It is really an optimization, in the following sense.
\begin{proposition}[Correction of the optimized cast]
	If $\Gamma \vdash A : \square{}$ and $\Gamma \vdash B : \square{}$ one has $\Gamma \vdash \ocast{A}{B}{t} : B$ if and only if $\Gamma \vdash \cast{A}{B}{t} : B$.
\end{proposition}

\begin{lemma}[Syntactic precision respects optimized cast]
	If $\Gamma.1 \vdash t : A$, $\Gamma.1 \vdash B : \square$, $\Gamma.2 \vdash t' : A'$, $\Gamma.2 \vdash B' : \square$ and we have the precision relations $\Gamma \vdash A \capre A'$, $\Gamma \vdash B \capre B'$ and $\Gamma \vdash t \capre t'$, then $\ocast{A'}{B'}{t'} \rtred s'$ for some $s'$ such that $\Gamma \vdash s \capre s'$.
\end{lemma}

\begin{proof}
	If both casts optimize or do not optimize, no reduction is needed and the pointwise precision suffice. If the cast on the left is not optimized, and the one on the right is, we use the corresponding non-diagonal rule for the cast, and also no reduction is needed. The hard case is when the cast on the left is optimized, but not the one on the right. Then we can take $T$ to be the normal form of both $A$ and $B$, and we have $\Gamma \vdash t : T$. We can simulate this reduction, and thus get $A' \rtred A''$ and $B' \rtred B''$ for some $A''$ and $B''$ such that $\Gamma \vdash T \capre A''$ and $\Gamma \vdash T \capre B''$. Thus, $\cast{A'}{B'}{t'} \rtred \cast{A''}{B''}{t'}$, and now we can conclude with the non-diagonal rule for the cast.
\end{proof}

\fi

\fi

\end{document}